\documentclass[10pt]{article}

\usepackage[a4paper,height=188mm,width=142mm]{geometry}

\usepackage[utf8]{inputenc}
\usepackage[T1]{fontenc}
\usepackage{amsthm}
\usepackage{amssymb}
\usepackage{latexsym}
\usepackage{amsmath}
\usepackage{amsfonts}
\usepackage{mathrsfs}
\usepackage{amstext}
\usepackage{enumitem}
\setlist{leftmargin=9mm}
\setitemize{noitemsep,topsep=1pt,parsep=1pt,partopsep=1pt}
\usepackage{mathtools}
\DeclarePairedDelimiter{\ceil}{\lceil}{\rceil}

\mathtoolsset{showonlyrefs}
\usepackage{stmaryrd}
\usepackage{dsfont}
\usepackage{tocloft}
\usepackage{cite}
\usepackage{slashed}
\usepackage{accents}
\newcommand{\ubar}[1]{\underaccent{\bar}{#1}}
\usepackage{graphicx}

\makeatletter
\newcommand*\botimes{{\mathpalette\botimes@{1.5}}}
\newcommand*\botimes@[2]{\mathbin{\vcenter{\hbox{\scalebox{#2}{\hspace{0.0mm}$\m@th#1\otimes$\hspace{0.0mm}}}}}}
\newcommand*\Cdot{{\mathpalette\Cdot@{.7}}}
\newcommand*\Cdot@[2]{\mathbin{\vcenter{\hbox{\scalebox{#2}{\hspace{0.1mm}$\m@th#1\bullet$\hspace{0.1mm}}}}}}
\makeatother
\newcommand{\VERT}[1]{{\left\vert\kern-0.25ex\left\vert\kern-0.25ex\left\vert #1 
    \right\vert\kern-0.25ex\right\vert\kern-0.25ex\right\vert}}

\newcommand{\bA}{\mathbb{A}}
\newcommand{\bF}{\mathbb{F}} 
 
\newcommand{\bN}{\mathbb{N}} 
\newcommand{\bE}{\mathbb{E}}
\newcommand{\bC}{\mathbb{C}}
\newcommand{\bZ}{\mathbb{Z}}
\newcommand{\bR}{\mathbb{R}}

\newcommand{\bT}{\mathbb{T}}

\newcommand{\bK}{\mathbb{K}}

\newcommand{\cS}{\mathcal{S}}
\newcommand{\cP}{\mathcal{P}}
\newcommand{\cX}{\mathcal{X}}

\newcommand{\fA}{\mathbf{A}}
\newcommand{\fB}{\mathbf{B}}
\newcommand{\fC}{\mathbf{C}}
\newcommand{\fD}{\mathbf{D}}
\newcommand{\fE}{\mathbf{E}}
\newcommand{\fV}{\mathbf{V}}
\newcommand{\fI}{\mathbf{I}}
\newcommand{\fJ}{\mathbf{J}}
\newcommand{\fF}{\mathbf{F}}
\newcommand{\fS}{\mathbf{S}}
\newcommand{\fL}{\mathbf{L}}
\newcommand{\fK}{\mathbf{K}}
\newcommand{\fH}{\mathbf{H}}

\newcommand{\fP}{\mathbf{P}}
\newcommand{\fG}{\mathbf{G}}
\newcommand{\fR}{\mathbf{R}}
\newcommand{\fT}{\mathbf{T}}

\newcommand{\fX}{\mathbf{X}}

\newcommand{\fZ}{\mathbf{Z}}

\newcommand{\fW}{\mathbf{W}}
\newcommand{\ft}{\mathbf{t}}
\newcommand{\fp}{\mathbf{p}}
\newcommand{\fs}{\mathbf{s}}
\newcommand{\sA}{\mathscr{A}}
\newcommand{\sS}{\mathscr{S}}
\newcommand{\sB}{\mathscr{B}}
\newcommand{\sC}{\mathscr{C}}

\newcommand{\sM}{\mathscr{M}}
\newcommand{\sN}{\mathscr{N}}
\newcommand{\sNN}{\mathscr{N}_{\sharp}}
\newcommand{\sW}{\mathscr{W}}
\newcommand{\sV}{\mathscr{V}}
\newcommand{\sH}{\mathscr{H}}
\newcommand{\sF}{\mathscr{F}}
\newcommand{\sG}{\mathscr{G}}
\newcommand{\sX}{\mathscr{X}}
\newcommand{\sY}{\mathscr{Y}}
\newcommand{\sZ}{\mathscr{Z}}
\newcommand{\sK}{\mathscr{K}}

\newcommand{\sP}{\mathscr{P}}
\newcommand{\sT}{\mathscr{T}}
\newcommand{\sU}{\mathscr{U}}

\newcommand{\bFF}{\mathbb{G}}
\newcommand{\ri}{\mathrm{i}}
\newcommand{\rd}{\mathrm{d}}
\newcommand{\rD}{\mathrm{D}}

\newcommand{\re}{\mathrm{e}}
\newcommand{\rc}{\mathrm{c}}

\newcommand{\rK}{\mathrm{K}}

\newcommand{\supp}{\mathrm{supp}}
\newcommand{\esssup}{\operatorname*{ess\,sup}}

\usepackage{color,hyperref}
\definecolor{darkblue}{rgb}{0.0,0.0,0.5}
\hypersetup{
  linkcolor  = darkblue,
  citecolor  = darkblue,
  urlcolor   = darkblue,
  colorlinks = true,
}

\makeatletter
\AtBeginDocument{
  \hypersetup{
    pdftitle = {\@title},
    pdfauthor = {\@author}
  }
}
\makeatother

\newtheorem{thm}{Theorem}
\newtheorem{lem}[thm]{Lemma}

\newtheorem{cor}[thm]{Corollary}

\theoremstyle{remark}
\newtheorem{rem}[thm]{Remark}

\theoremstyle{remark}

\theoremstyle{definition}
\newtheorem{dfn}[thm]{Definition}

\numberwithin{equation}{section}
\numberwithin{thm}{section}
\numberwithin{example}{section}

\begin{document}
\title{Construction of Gross--Neveu model using \\
Polchinski flow equation}
\author{Pawe{\l} Duch
\\
Faculty of Mathematics and Computer Science \\  
Adam Mickiewicz University in Pozna\'n\\
ul. Uniwersytetu Pozna\'nskiego 4, 61-614 Pozna\'n, Poland\\
pawel.duch@amu.edu.pl
}
\date{\today}

\maketitle

\begin{abstract}
 We develop a new method for solving the renormalization group flow equation, also known as the Polchinski equation, by formulating it as a well-posed infinite-dimensional partial differential equation on a Banach space. This yields the first direct construction of a solution of this equation via a fixed point argument, circumventing the need for perturbative or expansion-based techniques employed in earlier work. Our approach leverages recent advances in noncommutative probability theory, particularly a scale decomposition of the free fermionic field in a filtered noncommutative probability space equipped with a suitable conditional expectation.

 As a key application, we construct the Gross--Neveu model, a fermionic, asymptotically free, and scaling-critical quantum field theory. We rigorously establish convergence of the Schwinger functions as the cutoff is removed, and construct a random field on the noncommutative probability space of the free theory whose moments coincide with the Schwinger functions. This provides the first explicit strong coupling between the free and interacting fields for a scaling-critical quantum field theory.

 The method is expected to apply broadly to other critical or subcritical quantum field theories formulated in terms of bounded fields, including fermionic models such as the symplectic fermion model, as well as certain bosonic theories like the sine-Gordon model.
\end{abstract}

\tableofcontents

	\section{Introduction}

The renormalization group (RG) is a powerful framework for analyzing the behavior of physical systems across different length or energy scales. It is particularly effective in systems with many degrees of freedom, such as those arising in statistical mechanics and quantum field theory, and has proven instrumental in deriving bounds and understanding critical behavior in both bosonic and fermionic models.

Let us consider a measure $\mu$ formally given by
\begin{equation}
	\mu(F) \propto \bE\Big(\exp(-V(\varPsi))\,F(\varPsi)\Big),
\end{equation}
where $\varPsi$ is a typically rough Gaussian random field and $V$ is a potential functional. To make $\mu$ mathematically well-defined, one typically introduces a cutoff and assumes that the bare potential $V$ is a sum of local terms whose coefficients depend on the cutoff parameter. A~central challenge is to show that for a suitable choice of this dependence, the cutoff can be removed. This task, known as the renormalization problem, is often highly non-trivial.

To address this problem within the RG framework, one introduces a scale decomposition $(\varPsi_t)_{t \in [0,1]}$ of the field $\varPsi$, where $\varPsi_0 = \varPsi$ and $\varPsi_t$ captures the behavior of $\varPsi$ at spatial scales larger than $t\in(0,1]$. This decomposition yields a splitting $\varPsi = \varPsi_t + (\varPsi - \varPsi_t)$, where $\varPsi_t$ is a smooth component, approximately constant at scales smaller than $t$, while $\varPsi - \varPsi_t$ represents rough short-scale fluctuations. The effective potential $V_t$ at scale $t$ is defined by averaging over the fluctuations:
\begin{equation}
	\exp(-V_t(\varphi)):= \bE\Big(\exp(-V(\varPsi-\varPsi_t+\varphi))\Big).
\end{equation}
A key observation is that the family $(V_t)_{t \in (0,1]}$ contains complete information about the system and, in contrast to the bare potential $V$, remains well-defined even as the cutoff is removed.

In Wilson’s original formulation of the RG~\cite{Wil71}, one integrates out degrees of freedom in discrete steps, typically from high to low energy. This leads to a discrete RG map
\begin{equation}
	V_t \mapsto V_{L t}, \qquad
	\exp(-V_{L t}(\varphi)) := \mathbb{E}\left( \exp\left( -V_t(\varPsi_t - \varPsi_{L t} + \varphi) \right) \right),
\end{equation}
where $L > 1$ is a fixed scaling factor. The renormalization problem is then reformulated as the task of tuning the parameters of the bare potential $V$ so that the discrete flow $n \mapsto V_{L^n t}$ remains bounded as the cutoff is removed.

An alternative approach involves studying the continuous flow $t \mapsto V_t$, governed by a certain infinite-dimensional Hamilton--Jacobi--Bellman type equation known as the renormalization group flow equation or Polchinski equation. In principle, solving the renormalization problem then amounts to imposing appropriate boundary conditions for this PDE. This conceptual simplification allowed Polchinski~\cite{Pol84} to give a remarkably concise proof of the perturbative renormalizability of certain quantum field theories. 

Despite its elegance, this approach faces significant technical challenges. A major obstacle is the so-called large field problem, which limits the applicability of the equation to models where the fields are bounded in an appropriate sense. However, even for such models, it has remained difficult to formulate the flow equation as a well-posed PDE due to the lack of an appropriate norm compatible with its functional-analytic structure. As a result, all known nonperturbative constructions using the flow equation proceed indirectly, via power series expansions in the coupling constant~\cite{BK87} or tree expansion techniques~\cite{DR00,DR00b,DR00c}, followed by convergence proofs.

In this work, we present the first direct method for solving the RG flow equation, circumventing the need for auxiliary expansions. We demonstrate that the flow equation can be formulated as a well-posed infinite-dimensional PDE on a suitable Banach space and construct its solution via a fixed point argument. Our approach is made possible by recent advances in noncommutative probability theory, particularly the scale decomposition of the free fermionic field in a filtered noncommutative probability space equipped with a conditional expectation, as developed in~\cite{DFG25,ABDG20}.

To illustrate the effectiveness of our approach, we apply it to the Gross--Neveu model, a~prototypical example of a fermionic, asymptotically free, and barely-renormalizable (scaling critical) quantum field theory. Our analysis yields a rigorous construction of the model and establishes convergence of the Schwinger functions in the limit where the cutoff is removed. Furthermore, we construct a random field on the probability space of the free theory whose moments coincide with the Schwinger functions of the interacting Gross--Neveu model. This provides the first construction of a strong coupling between the free and interacting fields for a scaling-critical quantum field theory.

The scope of this method extends beyond the Gross--Neveu model. It is well-suited for constructing other critical or subcritical quantum field theories in which the field exhibits boundedness. In particular, it applies to other fermionic models such as the symplectic fermion model, which can be viewed as a fermionic analogue of the bosonic $\Phi^4$ theory. Furthermore, the method should be adaptable to certain bosonic theories, such as the sine-Gordon model, where boundedness arises due to the periodic nature of the interaction. More speculatively, the analytic framework developed here may also offer new insights into the long-time asymptotics of interacting particle systems that admit a functional integral representation. This opens a promising direction for connecting rigorously probabilistic particle models with the analytic machinery of RG flow equations.

\subsection{Main results}

The main result of the paper is a new construction of the massive Gross--Neveu model on the plane based entirely on the Polchinski flow equation. The model is defined in terms of $N$ Dirac flavors, represented by the Grassmann spinor fields
\begin{equation}
 \psi\equiv(\bar\psi,\ubar\psi)\equiv 
 ((\bar\psi^{1},\ldots,\bar\psi^{N}),(\ubar\psi^{1},\ldots,\ubar\psi^{N})).
\end{equation}
For every $n\in\{1,\ldots,N\}$ the fields $\bar\psi^{n}\equiv(\bar\psi^{n,\alpha})_{\alpha\in\{1,2\}}$, $\ubar\psi^{n}\equiv(\ubar\psi^{n,\alpha})_{\alpha\in\{1,2\}}$ transform under rotations of the plane as two-component spinors. The fields take values in a Grassmann algebra and, in particular, they all anti-commute. In what follows, we use the notation $\bar\psi\equiv\psi^-$, $\ubar\psi\equiv\psi^+$, $\psi\equiv(\psi^\sigma)_{\sigma\in\bFF}$, where $\bFF:=\{-,+\}\times\{1,\ldots,N\}\times\{1,2\}$, and
\begin{equation}
 \bar\psi\cdot\ubar\psi=\sum_{n=1}^N\sum_{\alpha=1}^2 \bar\psi^{n,\alpha}\,\ubar\psi^{n,\alpha},
 \qquad
 \bar\psi\cdot \slashed{\partial}\ubar\psi
 =
 \sum_{n=1}^N\sum_{\alpha_1,\alpha_2=1}^2\sum_{j=1}^2
 \bar\psi^{n,\alpha_1}\,\gamma^{\alpha_1,\alpha_2}_j\partial_j\ubar\psi^{n,\alpha_2},
\end{equation}
where $\gamma_j=(\gamma_j^{\alpha_1,\alpha_2})_{\alpha_1,\alpha_2\in\{1,2\}}$, $j\in\{1,2\}$, are the Pauli matrices and $\partial_j$, $j\in\{1,2\}$, denote the derivatives with respect to the Cartesian coordinates of the plane. We first introduce a Gross--Neveu model with cutoffs $\tau,\varepsilon\in(0,1]$ that is defined on a two-dimensional torus $\bT_\tau^2$ of size $1/\tau$ in terms of Dirac fields containing only the Fourier modes with frequencies less than $1/\varepsilon$. Subsequently, we study the limit $\tau,\varepsilon\searrow0$. The free part of the action takes the form
\begin{equation}\label{eq:free_action_intro}
 A_{\tau}(\psi) 
 = 
 \int_{\bT_\tau^2} 
 \bar\psi(x)\cdot\ubar\psi(x)
 \,\rd x
 +
 \int_{\bT_\tau^2} 
 \bar\psi(x)\cdot(\slashed{\partial}\ubar\psi)(x)\,\rd x.
\end{equation}
The first term on the RHS of the above equality is called the mass term whereas the second one is called the kinetic term. The interaction potential is given by
\begin{equation}\label{eq:potential_intro}
 U_{\tau,\varepsilon}(\psi) = 
 \int_{\bT_\tau^2} 
 1/g_{\tau,\varepsilon}~(\bar\psi(x)\cdot\ubar\psi(x))^2
 \,\rd x
 +
 \int_{\bT_\tau^2} 
 r_{\tau,\varepsilon}\,\bar\psi(x)\cdot\ubar\psi(x)\,\rd x.
\end{equation}
The parameters $1/g_{\tau,\varepsilon},r_{\tau,\varepsilon}\in\bR$ are called the bare coupling constant and the mass counterterm, respectively. The Grassmann measure of the Gross--Neveu model with infrared cutoff $\tau\in(0,1]$ and ultraviolet cutoff $\varepsilon\in(0,1]$ is defined by the formula
\begin{equation}\label{eq:measure_intro}
 \mu_{\tau,\varepsilon}(F) := 
 \frac{\int F(\vartheta_\varepsilon\ast\psi_{\tau,\varepsilon})\,\exp\!\big(-A_{\tau}(\psi_{\tau,\varepsilon})+U_{\tau,\varepsilon}(\vartheta_\varepsilon\ast\psi_{\tau,\varepsilon})\big)\,\rd\psi_{\tau,\varepsilon}}{\int \exp\!\big(-A_{\tau}(\psi_{\tau,\varepsilon})+U_{\tau,\varepsilon}(\vartheta_\varepsilon\ast\psi_{\tau,\varepsilon})\big)\,\rd\psi_{\tau,\varepsilon}}
\end{equation}
for functionals of polynomial type $F$ on Grassmann-valued Schwartz distributions on $\bR^2$. We use the notation $\int F(\psi_{\tau,\varepsilon}) \,\rd\psi_{\tau,\varepsilon}$ for the so-called Berezin integral~\cite{Ber87}. The Berezin integral can be viewed as a fermionic analog of the Lebesgue integral, even though it is not an integral in the usual sense. The Grassmann field $\psi_{\tau,\varepsilon}\equiv(\psi^\sigma_{\tau,\varepsilon})_{\sigma\in\bFF}$ on $\bR^2$ appearing in the above formula is periodic with period $1/\tau$ and contains only Fourier modes with frequencies smaller than $1/\varepsilon$. Since working with a sharp cutoff in momentum space would result in slow decay of correlations in position space we convolve the Grassmann field $\psi_{\tau,\varepsilon}$ with a Schwartz function $\vartheta_\varepsilon\in\sS(\bR^2)$ supported in Fourier space inside the ball of radius $1/\varepsilon$. We denote by $\sT^\bK$ the set of tuples $(V^k)_{k\in\bK}$ such that $V^k\in\sT$ for all $k\in\bK$. Now we are ready to state our main result.

\begin{thm}\label{thm:main}
Let $N\in\{2,3,\ldots\}$. There exists $\lambda_\star\in(0,1]$ such that for every $\lambda\in(0,\lambda_\star]$ there exists a choice of parameters $(g_{\tau,\varepsilon},r_{\tau,\varepsilon})_{\tau,\varepsilon\in(0,1]}$ such that the following statements are true.
\begin{itemize}
 \item[(A)] For every $m\in\bN_+$ and $\varphi_1,\ldots,\varphi_m\in\sS(\bR^2)^\bFF$ the limit
\begin{equation}
 \langle S^m_\lambda,\varphi_1\otimes\ldots\otimes\varphi_m\rangle :=\lim_{\tau,\varepsilon\searrow0} \int \psi(\varphi_1)\ldots\psi(\varphi_m)\,\mu_{\tau,\varepsilon}(\rd\psi)
\end{equation}
exists and defines the $m$-point Schwinger function $S^m_\lambda\in\sS'(\bR^{2m})^{\bFF^{m}}$, where $\mu_{\tau,\varepsilon}$ is the Grassmann measure of the Gross--Neveu model with cutoffs $\tau,\varepsilon\in(0,1]$ defined in terms of the parameters $g_{\tau,\varepsilon}$, $r_{\tau,\varepsilon}\in\bR$.

\item[(B)] For every $m\in\bN_+$ the distribution $S^m_\lambda\in\sS'(\bR^{2m})^{\bFF^{m}}$ is invariant under Euclidean transformations of the plane.

\item[(C)] For $m\in\bN_+$ let $T^m_\lambda\in\sS'(\bR^{2m})^{\bFF^{m}}$ be the $m$-point truncated Schwinger function associated to the hierarchy $(S^m_\lambda)_{m\in\bN_+}$. For every $m\in\bN_+$ there exist a collection $(\hat T_\lambda^{m,a,\sigma})_{a\in\bA^m,\sigma\in\bFF^m}$ of finite Borel measures on $\bR^{2(m-1)}$ such that
\begin{equation}
 \langle T^m_\lambda,\varphi\rangle= \sum_{a\in\bA^m}\sum_{\sigma\in\bFF^m} \int_{\bR^{2m}} \hat T_\lambda^{m,a,\sigma}(\rd y_1,\ldots,\rd y_{m-1})\,(\partial^a\varphi^\sigma)(x,x+y_1,\ldots,x+y_{m-1})\,\rd x
\end{equation}
for all $\varphi\in\sS(\bR^{2m})^{\bFF^m}$, where $\bA:=\{a\in\mathbb N_0^2\,|\,|a|\leq2\}$.

\item[(D)] The equalities
\begin{equation}
\begin{gathered}
 \frac{1}{2!}\sum_{\sigma\in\bFF^2}\psi^{\sigma_1}\psi^{\sigma_2}\int_{\bR^2}\hat T_\lambda^{2,0,\sigma}(\rd y) =\ubar\psi\cdot\bar\psi,
 \\
  \frac1{4!}\sum_{\sigma\in\bFF^4}\psi^{\sigma_1}\ldots\psi^{\sigma_4}\int_{\bR^6}\hat T_\lambda^{4,0,\sigma}(\rd y_1,\rd y_2,\rd y_3) = \lambda \,(\ubar\psi\cdot\bar\psi)^2
\end{gathered} 
\end{equation}
hold for all Grassmann-valued $\psi\equiv(\psi^\sigma)_{\sigma\in\bFF}\equiv(\bar\psi^{\alpha,\varsigma},\ubar\psi^{\alpha,\varsigma})_{\alpha\in\{1,2\},\varsigma\in\{1,\ldots,N\}}$.

\item[(E)] The truncated Schwinger functions decay  stretched exponentially, more precisely
\begin{equation}
 \lim_{|x|\to\infty}\exp(|x|^{1/2})\,|\langle T_\lambda^{m+n},\varphi_x\otimes\psi\rangle|=0
\end{equation}
for all $m,n\in\bN_+$ and $\varphi\in C^\infty_\rc(\bR^{2m})^{\bFF^m}$, \mbox{$\psi\in C^\infty_\rc(\bR^{2n})^{\bFF^n}$}, where for $x\in\bR^2$ we define $\varphi_x\in C^\infty_\rc(\bR^{2m})^{\bFF^m}$ by the equality $\varphi_x(y_1,\ldots,y_m):=\varphi(y_1-x,\ldots,y_m-x)$ for all $y_1,\ldots,y_m\in\bR^2$.
\end{itemize}
\end{thm}

\begin{proof}
The theorem is an immediate consequence of Theorem~\ref{thm:convergence}. 
\end{proof}

\begin{rem}
The precise definition of the Berezin integral and the Gross--Neveu measure with cutoffs can be found in Sec.~\ref{sec:gross-neveu}. The notion of invariance under Euclidean transformations used in Item~(B) of the above theorem is defined in Sec.~\ref{sec:symmetries}.
\end{rem}

\begin{rem}\label{rem:ren_conditions}
The equalities stated in Item~(D) of the above theorem should be viewed as renormalization conditions that fix implicitly the bare coupling constant $g_{\tau,\varepsilon}$ and the mass counterterm $r_{\tau,\varepsilon}$ appearing in the expression for the measure $\mu_{\tau,\varepsilon}$. We parameterize the Gross--Neveu models in terms of numerical constants linked directly to the Schwinger functions instead of using the parameters $g_{\tau,\varepsilon}$ and $r_{\tau,\varepsilon}$. Note that Item~(D) implies in particular that the truncated four-point Schwinger function is non-zero. As a result, the construction yields a non-trivial quantum field theory.
\end{rem}

\begin{rem}
The result stated in Item~(E) of the above theorem is not optimal. Actually, the truncated Schwinger functions decay exponentially. However, the exponential decay does not follow immediately from the construction presented in the paper and its proof requires a~separate argument, which is not included in the paper.
\end{rem}

\begin{thm}\label{thm:main2}
Let $N\in\{2,3,\ldots\}$. There exist a Banach algebra $\sB$ with the identity $\mathds{1}$, $\lambda_\star\in(0,1]$ such that for every $\lambda\in(0,\lambda_\star]$ there exist a choice of parameters $(g_{\tau,\varepsilon},r_{\tau,\varepsilon})_{\tau,\varepsilon\in(0,1]}$ and continuous linear maps
\begin{equation}
 \fE\,:\,\sB\to\bC,
 \qquad
 \varPsi,\varPhi,\varPsi_{\tau,\varepsilon},\varPhi_{\tau,\varepsilon} \,:\,\sS(\bR^2)^\bFF \to \sB,\quad \tau,\varepsilon\in(0,1],
\end{equation}
satisfying the following conditions.
\begin{enumerate}
 \item[(A)] The expected value fulfills the normalization condition $\fE \mathds{1}=1$.

 \item[(B)] For all $\tau,\varepsilon\in(0,1]$, $m\in\bN_+$, $\varphi_1,\ldots,\varphi_m\in\sS(\bR^2)^\bFF$ it holds
\begin{equation}
\begin{aligned}
\fE \varPsi_{\tau,\varepsilon}(\varphi_1)\ldots\varPsi_{\tau,\varepsilon}(\varphi_m) &=
 \int \psi(\varphi_1)\ldots\psi(\varphi_m)\,\nu_{\tau,\varepsilon}(\rd\psi),
 \\
\fE \varPhi_{\tau,\varepsilon}(\varphi_1)\ldots\varPhi_{\tau,\varepsilon}(\varphi_m)&=
 \int \psi(\varphi_1)\ldots\psi(\varphi_m)\,\mu_{\tau,\varepsilon}(\rd\psi)
\end{aligned}
\end{equation}
and
\begin{equation}
\begin{aligned}
 \fE \varPsi(\varphi_1)\ldots\varPsi(\varphi_m) &=\langle S_{0}^m,\varphi_1\otimes\ldots\otimes\varphi_m\rangle,
 \\
 \fE \varPhi(\varphi_1)\ldots\varPhi(\varphi_m)&=\langle S^m_\lambda,\varphi_1\otimes\ldots\otimes\varphi_m\rangle,
\end{aligned}
\end{equation}
where $\nu_{\tau,\varepsilon}$ is the free field measure with cutoffs $\tau,\varepsilon\in(0,1]$ defined by~\eqref{eq:measure_intro} with $U_{\tau,\varepsilon}$ set to zero, $\mu_{\tau,\varepsilon}$ is the Gross--Neveu measure with cutoffs $\tau,\varepsilon\in(0,1]$, $(S^m_0)_{m\in\bN_+}$ are the Schwinger functions of the free theory and $(S^m_\lambda)_{m\in\bN_+}$ are the Schwinger functions of the Gross--Neveu model constructed in Theorem~\ref{thm:main}.

 \item[(C)] For all $\sigma\in\bFF$ it holds $\varPsi^\sigma,\varPhi^\sigma\in\sC^{-1/2}$ and
\begin{equation}
 \lim_{\tau,\varepsilon\searrow0} \|\varPsi^\sigma - \varPsi_{\tau,\varepsilon}^\sigma\|_{\tilde\sC^\alpha} = 0,
 \qquad
 \lim_{\tau,\varepsilon\searrow0} \|\varPhi^\sigma - \varPhi_{\tau,\varepsilon}^\sigma\|_{\tilde\sC^\alpha} = 0
\end{equation}
for all $\alpha\in(-\infty,-1/2)$, where $\sC^\alpha$ is the Hölder--Besov space with regularity index $\alpha\in\bR$ of Schwartz distributions valued in $\sB$ and $\tilde\sC^\alpha$ is an analogous Hölder--Besov space with a polynomially decaying weight.

\item[(D)] For all $\sigma\in\bFF$ it holds
 \begin{equation}\label{eq:main2_limsup}
  \limsup_{i\to\infty} 2^{\alpha i} \sup_{x\in\bR^2}\|(\Delta_i \varPsi^\sigma)(x)\|_{\sB}>0
 \end{equation}
and
 \begin{equation}\label{eq:main2_diff}
  \sup_{i\in\{-1,0,1,\ldots\}} (i+2)^{1/2}\,2^{\alpha i} \sup_{x\in\bR^2}\|(\Delta_i (\varPhi^\sigma-\varPsi^\sigma))(x)\|_{\sB}<\infty,
 \end{equation}
where $(\Delta_i)_{i\in\{-1,0,1,\ldots\}}$ are the Littlewood--Paley blocks on $\bR^2$ and $\alpha=-1/2$.

\end{enumerate}
\end{thm}
\begin{proof}
The non-commutative probability space $(\sB,\fE)$ as well as the free fields $\varPsi,\varPsi_{\tau,\varepsilon}$ are constructed in Sec.~\ref{sec:free_field_decomposition}. The interacting fields $\varPhi,\varPhi_{\tau,\varepsilon}$ are constructed in Sec.~\ref{sec:interacting_field}. The statements regarding the free fields $\varPsi,\varPsi_{\tau,\varepsilon}$ are straightforward to prove and are consequences of Lemmas~\ref{lem:free_measure} and~\ref{lem:Psi_bounds}.  The statements regarding the interacting fields $\varPhi,\varPhi_{\tau,\varepsilon}$  follow from Corollary~\ref{cor:fbsde}.
\end{proof}

\begin{rem}
The above theorem states the existence of a coupling between the free and interacting fields. More specifically, we construct the Gross--Neveu field as the final solution $\varPhi=\varPhi_0$ to the backward SDE
\begin{equation}
 \varPhi_t = \int_t^1 \fJ\dot G_s\ast \rD V_s[\fJ\varPhi_s]\,\rd s + \varPsi_t,
\end{equation}
where:
\begin{itemize}
 \item $t\mapsto \dot G_t$ is the scale decomposition of the free field covariance such that $\dot G_t$ is supported in the Fourier space in a shell of radius of order $1/t$,
 \item $t\mapsto \varPsi_t$ is a Grassmann Brownian martingale, called  the scale decomposition of the free field, such that $\varPsi_t$ is supported in a Fourier space in a ball of radius of order $1/t$ and
 \item $\rD V_t$ is the gradient of the so-called effective potential solving the Polchinski equation.
\end{itemize}
Let us mention that similar couplings were constructed before in~\cite{BH22,BGH23,DFG25,GM24} in super-renormalizable models.
\end{rem}

\begin{rem}
Note that $\|\varPsi^\sigma\|_{\sC^{\alpha}}$ with $\alpha=-1/2$ coincides with the expression on the LHS of the bound~\eqref{eq:main2_limsup} with $\limsup$ replaced by $\sup$. Consequently, the condition $\varPsi^\sigma\in \sC^{-1/2}$ stated in Item~(C) is optimal. The sub-logarithmic (in the scale parameter $2^i$) improvement of the regularity for the difference $\varPhi-\varPsi$ in the bound~\eqref{eq:main2_diff} is a manifestation of the so-called asymptotic freedom of the Gross--Neveu model. Actually, a stronger bound with $(2+i)^{1/2}$ replaced by $(2+i)$ is also true. However, its proof requires some extra work and is omitted. A coupling between $\varPhi$ and $\varPsi$ such that $\varPhi-\varPsi$ satisfies the bound~\eqref{eq:main2_diff} with $(2+i)^{1/2}$ replaced by $(2+i)^\beta$ with $\beta>1$ is not expected to exist.
\end{rem}

\subsection{Method of the proof}

Our construction of the Gross--Neveu model, like the previous constructions~\cite{GK85a,FMRS86}, utilizes the Wilsonian renormalization group theory~\cite{Wil71}. The fundamental object of this theory is the so-called effective potential $U_{\tau,\varepsilon;t}$ depending on the spatial scale $t\in[0,1]$ and the cutoffs $\tau,\varepsilon\in(0,1]$. The effective potential $U_{\tau,\varepsilon;t}$ at the spatial scale $t=0$ coincides with the interaction potential $U_{\tau,\varepsilon}$ defined by~\eqref{eq:potential_intro}. The goal is to construct the effective potential $U_{\tau,\varepsilon;t}$ at the spatial scale $t=1$, which coincides with the generating functional of the so-called connected amputated Schwinger functions. To this end, one usually solves a certain equation that relates the effective potentials at different scales. The previous constructions~\cite{GK85a,FMRS86} of the Gross--Neveu model used the renormalization group transformation that relates the effective potentials $U_{\tau,\varepsilon;t}$ at different discrete values of the scale parameter $t\in\{L^{-k}\,|\,k\in\bN_0\}\subset (0,1]$, where $L\in(1,\infty)$ is sufficiently big. In contrast, the approach of this paper is based on the so-called Polchinski flow equation~\cite{WK74,Pol84} that is a certain differential equation in the scale parameter $t\in(0,1]$.

It was recognized a long time ago that the flow equation is a very powerful tool in perturbative quantum field theory. In fact, there is a short and general proof of perturbative renormalizability of QFT models based on the flow equation~\cite{Pol84,Kop07,Mul03,Sal99}. The renormalization problem is solved using a simple inductive argument and amounts to imposing appropriate boundary conditions for the flow equation. In particular, the proof avoids the problem of overlapping divergences appearing in all approaches based on the Feynman diagrammatical representations such as the BPHZ approach. However, the applicability of the Polchinski flow equation in non-perturbative constructions is limited. For bosonic theories it is generally believed that the flow equation is not useful non-perturbatively as it does not allow us to take advantage of the boundedness from above of the interaction potential. The exception is the sine-Gordon model, which up to the second threshold was constructed~\cite{BK87,BB21,GM24} with the use of the Polchinski equation. In quantum field theory models with bosons the boundedness from above of the effective potential is usually necessary to ensure that the interacting measure has finite total mass. This is the so-called large field problem. It is possible to address this problem in the approach based on the renormalization group transformation but the problem seems intractable in the approach based on the Polchinski equation. As an aside, let us mention that the large field problem can be solved in the framework for singular stochastic PDEs developed in~\cite{Duc25a,Duc25b} based on a certain flow equation that plays an analogous role to the Polchinski equation. In fact, this framework was used recently to construct the bosonic fractional $\Phi^4_3$ model in full super-renormalizable regime~\cite{DGR23}.  

Since fermionic fields can be represented by bounded operators the large field problem should be absent in purely fermionic models of quantum field theory and a construction of such models based on the Polchinski flow equation should be feasible. Such a construction is not straightforward because the effective potential solving the Polchinski equation is necessarily a highly non-local functional, which in particular is generically not of polynomial type. Progress in this direction has recently been made by De Vecchi, Fresta, Gubinelli in~\cite{DFG25}, where a new approach to super-renormalizable fermionic theories was developed based on a certain forward-backward stochastic differential equation (FBSDE). The main advantage of the FBSDE approach is the fact that in this approach one only needs to construct an effective potential that satisfies the flow equation up to a sufficiently small error term. In particular, for super-renormalizable models it is always possible to construct a suitable effective potential in the space of functionals of polynomial type using a certain recursive procedure. In principle the FBSDE approach can also be useful for barely-renormalizable models. However, it is clear that in the case of such models a suitable approximate solution of the Polchinski equation cannot be a functional of polynomial type. Consequently, allowing for an error term in the Polchinski equation does not provide an obvious benefit.

The primary difficulty in constructing an exact solution to the Polchinski equation lies in the absence of a natural norm on the space of functionals that is compatible with the structure of the equation. The use of flow equations to establish bounds on correlation functions in fermionic quantum field theories dates back to the work of Brydges and Wright~\cite{BW88}, which, however, was later shown to contain an error~\cite{BW99}, as clarified in~\cite{SW00}. On a technical level, the issue appears to stem from the fact that the scale decomposition of the free fermionic field, formally analogous to a Grassmann-valued cylindrical Wiener process, is not Lipschitz continuous with respect to the scale parameter, but only $1/2$-H{\"o}lder continuous. For further discussion, we refer the reader to~\cite{SW00} and to Sec.~\ref{sec:strategy}.

A central contribution of this paper is the introduction of a novel norm on the space of functionals, which enables a direct solution of the flow equation without relying on the traditional tree expansion ansatz~\cite{DR00,DR00b,DR00c}. This norm is constructed using the filtered non-commutative probability space developed in~\cite{DFG25}. With this framework, we reformulate the mild form of the Polchinski equation as a fixed-point problem in a suitably chosen complete metric space and establish existence and uniqueness of solutions via the Banach fixed-point theorem. Finally, we express the Schwinger functions in terms of the effective potential $U_{\tau,\varepsilon;t}$ at spatial scale $t=1$, and prove their convergence as $\tau, \varepsilon \searrow 0$.

\begin{rem}
The sign of the parameter $\lambda$ that appears in Item~(D) of Theorem~\ref{thm:main} plays a very important role in our construction. The Gross--Neveu model with $\lambda>0$ and $N\in\{2,3,\ldots\}$ that we construct is asymptotically free -- at high energies the Schwinger functions of the model are well approximated by the Schwinger functions of the free theory. We refer the reader to~\cite{GK89} for a discussion of the significance of asymptotic freedom in constructions of barely renormalizable models of quantum field theory. For $\lambda<0$ and $N\in\{2,3,\ldots\}$ the Gross--Neveu model is not expected to be asymptotically free and most likely does not exist non-perturbatively. For $N=1$ the Gross--Neveu model coincides with the so-called Thirring model. Even though the Thirring model is not asymptotically free it has some special properties that allow its non-perturbative construction~\cite{BFM07}.
\end{rem}

\begin{rem}
One of the manifestations of the asymptotic freedom of the Gross--Neveu model is the fact that the bare coupling constant $1/g_{\tau,\varepsilon}>0$ fixed by the renormalization condition stated in Item~(D) of Theorem~\ref{thm:main} vanishes logarithmically in the limit $\varepsilon\searrow0$. Thus, naively one could expect that the construction yields a free theory. However, as we pointed out in Remark~\ref{rem:ren_conditions} the hierarchy of the constructed Schwinger functions is non-trivial. The underlying intuitive reason for the non-triviality is the fact that in the limit $\varepsilon\searrow0$ the fields become genuine distributions and consequently the quartic interaction term involving pointwise products of the fields is very singular. See~\cite{Hai24} for a related phenomenon in the context of subcritical singular stochastic PDEs.
\end{rem}

\begin{rem}
We rewrite the Polchinski equation as a fixed point of a certain map that is well-defined also in the limit $\tau,\varepsilon\searrow0$. The Schwinger functions of the Gross--Neveu model without cutoffs are directly related to the above fixed point. Let us stress that one of the main difficulties in the proof of Theorem~\ref{thm:main} is to show that the Schwinger functions are well-approximated by the moments of the Grassmann measure~\eqref{eq:measure_intro} with the cutoffs $\tau,\varepsilon\in(0,1]$. In particular, the proof would simplify drastically if one was only interested in constructing the Schwinger functions without cutoffs and the existence of the limit $\tau,\varepsilon\searrow0$ was disregarded. The proof would also simplify significantly if one only wanted to show bounds uniform in the cutoffs $\tau,\varepsilon\in(0,1]$. Both simplifications are common in mathematical physics literature. We could not find a reference in which the existence of both the infinite volume $\tau\searrow0$ and the ultraviolet $\varepsilon\searrow0$ limits of the Gross--Neveu correlations is rigorously proved.
\end{rem}

\subsection{Possible generalizations}

The method developed in this work is applicable to other purely fermionic renormalizable models of quantum theory such as for example the symplectic fermion model with $N\neq8$. The symplectic fermion model is a fermionic analog of the bosonic $\Phi^4$ model in four dimensions. The model is barely renormalizable. It describes $N\in\{4,5,\ldots\}$ scalar fermionic fields. The kinetic part of the action of the model contains the Laplacian and the quartic interaction potential is invariant under symplectic transformations of the fields. See~\cite{GMR21} for a precise definition of the model. The assumption $N\neq8$ is probably not essential and is related to the fact that the construction presented in the paper requires that the one-loop beta-function $\beta_2$ is positive. Our construction can be also applied to super-renormalizable models such as the Gross--Neveu model or symplectic fermion model with modified propagators that are less singular at the origin than the standard propagators. Even though our construction simplifies drastically in the super-renormalizable regime the FBSDE approach developed in~\cite{DFG25}, which was discussed briefly above, may be advantageous. Let us also mention that the method developed in this paper should be applicable to the sine-Gordon model of quantum field theory up to the second threshold. In addition, the results of the paper can be useful in the construction of weakly interacting fixed points of the renormalization group transformation.

\subsection{Overview of the literature}

The Gross--Neveu model appeared for the first time in the work~\cite{MW73} by Mitter and Weisz who investigated the flow of renormalization group transformations in this model. In the paper~\cite{GN74} Gross and Neveu introduced the same model as a toy model of Yang--Mills theory in four dimensions. They studied the dynamical mass generation in a model with chiral symmetry. They presented an argument suggesting that despite the fact that the action of such a model does not contain a mass term, which is prohibited by the chiral symmetry, the truncated Schwinger functions decay exponentially, that is the model is massive. Note that in this paper we construct the version of the Gross--Neveu model with a mass term.

The first mathematical construction of the Gross--Neveu model was given by Gaw{\k e}dzki and Kupiainen in~\cite{GK85a}. Another construction was given shortly thereafter in~\cite{FMRS86} by Feldman, Magnen, Rivasseau and Seneor. Both constructions rely on the discrete renormalization group method and the estimates for fermionic correlations derived with the use of the Gram--Hadamard inequality and the cluster expansion. In~\cite{DR00} Disertori and Rivasseau constructed the Gross--Neveu model by proving convergence of appropriately rearranged perturbation theory. Note that while the Polchinski flow equation was studied in~\cite{DR00} the construction was not based directly on the Polchinski equation. The authors constructed the effective potential using a convergent tree expansion and subsequently verified that it satisfies the flow equation. Let us also mention the recent paper~\cite{DY24} by Dimock and Yuan who studied the flow of the renormalization group transformations in the massless Gross--Neveu model on a torus and established uniform boundedness of the partition function of the model in the UV cutoff.

The above-mentioned works concerned the Gross--Neveu model with a mass term on the plane or the model with chiral symmetry on the unit torus. The problem of the dynamical mass generation was addressed in~\cite{KMR95}, where it was proved that the two-point Schwinger functions of the chiral model with a fixed UV cutoff fall off exponentially. Some properties of the Gross--Neveu model related to the particle interpretation and the asymptotic completeness were established in \cite{IM87,IM88a,IM88b}. The super-renormalizable Gross--Neveu model with a less singular propagator was studied in~\cite{SW02}. A construction of the Gross--Neveu model with a more singular propagator, which is perturbatively non-renormalizable, was given in~\cite{GK85b}.

Let us discuss some related results. The two-dimensional Yukawa model was constructed by Lesniewski in~\cite{Les87} using the technique developed in~\cite{GK85a}. The important contribution of~\cite{Les87} was a new elegant proof of bounds for fermionic connected correlations based on the Brydges--Battle--Federbush formula. An alternative approach to bounds for fermionic correlations was given by Salmhofer and Wieczerkowski in~\cite{SW00}. The infrared stable non-Gaussian fixed point of the renormalization group transformation in the symplectic fermion model with weakly relevant quartic interaction was constructed in~\cite{GMR21}. An alternative construction of this fixed point using the Polchinski equation and the tree expansion ansatz was given in~\cite{Gre24}. Let us mention once again the stochastic quantization approach to fermionic models based on a FBSDE developed in~\cite{DFG25}, which was used to construct the symplectic fermion model with $N=4$ in infinite volume in full super-renormalizable regime. First steps towards the construction of fermionic quantum field theories using the parabolic stochastic quantization method were made in~\cite{ABDG20,CHP23,DFGG25}, where the framework of non-commutative probability suitable for fermionic stochastic analysis was developed and local well-posedness for the parabolic stochastic quantization equation of the two-dimensional Yukawa model was proved. We refer the reader to~\cite{FKT00} for a general overview of fermionic functional integrals with applications to constructive quantum field theory.

\section{Strategy of the proof}\label{sec:strategy}

In this section, we give an informal explanation of the main steps in the proofs of Theorems~\ref{thm:main} and~\ref{thm:main2}, focusing on the key ideas behind them. The goal is to help the reader understand the overall strategy before going into the full technical details. At the end of the section, we briefly describe how the rest of the paper is organized.

\begin{dfn}\label{dfn:kappa}
We fix $N\in\{2,3,\ldots\}$ and a small parameter $\kappa:=1/1000$.
\end{dfn}

\begin{dfn}\label{dfn:lambda}
Let $\beta_2:=2(N-1)/\pi$. For $t,\lambda\in(0,1]$ we define $\lambda_t:=(\lambda^{-1}-\beta_2\log t)^{-1}$. For $t=0$ we set $\lambda_t=0$.
\end{dfn}

\begin{rem}
We identify functions and distributions on the torus $\bT^2_\tau$ with periodic functions and distributions on $\bR^2$ with period $1/\tau$.  
\end{rem}

We use the framework of non-commutative probability. A non-commutative probability space is the pair $(\sF,\fE)$ consisting of a unital Banach subalgebra $\sF$ of the algebra of operators acting on a separable Hilbert space and a continuous normalized linear functional $\sF\to\bC$. Recall that the Grassmann measure~\eqref{eq:measure_intro} of the Gross--Neveu model is defined using the Berezin integral, which, like the Lebesgue integral, is only well-defined in finite dimension. For this reason, it is advantageous to use as a reference measure a certain Gaussian Grassmann measure, called the free field measure with cutoffs $\tau,\varepsilon\in(0,1]$, defined by the formula
\begin{equation}
 \nu_{\tau,\varepsilon}(F) := 
 \frac{\int F(\vartheta_\varepsilon\ast\psi_{\tau,\varepsilon})\exp(-A_{\tau}(\psi_{\tau,\varepsilon}))\,\rd\psi_{\tau,\varepsilon}}{\int \exp(-A_{\tau}(\psi_{\tau,\varepsilon}))\,\rd\psi_{\tau,\varepsilon}}
\end{equation}
for functionals $F$ of polynomial type. Note that the RHS of the above equality involves only the free part of the action $A_{\tau}(\psi_{\tau,\varepsilon})$, which is quadratic in the field $\psi_{\tau,\varepsilon}$. We stress that the above formula for the measure $\nu_{\tau,\varepsilon}$ is meaningful only if $\tau,\varepsilon\in(0,1]$. We would like to construct a collection of anti-commuting random variables $(\varPsi_{\tau,\varepsilon})_{\tau,\varepsilon\in[0,1]}$ in a non-commutative probability space such that $\nu_{\tau,\varepsilon}(F)=\fE(F(\varPsi_{\tau,\varepsilon}))$ for all $\tau,\varepsilon\in(0,1]$.

Such a family can be constructed using the Osterwalder--Schrader construction~\cite{OS72,OS73}; see Sec.~\ref{sec:free_field_decomposition} for details. The resulting fields $\varPsi_{\tau,\varepsilon}$ are centered anti-commuting Gaussian fields in a non-commutative probability space $(\sF,\fE)$ with covariance
\begin{equation}
 \fE \langle\varPsi_{\tau,\varepsilon},\phi\rangle =0,
 \qquad
 \fE \langle\varPsi_{\tau,\varepsilon},\phi\rangle
 \langle\varPsi_{\tau,\varepsilon},\eta\rangle
 =
 \langle \phi, G_{\tau,\varepsilon}\ast\eta\rangle.
\end{equation}
Their higher moments are determined by the Fermionic Wick rule: the expectation of a product of an odd number of fields vanishes, while
\begin{equation}
 \fE\prod_{i=1}^{2n}\langle\varPsi_{\tau,\varepsilon},\phi_i\rangle
 =
 \sum_{\pi\in\mathfrak P_{2n}}
 \mathrm{sgn}(\pi)
 \prod_{j=1}^n
 \langle
 \phi_{\pi(2j-1)},
 G_{\tau,\varepsilon}\ast\phi_{\pi(2j)}
 \rangle,
\end{equation}
where $\mathfrak P_{2n}$ denotes the set of pairings of $\{1,\ldots,2n\}$ and $\mathrm{sgn}(\pi)$ is the sign associated with the corresponding permutation. See Remark~\ref{rem:Wick}. For $\tau,\varepsilon\in(0,1]$, this construction agrees with the Gaussian Berezin integral in the sense that
\begin{equation}
 \nu_{\tau,\varepsilon}(F)=\fE F(\varPsi_{\tau,\varepsilon});
\end{equation}
see Lemma~\ref{lem:free_measure}.

If $\tau\in(0,1]$, then $\varPsi_{\tau,\varepsilon}$ is periodic with period $1/\tau$. The propagator $G_{\tau,\varepsilon}$ with the IR cutoff $\tau\in(0,1]$ and the UV cutoff $\varepsilon\in(0,1]$ is defined as a periodization of $G_\varepsilon$ with the period $1/\tau$, the propagator $G_\varepsilon$ with the UV cutoff $\varepsilon\in(0,1]$ is defined by $G_\varepsilon:=\vartheta_\varepsilon\ast G\ast\vartheta_\varepsilon$ and the propagator $G$ is the inverse of the differential operator 
\begin{equation}
\frac{1}{2}
\begin{pmatrix}
0 & \slashed{\partial}+1
\\
\slashed{\partial}{}^{\mathrm{t}}-1 & 0
\end{pmatrix}
\end{equation}
appearing in the free part of the action~\eqref{eq:free_action_intro}. The family of functions $\vartheta_\varepsilon\in\sS(\bR^2)$ parameterized by $\varepsilon\in(0,1]$ is chosen in such a way that it converges to the Dirac delta at the origin as $\varepsilon\searrow0$ and for all $\varepsilon\in(0,1]$ the Fourier transform of $\vartheta_\varepsilon$ is supported in a ball of radius $1/\varepsilon$. Moreover, if $\varepsilon=1$, then $\vartheta_\varepsilon=0$ and we define $\vartheta_\varepsilon$ with $\varepsilon=0$ to be the Dirac delta at the origin. If $\tau=0$, then $G_{\tau,\varepsilon}:=G_\varepsilon$ and if $\varepsilon=0$, then $G_\varepsilon:=G$. We call $\varPsi_{\tau,\varepsilon}$ the free field with cutoffs $\tau,\varepsilon\in(0,1]$.

Using $\varPsi_{\tau,\varepsilon}$ one rewrites the formula~\eqref{eq:measure_intro} in the following way
\begin{equation}\label{eq:measure_intro2}
 \mu_{\tau,\varepsilon}(F) 
 =
 \frac{\fE\big( F(\varPsi_{\tau,\varepsilon}) \re^{U_{\tau,\varepsilon}(\varPsi_{\tau,\varepsilon})}\big)}{\fE\big(\re^{U_{\tau,\varepsilon}(\varPsi_{\tau,\varepsilon})}\big)}.
\end{equation}

Note that for all $\varepsilon\in(0,1]$ the free field $\varPsi_{\tau,\varepsilon}$ belongs to $C^\infty(\bR^2,\sB(\mathscr{H}))^\bFF$. Consequently, the RHS of the above equality is well-defined for all $\tau,\varepsilon\in(0,1]$ (provided the denominator is not zero). Since, in the limit $\varepsilon\searrow0$, the free field $\varPsi_{\tau,\varepsilon}$ is not a function on $\bR^2$ but only a Schwartz distribution the pointwise products in the expression~\eqref{eq:potential_intro} for the potential $U_{\tau,\varepsilon}$ become meaningless. In the limit $\tau\searrow0$, the expression~\eqref{eq:potential_intro} for the potential $U_{\tau,\varepsilon}$ is ill-defined due to unbounded domain of integration. Thus, the expression~\eqref{eq:measure_intro2} for the measure $\mu_{\tau,\varepsilon}$ becomes singular in the limit $\varepsilon\searrow0$ or $\tau\searrow0$. Because of the presence of the unit mass term in the free part of the action the proof of the existence of the limit $\tau\searrow0$ is relatively easy. The limit $\varepsilon\searrow0$ is quite subtle as it only exists if the parameters $g_{\tau,\varepsilon}$ and $r_{\tau,\varepsilon}$ of the potential~\eqref{eq:potential_intro} diverge at particular rate as $\varepsilon\searrow0$.

In order to study the limit $\varepsilon\searrow0$ of $\mu_{\tau,\varepsilon}(F)$ we use renormalization group theory~\cite{Wil71,Pol84}. To this end, we introduce a certain scale decomposition of the propagator $G_{\tau,\varepsilon}$ and the free field $\varPsi_{\tau,\varepsilon}$. 
The scale decomposition of the propagator is constructed in
Sec.~\ref{sec:covariance_decomposition}, while the corresponding
decomposition of the free field is constructed in
Sec.~\ref{sec:free_field_decomposition}. The scale decomposition $[0,1]\ni t \mapsto G_{\varepsilon;t}\in L^1(\bR^2)^{\bFF^2}$ of the propagator $G_\varepsilon$ is defined by $G_{\varepsilon;t}:=\vartheta_t \ast G_{\varepsilon}\ast \vartheta_t$ for $t\in[0,1]$. Then $G_{\tau,\varepsilon;t}$ is defined as the periodization of $G_{\varepsilon;t}$. The parameter $t\in[0,1]$ plays the role of an artificial UV cutoff. Note that due to the properties of the function $\vartheta_t$ for every $t\in(0,1]$ and $\tau,\varepsilon\in[0,1]$ the propagator $G_{\tau,\varepsilon;t}$ is smooth and essentially constant at spatial scales smaller than $t$ and it captures the behavior of $G_{\tau,\varepsilon}$ at spatial scales larger than $t$.

The scale decomposition $[0,1]\mapsto \varPsi_{\tau,\varepsilon;t}$ of the free field $\varPsi_{\tau,\varepsilon}$ is defined in a non-commutative probability space containing the probability space of the free field as a collection of anti-commuting random variables such that for all $\tau,\varepsilon\in[0,1]$ the following conditions are satisfied:
\begin{itemize}
 \item[(1)] $\varPsi_{\tau,\varepsilon;t}$ is Gaussian with mean zero and covariance $G_{\tau,\varepsilon;t}$ for all $t\in[0,1]$,
 \item[(2)] $\varPsi_{\tau,\varepsilon;t}\in C^\infty(\bR^2,\sB(\mathscr{H}))$ for all $t\in(0,1]$,
 \item[(3)] $\varPsi_{\tau,\varepsilon;t}$ is essentially constant at spatial scales smaller than $t\in(0,1]$, 
 \item[(4)] $\varPsi_{\tau,\varepsilon;t}$ captures the behavior of $\varPsi_{\tau,\varepsilon}$ at spatial scales larger than $t\in(0,1]$,
 \item[(5)] $t\mapsto \varPsi_{\tau,\varepsilon;t}$ has independent increments, in the sense that increments corresponding to disjoint scale intervals have vanishing mixed covariances and satisfy the corresponding Fermionic Wick factorization.
\end{itemize}

Observe that condition (1) implies, in particular, that $\mathrm{Law}(\varPsi_{\tau,\varepsilon;t=0})=\mathrm{Law}(\varPsi_{\tau,\varepsilon})$ and $\varPsi_{\tau,\varepsilon;t=1}=0$. A~scale decomposition satisfying conditions (1)--(4) could be defined by the formula $\varPsi_{\tau,\varepsilon;t}=\vartheta_t\ast \varPsi_{\tau,\varepsilon}$ for all $t\in[0,1]$. However, with this definition, the process $t\mapsto \varPsi_{\tau,\varepsilon;t}$ does not have independent increments, a property that plays a crucial role in the whole construction.

In order to define a scale decomposition of the free field satisfying all conditions (1)--(5) we follow the strategy proposed in~\cite{DFG25}. To this end, we use the Fermionic white noise
\begin{equation}
 \xi\in\sS'(\bR^2\times[0,1],\sB(\mathscr H))^\bFF,
\end{equation}
which is a centered anti-commuting Gaussian field whose covariance is supported on the diagonal in $(\bR^2\times[0,1])^2$. We refer to Sec.~\ref{sec:free_field_decomposition} for its construction. The algebra $\sF$ contains the random variables $\langle\xi,f\rangle$ with $f\in L^2(\bR^2\times[0,1])^\bFF$.

We also introduce propagators $G^\pm_{\tau,\varepsilon;t}$ such that for all $t\in(0,1]$ the Fourier transform of $G^\pm_{\tau,\varepsilon;t}$ is supported at frequencies $\lesssim 1/t$ and $G^+_{\tau,\varepsilon;t}\ast G^-_{\tau,\varepsilon;t}:= -\partial_t G_{\tau,\varepsilon;t}$. The scale decomposition $[0,1]\mapsto \varPsi_{\tau,\varepsilon;t}\equiv(\bar\varPsi_{\tau,\varepsilon;t},\ubar \varPsi_{\tau,\varepsilon;t})$ of the free field $\varPsi_{\tau,\varepsilon}\equiv(\bar\varPsi_{\tau,\varepsilon},\ubar\varPsi_{\tau,\varepsilon})$ satisfying the conditions (1)--(5) is defined by the formulas
\begin{equation}
 \bar\varPsi_{\tau,\varepsilon;t}
 :=
 \int_{\bT_\tau^2\times [t,1]} 
 G^{-}_{\tau,\varepsilon;s}(\Cdot-y)\,\xi(\rd y,\rd s),
 \qquad
 \ubar\varPsi_{\tau,\varepsilon;t}
 :=
 \int_{\bT_\tau^2\times [t,1]} 
 G^{+}_{\tau,\varepsilon;s}(\Cdot-y)\,\xi(\rd y,\rd s).
\end{equation}
The free field is defined to be $\varPsi_{\tau,\varepsilon}:=\varPsi_{\tau,\varepsilon;0}$.

For $t,s\in[0,1]$ such that $s\leq t$ the increment $\varPsi_{\tau,\varepsilon;t,s}$ of the scale decomposition of the free field is defined by $\varPsi_{\tau,\varepsilon;t,s}:=\varPsi_{\tau,\varepsilon;s}-\varPsi_{\tau,\varepsilon;t}$ and the algebra $\sF_{t,s}$ is defined as the Banach subalgebra of $\sF\subset\sB(\mathscr{H})$ generated by $\xi(f)$, where $f\in L^2(\bR^2\times[0,1])$ is s.t. $\supp\,f\subset\bR^2\times [s,t]$. Note that $\varPsi_{\tau,\varepsilon;t,s}$ is supported in Fourier space in a shell with radii of order $1/t$ and $1/s$ and depends only on the white noise in the time interval $[s,t]$. Hence, $\langle\varPsi_{\tau,\varepsilon;t,s},\varphi\rangle\in\sF_{t,s}$ for all $\varphi\in\sS(\bR^2)^\bFF$. 

Let us verify explicitly the independent-increment property. Let
$0\leq s_1<t_1<s_2<t_2\leq1$. By construction,
$\varPsi_{\tau,\varepsilon;t_i,s_i}$ is obtained by integrating the
white noise $\xi$ over the scale interval $[s_i,t_i]$. Hence, for all
$\phi,\eta\in\sS(\bR^2)^\bFF$,
\begin{equation}
 \fE\left(
 \langle\varPsi_{\tau,\varepsilon;t_1,s_1},\phi\rangle
 \langle\varPsi_{\tau,\varepsilon;t_2,s_2},\eta\rangle
 \right)
 =0.
\end{equation}
Indeed, the covariance of the white noise is proportional to the
$L^2([0,1]\times\bR^2)$ pairing, while the functions against which
$\xi$ is integrated in the two increments have disjoint supports in
the scale variable. Since the fields are jointly Fermionic Gaussian,
the Fermionic Wick rule then implies that all joint moments of
increments associated with disjoint scale intervals contain no
contractions between different intervals. This is precisely the
independent-increment property in~(5); see Remark~\ref{rem:Wick} for
the Wick formula.

Observe also that for all $t\in(0,1]$ the free field can be decomposed into the low and high frequency part $\varPsi_{\tau,\varepsilon}=\varPsi_{\tau,\varepsilon;1,t}+\varPsi_{\tau,\varepsilon;t,0}$, where $\varPsi_{\tau,\varepsilon;1,t}$ contains only the Fourier modes with frequencies in the ball with radius of order $1/t$ and $\varPsi_{\tau,\varepsilon;t,0}$ contains only the Fourier modes with frequencies outside the ball with radius of order $1/t$. Finally, let us mention that in the non-commutative probability space $(\sF,\fE)$ of the white noise it is possible to construct the conditional expectation $\fE_t\,:\,\sF\to\sF_{1,t}$ given the algebra $\sF_{1,t}$ that has the usual properties of a conditional expectation; see~\cite[Appendix~A]{DFG25} and Def.~\ref{dfn:conditional_exp}. In intuitive terms, $\fE_t$ integrates out the high frequency part $\varPsi_{\tau,\varepsilon;t,0}$ of the free field $\varPsi_{\tau,\varepsilon}=\varPsi_{\tau,\varepsilon;1,t}+\varPsi_{\tau,\varepsilon;t,0}$ and acts trivially on the low frequency part $\varPsi_{\tau,\varepsilon;1,t}$.

Now let us introduce the effective potential and the Polchinski flow equation that play a central role in our construction. To this end, fix $t\in(0,1]$ and let $F$ be a functional of polynomial type such that $F(\varPsi_{\tau,\varepsilon})=F({\varPsi_{\tau,\varepsilon;1,t}})$, that is $F$ depends only on the low frequency part of the free field. For such a functional, the formula~\eqref{eq:measure_intro2} for the interacting measure $\mu_{\tau,\varepsilon}$ can be rewritten in the form
\begin{equation}\label{eq:measure_intro3}
 \mu_{\tau,\varepsilon}(F) 
 \,\propto\,
 \fE{\fE_t}\big( F({\varPsi_{\tau,\varepsilon;1,t}}) \re^{U_{\tau,\varepsilon}(\varPsi_{\tau,\varepsilon})}\big)
 =
 \fE\big( F({\varPsi_{\tau,\varepsilon;1,t}}) {\fE_t}\re^{U_{\tau,\varepsilon}({\varPsi_{\tau,\varepsilon;1,t}}+{\varPsi_{\tau,\varepsilon;t,0}})}\big),
\end{equation}
where we used the tower property of the conditional expectation as well as the fact that $F({\varPsi_{\tau,\varepsilon;1,t}})\in\sF_{1,t}$. 

We call a functional $U_{\tau,\varepsilon;t}$ an effective potential at
spatial scale $t\in(0,1]$ if it satisfies (see Def.~\ref{dfn:effective_potential})
\begin{equation}\label{eq:effective_potential_intro}
 \exp(U_{\tau,\varepsilon;t}(\phi))
 =
 \fE\exp(U_{\tau,\varepsilon}(\phi+\varPsi_{\tau,\varepsilon;t,0}))
\end{equation}
for all $\phi\in C^\infty(\bT^2_\tau)$ valued in the odd part of some Grassmann algebra $\sG$ independent of the white noise $\xi(f)$. We formulate this as a condition rather than defining $U_{\tau,\varepsilon;t}$ by taking a logarithm, since the right hand
side of~\eqref{eq:effective_potential_intro} is not a priori known to be invertible.
The existence of such an effective potential will follow from the construction of a solution to the Polchinski equation in Sec.~\ref{sec:relation_polchinski}, together with the results of Sec.~\ref{sec:effective}.
The choice of the Grassmann algebra $\sG$ does not play an important role. However, $\sG$ has to have at least as many generators as the Grassmann algebra of the free field with cutoffs  $\tau,\varepsilon\in(0,1]$. For this reason, it is convenient to assume
that $\sG$ is infinite-dimensional.

The relation with the conditional expectation introduced above is particularly transparent after evaluating the effective potential at the low frequency field. Using the extension of $\fE_t$ to the auxiliary Grassmann algebra and the defining properties of the conditional expectation,~\eqref{eq:effective_potential_intro} gives
\begin{equation}\label{eq:effective_potential_conditional_intro}
 \fE_t\exp(U_{\tau,\varepsilon}(\varPsi_{\tau,\varepsilon}))
 =
 \exp\big(
 U_{\tau,\varepsilon;t}(\varPsi_{\tau,\varepsilon;1,t})
 \big).
\end{equation}
Thus $U_{\tau,\varepsilon;t}$ describes the interaction obtained after integrating out the high frequency field $\varPsi_{\tau,\varepsilon;t,0}$ while keeping $\varPsi_{\tau,\varepsilon;1,t}$ fixed.

We conclude that
\begin{equation}
 \mu_{\tau,\varepsilon}(F) 
 \,\propto\,
 \fE\big( F({\varPsi_{\tau,\varepsilon;1,t}}) 
 \re^{U_{\tau,\varepsilon;t}({\varPsi_{\tau,\varepsilon;1,t}})}\big)
\end{equation}
for all functionals $F$ as above. Note that the RHS of the above equality involves only the low frequency part ${\varPsi_{\tau,\varepsilon;1,t}}$ of the free field, which is smooth for all $\tau,\varepsilon\in[0,1]$. Consequently, the above formula remains meaningful without cutoffs provided an effective potential $U_{\tau,\varepsilon;t}$ is well-defined. This suggests that the family of effective potentials $(U_{\tau,\varepsilon;t})_{t\in(0,1]}$ determines the measure $\mu_{\tau,\varepsilon}$ completely and one can prove the existence of the limit $\tau,\varepsilon\searrow0$ of $\mu_{\tau,\varepsilon}$ by studying the limit $\tau,\varepsilon\searrow0$ of $(U_{\tau,\varepsilon;t})_{t\in(0,1]}$. 

The relation between the effective potential and the generating functional can be directly seen  from the Fermionic analogue of the Cameron--Martin formula. Indeed, completing the square in the Gaussian Berezin integral and using the translational invariance of the Berezin integral gives
\begin{equation}
 \mu_{\tau,\varepsilon}(\exp(\langle\Cdot,\varphi\rangle))
 = 
 \exp(-\langle\varphi,G_{\tau,\varepsilon}\ast\varphi\rangle/2)
 \frac{
  \fE\exp\big(
   U_{\tau,\varepsilon}(
    \varPsi_{\tau,\varepsilon}
    +G_{\tau,\varepsilon}\ast\varphi
   )\big)
    }{\fE\exp\big(
   U_{\tau,\varepsilon}(\varPsi_{\tau,\varepsilon})
  \big)}.
\end{equation}
Since $\varPsi_{\tau,\varepsilon}=\varPsi_{\tau,\varepsilon;1,0}$, the definition of the effective potential at scale $t=1$ yields
\begin{equation}\label{eq:generating_effective_potential_intro}
 \mu_{\tau,\varepsilon}(\exp(\langle\Cdot,\varphi\rangle))
 =
 \exp(
  -\langle\varphi,G_{\tau,\varepsilon}\ast\varphi\rangle/2
  +U_{\tau,\varepsilon;1}(G_{\tau,\varepsilon}\ast\varphi)
  -U_{\tau,\varepsilon;1}(0)
 ).
\end{equation}
This provides a direct link between the effective potential at spatial scale $t=1$ and the generating functional of the Schwinger functions. See Lemma~\ref{lem:generating_functional} for the detailed argument.

The upshot is that in order to construct the Schwinger functions and prove their convergence as $\tau,\varepsilon\searrow0$ it is enough to construct the family of effective potentials and prove its convergence as $\tau,\varepsilon\searrow0$.

The basic idea of the renormalization group theory is to derive an equation relating the effective potentials at different scales and then use it to construct them. In the discrete renormalization group method one studies such an equation at discrete values of the scale parameter. The present  construction uses instead the Polchinski flow equation.

Let us sketch the derivation of this equation. Set $Z_{\tau,\varepsilon;t} := \exp(U_{\tau,\varepsilon;t})$. By the definition of the effective potential,
\begin{equation}
 Z_{\tau,\varepsilon;t}(\phi)
 =
 \fE\exp\big(  U_{\tau,\varepsilon}(
   \varPsi_{\tau,\varepsilon;t,0}+\phi) 
 \big).
\end{equation}
Moreover, the scale decomposition
\begin{equation}
 \varPsi_{\tau,\varepsilon;t,0}
 =
 \varPsi_{\tau,\varepsilon;s,0}
 +
 \varPsi_{\tau,\varepsilon;t,s}, \qquad 0\leq s\leq t,
\end{equation}
together with the independent-increment property and the tower property
of the conditional expectation gives
\begin{equation}
 Z_{\tau,\varepsilon;t}(\phi)
 =
 \fE Z_{\tau,\varepsilon;s}
 \big(
  \varPsi_{\tau,\varepsilon;t,s}+\phi
 \big).
\end{equation}
Differentiating with respect to~$t$ gives the linear equation (see Lemma~\ref{lem:expected_F_incremenet})
\begin{equation}
 \partial_tZ_{\tau,\varepsilon;t}(\phi)
 = -\frac12
 \langle
  \rD_\phi^2Z_{\tau,\varepsilon;t}(\phi),
  \dot G_{\tau,\varepsilon;t}(\Cdot-\Cdot)
 \rangle_\tau.
\end{equation}

It remains to rewrite this equation in terms of the effective potential. For an even functional $U$, the Fermionic chain rule gives
\begin{equation}
 \rD_\phi^2\exp(U(\phi))
 =
 \exp(U(\phi))\rD_\phi^2U(\phi)
 -
 \exp(U(\phi))
 \big(
  \rD_\phi U(\phi)\otimes\rD_\phi U(\phi)
 \big),
\end{equation}
where the minus sign in the second term is due to the odd Grassmann parity of $\rD_\phi U$. Consequently, the linear equation above is equivalent to the Polchinski equation
\begin{equation}
 \partial_tU_{\tau,\varepsilon;t}(\phi)
 =
 -\frac12
 \langle
  \rD_\phi^2U_{\tau,\varepsilon;t}(\phi),
  \dot G_{\tau,\varepsilon;t}(\Cdot-\Cdot)
 \rangle_\tau
 +
 \frac12
 \langle
  \rD_\phi U_{\tau,\varepsilon;t}(\phi),
  \dot G_{\varepsilon;t}\ast
  \rD_\phi U_{\tau,\varepsilon;t}(\phi)
 \rangle_\tau.
\end{equation}
See Lemma~\ref{lem:polchinski} for the rigorous derivation.

The mild form of the resulting Polchinski equation is
\begin{multline}\label{eq:flow_mild_intro}
 U_{\tau,\varepsilon;t}(\phi) = \fE U_{\tau,\varepsilon}(\varPsi_{\tau,\varepsilon;t,0}+\phi) 
 \\
 + \frac{1}{2} \int_0^t \fE\langle \rD_\phi U_{\tau,\varepsilon;s}(\varPsi_{\tau,\varepsilon;t,s}+\phi),\dot G_{\varepsilon;s}\ast \rD_\phi U_{\tau,\varepsilon;s}(\varPsi_{\tau,\varepsilon;t,s}+\phi)\rangle_\tau\,\rd s
\end{multline} 
for all $t\in[0,1]$ and $\phi\in C^\infty(\bT^2_\tau)^\bFF$ valued in the odd part of $\sG$. In the above formula, $\dot G_{\varepsilon;t}:=\partial_t G_{\varepsilon;t}$ and $\rD_\phi F$ stands for the functional derivative of $F$. 
The equivalence between this integral equation and the differential Polchinski equation is proved in Lemma~\ref{lem:polchinski2}.

The flow equation~\eqref{eq:flow_mild_intro} is well-posed provided $\tau,\varepsilon\in(0,1]$. Note that the first term on the RHS of~\eqref{eq:flow_mild_intro} involves the original potential $U_{\tau,\varepsilon}$ defined by~\eqref{eq:potential_intro}, which depends on the parameters $g_{\tau,\varepsilon},r_{\tau,\varepsilon}\in\bR$. Recall that these parameters are not given. Our task is to find $(g_{\tau,\varepsilon},r_{\tau,\varepsilon})_{\tau,\varepsilon\in(0,1]}$ such that the Schwinger functions converge as $\tau,\varepsilon\searrow0$. This is the so-called renormalization problem.

A comment on functionals is in order. We identify functionals $U$ with collections of kernels $(U^{m,\sigma})_{m\in\bN_0,\sigma\in\bFF^m}$ such that
\begin{equation}\label{eq:functional_U_intro}
 U(\phi)=U^0 + \sum_{m\in\bN_+} \sum_{\sigma\in\bFF^m}\langle U^{m,\sigma},\phi^{\sigma_1}\otimes\ldots\otimes\phi^{\sigma_m}\rangle_\tau
\end{equation}
for all $\phi\in C^\infty(\bT_\tau^2)^\bFF$ valued in the odd part of $\sG$, where $U^0\in\bC$ and $U^{m,\sigma}\in\sS'(\bT_\tau^{2m})$ is antisymmetric for $m\in\bN_+,\sigma\in\bFF^m$. 
The representation of functionals on $\bT_\tau^2$ is formalized in Sec.~\ref{sec:torus_functionals}. The infinite-volume spaces of kernels used below are introduced in Sec.~\ref{sec:topology}.

We interpret the flow equation as an equation for the hierarchy of kernels $(U^{m,\sigma})_{m\in\bN_0,\sigma\in\bFF^m}$. In order to study the infinite volume limit it is useful to consider instead the hierarchy of kernels $(V^{m,\sigma})_{m\in\bN_+,\sigma\in\bFF^m}$ such that $V^{m,\sigma}\in \sS'(\bR^{2m})$ for all $m\in\bN_+,\sigma\in\bFF^m$ and such that $V^{m,\sigma}$ is related to $U^{m,\sigma}$ by the equality
\begin{equation}\label{eq:U_V_intro}
 \langle U^{m,\sigma},\phi_1\otimes\ldots\otimes\phi_m\rangle_\tau 
 =
 \langle V^{m,\sigma},\chi_\tau\phi_1 \otimes \phi_2 \otimes\ldots\otimes\phi_m\rangle
\end{equation}
for all $m\in\bN_+,\sigma\in\bFF^m$ and $\phi_1,\ldots,\phi_m\in C^\infty(\bT^2_\tau)$, where $\chi_\tau\in C^\infty_\rc(\bR^2)$ is such that its periodization with period $1/\tau$ coincides with the constant function equal to one. The pairing $\langle\Cdot,\Cdot\rangle_\tau$ is the pairing between distributions and test functions on $\bT^{2m}_\tau$, whereas the pairing $\langle\Cdot,\Cdot\rangle$ is the pairing between distributions and test functions on $\bR^{2m}$. 
The pairing on the RHS of~\eqref{eq:U_V_intro} is understood in the extended sense of Remark~\ref{rem:sM_periodic_pairing}.
By the translational invariance of $V^{m,\sigma}$ and the periodicity of $\phi_1,\ldots,\phi_m$ the RHS of~\eqref{eq:U_V_intro} is independent of the choice of the function $\chi_\tau$. In informal terms, $U^{m,\sigma}$ is obtained by the periodization of $V^{m,\sigma}$ in $m-1$ variables. Because of the assumption that $V^{m,\sigma}$ is antisymmetric the choice of these $m-1$ variables is not important. Note that the hierarchy $(V^{m,\sigma})_{m\in\bN_+,\sigma\in\bFF^m}$ contains no information about $U^0$, which is redundant.
If the flow equation~\eqref{eq:flow_mild_intro} is satisfied modulo an additive constant, the constant can be chosen so that the equation holds exactly; see Sec.~\ref{sec:relation_polchinski}.

The advantage of working with the kernels $(V^{m,\sigma})_{m\in\bN_+,\sigma\in\bFF^m}$ rather than $(U^{m,\sigma})_{m\in\bN_0,\sigma\in\bFF^m}$ is that the kernels $(V^{m,\sigma})_{m\in\bN_+,\sigma\in\bFF^m}$ are not periodic and decay rapidly in directions transversal to the diagonal. Collections of kernels $(V^{m,\sigma})_{m\in\bN_+,\sigma\in\bFF^m}$ are identified with functionals $V$ by the formula 
\begin{equation}\label{eq:functional_V_intro}
 V[\varphi] = \sum_{m\in\bN_+}\sum_{\sigma\in\bFF^m} \langle V^{m,\sigma},\varphi^{\sigma_1}\otimes\ldots\otimes\varphi^{\sigma_m}\rangle
\end{equation}
for all $\varphi\in \sS(\bR^2)^\bFF$ valued in the odd part of $\sG$. In order to introduce a topology in the space of functionals $V\equiv(V^{m,\sigma})_{m\in\bN_+,\sigma\in\bFF^m}$ for all $m\in\bN_+$, $\sigma\in\bFF^m$ it is convenient to make the following ansatz
\begin{equation}\label{eq:functionals_derivatives}
 \langle V^{m,\sigma},\varphi_1\otimes\ldots\otimes\varphi_m\rangle =
 \sum_{a\in\bA^m}\langle V^{m,a,\sigma},\partial^{a_1}\varphi_1 \otimes\ldots\otimes\partial^{a_m}\varphi_m\rangle
\end{equation}
for all $\varphi_1,\ldots,\varphi_m\in\sS(\bR^2)$, where
$\bA:=\{a\in\mathbb N_0^2\,|\,|a|\leq2\}$
is the set of spatial multi-indices and for all $m\in\bN_+$, $a\in\bA^m$, $\sigma\in\bFF^m$ the distributions $V^{m,a,\sigma}\in\sS'(\bR^{2m})$ are measures such that the following norm
\begin{equation} 
 \|V^{m,a,\sigma}\|_{\sM^m} = \esssup_{x_1\in\bR^2}\int |V^{m,a,\sigma}(x_1,\rd x_2,\ldots,\rd x_m)|
\end{equation}
is finite. We also introduce the notation $\fJ\varphi=(\partial^a\varphi^\sigma)_{a\in\bA,\sigma\in\bFF}\in\sS(\bR^2)^{\bA\times\bFF}$ for the jet extension of the function $\varphi\in\sS(\bR^2)^\bFF$. 

The representation~\eqref{eq:functionals_derivatives} of
$V^{m,\sigma}$ in terms of $(V^{m,a,\sigma})_{a\in\bA^m}$ is not unique, since derivatives can be transferred between the kernels and the fields by integration by parts. For example, suppressing the other variables and indices,
\begin{equation}
 \langle \partial_i H,\varphi\rangle
 +
 \langle H,\partial_i\varphi\rangle
 =0.
\end{equation}
Thus replacing $V^0$ by $V^0+\partial_iH$ and $V^{e_i}$ by
$V^{e_i}+H$ does not change the represented functional. The same observation applies separately in each argument of an $m$-linear kernel.

The fact that the potential of the Gross--Neveu model
$U_{\tau,\varepsilon}$ does not involve derivatives suggests that one could initially set $V^{m,a,\sigma}=0$ unless $a=0$. We shall instead use the above freedom to choose representatives adapted to the renormalization problem. In particular, the localization and remainder
maps introduced in Sec.~\ref{sec:loc_ren} separate the local relevant and marginal terms from Taylor remainders while preserving the represented functional; see Lemma~\ref{lem:R_L_identity}.

In what follows, we identify functionals $V$ with collections of kernels $(V^{m,a,\sigma})_{m\in\bN_+,a\in\bA^m,\sigma\in\bFF^m}$ such that Eqs.~\eqref{eq:functional_V_intro} and~\eqref{eq:functionals_derivatives} are satisfied. We also use the notation $V^m$ for the collection of kernels $(V^{m,a,\sigma})_{a\in\bA^m,\sigma\in\bFF^m}$. For $\tau,\varepsilon\in[0,1]$ and $t,s\in(0,1]$ we introduce the following maps in the space of functionals
\begin{equation}\label{eq:A_B_def_intro}
\begin{gathered}
 (\fA_{\tau,\varepsilon;t,s}V)[\varphi]:= \Pi_\circ V[\fJ\varPsi_{\tau,\varepsilon;t\vee s,s}+\varphi],
 \\
 \fB_{\varepsilon;s}(V)[\varphi]=\frac12\Pi_\circ\langle \rD_\varphi V[\varphi]\otimes\rD_\varphi V[\varphi], (\fJ\otimes\fJ)\dot G_{\varepsilon;s}(\Cdot-\Cdot)\rangle
\end{gathered} 
\end{equation}
where $\Pi_\circ V[\varphi]:=V[\varphi]-V[0]$ and $\varphi\in\sS(\bR^2)^{\bA\times\bFF}$ is an arbitrary Schwartz function valued in the odd part of $\sG$. 
The map $\fA_{\tau,\varepsilon;t,s}$ is introduced and studied in
Def.~\ref{dfn:A_map}, while $\fB_{\varepsilon;s}$ is treated
in Def.~\ref{dfn:B_map} and Sec.~\ref{sec:maps}. Eqs.~\eqref{eq:A_B_def_intro} should be interpreted as hierarchies of equations for the kernels $V\equiv(V^{m,a,\sigma})_{m\in\bN_+,a\in\bA^m,\sigma\in\bFF^m}$. Note that the map $\Pi_\circ$  ensures that the functionals $(\fA_{\tau,\varepsilon;t,s}V)[\varphi]$ and $(\fB_{\varepsilon;s}(V))[\varphi]$ vanish for $\varphi=0$ and are of the form~\eqref{eq:functional_V_intro}. We stress that $\fA_{\tau,\varepsilon;t,s}V$ is a functional taking values in $\sF_{t,s}\otimes_{\mathrm{alg}}\sG$. We denote by $\fA^{m,a,\sigma}_{\tau,\varepsilon;t,s}V$ and $\fB^{m,a,\sigma}_{\varepsilon;s}(V)$ the kernels of the functionals $\fA_{\tau,\varepsilon;t,s}V$ and $\fB_{\varepsilon;s}(V)$. The symbols $\fA^{m}_{\tau,\varepsilon;t,s}V$ and $\fB^{m}_{\varepsilon;s}(V)$ denote collections of kernels.

Now we shall rewrite the flow equation~\eqref{eq:flow_mild_intro} in a form that is suitable for studying the infinite volume limit $\tau\searrow0$. To this end, let the functional $V_{\tau,\varepsilon}\equiv(V^{m,a,\sigma}_{\tau,\varepsilon})_{m\in\bN_+,a\in\bA^m,\sigma\in\bFF^m}$ be related to the potential~\eqref{eq:potential_intro} of the Gross--Neveu model by Eqs.~\eqref{eq:functional_U_intro},~\eqref{eq:U_V_intro} and~\eqref{eq:functionals_derivatives} and let $t\mapsto V_{\tau,\varepsilon;t}\equiv(V^{m,a,\sigma}_{\tau,\varepsilon;t})_{m\in\bN_+,a\in\bA^m,\sigma\in\bFF^m}$ be a solution of the following flow equation
\begin{equation}\label{eq:flow_2_intro}
 V_{\tau,\varepsilon;t}[\fJ\varphi] =
 \fE\fA_{\tau,\varepsilon;t,0} V_{\tau,\varepsilon}[\fJ\varphi]
 +
 \int_0^t \fE\fA_{\tau,\varepsilon;t,s}\fB_{\varepsilon;s}(V_{\tau,\varepsilon;s})[\fJ\varphi]\,\rd s
\end{equation} 
for all $\varphi\in \sS(\bR^2)^{\bFF}$ valued in the odd part of $\sG$. The above equation is written using the notation introduced in the previous paragraph. For $t\in[0,1]$ we define the functional \mbox{$U_{\tau,\varepsilon;t}\equiv(U^{m,\sigma}_{\tau,\varepsilon;t})_{m\in\bN_0,\sigma\in\bFF^m}$} related to $V_{\tau,\varepsilon;t}$ by Eqs.~\eqref{eq:U_V_intro} and~\eqref{eq:functionals_derivatives}. This leaves $U^0_{\tau,\varepsilon;t}\in\bC$ unspecified but for a suitable choice of $t\mapsto U^0_{\tau,\varepsilon;t}$ the function $(t,\phi)\mapsto U_{\tau,\varepsilon;t}(\phi)$ satisfies the flow equation~\eqref{eq:flow_mild_intro}. Hence, we are led to study the flow equation~\eqref{eq:flow_2_intro}. The advantage of~\eqref{eq:flow_2_intro} over~\eqref{eq:flow_mild_intro} is that it is well-posed in the limit $\tau\searrow0$.

Let us stress that a solution $t\mapsto V_{\tau,\varepsilon;t}\equiv(V^{m,a,\sigma}_{\tau,\varepsilon;t})_{m\in\bN_+,a\in\bA^m,\sigma\in\bFF^m}$ of~\eqref{eq:flow_2_intro} is not unique. Observe that a solution of the following equation
\begin{equation}\label{eq:flow_3_intro}
 V_{\tau,\varepsilon;t}[\varphi] =
 \fE\fA_{\tau,\varepsilon;t,0} V_{\tau,\varepsilon}[\varphi]
 +
 \int_0^t \fE\fA_{\tau,\varepsilon;t,s}\fB_{\varepsilon;s}(V_{\tau,\varepsilon;s})[\varphi]\,\rd s
\end{equation} 
for all $\varphi\in C^\infty(\bR^2)^{\bA\times\bFF}$ valued in the odd part of $\sG$ is also a solution of~\eqref{eq:flow_2_intro}. However, the reverse implication is not true because of the non-uniqueness of the representation~\eqref{eq:functionals_derivatives} of $V^{m,\sigma}$ in terms of $(V^{m,a,\sigma})_{a\in\bA^m}$. We shall take advantage of this non-uniqueness in order to solve the renormalization problem. In particular, a solution of~\eqref{eq:flow_2_intro} that we are going to construct will not satisfy~\eqref{eq:flow_3_intro}.

Let us now investigate the renormalization problem. To this end, we have to study the limit $\varepsilon\searrow0$. First note  that since for $\varepsilon=0$ the increment $\varPsi_{\tau,\varepsilon;t,0}$ is not a function on $\bR^2$ but only a distribution and the functional $V_{\tau,\varepsilon}$ involves pointwise products the first term on the RHS of~\eqref{eq:flow_2_intro} becomes singular in the limit $\varepsilon\searrow0$. We hope that the second term becomes singular in this limit as well and for specific choices of the parameters $g_{\tau,\varepsilon},r_{\tau,\varepsilon}\in\bR$ of the original potential the singularities of both terms cancel out. We would like to rewrite~\eqref{eq:flow_2_intro} in a form that exhibits the above-mentioned cancellation of singularities. To this end, we have to come up with an appropriate space of functionals. The norm in this space should control the dependence of the norm $\|V^{m,a,\sigma}_{\tau,\varepsilon;t}\|_{\sM^m}$ of the kernels of the functional $V_{\tau,\varepsilon;t}$ on $t\in(0,1]$ and $m\in\bN_+$, $a\in\bA^m$, $\sigma\in\bFF^m$. Let us first concentrate on the dependence on $t\in(0,1]$. Using the fact that the perturbative corrections to the kernels $V^{m,a,\sigma}_{\tau,\varepsilon;t}$ can be expressed in terms of Feynman diagrams by a naive power-counting argument for all $m\in\bN_+$, $a\in\bA^m$, $\sigma\in\bFF^m$ the following bound
\begin{equation}\label{eq:kernel_bound_intro}
 \|V^{m,a,\sigma}_{\tau,\varepsilon;t}\|_{\sM^m} \lesssim t^{m/2+|a|-2}
\end{equation}
should hold uniformly in $\tau,\varepsilon\in[0,1]$ and $t\in(0,1]$ up to logarithmic corrections, which we ignore for the moment. 

The effective potentials considered below are invariant under the internal $U(N)$ symmetry introduced in Sec.~\ref{sec:symmetries}. In particular, taking $\mathfrak u=e^{i\theta}I_N$, the fields transform
as
\begin{equation}
 \bar\varphi\longmapsto e^{-i\theta}\bar\varphi,
 \qquad \ubar\varphi\longmapsto e^{i\theta}\ubar\varphi.
\end{equation}
Hence a non-zero kernel can contain only terms with the same number of $\bar\varphi$ and $\ubar\varphi$. Hence, its fermionic degree $m$ is even. Consequently, the scalar kernels vanish for odd $m$ and we can restrict attention to kernels $V^{m,a,\sigma}_{\tau,\varepsilon;t}$ with $m\in2\bN_+$.

In view of the scaling estimate \eqref{eq:kernel_bound_intro}, for $\varepsilon=0$ the norm of the kernel 
\begin{equation}
 (V^{2,0,\sigma}_{\tau,\varepsilon;t})_{\sigma\in\bFF^2},
\end{equation}
may diverge polynomially and the norms of the kernels
\begin{equation}
 (V^{4,0,\sigma}_{\tau,\varepsilon;t})_{\sigma\in\bFF^4},
 \qquad
 (V^{2,a,\sigma}_{\tau,\varepsilon;t})_{a\in\bA^2,|a|=1,\sigma\in\bFF^2}
\end{equation}
may diverge logarithmically as $t\searrow0$. We call the above kernels relevant and marginal, respectively. The remaining kernels vanish in the limit $t\searrow0$ and are called irrelevant. Note that the functional $V_{\tau,\varepsilon}$, which is closely related to the original potential $U_{\tau,\varepsilon}$ defined by~\eqref{eq:potential_intro}, contains only relevant and marginal kernels.

Using the properties of the maps $\fA_{\tau,\varepsilon;t,s}$ and $\fB_{\varepsilon;s}$ as well as bound~\eqref{eq:kernel_bound_intro}, one shows that for all $m\in\bN_+$, $a\in\bA^m$, $\sigma\in\bFF^m$ the bound
\begin{equation}
 \|\fE\fA^{m,a,\sigma}_{\tau,\varepsilon;t,s}\fB_{\varepsilon;s}(V_{\tau,\varepsilon;s})\|_{\sM^m}
 \lesssim s^{m/2+|a|-3}
\end{equation}
holds uniformly in $\tau,\varepsilon\in[0,1]$ and $t,s\in(0,1]$ up to logarithmic corrections. By the Minkowski inequality this implies the bound
\begin{multline}\label{eq:I_B_A_intro}
 \big\|\textstyle\int_0^t \fE\fA^{m,a,\sigma}_{\tau,\varepsilon;t,s}\fB_{\varepsilon;s}(V_{s})\, \rd s\big\|_{\sM^m}\leq
 \textstyle\int_0^t \big\|\fE\fA^{m,a,\sigma}_{\tau,\varepsilon;t,s}\fB_{\varepsilon;s}(V_{s})\big\|_{\sM^m}\, \rd s
 \\
 \lesssim \textstyle\int_0^t s^{m/2+|a|-3} \, \rd s \lesssim t^{m/2+|a|-2}
\end{multline}
uniformly in $\tau,\varepsilon\in[0,1]$ and $t\in(0,1]$ up to logarithmic corrections unless $m=4$ and $a=0$, or $m=2$ and $|a|\leq1$. Consequently, for the irrelevant kernels, the bound~\eqref{eq:kernel_bound_intro} is consistent with the flow equation~\eqref{eq:flow_2_intro}. The relevant kernels require special treatment. Indeed, if $m=4$ and $a=0$, or $m=2$ and $|a|\leq1$, then the last estimate in~\eqref{eq:I_B_A_intro} is false.

In order to address the above-mentioned problem with the estimates for the relevant kernels we make the following ansatz
\begin{equation}\label{eq:ansatz_intro}
 V_{\tau,\varepsilon;t}= U(1/g_{\tau,\varepsilon;t},r_{\tau,\varepsilon;t},z_{\tau,\varepsilon;t})+W_{\tau,\varepsilon;t},
\end{equation}
where $g_{\tau,\varepsilon;t}\in(0,\infty)$, $r_{\tau,\varepsilon;t},z_{\tau,\varepsilon;t}\in\bR$ are some parameters, $W_{\tau,\varepsilon;t}$ is a functional and, for all $g,r,z\in\bR$, the functional $U(g,r,z)$ is defined by the equality
\begin{equation}
 U(g,r,z)[\psi] := \int_{\bR^2} 
 \big(
 g\,(\bar\psi(x)\cdot\ubar\psi(x))^2
 + r\,\bar\psi(x)\cdot\ubar\psi(x)
 + z\,\bar\psi(x)\cdot(\slashed{\partial}\ubar\psi)(x)
 \big)\,\rd x
\end{equation}
for all $\psi=(\psi^\sigma)_{\sigma\in\bFF}\in\sS(\bR^2)^\bFF$ valued in the odd part of $\sG$. Since the functional $U(g,r,z)$ is a linear combination of quartic and quadratic terms it holds that $V_{\tau,\varepsilon;t}^{m,a,\sigma}=W_{\tau,\varepsilon;t}^{m,a,\sigma}$ unless $m=4$ and $a=0$, or $m=2$ and $|a|\leq1$. 

In order to specify the relation between $V_{\tau,\varepsilon;t}$ and $(g_{\tau,\varepsilon;t},r_{\tau,\varepsilon;t},z_{\tau,\varepsilon;t},W_{\tau,\varepsilon;t})$ we have to first introduce the notion of the local part of a kernel and the remainder. The local part of a collection of kernels $V^{4}_{\tau,\varepsilon;t}=(V^{4,a,\sigma})_{a\in\bA^4,\sigma\in\bFF^4}$ is defined by the equality
\begin{equation}
 \fL V^4\,(\ubar\psi\cdot\bar\psi)^2 := \sum_{\sigma\in\bFF^4}\psi^{\sigma_1}\ldots\psi^{\sigma_4}\int_{\bR^6} V^{4,0,\sigma}(x_1,\rd x_2,\rd x_3,\rd x_4) 
\end{equation}
for all Grassmann-valued vectors $\psi\equiv(\psi^\sigma)_{\sigma\in\bFF}\equiv(\bar\psi^{\alpha,\varsigma},\ubar\psi^{\alpha,\varsigma})_{\alpha\in\{1,2\},\varsigma\in\{1,\ldots,N\}}$. Note that the RHS of the above formula does not depend on $x_1\in\bR^2$ because of the translational invariance of the kernels of the functional $V_{\tau,\varepsilon;t}$. Taylor expansion in the relative coordinates allows one to separate the local quartic contribution from a remainder carrying additional spatial derivatives. More concretely, for a collection of kernels $V^4=(V^{4,a,\sigma})_{a\in\bA^4,\sigma\in\bFF^4}$ satisfying the symmetries of the Gross--Neveu model, there exists a collection $\fR V^4=((\fR V^4)^{a,\sigma})_{a\in\bA^4,\sigma\in\bFF^4}$
such that $(\fR V^4)^{a,\sigma}=0$ if $a=0$ and
\begin{multline}\label{eq:L_R_4_intro}
 \fL V^4\,\int_{\bR^2}(\ubar\psi(x)\cdot\bar\psi(x))^2\,\rd x
 +
 \sum_{a\in\bA^4}\sum_{\sigma\in\bFF^4}\langle(\fR V^4)^{a,\sigma},
 \partial^{a_1}\psi^{\sigma_1}\otimes
 \ldots
 \otimes \partial^{a_4}\psi^{\sigma_4}
 \rangle
 \\
 =
 \sum_{a\in\bA^4}\sum_{\sigma\in\bFF^4}\langle V^{4,a,\sigma},
 \partial^{a_1}\psi^{\sigma_1}\otimes
 \ldots
 \otimes \partial^{a_4}\psi^{\sigma_4}
 \rangle
\end{multline}
for all $\psi\in\sS(\bR^2)^\bFF$ valued in the odd part of $\sG$.

Similarly, given a collection of kernels $V^2=(V^{2,a,\sigma})_{a\in\bA^2,\sigma\in\bFF^2}$ that possesses certain symmetries there exist numbers $\fL V^2,\fL_\partial V^2\in\bR$ and a collection of kernels $\fR V^2=((\fR V^2)^{a,\sigma})_{a\in\bA^2,\sigma\in\bFF^2}$
such that $(\fR V^2)^{a,\sigma}=0$ if $|a|\leq1$ and
\begin{multline}\label{eq:L_R_2_intro}
 \fL V^2\,\int_{\bR^2}\bar\psi(x)\cdot\ubar\psi(x)\,\rd x
 +
 \fL_\partial V^2\,\int_{\bR^2}\bar\psi(x)\cdot(\slashed{\partial}\ubar\psi)(x)\,\rd x
 \\
 +
 \sum_{a\in\bA^2}\sum_{\sigma\in\bFF^2}\langle(\fR V^2)^{a,\sigma},
 \partial^{a_1}\psi^{\sigma_1}\otimes
 \partial^{a_2}\psi^{\sigma_2}
 \rangle
 =
 \sum_{a\in\bA^2}\sum_{\sigma\in\bFF^2}\langle V^{2,a,\sigma},
 \partial^{a_1}\psi^{\sigma_1}\otimes
 \partial^{a_2}\psi^{\sigma_2}
 \rangle
\end{multline}
for all $\psi\in\sS(\bR^2)^\bFF$ valued in the odd part of $\sG$. 

The identities~\eqref{eq:L_R_4_intro} and~\eqref{eq:L_R_2_intro} are consequences of the Taylor theorem and the form of the local terms
\begin{equation}
 \int_{\bR^2} 
 \bar\psi(x)\cdot\ubar\psi(x)\,\rd x,
 \qquad
 \int_{\bR^2} 
 \bar\psi(x)\cdot(\slashed{\partial}\ubar\psi)(x)\,\rd x,
 \qquad
 \int_{\bR^2} 
 (\bar\psi(x)\cdot\ubar\psi(x))^2\,\rd x
\end{equation}
appearing in these identities is dictated by the symmetries of the Gross--Neveu model together with compatibility with the flow of charge and flavor (see Remark~\ref{rem:symmetries_flow_charge}). 
The maps $\fL$, $\fL_\partial$ and $\fR$ are defined precisely in
Def.~\ref{dfn:L} and~\ref{dfn:R_map}, and the above
decompositions are proved in Lemma~\ref{lem:R_L_identity}; see Sec.~\ref{sec:loc_ren}.

The parameter $1/g_{\tau,\varepsilon;t}\in(0,\infty)$ is called the effective coupling constant at spatial scale $t\in(0,1]$ and is defined so that it satisfies the equation
\begin{equation}\label{eq:flow_2_intro_g}
 1/g_{\tau,\varepsilon;t}
 =
 \fL \fE\fA^4_{\tau,\varepsilon;t,0} V_{\tau,\varepsilon}
 +
 \int_0^t \fL\fE\fA^4_{\tau,\varepsilon;1,s}\fB_{\varepsilon;s}(V_{\tau,\varepsilon;s})\,\rd s.
\end{equation} 
The parameter $r_{\tau,\varepsilon;t}$ satisfies the equation
\begin{multline}\label{eq:flow_2_intro_r}
 r_{\tau,\varepsilon;t}
 =
 \fL\fE\fA^2_{\tau,\varepsilon;t,0} V_{\tau,\varepsilon}
 +
 \fL\fE\fA^2_{\tau,\varepsilon;1,t} U(1/g_{\tau,\varepsilon;0}-1/g_{\tau,\varepsilon;t},0,0)
 \\
 +
 \int_0^t
 \fL\fE\fA^2_{\tau,\varepsilon;1,s}\fB_{\varepsilon;s}(V_{\tau,\varepsilon;s})\,\rd s.
\end{multline}
The parameter $z_{\tau,\varepsilon;t}$ satisfies the equation
\begin{equation}\label{eq:flow_2_intro_z}
 z_{\tau,\varepsilon;t}
 =
 \int_0^t
 \fL_\partial\fE\fA^2_{\tau,\varepsilon;1,s}\fB_{\varepsilon;s}(V_{\tau,\varepsilon;s})\,\rd s.
\end{equation}

Finally, the functional $W_{\tau,\varepsilon;t}$ satisfies the equations
\begin{equation}\label{eq:flow_2_intro_W1}
  W^m_{\tau,\varepsilon;t} 
  =
 \int_0^t \fE\fA^m_{\tau,\varepsilon;t,s}\fB_{\varepsilon;s}(V_{\tau,\varepsilon;s})\,\rd s,
 \qquad
 m\in\bN_+\setminus\{2,4\},
\end{equation}
\begin{equation}\label{eq:flow_2_intro_W2}
 W^m_{\tau,\varepsilon;t} 
 =
 \int_0^t \fR\fE\fA^m_{\tau,\varepsilon;1,s}\fB_{\varepsilon;s}(V_{\tau,\varepsilon;s})\,\rd s
 -
 \fE\fC^m_{\tau,\varepsilon;1,t} W_{\tau,\varepsilon;t},
 \qquad m\in\{2,4\},
\end{equation} 
where
$\fC_{\tau,\varepsilon;1,t}V
:=
\fA_{\tau,\varepsilon;1,t}V-V$.

The correction term in~\eqref{eq:flow_2_intro_W2} has a simple interpretation. Since 
$\fA_{\tau,\varepsilon;1,t}
=
\mathrm{Id}+\fC_{\tau,\varepsilon;1,t}$,
Eq.~\eqref{eq:flow_2_intro_W2} is equivalent to
\begin{equation}\label{eq:flow_2_intro_W2_transported}
 \fE\fA^m_{\tau,\varepsilon;1,t}W_{\tau,\varepsilon;t}
 =
 \int_0^t 
 \fR\fE\fA^m_{\tau,\varepsilon;1,s}
 \fB_{\varepsilon;s}(V_{\tau,\varepsilon;s})\,\rd s,
 \qquad  m\in\{2,4\}.
\end{equation}
Thus, in the degrees in which local terms have to be removed, all contributions are first moved to the common reference scale $1$ before applying $\fR$. Rewriting this identity in terms of $W_{\tau,\varepsilon;t}$ produces precisely the correction
$-\fE\fC^m_{\tau,\varepsilon;1,t}W_{\tau,\varepsilon;t}$
in~\eqref{eq:flow_2_intro_W2}. The complementary local terms are encoded in the equations for $g_{\tau,\varepsilon;t}$, $r_{\tau,\varepsilon;t}$ and $z_{\tau,\varepsilon;t}$. The fact that these local and remainder equations reconstruct the original flow follows from Lemma~\ref{lem:R_L_identity} and is established in Theorem~\ref{thm:polchinski}.
Finally, $\fC^m_{\tau,\varepsilon;1,t}W_{\tau,\varepsilon;t}$ depends only on $(W^{m+k}_{\tau,\varepsilon;t})_{k\in\bN_+}$. Hence, the RHS
of~\eqref{eq:flow_2_intro_W2} does not depend on
$W^m_{\tau,\varepsilon;t}$.

The functional $V_{\tau,\varepsilon;s}$ appearing in
Eqs.~\eqref{eq:flow_2_intro_g}, \eqref{eq:flow_2_intro_r},
\eqref{eq:flow_2_intro_z}, \eqref{eq:flow_2_intro_W1},
and~\eqref{eq:flow_2_intro_W2} is reconstructed from
$(g_{\tau,\varepsilon;s},r_{\tau,\varepsilon;s},
z_{\tau,\varepsilon;s},W_{\tau,\varepsilon;s})$
through the ansatz~\eqref{eq:ansatz_intro}. These equations therefore
form a closed system for
$(g_{\tau,\varepsilon;\Cdot},r_{\tau,\varepsilon;\Cdot},
z_{\tau,\varepsilon;\Cdot},W_{\tau,\varepsilon;\Cdot})$.
Conversely, using the decompositions~\eqref{eq:L_R_4_intro} and~\eqref{eq:L_R_2_intro}, one verifies that any solution of this
closed system determines through~\eqref{eq:ansatz_intro} a functional
$V_{\tau,\varepsilon;\Cdot}$ satisfying the flow equation~\eqref{eq:flow_2_intro}.
Notice that the reconstructed functional solves~\eqref{eq:flow_2_intro}, but its chosen kernel representative does not satisfy the stronger equation~\eqref{eq:flow_3_intro} for arbitrary independent jet variables. This is precisely where we exploit the non-uniqueness of the representation~\eqref{eq:functionals_derivatives}.

The power-counting argument leading to~\eqref{eq:I_B_A_intro}, together with the fact that $\fR$ annihilates the relevant and marginal kernels of
$\fE\fA_{\tau,\varepsilon;1,s}
\fB_{\varepsilon;s}(V_{\tau,\varepsilon;s})$,
suggests that the integrands in~\eqref{eq:flow_2_intro_W1} and~\eqref{eq:flow_2_intro_W2} are absolutely integrable at $s=0$ also when $\varepsilon=0$. 

Now let us study more closely~\eqref{eq:flow_2_intro_g} for the effective coupling constant. Noting that
\begin{equation}
 \partial_t g_{\tau,\varepsilon;t} 
 =
 - g_{\tau,\varepsilon;t}^2\,\partial_t (1/g_{\tau,\varepsilon;t})
 =
 - g_{\tau,\varepsilon;t}^2\,\fL\fE\fA^4_{\tau,\varepsilon;1,t}\fB_{\varepsilon;t}(V_{\tau,\varepsilon;t})
\end{equation}
and using the fact that $\fL \fE\fA^4_{\tau,\varepsilon;t,0} V_{\tau,\varepsilon}=1/g_{\tau,\varepsilon}$, where $g_{\tau,\varepsilon}$ is the parameter of the potential~\eqref{eq:potential_intro}, we obtain
\begin{equation}
 g_{\tau,\varepsilon;t}=g_{\tau,\varepsilon} 
 -
 \int_0^t g_{\tau,\varepsilon;s}^2\,\fL\fE\fA^4_{\tau,\varepsilon;1,s}\fB_{\varepsilon;s}(V_{\tau,\varepsilon;s})\,\rd s.
\end{equation}
We fix the parameter $g_{\tau,\varepsilon}=g_{\tau,\varepsilon;t=0}$ implicitly by imposing the following renormalization condition
\begin{equation}\label{eq:ren_condition_g_intro}
 g_{\tau,\varepsilon;t=1}=1/\lambda,
\end{equation}
where $\lambda\in(0,1]$ is assumed to be sufficiently small. Taking into account the above boundary condition, we rewrite the equation for $g_{\tau,\varepsilon;t}$ as
\begin{equation}\label{eq:flow_g_intro_ren}
 g_{\tau,\varepsilon;t}
 =
 1/\lambda
 +
 \int_t^1
 g_{\tau,\varepsilon;s}^2\,
 \fL\fE\fA^4_{\tau,\varepsilon;1,s}
 \fB_{\varepsilon;s}(V_{\tau,\varepsilon;s})\,\rd s.
\end{equation}

The dominant term in the integrand in~\eqref{eq:flow_g_intro_ren} is obtained by replacing $V_{\tau,\varepsilon;s}$ by $U(1/g_{\tau,\varepsilon;s},0,0)$. The contributions involving $r_{\tau,\varepsilon;s}$, $z_{\tau,\varepsilon;s}$ or $W_{\tau,\varepsilon;s}$ are sub-leading as $s\searrow0$ under the bounds imposed below in $\sY_{\tau,\varepsilon}$. On the other hand, by the bilinearity of $\fB_{\varepsilon;s}$ and the definition of $U(1/g,0,0)$, we have
\begin{equation}
 g_{\tau,\varepsilon;s}^2\,
 \fL\fE\fA^4_{\tau,\varepsilon;1,s}
 \fB_{\varepsilon;s}(U(1/g_{\tau,\varepsilon;s},0,0))
 =
 \fL\fE\fA^4_{\tau,\varepsilon;1,s}
 \fB_{\varepsilon;s}(U(1,0,0)).
\end{equation}
For $\tau=\varepsilon=0$, a direct one-loop computation gives
\begin{equation}
 \lim_{s\searrow0}
 s\,\fL\fE\fA^4_{0,0;1,s}
 \fB_{0;s}(U(1,0,0))
 =
 \beta_2,
 \qquad
 \beta_2=2(N-1)/\pi.
\end{equation}
The corresponding quantitative estimates, including the control of the
sub-leading terms and of the cutoff dependence, are proved in
Lemmas~\ref{lem:contraction_f_h} and~\ref{lem:contraction_g}.
Consequently, neglecting sub-leading corrections, the asymptotic
behavior of $t\mapsto1/g_{0,0;t}$ as $t\searrow0$ agrees with that of
$t\mapsto\lambda_t$, where $\lambda_t$ was introduced in
Def.~\ref{dfn:lambda}.

Note that the asymptotic behavior of the effective coupling constant $1/g_{\tau,\varepsilon;t}$ is easiest to determine by studying~\eqref{eq:flow_g_intro_ren} for its inverse $g_{\tau,\varepsilon;t}$. This is one of the reasons why we chose $g_{\tau,\varepsilon;t}$ rather than the effective coupling constant $1/g_{\tau,\varepsilon;t}$ to be one of the fundamental variables. Note that by the above argument the effective coupling constant $1/g_{\tau,\varepsilon;t}$ with $\varepsilon=0$ vanishes logarithmically as $t\searrow0$. Thus, at high energy the behavior of Schwinger functions of the Gross--Neveu model without the UV cutoff should not differ much from the behavior of the Schwinger functions of the free theory. In the literature, this property is called asymptotic freedom and plays a crucial role in our construction. Let us remark that our scale decomposition $G_{\varepsilon;t}$ of the propagator $G_{\varepsilon}$ satisfies the condition $\dot G_{\varepsilon;t}=0$ for $t\in(0,\varepsilon]$. As a result, $\fB_{\varepsilon;t}(\Cdot)=0$ and $g_{\tau,\varepsilon;t}=g_{\tau,\varepsilon;\varepsilon}$ for $t\in[0,\varepsilon]$.
For fixed $\varepsilon\in(0,1]$, the flow of the effective coupling
constant halts at the scale $t=\varepsilon$, so~$1/g_{\tau,\varepsilon;t}$ does not tend to zero as $t\searrow0$.
However, as the UV cutoff is removed,
\begin{equation}
 1/g_{\tau,\varepsilon} = 1/g_{\tau,\varepsilon;\varepsilon}
\end{equation}
vanishes logarithmically as $\varepsilon\searrow0$. In particular, the model with the UV cutoff is not asymptotically free (the flow stops at the scale of the cutoff) and the value of the parameter $g_{\tau,\varepsilon}=g_{\tau,\varepsilon;0}=g_{\tau,\varepsilon;\varepsilon}$ of the potential~\eqref{eq:potential_intro} fixed by the renormalization condition~\eqref{eq:ren_condition_g_intro} diverges logarithmically as $\varepsilon\searrow0$. As a side remark, we mention that the assumption $\lambda>0$ is necessary for the asymptotic freedom of the model without the UV cutoff as for $\lambda<0$ the function $t\mapsto\lambda_t$ blows up at the scale $t=\re^{1/(\lambda\beta_2)}$, which is called in the literature the Landau pole.

Because of asymptotic freedom, the logarithmic corrections, which were ignored in the bound~\eqref{eq:kernel_bound_intro}, improve the asymptotic behavior of $t\mapsto\|V^{m,a,\sigma}_{\tau,\varepsilon;t}\|_{\sM^m}$ at $t=0$. Actually, taking into account the logarithmic vanishing of the effective coupling constant, a direct inspection of perturbative corrections yields the following bound
\begin{equation}
 \|V^{m,a,\sigma}_t\|_{\sM^m} \lesssim \lambda_{\varepsilon\vee t}^{(m/2-1)\vee1}\, t^{m/2+|a|-2}
\end{equation}
uniformly in $\tau,\varepsilon\in[0,1]$ and $t\in(0,1]$ for all $m\in\bN_+$, $a\in\bA^m$, $\sigma\in\bFF^m$. Using the above bound one shows that 
\begin{equation}
 \|\fE\fA^{m,a,\sigma}_{\tau,\varepsilon;t,s}\fB_{\varepsilon;s}(V_{\tau,\varepsilon;s})\|_{\sM^m}
 \lesssim \lambda_s^{(m/2-1)\vee2}\, s^{m/2+|a|-3}
\end{equation}
uniformly in $\tau,\varepsilon\in[0,1]$ and $t,s\in(0,1]$. Hence, using Lemma~\ref{lem:bounds_relevant_irrelevant}~(A) we obtain
\begin{multline}\label{eq:A_B_irrelevant_intro}
 \big\|\textstyle\int_0^t \fE\fA^{m,a,\sigma}_{\tau,\varepsilon;u,s}\fB_{\varepsilon;s}(V_{\tau,\varepsilon;s})\, \rd s\big\|_{\sM^m} \lesssim \int_0^t \lambda_s^{(m/2-1)\vee2}\,s^{m/2+|a|-3} \, \rd s 
 \\
 \lesssim \lambda_t^{(m/2-1)\vee2}\, t^{m/2+|a|-2}
\end{multline}
uniformly in $\tau,\varepsilon\in[0,1]$ and $t,u\in(0,1]$ unless $m=4$ and $a=0$, or $m=2$ and $|a|\leq1$. By Lemma~\ref{lem:bounds_relevant_irrelevant}~(D) applied with $\rho=1$ we have
\begin{equation}\label{eq:A_B_marginal_intro}
 \big\|\textstyle\int_0^t \fE\fA^{m,a,\sigma}_{\tau,\varepsilon;u,s}\fB_{\varepsilon;s}(V_{\tau,\varepsilon;s})\, \rd s\big\|_{\sM^m} \lesssim \int_0^t \lambda_s^{2}\,s^{-1} \, \rd s \lesssim \lambda_t
\end{equation}
uniformly in $\tau,\varepsilon\in[0,1]$ and $t,u\in(0,1]$ for $m=4$ and $a=0$, or $m=2$ and $|a|=1$.

If $m=2$ and $a=0$, then for $\varepsilon=0$ the function $s\mapsto\fE\fA^{m,a,\sigma}_{\tau,\varepsilon;u,s}\fB_{\varepsilon;s}(V_{s})$ is not integrable at $s=0$. However, by Lemma~\ref{lem:bounds_relevant_irrelevant}~(C) applied with $\rho=2$, $\varrho=-1$ and $\lambda\in(0,1]$ small enough we have
\begin{equation}\label{eq:A_B_relevant_intro}
 \big\|\textstyle\int_t^1 \fE\fA^{m,a,\sigma}_{\tau,\varepsilon;u,s}\fB_{\varepsilon;s}(V_{\tau,\varepsilon;s})\, \rd s\big\|_{\sM^m} \lesssim \int_t^1 \lambda_s^{2}\,s^{-2} \, \rd s \lesssim \lambda_t^2\, t^{-1}
\end{equation}
uniformly in $\tau,\varepsilon\in[0,1]$ and $t,u\in(0,1]$ for $m=2$ and $a=0$. The bound~\eqref{eq:A_B_marginal_intro} shows that the integral in~\eqref{eq:flow_2_intro_z} converges at $s=0$ also for $\varepsilon=0$, even though $z_{\tau,\varepsilon;t}$ is only of order $\lambda_t$. The corresponding rigorous estimate is proved in Lemma~\ref{lem:contraction_z}. On the other hand, \eqref{eq:flow_2_intro_r} for $r_{\tau,\varepsilon;t}$ is not expected to be well-posed for $\varepsilon=0$. To address the problem with~\eqref{eq:flow_2_intro_r} we impose the following renormalization condition
\begin{equation}\label{eq:ren_condition_r_intro}
 r_{\tau,\varepsilon;t=1}=0
\end{equation}
that fixes implicitly the mass counterterm $r_{\tau,\varepsilon}=r_{\tau,\varepsilon;t=0}$ and rewrite this equation so that it involves the integral appearing on the LHS of the bound~\eqref{eq:A_B_relevant_intro}. Using Eq.~\eqref{eq:flow_2_intro_r} we obtain the following equation
\begin{equation}\label{eq:flow_r_intro_ren}
 r_{\tau,\varepsilon;t}
 =
 -\fL\fE\fA^2_{\tau,\varepsilon;1,t} U(1/g_{\tau,\varepsilon;t},0,0)
 -
 \int_t^1
 \fL\fE\fA^2_{\tau,\varepsilon;1,s}\fB_{\varepsilon;s}(V_{\tau,\varepsilon;s})\,\rd s.
\end{equation}
Let us remark that since we do not impose any non-trivial renormalization condition for the parameter $z_{\tau,\varepsilon;t}$ introducing this parameter in the ansatz~\eqref{eq:ansatz_intro} is not necessary and is done mainly because it is convenient when proving estimates.

The claim is that the system of equations \eqref{eq:flow_g_intro_ren}, \eqref{eq:flow_r_intro_ren},
\eqref{eq:flow_2_intro_z}, \eqref{eq:flow_2_intro_W1}, 
\eqref{eq:flow_2_intro_W2} for the tuple $(g_{\tau,\varepsilon;\Cdot},r_{\tau,\varepsilon;\Cdot},z_{\tau,\varepsilon;\Cdot},W_{\tau,\varepsilon;\Cdot})$ remains well-posed in the limit $\tau,\varepsilon\searrow0$. Since we would like to construct a solution of this system of equations using the contraction principle, we first have to find an appropriate complete metric space of tuples $(g_{\tau,\varepsilon;\Cdot},r_{\tau,\varepsilon;\Cdot},z_{\tau,\varepsilon;\Cdot},W_{\tau,\varepsilon;\Cdot})$. We will define this metric space as a closed subset of a certain Banach space. The norm on the Banach space has to control the dependence of $g_{\tau,\varepsilon;t},r_{\tau,\varepsilon;t},z_{\tau,\varepsilon;t}\in\bR$ on $t\in(0,1]$ as well as the dependence of $\|W^{m,a,\sigma}_{\tau,\varepsilon;t}\|_{\sM^m}$ on $t\in(0,1]$ and $m\in\bN_+$, $a\in\bA^m$, $\sigma\in\bFF^m$. To this end, for $\alpha,\beta\in[1,\infty)$ and $\gamma\in\bR$ we introduce the Banach space of functionals $s\mapsto W_s=(W^{m,a,\sigma}_s)_{m\in\bN_+,a\in\bA^m,\sigma\in\bFF^m}$ depending on the scale parameter $s\in(0,1]$ equipped with the following norm
\begin{equation}
 \|W_\Cdot\|_{\sV^{\alpha,\beta;\gamma}}: = \sup_{m\in\bN_+} \alpha^m m^\beta\, \|W^m_\Cdot\|_{\sV^{m;\gamma}},
\end{equation}
where
\begin{equation}\label{eq:sV_m_intro}
 \|W^m_\Cdot\|_{\sV^{m;\gamma}}:=\sum_{a\in\bA^m}\sum_{\sigma\in\bFF^m}\sup_{s\in(0,1]} \lambda_s^{-\rho_{\gamma,\kappa}(m)} \, s^{2-m/2-|a|}\, \|w^m_s W^{m,a,\sigma}_s\|_{\sM^m}.
\end{equation}

The weight $w_s^m$ depends only on the relative coordinates and grows
stretched exponentially. It is needed to establish stretched
exponential decay of truncated correlations as well as to prove the
desired estimates for the map $\fR$ introduced above.
The weights $w_s^m$ are defined and studied in Sec.~\ref{sec:weights}. The spaces $\sV^{\alpha,\beta;\gamma}$ are defined rigorously in Sec.~\ref{sec:topology}; see in particular Def.~\ref{dfn:V_space}.

The power of $\lambda_s$ in~\eqref{eq:sV_m_intro} keeps track of the perturbative order of the kernels in the running coupling. More precisely, the perturbative estimates above suggest that the kernels of the full effective functional $V_{\tau,\varepsilon;s}$ behave as
\begin{equation}
 \lambda_s^{(m/2-1)\vee1},
\end{equation}
whereas the kernels of the remainder $W_{\tau,\varepsilon;s}$ behave as
\begin{equation}
 \lambda_s^{(m/2-1)\vee2};
\end{equation}
see in particular~\eqref{eq:A_B_irrelevant_intro}. Thus the parameter $\gamma$ in~\eqref{eq:sV_m_intro} may be viewed as the perturbative order tracked by the norm: its natural value is $1$ for the full effective functional and $2$ for the remainder $W$. Instead of using the above powers, in the fixed point argument we use
\begin{equation}
 \rho_{\gamma,\kappa}(m):=\gamma+2\kappa m
\end{equation}
and allow a small loss from the natural values of $\gamma$. We take $\gamma=2-80\kappa$ for $W_{\tau,\varepsilon;\Cdot}$ and $\gamma=1-40\kappa$ for the full effective functional. In the presence of the UV cutoff, the flow of the running coupling constant stops below the scale $\varepsilon$ and therefore no longer has the small-scale decay required by this norm. We compensate for this by introducing a cutoff-adapted effective potential in which the quartic coupling is multiplied by $\theta_{\varepsilon;t}^2$, where
\begin{equation}
 \theta_{\varepsilon;t}
 =
 \frac{\lambda_t}{\lambda_{\varepsilon\vee t}}.
\end{equation}
This restores the required uniform small-scale bounds; see
Lemma~\ref{lem:contraction_V} and Corollary~\ref{cor:contraction}, in particular~\eqref{eq:cor_bound_2}. This modification is used only for the estimates and does not alter the flow, since $\fB_{\varepsilon;t}=0$ below the cutoff, whereas
$\theta_{\varepsilon;t}=1$ above it.

The usefulness of the norm $\|\Cdot\|_{\sV^{\alpha,\beta;\gamma}}$ comes from the following estimates
\begin{equation}\label{eq:bound_B0_intro}
 \big\|s\mapsto s\,\fB_{\varepsilon;s}(V_s)\big\|_{\sV^{\alpha,\beta-1;2\gamma}}\leq C\,\|s\mapsto V_s\|_{\sV^{\alpha,\beta;\gamma}}^2,
 \qquad \gamma\in\bR,
\end{equation}
and
\begin{equation}\label{eq:bound_I_intro}
 \big\|t\mapsto \textstyle\int_0^t \Pi_{>4}V_s/s\,\rd s\big\|_{\sV^{\alpha,\beta;\gamma}}\leq C\,\|s\mapsto V_s\|_{\sV^{\alpha,\beta-1;\gamma}},
 \qquad
 \gamma\in[0,\infty),
\end{equation}
valid for $\alpha\in[1,\infty)$, $\beta\in(2,\infty)$, where $C\in(0,\infty)$ is a universal constant and given a functional $V=(V^{m,a,\sigma})_{m\in\bN_+,a\in\bA^m,\sigma\in\bFF^m}$ we define the functional $\Pi_{>4}V$ by the equalities $(\Pi_{>4}V)^{m,a,\sigma}=0$ if $m\leq4$ and $(\Pi_{>4}V)^{m,a,\sigma}=V^{m,a,\sigma}$ if $m>4$. 
Rigorous versions of these estimates are proved in Lemmas~\ref{lem:B_map_bound} and~\ref{lem:I_bound} in Sec.~\ref{sec:maps}.
The presence of the operator $\Pi_{>4}$ in the second of the above estimates is related to the fact that the  relevant kernels of $s\mapsto \fB_{\varepsilon;s}(V_s)$ are not absolutely integrable at $s=0$ and need a special treatment. Observe also that since the scale dependent functional $s\mapsto V_s-\Pi_{>4}V_s$ has only finitely many non-zero kernels the parameters $\alpha$ and $\beta$ do not play any role when estimating its $\|\Cdot\|_{\sV^{\alpha,\beta;\gamma}}$ norm. The estimates~\eqref{eq:bound_B0_intro} and~\eqref{eq:bound_I_intro} imply in particular that
\begin{equation}\label{eq:bound_B_intro}
 \big\|t\mapsto \textstyle\int_0^t \Pi_{>4}\fB_{\varepsilon;s}(V_s)\,\rd s\big\|_{\sV^{\alpha,\beta;2\gamma}}\leq C\,\|s\mapsto V_s\|_{\sV^{\alpha,\beta;\gamma}}^2.
\end{equation}

The remarkable property of the above bound is the fact that the norms on both sides of this bound have the same parameter $\alpha$ and $\beta$. Recall that the RHS of~\eqref{eq:flow_2_intro_W1} involves $\fE\fA_{\tau,\varepsilon;t,s}\fB_{\varepsilon;s}(V_{\tau,\varepsilon;s})$. In particular, to be able to solve the above system in the space $\sV^{\alpha,\beta;\gamma}$ of functionals $W_{\tau,\varepsilon;\Cdot}$ the validity of an estimate of the form
\begin{equation}\label{eq:bound_A_intro}
 \|t\mapsto \,\textstyle\int_0^t \Pi_{>4}\fE\fA_{\tau,\varepsilon;t,s}V_s/s\,\rd s\|_{\sV^{\alpha,\beta;\gamma}} \leq C\,\|s\mapsto V_s\|_{\sV^{\alpha,\beta-1;\gamma}}
 \qquad\mathrm{(false)}
\end{equation}
with a universal constant $C\in(0,\infty)$ appears to be crucial. Indeed, the bounds~\eqref{eq:bound_B0_intro} and~\eqref{eq:bound_A_intro} imply the bound
\begin{equation}\label{eq:bound_A_B_intro}
 \big\|t\mapsto \textstyle\int_0^t \Pi_{>4}\fE\fA_{\tau,\varepsilon;t,s}\fB_{\varepsilon;s}(V_s)\,\rd s\big\|_{\sV^{\alpha,\beta;2\gamma}}\leq C\,\|s\mapsto V_s\|_{\sV^{\alpha,\beta;\gamma}}^2\qquad\mathrm{(false)}
\end{equation}
with a universal constant $C\in(0,\infty)$, which would allow us to control the $\|\Cdot\|_{\sV^{\alpha,\beta;\gamma}}$ norm of the RHS of~\eqref{eq:flow_2_intro_W1}.

Recall that the definition of the map $\fA_{\tau,\varepsilon;t,s}$ involves the increment $\varPsi_{\tau,\varepsilon;t,s}$ of the scale decomposition of the free field. Using the bound
\begin{equation}
 |\fE(\varPsi_{\tau,\varepsilon;t,s}^{\sigma_1}(x_1)\ldots\varPsi_{\tau,\varepsilon;t,s}^{\sigma_k}(x_k))|
 \leq
 \sup_{\sigma\in\bFF}\|\varPsi_{\tau,\varepsilon;t,s}^\sigma\|^k_\sC
\end{equation}
where $\|\phi\|_{\sC}:=\sup_{x\in\bR^2}\|\phi(x)\|_{\sB(\mathscr{H})}$, as well as a similar bound for $\partial^a \varPsi_{\tau,\varepsilon;t,s}^\sigma$ one shows, with some effort, that the bound~\eqref{eq:bound_A_intro} would be true for sufficiently small $\lambda\in(0,1]$ if the following bound
\begin{equation}\label{eq:intro_psi_estimate_wrong}
 \|\partial^a\varPsi_{\tau,\varepsilon;t,s}^\sigma\|_\sC \leq c\,(s^{-1/2-|a|}-t^{-1/2-|a|}),
 \qquad 0<s\leq t\leq 1,\qquad\mathrm{(false)}
\end{equation}
with a universal constant $c\in(0,\infty)$ were true. The above estimate would imply in particular that the scale decomposition of the free field is locally Lipschitz continuous. However, the scale decomposition of the free field should have similar regularity in the scale parameter, which plays the role of time, to the Grassmann Brownian motion, which is only H{\"o}lder continuous with exponent $1/2$; see Lemma~\ref{lem:Psi_bounds}. In fact, the following bound
\begin{equation}\label{eq:intro_psi_estimate}
 \|\partial^a\varPsi_{\tau,\varepsilon;t,s}^\sigma\|_\sC \leq c\,(s^{-1-2|a|}-t^{-1-2|a|})^{1/2},
 \qquad 0<s\leq t\leq 1,
\end{equation}
holds true with a universal constant $c\in(0,\infty)$. This indicates that the bound~\eqref{eq:bound_A_intro} is false. The above argument is obviously not conclusive. However, closer analysis reveals that it is unlikely that the estimate~\eqref{eq:bound_A_B_intro} is true. We refer the reader to~\cite{SW00} for related comments. As an aside, we mention that, as observed in~\cite{SW00}, the proof of estimates for fermionic correlations given in~\cite{BW88,BW99} is not correct precisely because the scale decomposition of the free field is not Lipschitz continuous. This suggests that, in a sense, the difficulty we face is related to the fact that the scale decomposition of the free field is only $1/2$ H{\"o}lder continuous. Note also that in the discrete renormalization group one studies scales $t,s\in(0,1]$ such that $s=t/L$ with $L\in(1,\infty)$ fixed. For scales $t,s\in(0,1]$ satisfying the above relation the estimates~\eqref{eq:intro_psi_estimate_wrong} and~\eqref{eq:intro_psi_estimate} with some constants $c\in(0,\infty)$ are equivalent.

One of the novel ideas of this work is the use of the following norm 
\begin{equation}\label{eq:sW_intro}
 \|W_\Cdot\|_{\sW^{\alpha,\beta;\gamma}_{\tau,\varepsilon}}:=\sup_{u\in[0,1]}\|s\mapsto\fA_{\tau,\varepsilon;u\vee s,s} W_s\|_{\sV^{\alpha,\beta;\gamma}}
\end{equation}
in the space of scale dependent functionals $s\mapsto W_s=(W^{m,a,\sigma}_s)_{m\in\bN_+,a\in\bA^m,\sigma\in\bFF^m}$. The corresponding Banach space and its completeness are established
in Def.~\ref{dfn:sW} and Lemma~\ref{lem:W_banach}. Since $\fA_{\tau,\varepsilon;s,s} V_s=V_s$, the norm $\|\Cdot\|_{\sW^{\alpha,\beta;\gamma}_{\tau,\varepsilon}}$ is stronger than $\|\Cdot\|_{\sV^{\alpha,\beta;\gamma}}$. We stress that the norm $\|\Cdot\|_{\sW^{\alpha,\beta;\gamma}_{\tau,\varepsilon}}$ depends on the cutoffs $\tau,\varepsilon\in(0,1]$. We also note that the functional $\fA_{\tau,\varepsilon;u\vee s,s} V_s$, which appears in the definition of the norm $\|\Cdot\|_{\sW^{\alpha,\beta;\gamma}_{\tau,\varepsilon}}$, takes values in $\sF_{u\vee s,s}\subset\sF\subset\sB(\mathscr{H})$. For this reason we are forced to work with kernels that are measures valued in a Banach space $\sB(\mathscr{H})$. The main advantage of the norm $\|\Cdot\|_{\sW^{\alpha,\beta;\gamma}_{\tau,\varepsilon}}$ over $\|\Cdot\|_{\sV^{\alpha,\beta;\gamma}}$ is the validity of the following estimate
\begin{equation}\label{eq:bound_A_B_2_intro}
 \big\|t\mapsto{\textstyle\int_0^t} \Pi_{>4}\fE\fA_{\tau,\varepsilon;t,s}\fB_{\varepsilon;s}(V_s)\,\rd s\big\|_{\sW^{\alpha,\beta;2\gamma}_{\tau,\varepsilon}}
 \leq
 C\,\|s\mapsto V_s\|_{\sW^{\alpha,\beta;\gamma}_{\tau,\varepsilon}}^2,
\end{equation}
where $C\in(0,\infty)$ is some universal constant. As we argued above, an analogous estimate is almost certainly false if the norm $\|\Cdot\|_{\sW^{\alpha,\beta;\gamma}_{\tau,\varepsilon}}$ is replaced by the more standard norm $\|\Cdot\|_{\sV^{\alpha,\beta;\gamma}}$. Let us stress that the lack of a norm in the space of scale dependent functionals for which the estimate of the form~\eqref{eq:bound_A_B_2_intro} holds true was the main obstacle in the construction of Grassmann measures of fermionic quantum field theories with the use of the Polchinski flow equation.

Since the estimate~\eqref{eq:bound_A_B_2_intro} is of crucial importance let us sketch the main ideas behind its proof. First, observe that in order to prove the estimate~\eqref{eq:bound_A_B_2_intro} we have to control the $\|\Cdot\|_{\sV^{\alpha,\beta;2\gamma}}$ norm of the functional $\fA_{\tau,\varepsilon;u,t} \fE\fA_{\tau,\varepsilon;t,s}\fB_{\varepsilon;s}(V_s)$ in terms of the $\|\Cdot\|_{\sV^{\alpha,\beta;\gamma}}$ norm of the functional $\fA^m_{\tau,\varepsilon;u,s} V_s$. Next, note that for $0<s\leq t\leq u\leq 1$ it holds
\begin{equation}
\begin{gathered}
 \fA_{\tau,\varepsilon;u,s}V=\fA_{\tau,\varepsilon;u,t}\fA_{\tau,\varepsilon;t,s}V,
 \qquad
 \fA_{\tau,\varepsilon;u,s} \fB_{\varepsilon;s}(V) = \fB_{\varepsilon;s}(\fA_{\tau,\varepsilon;u,s}V),
 \\
 \fA_{\tau,\varepsilon;u,t}\fE\fA_{\tau,\varepsilon;t,s}V
 =
 \fA_{\tau,\varepsilon;u,t}\fE_t\fA_{\tau,\varepsilon;t,s}V
 = 
 \fE_t\fA_{\tau,\varepsilon;u,t}\fA_{\tau,\varepsilon;t,s}V
 =
 \fE_t\fA_{\tau,\varepsilon;u,s}V.
\end{gathered} 
\end{equation}
The identities in the first line above follow easily from the definition~\eqref{eq:A_B_def_intro} of the maps $\fA_{\tau,\varepsilon;t,s}$ and $\fB_{\varepsilon;s}$. The first equality in the second line is a consequence of the fact that the kernels of the functional $\fA_{\tau,\varepsilon;t,s}V$ take values in $\sF_{t,s}$ and on $\sF_{t,s}\subset\sF_{t,0}$ the expected value $\fE$ coincides with the conditional expectation $\fE_t$. To prove the second equality in the second line we use the fact that the conditional expectation $\fE_t$ acts trivially on $\sF_{u,t}\subset\sF_{1,t}$. As a result, the following identity
\begin{equation}\label{eq:A_B_identity}
 \fA_{\tau,\varepsilon;u,t}\fE\fA_{\tau,\varepsilon;t,s}\fB_{\varepsilon;s}(V) 
 = 
 \fE_t\fB_{\varepsilon;s}(\fA_{\tau,\varepsilon;u,s}V)
\end{equation}
holds true. We also note that $\|\fE_t V\|_{\sM^m}\leq \|V\|_{\sM^m}$ and
\begin{equation}\label{eq:sV_E_intro}
 \|u\mapsto\fE_t V_u\|_{\sV^{\alpha,\beta;\gamma}}
 \leq
 \|u\mapsto V_u\|_{\sV^{\alpha,\beta;\gamma}}.
\end{equation}

Consequently, ignoring the presence of the operator $\Pi_{>4}$ in the estimate~\eqref{eq:bound_B_intro} we obtain
\begin{multline}
 \big\|t\mapsto{\textstyle\int_0^t} \fE\fA_{\tau,\varepsilon;t,s}\fB_{\varepsilon;s}(V_s)\,\rd s\big\|_{\sW^{\alpha,\beta;2\gamma}_{\tau,\varepsilon}}
 =
 \sup_{u\in(0,1]}\big\|t\mapsto\textstyle\int_0^t \fA_{\tau,\varepsilon;u,t}\fE\fA_{\tau,\varepsilon;t,s}\fB_{\varepsilon;s}(V_s)\,\rd s\big\|_{\sV^{\alpha,\beta;2\gamma}}
 \\
 =
 \sup_{u\in(0,1]}
 \big\|t\mapsto{\textstyle\int_0^t} \fE_t\fB_{\varepsilon;s}(\fA_{\tau,\varepsilon;u,s} V_s)\,\rd s\big\|_{\sV^{\alpha,\beta;2\gamma}}
 \leq
 \sup_{u\in(0,1]}
 \big\|t\mapsto\textstyle\int_0^t \fB_{\varepsilon;s}(\fA_{\tau,\varepsilon;u,s} V_s)\,\rd s\big\|_{\sV^{\alpha,\beta;2\gamma}}
 \\
 \leq C\,
 \sup_{u\in(0,1]}\|s\mapsto \fA_{\tau,\varepsilon;u,s}V_s\|_{\sV^{\alpha,\beta;\gamma}}^2
 =
 C\,\|s\mapsto V_s\|_{\sW^{\alpha,\beta;\gamma}_{\tau,\varepsilon}}^2.
\end{multline}
The first and the last equality above follow from the definition of the norm $\|\Cdot\|_{\sW^{\alpha,\beta;\gamma}_{\tau,\varepsilon}}$. The second equality above follows from the identity~\eqref{eq:A_B_identity}. The first bound above follows from the estimate~\eqref{eq:sV_E_intro}. Finally, to prove the second bound above we used the estimate~\eqref{eq:bound_B_intro} with the operator $\Pi_{>4}$ omitted. Since the presence of the operator $\Pi_{>4}$ in the estimate~\eqref{eq:bound_B_intro} is crucial, the above reasoning is of course not correct. However, a more complicated argument (see
Lemmas~\ref{lem:B_map_bound} and~\ref{lem:I_bound}) not ignoring the presence of the operator $\Pi_{>4}$ in the estimate~\eqref{eq:bound_B_intro} allows us to establish rigorously the desired bound~\eqref{eq:bound_A_B_2_intro}.

We are ready to define the complete metric space in which we will solve the system of equations~\eqref{eq:flow_g_intro_ren}, \eqref{eq:flow_r_intro_ren},
\eqref{eq:flow_2_intro_z}, \eqref{eq:flow_2_intro_W1}, 
\eqref{eq:flow_2_intro_W2}. Let us first introduce the Banach space of tuples
\begin{equation}
 X_\Cdot\equiv (g_\Cdot,r_\Cdot,z_\Cdot,W_\Cdot)\in C((0,1],\bR)\times C((0,1],\bR)\times C((0,1],\bR)\times \sW^{8,4;2-80\kappa}_{\tau,\varepsilon}=: \sX_{\tau,\varepsilon},
\end{equation}
equipped with the norm
\begin{equation}
 \|X_\Cdot\|
 := 
 \sup_{t\in(0,1]}\lambda_t^{\kappa_1+1}\, |g_t|
 +
 \sup_{t\in(0,1]}\lambda_t^{\kappa_2-1}\, t\,|r_t|
 +
 \sup_{t\in(0,1]}\lambda_t^{\kappa_3-1}\,|z_t|
 +
 \|W_\Cdot\|_{\sW^{8,4;2-80\kappa}_{\tau,\varepsilon}},
\end{equation} 
where $\kappa_1,\kappa_2,\kappa_3\in(0,1)$ are certain small parameters.

Next, for $\tau,\varepsilon\in[0,1]$ we define the complete metric space
\begin{equation}
 \sY_{\tau,\varepsilon}:=\{X_\Cdot\in \sX_{\tau,\varepsilon}\,|\,\|X_\Cdot\|_{\sX_{\tau,\varepsilon}}\leq 1,~\lambda_{\varepsilon\vee t}\, g_t\geq\lambda^\kappa\text{ for all }t\in(0,1]\},
\end{equation}
and a map
\begin{equation}
 \sY_{\tau,\varepsilon}\ni X_\Cdot\mapsto \fX_{\tau,\varepsilon;t}(X_\Cdot):=(\mathbf{g}_{\tau,\varepsilon;t}(X_\Cdot),\mathbf{r}_{\tau,\varepsilon;t}(X_\Cdot),\mathbf{z}_{\tau,\varepsilon;t}(X_\Cdot),\fW_{\tau,\varepsilon;t}(X_\Cdot)),
 \qquad t\in(0,1],
\end{equation}
such that a fixed point of $\fX_{\tau,\varepsilon;\Cdot}$ coincides with a solution of the system of equations~\eqref{eq:flow_g_intro_ren}, \eqref{eq:flow_r_intro_ren},
\eqref{eq:flow_2_intro_z}, \eqref{eq:flow_2_intro_W1}, 
\eqref{eq:flow_2_intro_W2}. In particular, for all $\tau,\varepsilon\in[0,1]$ it holds that $\mathbf{g}_{\tau,\varepsilon;1}(X_\Cdot)=1/\lambda$ and $\mathbf{r}_{\tau,\varepsilon;1}(X_\Cdot)=0$. The fixed-point space and the contraction map are constructed rigorously in Sec.~\ref{sec:fixed_point}. We prove in Corollary~\ref{cor:contraction} that there exists $\lambda_\star\in(0,1]$ such that for all $\lambda\in(0,\lambda_\star]$ and $\tau,\varepsilon\in[0,1]$ the map $\fX_{\tau,\varepsilon;\Cdot}\,:\,\sY_{\tau,\varepsilon}\to\sY_{\tau,\varepsilon}$ is well-defined and is a contraction. We denote by $X_{\tau,\varepsilon;\Cdot}$ the fixed point of $\fX_{\tau,\varepsilon;\Cdot}$. We omit $\tau$ and $\varepsilon$ if $\tau=0$ and $\varepsilon=0$.

In order to establish convergence of the Schwinger functions as $\tau,\varepsilon\searrow0$ we have to control $X_{\Cdot}-X_{\tau,\varepsilon;\Cdot}$. To this end, we introduce a certain norm $\|\Cdot\|_{\tilde\sX_{\tau,\varepsilon}}$ in $\sX\cup\sX_{\tau,\varepsilon}$ weaker than $\|\Cdot\|_{\sX}\vee\|\Cdot\|_{\sX_{\tau,\varepsilon}}$ and prove that
\begin{equation}\label{eq:limit_X_intro}
 \lim_{\tau,\varepsilon\searrow0}\|X_{\Cdot}-X_{\tau,\varepsilon;\Cdot}\|_{\tilde\sX_{\tau,\varepsilon}}=0,
\end{equation}
which follows from Corollary~\ref{cor:contraction}.

The norm $\|(g_{\Cdot},r_{\Cdot},z_{\Cdot},W_{\Cdot})\|_{\tilde\sX_{\tau,\varepsilon}}$ is defined in terms of a norm $\|W_{\Cdot}\|_{\tilde\sW^{2,3;2-80\kappa}_{\tau,\varepsilon}}$ that has a form similar to the norm $\|W_{\Cdot}\|_{\sW^{2,3;2-80\kappa}_{\tau,\varepsilon}}$ but involves a different weight. The use of a different weight having some decay at infinity is crucial for the existence of the infinite volume limit $\tau\searrow0$. Let us also mention that the norm $\|\Cdot\|_{\tilde\sW^{2,3;2-80\kappa}_{\tau,\varepsilon}}$ is stronger than a certain norm $\|\Cdot\|_{\tilde\sV^{2,3;2-80\kappa}}$ independent of $\tau,\varepsilon\in(0,1]$ whose definition is very similar to the definition of the norm $\|\Cdot\|_{\sV^{2,3;2-80\kappa}}$ but involves a different weight.

Finally, let us discuss the convergence of the Schwinger functions as $\tau,\varepsilon\searrow0$. To this end, first recall that the fixed point $X_{\tau,\varepsilon;\Cdot}=(g_{\tau,\varepsilon;\Cdot},r_{\tau,\varepsilon;\Cdot},z_{\tau,\varepsilon;\Cdot},W_{\tau,\varepsilon;\Cdot})$ of the map $\fX_{\tau,\varepsilon;\Cdot}$ is related to a solution $V_{\tau,\varepsilon;\Cdot}= U(1/g_{\tau,\varepsilon;\Cdot},r_{\tau,\varepsilon;\Cdot},z_{\tau,\varepsilon;\Cdot})+W_{\tau,\varepsilon;\Cdot}$ of the flow equation~\eqref{eq:flow_2_intro} and $V_{\tau,\varepsilon;\Cdot}=(V^{m,a,\sigma}_{\tau,\varepsilon;\Cdot})_{m\in\bN_+,a\in\bA^m,\sigma\in\bFF^m}$ is related by Eqs.~\eqref{eq:functional_U_intro},~\eqref{eq:U_V_intro} and~\eqref{eq:functionals_derivatives} to a solution $U_{\tau,\varepsilon;\Cdot}$ of the flow equation~\eqref{eq:flow_mild_intro}. We conclude that the functional $U_{\tau,\varepsilon;t}$ is an effective potential at the scale $t\in(0,1]$, that is $U_{\tau,\varepsilon;t}$ satisfies~\eqref{eq:effective_potential_intro}, and $U_{\tau,\varepsilon;t=1}$ is related by~\eqref{eq:generating_effective_potential_intro} to the generating functional of the Schwinger functions. The convergence of the Schwinger functions as $\tau,\varepsilon\searrow0$ follows now from the existence of the limit~\eqref{eq:limit_X_intro}.

The rest of this article is organised as follows. In Sec.~\ref{sec:torus_functionals}, we define the spaces of distributions and functionals on a torus and introduce the notation used in the paper. Sec.~\ref{sec:gross-neveu} contains the definition of the Berezin integral and the  Grassmann measure $\mu_{\tau,\varepsilon}$ of the Gross--Neveu model with cutoffs. In Sec.~\ref{sec:covariance_decomposition}, we define the scale decomposition $G_{\tau,\varepsilon;t}$ of the propagator $G_{\tau,\varepsilon}$ and discuss its properties. In Sec.~\ref{sec:free_field_decomposition}, we introduce the filtered non-commutative probability space $(\sF,\fE)$ of the spacetime white noise $\xi$ and define the conditional expectation $\fE_t$ as well as the scale decomposition of the free field $\varPsi_{\tau,\varepsilon;t}$. In Sec.~\ref{sec:effective}, we present the equation for an effective potential and examine its relation to the Polchinski equation. In Sec.~\ref{sec:symmetries}, we discuss the symmetries of the Gross--Neveu model.

In Sec.~\ref{sec:weights}, we define the weights. Sec.~\ref{sec:topology} contains the definitions of the spaces $\sV^{m;\gamma}$, $\sV^{\alpha,\beta;\gamma}$, $\sW^{\alpha,\beta;\gamma}_{\tau,\varepsilon}$ of kernels and functionals in infinite volume. In Sec.~\ref{sec:loc_ren} we introduce the decomposition of kernels into the local part and the remainder and discuss its properties. In Sec.~\ref{sec:maps} we establish bounds for the maps $\fA_{\tau,\varepsilon;t,s}$ and $\fB_{\varepsilon;t}$. In Sec.~\ref{sec:fixed_point} we introduce the space $\sX_{\tau,\varepsilon}$ and the map $\fX_{\tau,\varepsilon}$ and prove that the map $\fX_{\tau,\varepsilon}$ is a contraction.

In Sec.~\ref{sec:relation_polchinski}, we prove that the fixed point of the map $\fX_{\tau,\varepsilon}$ yields a solution of the Polchinski equation. In Sec.~\ref{sec:convergence}, we show that the effective potential at unit scale is directly related to the generating functional of the Schwinger functions, and we prove Theorem~\ref{thm:main}. In Sec.~\ref{sec:interacting_field}, we construct a coupling between the interacting Gross--Neveu field and the free field, thereby completing the proof of Theorem~\ref{thm:main2}. The core arguments establishing the main results are developed in Sec.~\ref{sec:fixed_point}-\ref{sec:interacting_field}.

\section{Distributions and functionals on torus}\label{sec:torus_functionals}

In this section, we introduce the basic notation we use throughout the paper. In particular, we define the spaces of distributions and functionals on the torus.

\begin{dfn}
For $m\in\bN_+$ we denote by $\sS(\bR^{2m})$ and $\sS'(\bR^{2m})$ the space of Schwartz functions and distributions, respectively. We denote by $\langle V,\phi\rangle\in\bC$ the standard pairing between a distribution \mbox{$V\in\sS'(\bR^{2m})$} and a test function $\phi\in\sS(\bR^{2m})$. By $\ast$ we mean the convolution in $\sS'(\bR^{2m})$ whenever it is well-defined. We use the following convention for the Fourier transform
\begin{equation}
 (\fF f)(p):= \int_{\bR^{2m}} f(x)\,\re^{-\ri p\cdot x}\,\rd x,
 \qquad
 f(x) = \frac{1}{(2\pi)^{2m}}\int_{\bR^{2m}} (\fF f)(p)\,\re^{\ri p\cdot x}\,\rd p.
\end{equation}
\end{dfn}

\begin{dfn}\label{dfn:distributions_torus}
Let $\tau\in(0,1]$. We define $\bT^2_\tau := (\bR/\tau^{-1}\bZ)^2$. Occasionally, we identify $\bT^2_\tau$ with the set $(-1/(2\tau),1/(2\tau))^2\subset\bR^2$. For $m\in\bN_+$ we denote by $\sS'(\bT_\tau^{2m})$ the subspace of $\sS'(\bR^{2m})$ consisting of distributions that are $1/\tau$ periodic. Given $\varphi\in L^1(\bR^2)$ we define $\fP_\tau \varphi\in L^1(\bT_\tau^2)$ by the equality
\begin{equation}
 \fP_\tau \varphi =  \sum_{n\in\bZ^2} \varphi(\Cdot+n/\tau).
\end{equation}
For $V\in\sS'(\bT_\tau^{2m})$ and $\phi\in C^\infty(\bT_\tau^{2m})$ we write
\begin{equation}\label{eq:bracket_tau}
 \langle V,\phi\rangle_\tau:=
  \langle V,\chi_\tau^{\otimes m}\phi\rangle \in\bC,
\end{equation} 
where $\chi_\tau\in C^\infty_\rc(\bR^2)$ is such that $\fP_\tau\chi_\tau=1$.  By $\ast_\tau$ we mean the convolution in $\sS'(\bT_\tau^{2m})$, whenever it is well-defined. We use the following convention for the Fourier transform
\begin{equation}
 (\fF_\tau f)(p):= \int_{\bT^{2m}_\tau} f(x)\,\re^{-\ri p\cdot x}\,\rd x,
 \qquad
 f(x) = \tau^{2m}\sum_{p\in(2\pi\tau\bZ)^{2m}} (\fF_\tau f)(p)\,\re^{\ri p\cdot x}.
\end{equation}
\end{dfn}
\begin{rem}
Using periodicity of $V\in\sS'(\bT_\tau^{2m})$ and $\phi\in C^\infty(\bT_\tau^{2m})$ one shows that the bracket $\langle V,\phi\rangle_\tau$ does not depend on the choice of $\chi_\tau$.
\end{rem}
\begin{rem}
For $\varphi\in\sS(\bR^2)$ and $p\in(2\pi\tau\bZ)^2$ we have
$
 (\fF_\tau\fP_\tau\varphi)(p)=(\fF\varphi)(p)
$.
In particular,
$\supp\fF_\tau\fP_\tau\varphi\subset\supp\fF\varphi$.
\end{rem}
\begin{rem}\label{rem:convolution_tau}
Note that $G\ast \phi = \fP_\tau G\ast_\tau \phi\in C(\bT_\tau^2)$ for all $\tau\in(0,1]$, $G\in L^1(\bR^2)$, $\phi\in C(\bT_\tau^2)$. 
\end{rem}
\begin{rem}
For $\tau=0$ we identify $\bT_\tau^2$ with $\bR^2$. 
\end{rem}

\begin{dfn}
Let $\bK$ be a finite set and let $\sT$ be a topological space. We denote by $\sT^\bK$ the set of maps $\bK\to\sT$. We identity elements $V$ of $\sT^\bK$ with tuples $(V^k)_{k\in\bK}$ and equip $\sT^\bK$ with the product topology. 
\end{dfn}

\begin{dfn}\label{dfn:dist_K}
Fix $\tau\in(0,1]$ and a finite set $\bK$. For $m\in\bN_+$, \mbox{$U\equiv(U^k)_{k\in\bK^m}\in\sS'(\bT_\tau^{2m})^{\bK^m}$} and $\varphi\equiv(\varphi^k)_{k\in\bK^m}\in C^\infty(\bT_\tau^{2m})^{\bK^m}$ we write $\langle U,\varphi\rangle_\tau:=\sum_{k\in\bK^m} \langle U^k,\varphi^k\rangle_\tau$. For $m,n\in\bN_+$, \mbox{$\varphi\in C^\infty(\bT_\tau^{2m})^{\bK^m}$} and $\psi\in C^\infty(\bT_\tau^{2n})^{\bK^n}$ we define $\varphi\otimes\psi\in C^\infty(\bT_\tau^{2(m+n)})^{\bK^{m+n}}$ by the equality 
\begin{equation}
 (\varphi\otimes\psi)^{(k,l)}:=\varphi^{k}\otimes\psi^{l},
 \qquad
 k\in\bK^m,~l\in\bK^n.
\end{equation}
For $G\in L^1(\bR^2)^{\bK^2}$ and $\phi\in C^\infty(\bT_\tau^{2})^{\bK}$ we define $G\ast\phi \in C^\infty(\bT_\tau^{2})^{\bK}$ by the equality
\begin{equation}
 (G\ast\phi)^k = \sum_{l\in\bK} G^{k,l}\ast \phi^{l}.
\end{equation}
We define $\ast_\tau$ in an analogous way.
\end{dfn}

\begin{dfn}
Let $\tau\in(0,1]$, $m\in\bN_+$ and $\bK$ be a finite set. We say that $U\in\sS'(\bT_\tau^{2m})^{\bK^m}$ is antisymmetric iff 
\begin{equation}
 \langle U,\varphi_1\otimes\ldots\otimes\varphi_m\rangle _\tau
 = \mathrm{sgn}(\pi)
 \langle U,\varphi_{\pi(1)}\otimes\ldots\otimes\varphi_{\pi(m)}\rangle_\tau
\end{equation} 
for all $\varphi_1,\ldots,\varphi_m\in C^\infty(\bT_\tau^2)^{\bK}$ and all permutations $\pi\in\mathcal{P}_m$.
\end{dfn}

\begin{dfn}[Grassmann algebra]\label{dfn:grassmann_algebra}
Let $E$ be a vector space over $\bC$. By definition the Grassmann algebra $\sG(E)$ is the exterior algebra of $E$. We define a unique grading \mbox{$\sG=\sG^+\oplus\sG^-$} such that $E\subset\sG^-$ and such that the exterior product satisfies the conditions $\sG^\pm\sG^\pm\subset\sG^+$ and $\sG^\pm\sG^\mp\subset\sG^-$. We say that an algebra $\sA$ is a Grassmann algebra if $\sA=\sG(E)$ for some vector space $E$. A Grassmann algebra $\sG(E)$ is finite/infinite dimensional if $E$ is finite/infinite dimensional.
\end{dfn}

\begin{rem}
In what follows, the symbol $\otimes_{\mathrm{alg}}$ denotes the algebraic tensor product of vector spaces (the linear span of the set of simple tensors). 
\end{rem}

\begin{dfn}
Let $\tau\in(0,1]$, $\sG$ be a Grassmann algebra and $\bK$ be a finite set. For $m\in\bN_+$, $U\in\sS'(\bT_\tau^{2m})^{\bK^m}\otimes_{\mathrm{alg}}\sG$ and $\phi\in C^\infty(\bT_\tau^{2m})^{\bK^m}\otimes_{\mathrm{alg}}\sG$ we define $\langle U,\phi\rangle_\tau\in\sG$ by the equality
\begin{equation}\label{eq:bracket_tau_sG}
 \langle U,\phi\rangle_\tau:=
 \sum_{i=1}^k\sum_{j=1}^l\langle U_i,\phi_j\rangle_\tau \,g_i h_j\in\sG,
 \qquad
 U=\sum_{i=1}^k U_i \otimes g_i,
 \quad
 \phi=\sum_{j=1}^l\phi_j \otimes h_j,
\end{equation} 
for all $k,l\in\bN_+$, $U_1,\ldots,U_k\in \sS'(\bT_\tau^{2m})^{\bK^m}$, $g_1,\ldots,g_k\in\sG$ and $\phi_1,\ldots,\phi_l\in C^\infty(\bT_\tau^{2m})^{\bK^m}$, $h_1,\ldots,h_l\in\sG$. For $m,n\in\bN_+$ and $\varphi\in C^\infty(\bT_\tau^{2n})^{\bK^n}\otimes_{\mathrm{alg}}\sG$, $\psi\in C^\infty(\bT_\tau^{2m})^{\bK^m}\otimes_{\mathrm{alg}}\sG$ we define $\varphi\otimes\psi\in C^\infty(\bT_\tau^{2(n+m)})^{\bK^{n+m}}\otimes_{\mathrm{alg}}\sG$ by the equality
\begin{equation}
 \varphi\otimes\psi:=\sum_{i=1}^k\sum_{j=1}^l (\varphi_i\otimes \psi_j) \otimes g_i h_j,
 \qquad
 \varphi=\sum_{i=1}^k\varphi_i \otimes g_i,
 \quad
 \psi=\sum_{j=1}^l\psi_j \otimes h_j,
\end{equation}
for all $k,l\in\bN_+$, $\varphi_1,\ldots,\varphi_k\in C^\infty(\bT_\tau^{2n})^{\bK^n}$, $g_1,\ldots,g_k\in\sG$ and $\psi_1,\ldots,\psi_l\in C^\infty(\bT_\tau^{2m})^{\bK^m}$, $h_1,\ldots,h_l\in\sG$.
\end{dfn}

\begin{rem}
The bracket $\langle U,\phi\rangle_\tau$ and the tensor product $\varphi\otimes\psi$ do not depend on the choice of a representation of $U,\phi,\varphi,\psi$ as sums of simple tensors.
\end{rem}

\begin{rem}\label{rem:graded_uniqueness}
Let $m\in\bN_+$ and $E$ be a vector space of dimension greater or equal to $m$. An~antisymmetric distribution $U\in\sS'(\bT_\tau^{2m})^{\bK^m}$ is uniquely determined by the map 
\begin{equation}\label{eq:map_unique_grassmann}
 C^\infty(\bT_\tau^2)^{\bK}\otimes_{\mathrm{alg}}\sG^-\ni\varphi \mapsto \langle U,\varphi^{\otimes m}\rangle_\tau\in \sG,
\end{equation}
where $\sG:=\sG(E)$. Indeed, given $\psi_1,\ldots,\psi_m\in C^\infty(\bT_\tau^2)^{\bK}$ we choose $\varphi=\sum_{j=1}^m \psi_j\otimes g_j$, where $g_1,\ldots,g_m\in\sG^-$ are such that $g=g_1\ldots g_m\neq 0$. Then $m!\,\langle U,\psi_1\otimes\ldots\otimes\psi_m\rangle_\tau\, g =\langle U,\varphi^{\otimes m}\rangle_\tau$. Note also that if $\sG$ is an infinite-dimensional Grassmann algebra, then for every $m\in\bN_+$ an antisymmetric distribution $U\in\sS'(\bT_\tau^{2m})^{\bK^m}$ is uniquely determined by the map~\eqref{eq:map_unique_grassmann}. 
\end{rem}

\begin{dfn}\label{dfn:functional}
Let $\tau\in(0,1]$ and $\bK$ be a finite set. We call a functional a collection $U=(U^m)_{m\in\bN_0}$ such that $U^0\in\bC$ and $U^m\in \sS'(\bT_\tau^{2m})^{\bK^m}$ is antisymmetric for $m\in\bN_+$. We denote the vector space of functionals by
\begin{equation}
 \sN(C^\infty(\bT_\tau^2)^\bK)\subset\bC\times\bigtimes_{m\in\bN_+} 
 \sS'(\bT_\tau^{2m})^{\bK^m}
\end{equation}
and endow it with the product topology. Let $\sG$ be an infinite-dimensional Grassmann algebra. For $U\in\sN(C^\infty(\bT_\tau^2)^\bK)$ and $\phi\in C^\infty(\bT_\tau^2)^\bK\otimes_{\mathrm{alg}}\sG^-$ we define
\begin{equation}\label{eq:functional_sum}
 U(\phi):= 
 U^0+\sum_{m\in\bN_+} \langle U^m,\phi^{\otimes m}\rangle_\tau \in \sG.
\end{equation}
For $U\in\sN(C^\infty(\bT_\tau^2)^\bK)$, an entire function $f:\bC\to\bC$ and
$\phi\in C^\infty(\bT_\tau^2)^\bK\otimes_{\mathrm{alg}}\sG^-$
we define
\begin{equation}
 f(U(\phi))
 :=
 \sum_{j\in\bN_0}
 \frac{f^{(j)}(U^0)}{j!}
 \bigl(U(\phi)-U^0\bigr)^j.
\end{equation}
The sum on the RHS is finite since $U(\phi)-U^0$ is nilpotent. For $k\in\bN_0$, $U\in\sN(C^\infty(\bT_\tau^2)^\bK)$ and $\phi\in C^\infty(\bT_\tau^2)^{\bK}\otimes_{\mathrm{alg}}\sG^-$ the antisymmetric distribution $$\rD_\phi^k U(\phi)\in \sS'(\bT_\tau^{2k})^{\bK^k}\otimes_{\mathrm{alg}}\sG$$ is uniquely defined by the equality
\begin{equation}
 \langle \rD_\phi^k U(\phi),\psi^{\otimes k}\rangle_\tau := \partial_u^k U(\phi+u\psi)\big|_{u=0}
\end{equation}
for all $\phi,\psi\in C^\infty(\bT_\tau^2)^{\bK}\otimes_{\mathrm{alg}}\sG^-$. 
\end{dfn}

\begin{rem}
The assumption that $\sG$ is an infinite-dimensional Grassmann algebra is only needed to ensure that for all $k\in\bN_0$ the antisymmetric distribution $\rD_\phi^k U(\phi)$ is uniquely defined, cf.~Remark~\ref{rem:graded_uniqueness}. Moreover, observe that given $U\equiv (U^m)_{m\in\bN_0}\in\sN(C^\infty(\bT_\tau^2)^\bK)$ for all $k\in\bN_0$ it holds $\rD_\phi^k U(\phi)|_{\phi=0}=k!\,U^k$. 
\end{rem}

\begin{rem}
Note that for every $\phi\in C^\infty(\bT_\tau^2)^\bK\otimes_{\mathrm{alg}}\sG^-$
there are only finitely many non-zero terms in the series on the RHS of~\eqref{eq:functional_sum}. Indeed, if $\phi=\sum_{i=1}^n\phi_i\otimes g_i$, then $\phi^{\otimes m}=0$ for $m>n$. Moreover, $U(\phi)-U^0$ belongs to the ideal generated by $g_1,\ldots,g_n$ and is therefore nilpotent. In fact, $(U(\phi)-U^0)^m=0$ for $m>n$. Hence the definition of $f(U(\phi))$ above involves only finitely many non-zero terms.
For every $k\in\bN_0$, the $k$-th derivative at $\phi=0$ of the map $\phi\mapsto f(U(\phi))$ is a continuous antisymmetric $k$-linear functional and hence determines an element $W^k\in\sS'(\bT_\tau^{2k})^{\bK^k}$. Thus there exists a unique $W=(W^m)_{m\in\bN_0}\in\sN(C^\infty(\bT_\tau^2)^\bK)$ such that $f(U(\phi))=W(\phi)$, and $W^k = \frac{1}{k!}\rD_\phi^k f(U(\phi))\big|_{\phi=0}$ for $k\in\bN_0$.
\end{rem}

\begin{rem}
As an example, let us consider a functional $U\in\sN(C^\infty(\bT_\tau^2))$ of the form \mbox{$U(\phi)=\phi(x_1)\phi(x_2)$} for some fixed $x_1,x_2\in\bR^2$. Then $\rD_\phi U(\phi) = -\phi(x_2)\delta_{x_1}+\phi(x_1)\delta_{x_2}$ and \mbox{$\rD^2_\phi U(\phi) = \delta_{x_1}\otimes\delta_{x_2}-\delta_{x_2}\otimes \delta_{x_1}$}, where $\delta_x\in\sS'(\bT^2_\tau)$ is the periodization of the Dirac delta at $x\in\bR^2$. Furthermore, we have $\langle \rD_\phi U(\phi),\psi\rangle_\tau = \psi(x_1) \phi(x_2)+\phi(x_1)\psi(x_2)$ as well as \mbox{$\langle \rD_\phi^2 U(\phi),\psi^{\otimes2}\rangle_\tau = 2\psi(x_1) \psi(x_2)$}.
\end{rem}

\section{Gross--Neveu model with cutoffs}\label{sec:gross-neveu}

In this section we introduce the Grassmann algebra of the Gross--Neveu model and use it to define the Berezin integral and the  Grassmann interacting measure $\mu_{\tau,\varepsilon}$ of the Gross--Neveu model with cutoffs $\tau,\varepsilon\in(0,1]$. We also state the conditions for the bump function $\vartheta_\varepsilon$ used to implement the UV cutoff.

\begin{dfn}
Let $\bF:=\{1,\ldots,N\}\times\{1,2\}$, $\bFF:=\{-,+\}\times\bF$, $\bA:=\{a\in\mathbb N_0^2\,|\,|a|\leq2\}$. We call $\varsigma\in\bF$ and $\sigma\in\bFF$ spinor indices and $a\in\bA$ a spatial multi-index.
\end{dfn}

\begin{dfn}\label{dfn:gamma}
The so-called gamma matrices are defined by
\begin{equation}
 \gamma_1
 :=
 \begin{pmatrix}
  0 & 1\\1 & 0
 \end{pmatrix},
 \quad
 \gamma_2
 :=
 \begin{pmatrix}
  0 & -\ri\\\ri & 0
 \end{pmatrix},
 \quad
 \Gamma_j
 \equiv(\Gamma_j^{\varsigma_1,\varsigma_2})_{\varsigma_1,\varsigma_2\in\bF}
 :=(\gamma_j)^{\oplus N},
 \quad
 j\in\{1,2\}.
\end{equation}
\end{dfn}
\begin{rem}\label{rem:gamma}
Note that $\gamma_1,\gamma_2$ are complex $2\times2$ matrices and $\Gamma_1,\Gamma_2$ are complex block diagonal $2N\times2N$ matrices. We have
\begin{equation}
 \gamma_2^{\mathrm{t}} \gamma_2 = -1,
 \qquad
 \gamma_2^{\mathrm{t}}\gamma_1 \gamma_2 = \gamma_1^{\mathrm{t}},
 \qquad 
 \gamma_2^{\mathrm{t}} \gamma_2 \gamma_2 = \gamma_2^{\mathrm{t}},
\end{equation}
where $A^{\mathrm{t}}$ denotes the transpose of the matrix $A$. The matrices $\Gamma_1,\Gamma_2$ satisfy the same identities.
Moreover, $\Gamma_1$ is real and symmetric and
\begin{equation}
 \Gamma_1^2=1,
 \qquad
 \Gamma_1^{\mathrm{t}}=\Gamma_1,
 \qquad
 \Gamma_1\overline{\Gamma_j}\Gamma_1=\Gamma_j,
 \qquad
 j\in\{1,2\}
 \, .
\end{equation}
Indeed, $\overline{\Gamma_1}=\Gamma_1$,
$\overline{\Gamma_2}=-\Gamma_2$ and
$\Gamma_1\Gamma_2\Gamma_1=-\Gamma_2$.
\end{rem}

\begin{dfn}
For $\phi=(\phi^\varsigma)_{\varsigma\in\bF}\in \bC^\bF$ and $\psi=(\psi^\varsigma)_{\varsigma\in\bF}\in \bC^\bF$ we write $ \phi\cdot\psi = \sum_{\varsigma\in\bF} \phi^\varsigma \psi^\varsigma$. 
\end{dfn}

\begin{dfn}\label{dfn:GN_grassmann}
Let $\tau,\varepsilon\in(0,1]$, $\omega(p):=(|p|^2+1)^{1/2}$, $\Lambda_{\tau,\varepsilon}:=\{p\in (2\pi\tau\bZ)^2\,|\,\varepsilon\omega(p)\leq 4\}$ and $\mathcal{C}_{\tau,\varepsilon}:=(\mathrm{Span}\{x\mapsto\re^{\ri p\cdot x}\,|\,p\in \Lambda_{\tau,\varepsilon}\})^\bFF\subset C^\infty(\bT_\tau^2)^\bFF$. The Grassmann algebra of the Gross--Neveu model $\sG_{\tau,\varepsilon}$ is the unital complex algebra whose generators
\begin{equation}\label{eq:generators_GN}
 \{(\fF_\tau\psi_{\tau,\varepsilon}^\sigma)(p)\,|\,p\in\Lambda_{\tau,\varepsilon},\ \sigma\in\bFF\}
\end{equation}
satisfy the conditions
\begin{equation}
\begin{aligned}
 (\fF_\tau\psi_{\tau,\varepsilon}^{\sigma_1})(p_1)\, (\fF_\tau\psi_{\tau,\varepsilon}^{\sigma_2})(p_2)
 +
 (\fF_\tau\psi_{\tau,\varepsilon}^{\sigma_2})(p_2)\, (\fF_\tau\psi_{\tau,\varepsilon}^{\sigma_1})(p_1)=0
\end{aligned}
\end{equation} 
for all $p_1,p_2\in\Lambda_{\tau,\varepsilon}$ and $\sigma_1,\sigma_2\in\bFF$. We define a unique grading \mbox{$\sG_{\tau,\varepsilon}=\sG_{\tau,\varepsilon}^+\oplus\sG_{\tau,\varepsilon}^-$} such that $(\fF_\tau\psi_{\tau,\varepsilon}^{\sigma})(p)\in \sG_{\tau,\varepsilon}^-$ for all $p\in\Lambda_{\tau,\varepsilon}$, $\sigma\in\bFF$ and such that the product in $\sG_{\tau,\varepsilon}$ satisfies the conditions $\sG_{\tau,\varepsilon}^\pm\sG_{\tau,\varepsilon}^\pm\subset\sG_{\tau,\varepsilon}^+$ and $\sG_{\tau,\varepsilon}^\pm\sG_{\tau,\varepsilon}^\mp\subset\sG_{\tau,\varepsilon}^-$. Moreover, we define the free Grassmann field
\begin{equation}
 \psi_{\tau,\varepsilon}=(\psi_{\tau,\varepsilon}^\sigma)_{\sigma\in\bFF}\in \mathcal{C}_{\tau,\varepsilon}\otimes \sG_{\tau,\varepsilon}^-\subset C^\infty(\bT_\tau^2)^\bFF\otimes\sG_{\tau,\varepsilon}^-
\end{equation}
by the following equality
\begin{equation}\label{eq:Grassmann_field}
 \psi_{\tau,\varepsilon}^\sigma(x) := \tau^2\sum_{p\in\Lambda_{\tau,\varepsilon}} (\fF_\tau\psi_{\tau,\varepsilon}^\sigma)(p)\,\re^{\ri p\cdot x}
\end{equation}
for all $x\in\bR^2$ and $\sigma\in\bFF$. For $\varsigma\in\bF$ we set $\bar\psi_{\tau,\varepsilon}^\varsigma(x) := \psi_{\tau,\varepsilon}^{-,\varsigma}(x)$ and $\ubar\psi_{\tau,\varepsilon}^\varsigma(x) := \psi_{\tau,\varepsilon}^{+,\varsigma}(x)$.
\end{dfn}

\begin{rem}
For all $\tau,\varepsilon\in(0,1]$ it holds $\sG_{\tau,\varepsilon}=\sG(E_{\tau,\varepsilon})$, where the vector space $E_{\tau,\varepsilon}$ coincides with the span of the set~\eqref{eq:generators_GN}. In particular, $\sG_{\tau,\varepsilon}$ is a finite-dimensional Grassmann algebra in the sense of Def.~\ref{dfn:grassmann_algebra}.
\end{rem}

\begin{rem}\label{rem:psi_finite}
Observe that for every $\tau,\varepsilon\in(0,1]$ it holds $\psi_{\tau,\varepsilon}^{\otimes m}=0$ for all $m\in\bN_+$ provided $m>|\bFF|\,|\Lambda_{\tau,\varepsilon}|$, where $|S|$ denotes the number of elements of a set $S$.
\end{rem}

\begin{dfn}\label{dfn:berezin_def}
Let $\tau,\varepsilon\in(0,1]$. The Berezin integral is a linear map 
\begin{equation}
 \sG_{\tau,\varepsilon}\ni g\mapsto \int g
 \,\rd\psi_{\tau,\varepsilon}\in\bC
\end{equation}
such that 
\begin{equation}\label{eq:berezin_def}
 \int g
 \,\rd\psi_{\tau,\varepsilon}=1,
 \qquad
 g=
 \prod_{p\in\Lambda_{\tau,\varepsilon}}\prod_{\varsigma\in\mathbb{F}}
 (\fF_\tau\bar\psi_{\tau,\varepsilon}^\varsigma)(p)\,
 (\fF_\tau\ubar\psi_{\tau,\varepsilon}^\varsigma)(p),
\end{equation}
and such that it vanishes on any other monomial in the generators of $\sG_{\tau,\varepsilon}$.
\end{dfn}

\begin{rem}
Note that the order in the product in~\eqref{eq:berezin_def} is unimportant because the elements of the set $\{(\fF_\tau\bar\psi_{\tau,\varepsilon}^\varsigma)(p)\,(\fF_\tau\ubar\psi_{\tau,\varepsilon}^\varsigma)(p)\,|\,p\in\Lambda_{\tau,\varepsilon},\,\varsigma\in\bF\}$ commute. For a discussion of properties of the Berezin integral see for example~\cite[Appendix~A]{CCS13}.  
\end{rem}

\begin{rem}
Let $\sG_1,\sG_2$ be Grassmann algebras. We denote by $\sG_1\otimes_{\mathrm{alg}}\sG_2$ the algebraic graded tensor product of graded algebras $\sG_1$ and $\sG_2$. Note that $\sG_1\otimes_{\mathrm{alg}}\sG_2$ is a Grassmann algebra. Given $g_1\in\sG_1$ and $g_2\in\sG_2$ we define $g_1+g_2:=g_1\otimes1+1\otimes g_2\in \sG_1\otimes_{\mathrm{alg}}\sG_2$. 
\end{rem}

\begin{rem}
Let $\sG$ be a Grassmann algebra. We define the Berezin integral on $\sG_{\tau,\varepsilon}\otimes_{\mathrm{alg}}\sG$ as a linear map 
\begin{equation}
 \sG_{\tau,\varepsilon}\otimes_{\mathrm{alg}}\sG\ni g\mapsto \int g
 \,\rd\psi_{\tau,\varepsilon}\in\sG
\end{equation}
such that
\begin{equation}
 \int (g\otimes h) \,\rd\psi_{\tau,\varepsilon}:=
\left(\int g \,\rd\psi_{\tau,\varepsilon}\right) h
\end{equation}
for all $g\in \sG_{\tau,\varepsilon}$ and $h\in\sG$.
\end{rem}
\begin{lem}\label{lem:berezin_translation}
Let $\tau,\varepsilon\in(0,1]$ and $\sG$ be a Grassmann algebra. For all $F\in\sN(C^\infty(\bT_\tau^2)^\bFF)$ and $\phi\in \mathcal{C}_{\tau,\varepsilon}\otimes \sG^-\subset C^\infty(\bT_\tau^2)^\bFF\otimes\sG^-$ it holds
\begin{equation}
 \int F(\psi_{\tau,\varepsilon}+\phi)\,\rd\psi_{\tau,\varepsilon}
 =
 \int F(\psi_{\tau,\varepsilon})\,\rd\psi_{\tau,\varepsilon},
\end{equation}
where $\psi_{\tau,\varepsilon}+\phi\in C^\infty(\bT_\tau^2)^\bFF\otimes(\sG_{\tau,\varepsilon}\otimes\sG)^-$.
\end{lem}
\begin{proof}
The lemma follows from the invariance under translation of the Berezin integral, cf.~\cite[Proposition A.12]{CCS13}. 
\end{proof}

\begin{dfn}\label{dfn:theta}
Let $\zeta_\flat\in(4/5,1)$ and let $\vartheta:\bR\to[0,1]$ be a Gevrey function of order $1/\zeta_\flat$ such that $\vartheta^{1/2}$ is also Gevrey of order
$1/\zeta_\flat$, and $\vartheta(t)=1$ for $t\in(-\infty,1/2]$ and $\vartheta(t)=0$ for $t\in[1,\infty)$. For $\varepsilon\in(0,1]$ we define $\vartheta_\varepsilon\in C^\infty(\bR^2)$ by $\fF\vartheta_\varepsilon(p)=\vartheta(2\varepsilon\omega(p))^{1/2}$.
\end{dfn}
\begin{rem}
Recall that the class of Gevrey functions of order $s\in[1,\infty)$ on $\bR^d$ consists of $\varphi\in C^\infty(\bR^d)$ such that for some $C\in(0,\infty)$ the bound
\begin{equation}
 \sup_{p\in\bR^d} |\partial^a\varphi(p)|\leq C^{1+|a|}\,(a!)^{s}
\end{equation}
is satisfied for all $a\in\bN_0^d$. 
\end{rem}
\begin{rem}
For an explicit construction of a function $\vartheta$ satisfying the conditions stated in the above definition see~\cite[App. A.1]{GMR21}. 
\end{rem}

\begin{dfn}\label{dfn:measures}
The free action functional $A_\tau\in \sN(C^\infty(\bT_\tau^2)^\bFF)$ and the interaction term $U_{\tau,\varepsilon}\equiv U_\tau(g_{\tau,\varepsilon},r_{\tau,\varepsilon})\in\sN(C^\infty(\bT_\tau^2)^\bFF)$ with cutoffs $\tau,\varepsilon\in(0,1]$ are defined uniquely by Eqs.~\eqref{eq:free_action_intro} and~\eqref{eq:potential_intro}
for all $\psi\equiv(\bar\psi,\ubar\psi)\in C^\infty(\bT_\tau^2)^\bFF\otimes_{\mathrm{alg}}\sG^-\equiv (C^\infty(\bT_\tau^2)^\bF\times C^\infty(\bT_\tau^2)^\bF)\otimes_{\mathrm{alg}}\sG^-$, where $\sG$ is an infinite-dimensional Grassmann algebra. The free and interacting measures with cutoffs $\tau,\varepsilon\in(0,1]$ are defined using the Berezin integral introduced in Def.~\ref{dfn:berezin_def} by the equations
\begin{equation}
 \nu_{\tau,\varepsilon}(F) := 
 \frac{\int F(\vartheta_\varepsilon\ast\psi_{\tau,\varepsilon})\exp(-A_{\tau}(\psi_{\tau,\varepsilon}))\,\rd\psi_{\tau,\varepsilon}}{\int \exp(-A_{\tau}(\psi_{\tau,\varepsilon}))\,\rd\psi_{\tau,\varepsilon}}
\end{equation}
and
\begin{equation}
 \mu_{\tau,\varepsilon}(F) := 
 \frac{\int F(\vartheta_\varepsilon\ast\psi_{\tau,\varepsilon})\exp(-A_{\tau}(\psi_{\tau,\varepsilon})+U_{\tau,\varepsilon}(\vartheta_\varepsilon\ast\psi_{\tau,\varepsilon}))\,\rd\psi_{\tau,\varepsilon}}{\int \exp(-A_{\tau}(\psi_{\tau,\varepsilon})+U_{\tau,\varepsilon}(\vartheta_\varepsilon\ast\psi_{\tau,\varepsilon}))\,\rd\psi_{\tau,\varepsilon}},
\end{equation}
respectively, for all functionals $F\in\sN(C^\infty(\bT_\tau^2)^\bFF)$, where $\psi_{\tau,\varepsilon}\in C^\infty(\bT_\tau^2)^\bFF\otimes\sG_{\tau,\varepsilon}^-$ is defined by~\eqref{eq:Grassmann_field}.
\end{dfn}

\begin{rem}
Note that by~\cite[Proposition~A.14]{CCS13} it holds
\begin{multline}
 \int \exp(-A_{\tau}(\psi_{\tau,\varepsilon}))\,\rd\psi_{\tau,\varepsilon}
 =
 \tau^{4N|\Lambda_{\tau,\varepsilon}|}
 \prod_{p\in\Lambda_{\tau,\varepsilon}}
 \det((\Xi^{\varsigma_1,\varsigma_2}(p))_{\varsigma_1,\varsigma_2\in\bF})
 \\
 =
 \tau^{4N|\Lambda_{\tau,\varepsilon}|}
 \prod_{p\in\Lambda_{\tau,\varepsilon}}(1+|p|^2)^N
 \neq 0,
\end{multline}
where $\Xi^{\varsigma_1,\varsigma_2}(p):=1^{\varsigma_1,\varsigma_2}-\sum_{j\in\{1,2\}}\ri\Gamma_j^{\varsigma_1,\varsigma_2} p^j$ for $\varsigma_1,\varsigma_2\in\bF$. As a result, the free field measure is well-defined for all $\tau,\varepsilon\in(0,1]$. The interacting measure with cutoffs $\tau,\varepsilon\in(0,1]$ is well-defined provided
\begin{equation}
 \frac{\int \exp(-A_{\tau}(\psi_{\tau,\varepsilon})+U_{\tau,\varepsilon}(\vartheta_\varepsilon\ast\psi_{\tau,\varepsilon}))\,\rd\psi_{\tau,\varepsilon}}{\int \exp(-A_{\tau}(\psi_{\tau,\varepsilon}))\,\rd\psi_{\tau,\varepsilon}}
 =
 \nu_{\tau,\varepsilon}(\exp(U_{\tau,\varepsilon}(\Cdot)))
 \neq 0.
\end{equation}
This is proved in Remark~\ref{rem:Z_not_zero} for a specific choice of parameters $g_{\tau,\varepsilon}$, $r_{\tau,\varepsilon}$ fixed as in Theorem~\ref{thm:main}.
\end{rem}

\section{Scale decomposition of propagator}\label{sec:covariance_decomposition}

In this section we define the scale decomposition $G_{\tau,\varepsilon;t}$ of the propagator $G_{\tau,\varepsilon}$. We also introduce auxiliary propagators $G^{\pm;\varsigma,\tilde\varsigma}_{\varepsilon;t}$ that will be used in Sec.~\ref{sec:free_field_decomposition} to construct the scale decomposition $\varPsi_{\tau,\varepsilon;t}$ of the free field that satisfies the conditions discussed in Sec.~\ref{sec:strategy}.  

\begin{dfn}\label{dfn:covariance}
For $\varsigma_1,\varsigma_2\in\bF$ and $p\in\bR^2$ let $\Xi^{\varsigma_1,\varsigma_2}(p):=1^{\varsigma_1,\varsigma_2}-\sum_{j\in\{1,2\}}\ri\Gamma_j^{\varsigma_1,\varsigma_2} p^j$ and $\omega(p):=(|p|^2+1)^{1/2}$. The free field covariance $$G\equiv(G^\sigma)_{\sigma\in\bFF^2}\in L^1(\bR^2)^{\bFF^2}$$ is defined by
\begin{equation}
 G^\sigma(x) =
 \frac{1}{(2\pi)^2}\int_{\bR^2} (\fF G^\sigma)(p)\,\re^{\ri p\cdot x}\,\rd p,
\end{equation}
where
\begin{equation}
\begin{gathered}
 (\fF G^{+,\varsigma_1,+,\varsigma_2})(p):=0,
 \qquad
 (\fF G^{+,\varsigma_1,-,\varsigma_2})(p):=
 \Xi^{\varsigma_1,\varsigma_2}(p)/(\omega(p))^2,
 \\
 (\fF G^{-,\varsigma_1,-,\varsigma_2})(p):=0,
 \qquad
 (\fF G^{-,\varsigma_1,+,\varsigma_2})(p):=-(\fF G^{+,\varsigma_2,-,\varsigma_1})(-p)
\end{gathered} 
\end{equation}
for all $\varsigma_1,\varsigma_2\in\bF$. For $\varepsilon,t\in[0,1]$ we define
\begin{equation}
 G_\varepsilon\equiv(G^\sigma_\varepsilon)_{\sigma\in\bFF^2},
 ~
 G_{\varepsilon;t}\equiv(G^\sigma_{\varepsilon;t})_{\sigma\in\bFF^2},
 ~
 \dot G_{\varepsilon;t}\equiv(\dot G^\sigma_{\varepsilon;t})_{\sigma\in\bFF^2}\in L^1(\bR^2)^{\bFF^2},
\end{equation}
by the equalities
\begin{equation}
 (\fF G^\sigma_\varepsilon)(p):=\vartheta(2\varepsilon\omega(p))\,(\fF G^\sigma)(p),
 \quad
 (\fF G^\sigma_{\varepsilon;t})(p):=\vartheta(t\omega(p))\,(\fF G^\sigma_\varepsilon)(p),
 \quad
 \dot G^\sigma_{\varepsilon;t}:=\partial_t G^\sigma_{\varepsilon;t},
\end{equation}
where the bump function $\vartheta\in C^\infty(\bR)$ was introduced in Def.~\ref{dfn:theta}. For $\varepsilon\in[0,1]$, $t\in(0,1]$ and $\sigma\in\bFF^2$, $\varsigma,\tilde\varsigma\in\bF$ we define
$
 G^{\pm;\varsigma,\tilde\varsigma}_{\varepsilon;t}\in L^1(\bR^2),
$
by the equalities
\begin{equation}
 (\fF G^{+;\varsigma,\tilde\varsigma}_{\varepsilon;t})(p):=-
 (\fF \dot G^{+,\varsigma,-,\tilde\varsigma}_{\varepsilon;t})(p),
 \qquad
 (\fF G^{-;\varsigma,\tilde\varsigma}_{\varepsilon;t})(p):=\vartheta(\varepsilon/2t)\,\vartheta(t\omega(p)/2)\,1^{\varsigma,\tilde\varsigma}.
\end{equation}
For $\tau\in(0,1]$ the periodizations of the kernels $G^\sigma, G^\sigma_{\varepsilon},G^\sigma_{\varepsilon;t},\dot G^\sigma_{\varepsilon;t},G^{\pm;\varsigma,\tilde\varsigma}_{\varepsilon;t}$ are defined using the operator $\fP_\tau$ introduced in Def.~\ref{dfn:distributions_torus} and are denoted by $G^\sigma_{\tau}, G^\sigma_{\tau,\varepsilon},G^\sigma_{\tau,\varepsilon;t},\dot G^\sigma_{\tau,\varepsilon;t},G^{\pm;\varsigma,\tilde\varsigma}_{\tau,\varepsilon;t}$, respectively. We use the symbols $\dot G^\sigma_{t},G^{\pm;\varsigma,\tilde\varsigma}_{t}$ to denote $\dot G^\sigma_{\varepsilon;t},G^{\pm;\varsigma,\tilde\varsigma}_{\varepsilon;t}$ with $\varepsilon=0$.
\end{dfn}

\begin{rem}\label{rem:G}
The kernels introduced in the above definition have the following properties.
\begin{itemize}
 \item[(A)] If $t\in(0,\varepsilon]$, then $\vartheta(2\varepsilon\omega(p))\vartheta(t\omega(p))=\vartheta(2\varepsilon\omega(p))$ and $G^\sigma_{\varepsilon;t}=G^\sigma_\varepsilon$, $\dot G^\sigma_{\varepsilon;t}=0$. 
 
 \item[(B)] If $t\in(0,\varepsilon/2]$, then $G^{\pm;\varsigma,\tilde\varsigma}_{\varepsilon;t}=0$. Moreover, $\supp\,\fF_\tau G^{\pm;\varsigma,\tilde\varsigma}_{\tau,\varepsilon;t}\subset\Lambda_{\tau,\varepsilon}$ for all $t\in(0,1]$.
 
 \item[(C)] If $t\in[4\varepsilon,1]$, then $\vartheta(2\varepsilon\omega(p))\dot\vartheta(t\omega(p))=\dot\vartheta(t\omega(p))$ and $\dot G^\sigma_{\varepsilon;t}=\dot G^\sigma_{t}$, $G^{\pm;\varsigma,\tilde\varsigma}_{\varepsilon;t}=G^{\pm;\varsigma,\tilde\varsigma}_{t}$.
 
 \item[(D)] If $t\in(0,1/2)$, then $\dot\vartheta(t\omega(0))=0$ and $\int_{\bR^2} \dot G^\sigma_{\varepsilon;t}=0$.
 
 \item[(E)] Since $\vartheta(\varepsilon/2t)\,\vartheta(t\omega(p)/2)=1$ on the support of $p\mapsto \vartheta(2\varepsilon\omega(p))\,\dot\vartheta(t\omega(p))$, it holds
\begin{equation}
 \dot G^{+,\varsigma,-,\tilde\varsigma}_{\varepsilon;t} = - \textstyle\sum_{\hat\varsigma\in\bF} G^{+;\varsigma,\hat\varsigma}_{\varepsilon;t} \ast G^{-;\hat\varsigma,\tilde\varsigma}_{\varepsilon;t},
 \qquad
 \dot G^{+,\varsigma,-,\tilde\varsigma}_{\tau,\varepsilon;t} = - \textstyle\sum_{\hat\varsigma\in\bF} G^{+;\varsigma,\hat\varsigma}_{\tau,\varepsilon;t} \ast_\tau G^{-;\hat\varsigma,\tilde\varsigma}_{\tau,\varepsilon;t}.
\end{equation}
Note also that $\dot G^{+,\varsigma,-,\tilde\varsigma}_{\varepsilon;t}(x)
= -\dot G^{-,\tilde\varsigma,+,\varsigma}_{\varepsilon;t}(-x)$ as well as $\dot G^{+,\varsigma,+,\tilde\varsigma}_{\varepsilon;t}=0$ and $\dot G^{-,\varsigma,-,\tilde\varsigma}_{\varepsilon;t}=0$.
\item[(F)] Since $G_{\varepsilon;1}=0$ it holds $G_\varepsilon=-\int_\varepsilon^1 \dot G_{\varepsilon;t}\,\rd t$ and $G-G_\varepsilon=-\int_0^{4\varepsilon\wedge 1} (\dot G_t-\dot G_{\varepsilon;t})\,\rd t$.
\end{itemize}
\end{rem}

\begin{lem}\label{lem:dot_G_infty}
For all $a\in\bN_0^2$ there exist $\delta\in(0,1)$ and $C\in(0,\infty)$ such that for all $\varepsilon\in[0,1]$, $t\in(0,1]$, $\varsigma_1,\varsigma_2\in\bF$ and $x\in\bR^2$ it holds
\begin{equation}
 \exp(\delta\,t^{-\zeta_\flat}\,|x|^{\zeta_\flat})~|\partial^a G_{\varepsilon;t}^{\pm;\varsigma_1,\varsigma_2}(x)|
 \leq C\,t^{-|a|-2},
\end{equation} 
where $\zeta_\flat\in(4/5,1)$ is the parameter introduced in Def.~\ref{dfn:theta}.
\end{lem}
\begin{proof}
We first note the following identities
\begin{equation}\label{eq:dot_G_b_a}
 x^b \partial^a  G_{\varepsilon;t}^{+,\varsigma_1,\varsigma_2}(x) 
 =
 -\frac{\ri^{|a|+|b|}}{(2\pi)^2}\int_{\bR^2} 
 \partial^b_p \Big(p^a\,\vartheta(2\varepsilon\omega(p))\,\dot\vartheta(t\omega(p))\,
 \Xi^{\varsigma_1,\varsigma_2}(p)/\omega(p)\Big)
 \,\re^{\ri p\cdot x}\,\rd p
\end{equation}
and
\begin{equation}\label{eq:dot_G_b_a_minus}
 x^b \partial^a G_{\varepsilon;t}^{-,\varsigma_1,\varsigma_2}(x) 
 =
 \frac{\ri^{|a|+|b|}}{(2\pi)^2}\int_{\bR^2} 
 \partial^b_p \Big(p^a\,\vartheta(\varepsilon/2t)\,\vartheta(t\omega(p)/2)\,1^{\varsigma_1,\varsigma_2}\Big)
 \,\re^{\ri p\cdot x}\,\rd p. 
\end{equation}
Note that the composition of a Gevrey function of order $s\in[1,\infty)$ with an analytic function is a Gevrey function of order $s\in[1,\infty)$, cf.~\cite[Prop. 1.4.6]{Rod93}. Consequently, since the function $\vartheta$ is flat in the vicinity of the origin the function
\begin{equation}
 \bR\times \bR^2\ni (\varepsilon,p)\mapsto \vartheta((\varepsilon^2+p^2)^{1/2})\in\bR
\end{equation}
belongs to the Gevrey class of order $1/{\zeta_\flat}$. As a result, there exists $C\in(0,\infty)$ such that for all $\varepsilon\in[0,1]$ and $a\in\bN_0^2$ it holds
\begin{equation}
 \sup_{p\in\bR^2} |\partial^a_p\vartheta(\varepsilon\omega(p))|
 =
 \varepsilon^{|a|}
 \sup_{p\in\bR^2} |\partial^a_p\vartheta((\varepsilon^2+p^2)^{1/2})|
 \leq C^{1+|a|}\,(a!)^{1/{\zeta_\flat}}\,\varepsilon^{|a|}.
\end{equation}
By a similar argument, there exists $C\in(0,\infty)$ such that for all $t\in[0,1]$ and $a\in\bN_0^2$ it holds
\begin{equation}
\begin{aligned}
 \sup_{p\in\bR^2} |\partial^a_p\vartheta(t\omega(p)/2)|
 &=
 t^{|a|}
 \sup_{p\in\bR^2} |\partial^a_p\vartheta((t^2+p^2)^{1/2}/2)|
 \leq C^{1+|a|}\,(a!)^{1/{\zeta_\flat}}\,t^{|a|},
 \\
 \sup_{p\in\bR^2} |\partial^a_p\dot\vartheta(t\omega(p))|
 &=
 t^{|a|}
 \sup_{p\in\bR^2} |\partial^a_p\dot\vartheta((t^2+p^2)^{1/2})|
 \leq C^{1+|a|}\,(a!)^{1/{\zeta_\flat}}\,t^{|a|}.
\end{aligned} 
\end{equation}
Since the function $p\mapsto \Xi^{\varsigma_1,\varsigma_2}(p)/\omega(p)\in\bC$ is analytic in a strip of size one around the real axis there exists $C\in(0,\infty)$ such that for all $a\in\bN_0^2$ it holds
\begin{equation}
 \sup_{p\in\bR^2} |\partial^a_p(\Xi^{\varsigma_1,\varsigma_2}(p)/\omega(p))|\leq C^{1+|a|}\,a!\,.
\end{equation}
Noting that $\vartheta(2\varepsilon\omega(p))\,\dot\vartheta(t\omega(p))$ vanishes unless $t\omega(p)\in(1/2,1)$ and $t>\varepsilon$, that $\vartheta(t\omega(p)/2)$ vanishes unless $t\omega(p)\in[0,2]$, and using the general Leibniz rule we show that there exists \mbox{$C\in(0,\infty)$} such that for all $\varepsilon\in[0,1]$, $t\in(0,1]$
\begin{equation}
 \left|\partial^b_p \left(p^a\,\vartheta(2\varepsilon\omega(p))\,\dot\vartheta(t\omega(p))\,
 \Xi^{\varsigma_1,\varsigma_2}(p)/\omega(p)\right)\right|
 \leq C^{1+|b|}\,(b!)^{1/{\zeta_\flat}}\,t^{|b|-|a|}
 \,1_{[0,2]}(t\omega(p))
\end{equation}
and
\begin{equation}
 \left|\partial^b_p \left(p^a\,\vartheta(\varepsilon/2t)\,\vartheta(t\omega(p)/2)\,1^{\varsigma_1,\varsigma_2}
 \right)\right|
 \leq C^{1+|b|}\,(b!)^{1/{\zeta_\flat}}\,t^{|b|-|a|}
 \,1_{[0,2]}(t\omega(p)).
\end{equation}
Using Eqs.~\eqref{eq:dot_G_b_a} and~\eqref{eq:dot_G_b_a_minus} we obtain the existence of $C\in(0,\infty)$ such that for all $\varepsilon\in[0,1]$, $t\in(0,1]$ and $k\in\bN_0$
\begin{equation}
 (|x|/t)^k~|\partial^a G_{\varepsilon;t}^{\pm,\varsigma_1,\varsigma_2}(x)|
 \leq C^{1+k}\,(k!)^{1/{\zeta_\flat}}\,t^{-|a|-2}/2
\end{equation}
If $|x|/t\leq1$, the desired estimate follows from the above bound with $k=0$. If $|x|/t\geq1$, we take $k=\ceil{{\zeta_\flat} l}$ and use $(|x|/t)^{{\zeta_\flat} l}\leq(|x|/t)^k$.
We conclude that there exists $C\in[\re,\infty)$ such that for all $\varepsilon\in[0,1]$, $t\in(0,1]$ and $l\in\bN_+$ it holds
\begin{equation}
 (|x|/t)^{{\zeta_\flat} l}~|\partial^a G_{\varepsilon;t}^{\pm,\varsigma_1,\varsigma_2}(x)|
 \leq C^{1+k}\,(k!)^{1/{\zeta_\flat}}\,t^{-|a|-2}/2
 \leq C^{1+3l/{\zeta_\flat}}\,l!\,t^{-|a|-2}/2
\end{equation}
where $k=\ceil{{\zeta_\flat} l}\leq l\wedge ({\zeta_\flat} l+1)$. The second of the above bounds follows from the estimates $(k!)^{1/{\zeta_\flat}}\leq k^{k/{\zeta_\flat}}\leq l^{l+1/{\zeta_\flat}}\leq \re^{2l/{\zeta_\flat}} l!$. This implies that the bound stated in the lemma holds true for $\delta=C^{-3/{\zeta_\flat}}/2$.
\end{proof}

\section{Scale decomposition of free field}\label{sec:free_field_decomposition}

In this section we introduce the scale decomposition of the free field $\varPsi_{\tau,\varepsilon;t}$ satisfying properties listed in Sec.~\ref{sec:strategy}. The scale decomposition of the free field is defined in the non-commutative probability space $(\sF,\fE)$ of the fermionic spacetime white noise $\xi$. We also introduce a filtration $(\sF_t)_{t\in[0,1]}$ of $\sF$ as well as the conditional expectation $\fE_t$ that will play an essential role in the proof of the estimates stated in Sec.~\ref{sec:maps}.

\begin{dfn}
Let $\mathfrak{h}$ be a Hilbert space. We denote by $\Gamma_{\mathrm{a}}(\mathfrak{h})$ the antisymmetric Fock space built on $\mathfrak{h}$. We denote by $\mathds{N}$ the number operator and by $\mathds{P}:=(-1)^{\mathds{N}}$ the parity operator. 
\end{dfn}

\begin{dfn}\label{dfn:white_noise}
Let
$
 \sH = \Gamma_{\mathrm{a}}(L^2([0,1]\times\bR^2\times\bF))^{\otimes 2}.
$
The creation and annihilation operators in $\Gamma_{\mathrm{a}}(L^2([0,1]\times\bR^2\times\bF))$ are denoted by $a^*(f,\varsigma)$ and $a(f,\varsigma)$, where $f\in L^2([0,1]\times\bR^2)$ and $\varsigma\in\bF$. 
The vacuum state in $\sH$ is denoted by $\Omega$. Let $\sB:=\sB(\sH)$ be the Banach algebra of bounded operators on $\sH$. The fermionic white noise is defined by
\begin{gather}
 \xi^{-,\varsigma}(f)\equiv \int_{[0,1]\times\bR^2} f(s,x)\, \xi^{-,\varsigma}(\rd s,\rd x)
 :=
 (\mathds{P}\otimes a^*(f,\varsigma) - a(f^\rc,\varsigma)\otimes\mathds{1}),
 \\
 \xi^{+,\varsigma}(f)\equiv \int_{[0,1]\times\bR^2} f(s,x)\, \xi^{+,\varsigma}(\rd s,\rd x) 
 :=
 (a^*(f,\varsigma)\otimes\mathds{1} + \mathds{P}\otimes a(f^\rc,\varsigma))
\end{gather}
for all $f\in L^2([0,1]\times\bR^2)$, $\varsigma\in\bF$, where $f^\rc$ denotes the complex conjugate of the function~$f$. We define $\sF$ to be the unital Banach subalgebra of $\sB$ generated by $\xi^{\pm,\varsigma}(f)$, where \mbox{$f\in L^2([0,1]\times \bR^2)$}, $\varsigma\in\bF$. We define a unique grading \mbox{$\sF=\sF^+\oplus\sF^-$} such that $\xi^{\pm,\varsigma}(f)\in\sF^-$ and such that the product in $\sF$ satisfies the conditions $\sF^\pm\sF^\pm\subset \sF^+$ and $\sF^\pm\sF^\mp\subset \sF^-$. For $s\in[0,1]$ and $t\in[s,1]$ we define $\sF_{t,s}$ to be the unital Banach subalgebra of $\sF$ generated by $\xi^{\pm,\varsigma}(f)$, where $f\in L^2([0,1]\times \bR^2)$, $\varsigma\in\bF$, $\supp\,f\subset[s,t]\times \bR^2$. For $t\in[0,1]$ we set $\sF_t:=\sF_{1,t}$. We write $\sF_{t,s}^{\pm}:=\sF_{t,s}\cap\sF^\pm$ and $\sF_t^\pm:=\sF_t\cap\sF^\pm$.
\end{dfn}

\begin{rem}
The algebra $\sF$ is a graded Banach algebra satisfying the Grassmann
commutation relations. In particular, elements of $\sF^-$
anti-commute with each other, while elements of $\sF^+$ commute with
all elements of $\sF$.
\end{rem}

\begin{rem}\label{rem:bound_noise}
In what follows, the following bound
\begin{equation}
 \textstyle\sum_{\varsigma\in\bF}\|f^\varsigma\|_{L^2([0,1]\times\bR^2)}^2\leq \left\|\textstyle\sum_{\varsigma\in\bF}\xi^{\pm,\varsigma}(f^\varsigma)\right\|_\sB^2\leq  2\textstyle\sum_{\varsigma\in\bF}\|f^\varsigma\|_{L^2([0,1]\times\bR^2)}^2
\end{equation}
valid for all $f=(f^\varsigma)_{\varsigma\in\bF} \in L^2([0,1]\times\bR^2)^\bF$ will play an important role. To prove the first bound we apply the operator to the vacuum state $\Omega$. The second of the above bounds is an immediate consequence of the identity
\begin{equation}
 \|\textstyle\sum_{\varsigma\in\bF}a(f^\varsigma,\varsigma)\|^2_\sB  = \|\textstyle\sum_{\varsigma\in\bF}a^*(f^\varsigma,\varsigma)\|^2_\sB= \textstyle\sum_{\varsigma\in\bF}\|f^\varsigma\|^2_{L^2([0,1]\times\bR^2)},
\end{equation}
which follows easily from the anti-commutation relations of the fermionic creation and annihilation operators, for details see for example~\cite[Proposition 5.2.2]{BR97}.
\end{rem}

\begin{dfn}[Expected value]
We define a continuous linear functional $\fE\,:\,\sF\to\bC$ by the equality $\fE F:=(\Omega,F\Omega)$ for all $F\in\sF$.
\end{dfn}

\begin{rem}\label{rem:xi_E}
Note that it holds
\begin{equation}
 \fE(\xi^{\pm,\varsigma}(f)\,\xi^{\mp,\tilde\varsigma}(\tilde f)) = \pm 1_{\varsigma,\tilde\varsigma}\,\int_{[0,1]\times\mathbb R^2}
 f(s,x)\tilde f(s,x)\,\rd s\,\rd x,
 \qquad
 \fE(\xi^{\pm,\varsigma}(f)\,\xi^{\pm,\tilde\varsigma}(\tilde f)) = 0
\end{equation}
for all $f,\tilde f\in L^2([0,1]\times\bR^2)$, $\varsigma,\tilde\varsigma\in\bF$. 
\end{rem}

\begin{dfn}[Conditional expectation]\label{dfn:conditional_exp}
A family of continuous linear maps $$(\fE_t\,:\,\sF\to\sF)_{t\in[0,1]}$$ is called the conditional expectation if
\begin{enumerate}
 \item[(0)] $\fE_t(\sF^\pm)\subset\sF_t^\pm$ for all $t\in[0,1]$, 
 
 \item[(1)] $\fE_t F=\mathds{1}\,\fE F$ for all $t\in[0,1]$ and $F\in\sF_{t,0}$,
 \item[(2)] $\fE_t=\fE_t\fE_s:\,\sF \to \sF_t$ for $s\in[0,1]$ and $t\in[s,1]$,
 \item[(3)] $\fE_t(GFH)=G\fE_t(F)H$ for all $t\in[0,1]$ and $G,H\in\sF_t$, $F\in\sF$,
 \item[(4)] $\|\fE_t F\|_\sB \leq \|F\|_\sB$ for all $t\in[0,1]$ and $F\in\sF$.
\end{enumerate}
\end{dfn}
\begin{rem}
For the construction of the conditional expectation see~\cite[Appendix~A]{DFG25}.
\end{rem}
\begin{rem}
Note that (0) implies $\fE_t(\sF)\subset\sF_t$, and (1) implies $\fE_1=\fE$ and $\fE(\sF^-)=0$.
\end{rem}
\begin{rem}
Let $\sG$ be a Grassmann algebra. We define $\sF\otimes_{\mathrm{alg}}\sG$ to be the algebraic graded tensor product of graded algebras $\sF$ and $\sG$. We introduce the family of functionals \mbox{$(\fE_t\,:\,\sF\otimes_{\mathrm{alg}}\sG\to\sF_t\otimes_{\mathrm{alg}}\sG)_{t\in[0,1]}$} defined uniquely by the equality $\fE_t(F\otimes g)=(\fE_tF)\otimes g$ for all $F\in\sF$, $g\in\sG$ and $t\in[0,1]$. The functional $\fE\,:\,\sF\otimes_{\mathrm{alg}}\sG\to\sG$ is defined in an analogous way.
\end{rem}

\begin{dfn}\label{dfn:Psi}
Let $\tau,\varepsilon\in[0,1]$, $s\in[0,1]$ and $t\in[s,1]$. We define operator-valued distributions
\begin{equation}
 \varPsi_{\tau,\varepsilon;t,s}
 \equiv
 (\varPsi^-_{\tau,\varepsilon;t,s},\varPsi^+_{\tau,\varepsilon;t,s})
 \equiv
 (\bar\varPsi_{\tau,\varepsilon;t,s},\ubar\varPsi_{\tau,\varepsilon;t,s})
 \equiv (\varPsi^{\sigma}_{\tau,\varepsilon;t,s})_{\sigma\in\bFF}\in \sS'(\bT_\tau^2,\sF_{t,s}\cap\sF^-)^\bFF
\end{equation}
by the equation
\begin{equation}
 \langle\varPsi^{\pm,\varsigma}_{\tau,\varepsilon;t,s},\varphi\rangle
 :=
 \sum_{\tilde\varsigma\in\bF}\int_{[s,t]\times\bT_\tau^2} 
 \int_{\bR^2} G^{\pm;\varsigma,\tilde\varsigma}_{\tau,\varepsilon;u}(y-x)\varphi(y)\,\rd y\,\xi^{\pm,\tilde\varsigma}(\rd u,\rd x)
\end{equation}
for all $\varsigma\in\bF$ and $\varphi\in\sS(\bR^2)$. 
Here integration over $\bT_\tau^2$ against $\xi$ is understood using the identification of $\bT_\tau^2$ with the fundamental domain specified in Def.~\ref{dfn:distributions_torus}.
We call $\varPsi_{\tau,\varepsilon}:=\varPsi_{\tau,\varepsilon;1,0}$ the free field with IR cutoff $\tau\in[0,1]$ and UV cutoff $\varepsilon\in[0,1]$. We also define $\varPsi_{\tau,\varepsilon;s}:=\varPsi_{\tau,\varepsilon;1,s}$. We omit $\tau$ and $\varepsilon$ if $\tau=0$ and $\varepsilon=0$.
\end{dfn}

\begin{rem}\label{rem:Psi_finite}
Observe that by Remark~\ref{rem:G}~(B) for $\tau,\varepsilon\in(0,1]$ it holds
\begin{equation}
 \varPsi_{\tau,\varepsilon}^\sigma(x) := \tau^2\sum_{p\in\Lambda_{\tau,\varepsilon}} (\fF_\tau\varPsi_{\tau,\varepsilon}^\sigma)(p)\,\re^{\ri p\cdot x}.
\end{equation} 
In particular, $\varPsi_{\tau,\varepsilon}^{\otimes m}=0$ for all $m\in\bN_+$ provided $m>|\bFF|\,|\Lambda_{\tau,\varepsilon}|$, where $|S|$ denotes the number of elements of a set $S$.
\end{rem}

\begin{rem}
Let $\sG$ be a Grassmann algebra. For $\phi\in C^\infty(\bT_\tau^2)^{\bFF}\otimes_{\mathrm{alg}} \sG$ we define the pairing $\langle\varPsi_{\tau,\varepsilon;t,s},\phi\rangle_\tau\in \sF\otimes_{\mathrm{alg}}\sG$ by the equality
\begin{equation}
 \langle\varPsi_{\tau,\varepsilon;t,s},\phi\rangle_\tau
 :=
 \sum_{i=1}^n\langle\varPsi_{\tau,\varepsilon;t,s},\phi_i\rangle_\tau\otimes g_i,
 \qquad
 \phi = \sum_{i=1}^n \phi_i\otimes g_i,
\end{equation}
for all $n\in\bN_+$, $\phi_1,\ldots,\phi_n\in C^\infty(\bT_\tau^2)^\bFF$ and $g_1,\ldots,g_n\in \sG$. We identify $A\in\sF$ and $g\in\sG$ with 
\begin{equation}
 A\otimes 1\in \sF\otimes_{\mathrm{alg}}\sG
 \qquad\mathrm{and}\qquad
 \mathds{1}\otimes g\in \sF\otimes_{\mathrm{alg}}\sG,
\end{equation}
respectively, and use these identifications to make sense of $A+g,Ag\in \sF\otimes_{\mathrm{alg}}\sG$.
\end{rem}

\begin{rem}\label{rem:Wick}
Let $\tau,\varepsilon\in[0,1]$, $s\in[0,1]$ and $t\in[s,1]$ be such that $\varepsilon>0$ or $s>0$. Using Remark~\ref{rem:bound_noise}, the fact that $G^{\pm;\varsigma,\tilde\varsigma}_{\varepsilon;t}=0$ for $t\in(0,\varepsilon/2]$ by Remark~\ref{rem:G}~(B) and Lemma~\ref{lem:dot_G_infty} one shows  that $\varPsi_{\tau,\varepsilon;t,s}\in C^\infty(\bT_\tau^2,\sF_{t,s})^\bFF$. Note that it holds $\fE\varPsi^{\sigma}_{\tau,\varepsilon;t}(x)=0$ and
\begin{equation}\label{eq:covariance_Psi_H}
 \fE(\varPsi^{\sigma_1}_{\tau,\varepsilon;t}(x)\,\varPsi^{\sigma_2}_{\tau,\varepsilon;s}(y)) 
 =
 G^{\sigma_1,\sigma_2}_{\tau,\varepsilon;t}(x-y),
\end{equation} 
where $G_{\tau,\varepsilon;t}\in C^\infty(\bR^2)^{\bFF^2}$ was introduced in Def.~\ref{dfn:covariance}. Using the equality $\varPsi_{\tau,\varepsilon;t,s}=\varPsi_{\tau,\varepsilon;s}-\varPsi_{\tau,\varepsilon;t}$ one shows that it holds $\fE\varPsi^{\sigma}_{\tau,\varepsilon;t,s}(x)=0$ and 
\begin{multline}
 \fE(\varPsi^{\sigma_1}_{\tau,\varepsilon;t,s}(x)\,\varPsi^{\sigma_2}_{\tau,\varepsilon;t,s}(y)) 
 =
 \fE(\,(\varPsi^{\sigma_1}_{\tau,\varepsilon;s}(x)-\varPsi^{\sigma_1}_{\tau,\varepsilon;t}(x))\,(\varPsi^{\sigma_2}_{\tau,\varepsilon;s}(y)
 -\varPsi^{\sigma_2}_{\tau,\varepsilon;t}(y))\,)
 \\
 =
 (G_{\tau,\varepsilon;s}^{\sigma_1,\sigma_2}
 -
 G_{\tau,\varepsilon;t}^{\sigma_1,\sigma_2})(x-y).
\end{multline} 
The expected values of products of more than two fields $\varPsi_{\tau,\varepsilon;t,s}$ can be expressed in terms of the covariance with the use of the generating functional
\begin{equation}\label{eq:generating_functional_free_measure}
 C^\infty(\bT_\tau^2)^{\bFF}\otimes_{\mathrm{alg}} \sG^-\ni\phi\mapsto\fE\exp(\langle\varPsi_{\tau,\varepsilon;t,s},\phi\rangle_\tau)=\exp(-\langle \phi,(G_{\varepsilon;s}-G_{\varepsilon;t})\ast\phi\rangle_\tau/2)\in\sG,
\end{equation}
where $\sG$ is an infinite-dimensional Grassmann algebra. In order to prove~\eqref{eq:generating_functional_free_measure} one decomposes $\varPsi_{\tau,\varepsilon;t,s}$ into a sum of two terms $\varPsi_{\tau,\varepsilon;t,s}^{(\pm)}$ involving only the creation and annihilation operators, applies the Baker–Campbell–Hausdorff formula 
\begin{equation}
 \exp(\langle\varPsi_{\tau,\varepsilon;t,s},\phi\rangle_\tau)
 =
 \exp(\langle\varPsi_{\tau,\varepsilon;t,s}^{(+)},\phi\rangle_\tau)
 \,
 \exp(\langle\varPsi_{\tau,\varepsilon;t,s}^{(-)},\phi\rangle_\tau)
 \,
 \exp(-[\langle\varPsi_{\tau,\varepsilon;t,s}^{(+ )},\phi\rangle_\tau,\langle\varPsi_{\tau,\varepsilon;t,s}^{(-)},\phi\rangle_\tau]/2),
\end{equation}
where $[\Cdot,\Cdot]$ denotes the commutator, and uses the fact that the vacuum state $\Omega$ belongs to the intersection of the kernels of the annihilation operators.
\end{rem}

\begin{lem}\label{lem:free_measure}
Let $\tau,\varepsilon\in(0,1]$. For all $F\in\sN(C^\infty(\bT_\tau^2)^\bFF)$ it holds
\begin{equation}
 \nu_{\tau,\varepsilon}(F)
 \equiv
 \frac{\int F(\vartheta_\varepsilon\ast\psi_{\tau,\varepsilon})\exp(-A_{\tau}(\psi_{\tau,\varepsilon}))\,\rd\psi_{\tau,\varepsilon}}{\int \exp(-A_{\tau}(\psi_{\tau,\varepsilon}))\,\rd\psi_{\tau,\varepsilon}}
 = \fE F(\varPsi_{\tau,\varepsilon}).
\end{equation}
\end{lem}
\begin{rem}
Note that in particular it holds
\begin{equation}
 \int\psi^{\sigma_1}(x)\psi^{\sigma_2}(y)\, \nu_{\tau,\varepsilon}(\rd\psi)
 =
 G_{\tau,\varepsilon}^{\sigma_1,\sigma_2}(x-y).
\end{equation}
\end{rem}

\begin{proof}
In view of Remarks~\ref{rem:psi_finite} and~\ref{rem:Psi_finite} it
suffices to prove the stated equality for functionals
$F=(F^m)_{m\in\bN_0}\in\sN(C^\infty(\bT_\tau^2)^\bFF)$ such that
$F^m\neq0$ only for finitely many $m\in\bN_0$.
By polarization and Remark~\ref{rem:graded_uniqueness}, it is enough
to prove the corresponding generating-functional identity for all
\begin{equation}
 \eta\in
 C^\infty(\bT_\tau^2)^\bFF\otimes_{\mathrm{alg}}\sG^-.
\end{equation}
For such $\eta$ set
\begin{equation}
 F_\eta(\phi)
 :=
 \exp(\langle\phi,\eta\rangle_\tau),
 \qquad
 \phi\in
 C^\infty(\bT_\tau^2)^\bFF\otimes_{\mathrm{alg}}\sG^-.
\end{equation}
By~\cite[Theorem~A.16]{CCS13} it holds
\begin{equation}
 \frac{\int F_\eta(\psi_{\tau,\varepsilon})\exp(-A_{\tau}(\psi_{\tau,\varepsilon}))\,\rd\psi_{\tau,\varepsilon}}{\int \exp(-A_{\tau}(\psi_{\tau,\varepsilon}))\,\rd\psi_{\tau,\varepsilon}} 
 =
 \exp(-\langle \eta,\tilde G_{\tau,\varepsilon}\ast_\tau\eta\rangle_\tau/2),
\end{equation}
where $\fF_\tau\tilde G_{\tau,\varepsilon}(p):=1_{\Lambda_{\tau,\varepsilon}}(p)\,\fF G(p)$. Note that 
\begin{equation}
 F_\eta(\vartheta_\varepsilon\ast\psi_{\tau,\varepsilon})=F_{\vartheta_\varepsilon\ast\eta}(\psi_{\tau,\varepsilon}).
\end{equation}
Using the fact that $\supp\,\fF_\tau\fP_\tau\vartheta_\varepsilon\subset \Lambda_{\tau,\varepsilon}$ we show that
\begin{equation}
 \vartheta_\varepsilon\ast\tilde G_{\tau,\varepsilon}\ast \vartheta_\varepsilon
 =
 \vartheta_\varepsilon\ast
 G_{\tau}
 \ast\vartheta_\varepsilon
 =
 G_{\tau,\varepsilon}
\end{equation}
and
\begin{equation}
 \frac{\int F_\eta(\vartheta_\varepsilon\ast\psi_{\tau,\varepsilon})\exp(-A_{\tau}(\psi_{\tau,\varepsilon}))\,\rd\psi_{\tau,\varepsilon}}{\int \exp(-A_{\tau}(\psi_{\tau,\varepsilon}))\,\rd\psi_{\tau,\varepsilon}} 
 =
 \exp(-\langle \eta,G_{\tau,\varepsilon}\ast_\tau\eta\rangle_\tau/2).
\end{equation}
On the other hand,
\begin{equation}
 \fE F_\eta(\varPsi_{\tau,\varepsilon})=\exp(-\langle \eta,G_{\tau,\varepsilon}\ast_\tau\eta\rangle_\tau/2)
\end{equation}
by~\eqref{eq:generating_functional_free_measure} and Remark~\ref{rem:convolution_tau}. This finishes the proof.
\end{proof}

\begin{lem}\label{lem:expected_F_incremenet}
Let $\tau,\varepsilon\in(0,1]$, $s\in[0,1]$ and $t\in[s,1]$. For all $F\in\sN(C^\infty(\bT_\tau^2)^\bFF)$ and $\phi\in C^\infty(\bT_\tau^2)^{\bFF}\otimes_{\mathrm{alg}}\sG^-$ it holds
\begin{equation}
 \fE F(\varPsi_{\tau,\varepsilon;t,s}+\phi) = (\exp(\fD_{\tau,\varepsilon;t,s}/2)F)(\phi),
\end{equation}
where the map $\fD_{\tau,\varepsilon;t,s}$ is defined by the equation
\begin{equation}
 (\fD_{\tau,\varepsilon;t,s}F)(\phi) := \langle\rD_\phi^2 F(\phi),(G_{\tau,\varepsilon;s}-G_{\tau,\varepsilon;t})(\Cdot-\Cdot)\rangle_\tau
\end{equation}
for all $F\in\sN(C^\infty(\bT_\tau^2)^\bFF)$ and $\phi\in C^\infty(\bT_\tau^2)^{\bFF}\otimes_{\mathrm{alg}}\sG^-$.
\end{lem}
\begin{proof}
By polarization and Remark~\ref{rem:graded_uniqueness}, it suffices to prove the identity for all functionals of the form $F=F_\eta$ with $\eta\in C^\infty(\bT_\tau^2)^{\bFF}\otimes_{\mathrm{alg}} \sG^-$, where $F_\eta(\phi):=\exp(\langle\phi,\eta\rangle_\tau)$ for all $\phi\in C^\infty(\bT_\tau^2)^{\bFF}\otimes_{\mathrm{alg}} \sG^-$. Direct differentiation in the finite Grassmann algebra generated by the components of $\phi$ and $\eta$ gives
\begin{equation}
 \rD_\phi^2F_\eta(\phi)
 =
 -F_\eta(\phi)\,\eta^{\otimes2}.
\end{equation}
Indeed, for every
$\psi\in C^\infty(\bT_\tau^2)^{\bFF}\otimes_{\mathrm{alg}}\sG^-$,
\begin{equation}
 \langle\rD_\phi^2F_\eta(\phi),\psi^{\otimes2}\rangle_\tau
 =
 \langle\psi,\eta\rangle_\tau^2F_\eta(\phi)
 =
 -\langle\eta^{\otimes2},\psi^{\otimes2}\rangle_\tau F_\eta(\phi).
\end{equation}
Consequently,
\begin{equation}
 \fD_{\tau,\varepsilon;t,s}F_\eta
 =
 -
 \langle \eta,(G_{\tau,\varepsilon;s}-G_{\tau,\varepsilon;t})\ast_\tau\eta\rangle_\tau
 ~
 F_\eta.
\end{equation}
Since $\langle \eta,(G_{\tau,\varepsilon;s}-G_{\tau,\varepsilon;t})\ast_\tau\eta\rangle_\tau\in\sG^+$ we arrive at
\begin{multline}
 \exp(\fD_{\tau,\varepsilon;t,s}/2)F_\eta
 =
 \sum_{n=0}^\infty \frac{1}{2^n n!}\fD^n_{\tau,\varepsilon;t,s}F_\eta
 \\
 =
 \sum_{n=0}^\infty \frac{(-1)^n}{2^n n!}
 \langle \eta,(G_{\tau,\varepsilon;s}-G_{\tau,\varepsilon;t})\ast_\tau\eta\rangle_\tau^n
 ~
 F_\eta
 \\
 =
 \exp(-\langle \eta,(G_{\tau,\varepsilon;s}-G_{\tau,\varepsilon;t})\ast_\tau\eta\rangle_\tau/2)
 ~
 F_\eta.
\end{multline}
On the other hand,
\begin{equation}
 \fE F_\eta(\varPsi_{\tau,\varepsilon;t,s}+\phi)
 =
 \fE F_\eta(\varPsi_{\tau,\varepsilon;t,s})~F_\eta(\phi)
 =
 \exp(-\langle \eta,(G_{\tau,\varepsilon;s}-G_{\tau,\varepsilon;t})\ast_\tau\eta\rangle_\tau/2)
 ~
 F_\eta(\phi)
\end{equation}
by~\eqref{eq:generating_functional_free_measure} and Remark~\ref{rem:convolution_tau}. This finishes the proof.
\end{proof}

\begin{dfn}\label{dfn:sC}
The vector space $\sC$ consists of continuous functions $\phi\,:\,\bR^2\to\sB$ such that the following norm
\begin{equation}
 \|\phi\|_{\sC}:=\sup_{x\in\bR^2}\|\phi(x)\|_\sB
\end{equation}
is finite. For $\phi\in\sC$ we define $\|\phi\|_{\tilde\sC}:=\|\tilde w\,\phi\|_\sC$, where $\tilde w\in C(\bR^2)$ is defined by the equality $\tilde w(x)=(1+|x|)^{-1/2}$ for all $x\in\bR^2$.
\end{dfn}

\begin{dfn}\label{dfn:Besov}
We denote by $(\Delta_i)_{i\in\{-1,0,1,\ldots\}}$ the Littlewood--Paley blocks on $\bR^2$. For $\alpha\in\bR$ the Hölder--Besov space $\sC^\alpha$ consists of linear and continuous maps
\begin{equation}
 \phi\,:\,\sS(\bR^2) \to \sB
\end{equation}
such that the following norm
\begin{equation}
 \|\phi\|_{\sC^\alpha}:=\sup_{i\in\{-1,0,1,\ldots\}}2^{\alpha i}\,\|\Delta_i \phi\|_\sC
\end{equation}
is finite. For $\phi\in\sC^\alpha$ we define
\begin{equation}
 \|\phi\|_{\tilde\sC^\alpha}\equiv \|\phi\|_{\sC^\alpha(\tilde w)}:=\sup_{i\in\{-1,0,1,\ldots\}}2^{\alpha i}\,\|\tilde w\,(\Delta_i \phi)\|_\sC.
\end{equation}
\end{dfn}

\begin{lem}\label{lem:Psi_bounds}
There exists $C\in(0,\infty)$ such that for all $\lambda\in(0,1]$, $\tau,\varepsilon\in[0,1]$, $t\in(0,1]$, $s\in(0,t]$, $\alpha\in(-\infty,-1/2)$ and $a\in\bA:=\{a\in\mathbb N_0^2\,|\,|a|\leq2\}$, $\sigma\in\bFF$ it holds
\begin{enumerate}
 \item[(A)] $\|\partial^a\varPsi^\sigma_{\tau,\varepsilon;t,s}\|_\sC \leq C\, s^{-1/2-|a|}$,
 \item[(B)] $\|\tilde w(\partial^a\varPsi^\sigma_{t,s}-
 \partial^a\varPsi^\sigma_{\tau,\varepsilon;t,s})\|_\sC \leq C\,\lambda_{\tau\vee\varepsilon}^\kappa\,\lambda_s^{-\kappa}\,s^{-1/2-|a|}$,
 \item[(C)] $\varPsi^\sigma\in \sC^{-1/2}$ and $\lim_{\tau,\varepsilon\searrow0}\|\varPsi^\sigma-\varPsi^\sigma_{\tau,\varepsilon}\|_{\tilde\sC^\alpha}=0$,
 \item[(D)] $\limsup_{i\to\infty} 2^{-i/2} \|\Delta_i \varPsi^\sigma\|_{\sC}>0$,
\end{enumerate}
where $\kappa$ is fixed in Def.~\ref{dfn:kappa} and $\lambda_s$ is introduced in Def.~\ref{dfn:lambda}.
\end{lem}
\begin{proof}
Suppose that $\sigma=(\pm,\varsigma)$ for some $\varsigma\in\bF$. By Remark~\ref{rem:bound_noise} we obtain
\begin{equation}\label{eq:Psi_proof}
 \|\partial^a\varPsi^\sigma_{\tau,\varepsilon;t,s}\|^2_\sC
 =
 \sup_{x\in\bR^2}\|\partial^a\varPsi_{\tau,\varepsilon;t,s}^{\pm,\varsigma}(x)\|_\sB^2
 \leq 2
 \sup_{x\in\bR^2}\sum_{\tilde\varsigma\in\bF}\int_s^t\int_{\bT_\tau^2}|\partial^a G^{\pm;\varsigma,\tilde\varsigma}_{\tau,\varepsilon;u}(x-y)|^2\,\rd y\,\rd u.
\end{equation}
For $\tau=0$ the required bound follows directly from Lemma~\ref{lem:dot_G_infty}. Let $\tau\in(0,1]$. By Minkowski's inequality and the definition of the periodization,
\begin{equation}
 \left(
 \int_{\bT_\tau^2}
 |\partial^a G^{\pm;\varsigma,\tilde\varsigma}_{\tau,\varepsilon;u}(y)|^2
 \,\rd y
 \right)^{1/2}
 \leq
 \sum_{n\in\bZ^2}
 \left(
 \int_{\bT_\tau^2}
 |\partial^a G^{\pm;\varsigma,\tilde\varsigma}_{\varepsilon;u}(y+n/\tau)|^2
 \,\rd y
 \right)^{1/2}.
\end{equation}
Using Lemma~\ref{lem:dot_G_infty} and the fact that $|y+n/\tau|_\infty\geq (|n|_\infty-1/2)/\tau$ for $y\in\bT_\tau^2$ and $n\in\bZ^2\setminus\{0\}$, we obtain
\begin{multline}
 \left(
 \int_{\bT_\tau^2}
 |\partial^a G^{\pm;\varsigma,\tilde\varsigma}_{\tau,\varepsilon;u}(y)|^2
 \,\rd y
 \right)^{1/2}
 \\
 \leq 
 C\,u^{-1-|a|}
 +
 C\,\tau^{-1}u^{-2-|a|}
 \sum_{n\in\bZ^2\setminus\{0\}}
 \exp\left( -c\left(\frac{|n|_\infty}{\tau u}\right)^{\zeta_\flat}
 \right)
 \leq
 C\,u^{-1-|a|}.
\end{multline}
Consequently,
\begin{equation}
 \int_{\bT_\tau^2}
 |\partial^a G^{\pm;\varsigma,\tilde\varsigma}_{\tau,\varepsilon;u}(y) |^2
 \,\rd y 
 \leq
 C\,u^{-2-2|a|}.
\end{equation}
As a result, we conclude that there exists $C\in(0,\infty)$ such that for all $\varepsilon\in[0,1]$, $t\in(0,1]$, $s\in(0,t]$ and $a\in\bA$, $\sigma\in\bFF$ it holds
\begin{equation}\label{eq:Psi_proof_bound}
 \|\partial^a\varPsi^\sigma_{\tau,\varepsilon;t,s}\|_\sC \leq C\, s^{-1/2-|a|}.
\end{equation}
This proves Item~(A). 

Next, we observe that
\begin{multline}\label{eq:Psi_diff_proof}
 \|\partial^a\varPsi^\sigma_{t,s}-
 \partial^a\varPsi^\sigma_{\varepsilon;t,s}\|^2_\sC
 =
 \sup_{x\in\bR^2}\|\partial^a\varPsi_{t,s}^{\pm,\varsigma}(x)-\partial^a\varPsi_{\varepsilon;t,s}^{\pm,\varsigma}(x)\|_\sB^2
 \\
 \leq 2
 \sup_{x\in\bR^2}\sum_{\tilde\varsigma\in\bF}\int_s^t\int_{\bR^2}
 |\partial^a G^{\pm;\varsigma,\tilde\varsigma}_{u}(x-y)
 -\partial^a G^{\pm;\varsigma,\tilde\varsigma}_{\varepsilon;u}(x-y)|^2
 \,\rd y\,\rd u.
\end{multline}
The integrand on the RHS of the above equality vanishes identically if
$u\in[4\varepsilon,1]$ by Remark~\ref{rem:G}~(C). Set
$q_\varepsilon:=(4\varepsilon)\wedge1$. Using
Lemma~\ref{lem:dot_G_infty} and the bound proved above we obtain
\begin{equation}
 \|\partial^a\varPsi^\sigma_{t,s}-
 \partial^a\varPsi^\sigma_{\varepsilon;t,s}\|_\sC^2
 \leq
 C\,1_{\{s\leq q_\varepsilon\}}
 \int_s^{t\wedge q_\varepsilon}u^{-2-2|a|}\,\rd u
 \leq
 C\,1_{\{s\leq q_\varepsilon\}}s^{-1-2|a|}.
\end{equation}
If $\varepsilon>0$ and $s\leq q_\varepsilon$, then Def.~\ref{dfn:lambda} implies
\begin{equation}
 1\leq C\,\lambda_\varepsilon^{2\kappa}\lambda_s^{-2\kappa}.
\end{equation}
For $\varepsilon=0$ the difference on the left-hand side vanishes.
Consequently,
\begin{equation}\label{eq:Psi2_proof_bound}
 \|\partial^a\varPsi^\sigma_{t,s}-
 \partial^a\varPsi^\sigma_{\varepsilon;t,s}\|_\sC
 \leq
 C\,\lambda_\varepsilon^\kappa\,\lambda_s^{-\kappa}\,
 s^{-1/2-|a|}.
\end{equation}
Next, note that
\begin{multline}\label{eq:Psi_diff2_proof}
 \|\tilde w(\partial^a\varPsi^\sigma_{\varepsilon;t,s}-
 \partial^a\varPsi^\sigma_{\tau,\varepsilon;t,s})\|^2_\sC
 =
 \sup_{x\in\bR^2}\tilde w(x)^2\,\|\partial^a\varPsi_{\varepsilon;t,s}^{\pm,\varsigma}(x)-\partial^a\varPsi_{\tau,\varepsilon;t,s}^{\pm,\varsigma}(x)\|_\sB^2
 \\
 \leq 2
 \sup_{x\in\bR^2}\tilde w(x)^2\sum_{\tilde\varsigma\in\bF}\int_s^t\int_{\bR^2}
 |\partial^a G^{\pm;\varsigma,\tilde\varsigma}_{\varepsilon;u}(x-y)
 -\partial^a G^{\pm;\varsigma,\tilde\varsigma}_{\tau,\varepsilon;u}(x-y)\,1_{\bT^2_{\tau}}(y)|^2
 \,\rd y\,\rd u,
\end{multline}
where $1_{\bT^2_{\tau}}$ is the characteristic function of the set $\bT^2_{\tau}=(-1/(2\tau),1/(2\tau))^2$.
Observe that $\tilde w(x)^2\leq\tilde w(y)^2\, \tilde w(x-y)^{-2}$. Consequently, it holds
\begin{multline}
 \sup_{x\in\bR^2}\tilde w(x)^2\int_{\bR^2}
 |\partial^a G^{\pm;\varsigma,\tilde\varsigma}_{\varepsilon;u}(x-y)
 -\partial^a G^{\pm;\varsigma,\tilde\varsigma}_{\tau,\varepsilon;u}(x-y)\,1_{\bT^2_{\tau}}(y)|^2
 \,\rd y
 \\
 \leq
 \sup_{x\in\bT^2_{2\tau}}
 \int_{\bT^2_{\tau}}
 |\partial^a G^{\pm;\varsigma,\tilde\varsigma}_{\varepsilon;u}(x-y)
 -\partial^a G^{\pm;\varsigma,\tilde\varsigma}_{\tau,\varepsilon;u}(x-y)|^2
 \,\rd y
 \\
 +
 \sup_{x\in\bR^2\setminus\bT^2_{2\tau}}
 \tilde w(x)^2\int_{\bT^2_{\tau}}
 |\partial^a G^{\pm;\varsigma,\tilde\varsigma}_{\varepsilon;u}(x-y)
 -\partial^a G^{\pm;\varsigma,\tilde\varsigma}_{\tau,\varepsilon;u}(x-y)|^2
 \,\rd y
 \\
 +
 \sup_{x\in\bR^2}\int_{\bR^2\setminus\bT^2_{\tau}}
 \tilde w(y)^2\, \tilde w(x-y)^{-2}\,
 |\partial^a G^{\pm;\varsigma,\tilde\varsigma}_{\varepsilon;u}(x-y)|^2
 \,\rd y.
\end{multline}
Using the fact that $\tilde w(x)=(1+|x|)^{-1/2}$ we obtain
\begin{multline}
 \sup_{x\in\bR^2}\tilde w(x)^2\int_{\bR^2}
 |\partial^a G^{\pm;\varsigma,\tilde\varsigma}_{\varepsilon;u}(x-y)
 -\partial^a G^{\pm;\varsigma,\tilde\varsigma}_{\tau,\varepsilon;u}(x-y)\,1_{\bT^2_{\tau}}(y)|^2
 \,\rd y
 \\
 \leq
 \sup_{x\in\bT^2_{2\tau}}\int_{\bT^2_{\tau}}
 |\partial^a G^{\pm;\varsigma,\tilde\varsigma}_{\varepsilon;u}(x-y)
 -\partial^a G^{\pm;\varsigma,\tilde\varsigma}_{\tau,\varepsilon;u}(x-y)|^2
 \,\rd y
 \\
 +
 C\tau
 \left(
 \int_{\bR^2}
 |\partial^a G^{\pm;\varsigma,\tilde\varsigma}_{\varepsilon;u}(y)|^2
 \,\rd y
 +
 \int_{\bT^2_{\tau}}
 |\partial^a G^{\pm;\varsigma,\tilde\varsigma}_{\tau,\varepsilon;u}(y)|^2
 \,\rd y
 \right)
 \\
 +
 \tau \int_{\bR^2}
 \tilde w(y)^{-2}\,
 |\partial^a G^{\pm;\varsigma,\tilde\varsigma}_{\varepsilon;u}(y)|^2
 \,\rd y.
\end{multline}
We may assume that $\tau\in(0,1]$, since the claim is trivial for
$\tau=0$. For all $x\in\bT^2_{2\tau}$,
$y\in\bT^2_\tau$ and $n\in\bZ^2\setminus\{0\}$ it holds
$|x-y+n/\tau|_\infty\geq |n|_\infty/(4\tau)$. Hence,
Lemma~\ref{lem:dot_G_infty} yields
\begin{equation}
 \sup_{x\in\bT^2_{2\tau}}\sup_{y\in\bT^2_\tau}
 |\partial^a G^{\pm;\varsigma,\tilde\varsigma}_{\varepsilon;u}(x-y)
 -
 \partial^a G^{\pm;\varsigma,\tilde\varsigma}_{\tau,\varepsilon;u}(x-y)|
 \leq
 C\,u^{-2-|a|}
 \exp\left(-c(\tau u)^{-\zeta_\flat}\right).
\end{equation}
Consequently,
\begin{equation}
 \sup_{x\in\bT^2_{2\tau}}
 \int_{\bT^2_\tau}
 |\partial^a G^{\pm;\varsigma,\tilde\varsigma}_{\varepsilon;u}(x-y)
 -
 \partial^a G^{\pm;\varsigma,\tilde\varsigma}_{\tau,\varepsilon;u}(x-y)|^2
 \,\rd y
 \leq
 C\,\tau\,u^{-2-2|a|}.
\end{equation}
Indeed, the bound before the last inequality is
$C\tau^{-2}u^{-4-2|a|}
\exp(-c(\tau u)^{-\zeta_\flat})$, and the stretched exponential
absorbs the additional powers of $\tau^{-1}$ and $u^{-1}$.

Moreover, by Lemma~\ref{lem:dot_G_infty} and the estimate proved in Item~(A),
\begin{equation}
 \int_{\bR^2}
 |\partial^a G^{\pm;\varsigma,\tilde\varsigma}_{\varepsilon;u}(y)|^2
 \,\rd y
 +
 \int_{\bT^2_\tau}
 |\partial^a G^{\pm;\varsigma,\tilde\varsigma}_{\tau,\varepsilon;u}(y)|^2
 \,\rd y
 \leq
 C\,u^{-2-2|a|}.
\end{equation}
Moreover, Lemma~\ref{lem:dot_G_infty} gives
\begin{equation}
 \int_{\bR^2}\tilde w(y)^{-2}
 |\partial^a G^{\pm;\varsigma,\tilde\varsigma}_{\varepsilon;u}(y)|^2
 \,\rd y
 \leq
 C\,u^{-2-2|a|}.
\end{equation}
Inserting these three bounds into the preceding decomposition yields
\begin{equation}
 \sup_{x\in\bR^2}\tilde w(x)^2
 \int_{\bR^2}
 |\partial^a G^{\pm;\varsigma,\tilde\varsigma}_{\varepsilon;u}(x-y)
 -
 \partial^a G^{\pm;\varsigma,\tilde\varsigma}_{\tau,\varepsilon;u}(x-y)
 1_{\bT^2_\tau}(y)|^2
 \,\rd y
 \leq
 C\,\tau\,u^{-2-2|a|}.
\end{equation}
Therefore, by~\eqref{eq:Psi_diff2_proof},
\begin{equation}\label{eq:Psi3_proof_bound}
 \|\tilde w(\partial^a\varPsi^\sigma_{\varepsilon;t,s}-
 \partial^a\varPsi^\sigma_{\tau,\varepsilon;t,s})\|_\sC
 \leq
 C\,\tau^{1/2}\,s^{-1/2-|a|}.
\end{equation}
Item~(B) follows now from the bounds~\eqref{eq:Psi2_proof_bound}, \eqref{eq:Psi3_proof_bound}, the inequality $\tilde w\leq1$ and the bounds
\begin{equation}
 \lambda_\varepsilon^\kappa
 \leq
 \lambda_{\tau\vee\varepsilon}^\kappa,
 \qquad
 \tau^{1/2}
 \leq
 C\,\lambda_\tau^\kappa\lambda_s^{-\kappa}
 \leq
 C\,\lambda_{\tau\vee\varepsilon}^\kappa\lambda_s^{-\kappa}.
\end{equation}

Item~(C) follows from a direct Littlewood--Paley estimate. Let $\chi_i$ denote the Fourier multiplier of $\Delta_i$. By Remark~\ref{rem:bound_noise}, Plancherel's identity and Def.~\ref{dfn:covariance},
\begin{equation}
 \|\Delta_i\varPsi^\sigma_{\tau,\varepsilon}\|_\sC^2
 \leq
 C\,2^i
\end{equation}
uniformly in $\tau,\varepsilon\in[0,1]$ and $i\in\{-1,0,1,\ldots\}$. Indeed, for $\tau=0$ the left-hand side is bounded by a constant times
\begin{equation}
 \sum_{\tilde\varsigma\in\bF}
 \int_{\bR^2}
 |\chi_i(p)|^2
 \int_0^1
 |(\fF G^{\pm;\varsigma,\tilde\varsigma}_{\varepsilon;u})(p)|^2
 \,\rd u\,\rd p,
\end{equation}
and
\begin{equation}
 \sum_{\tilde\varsigma\in\bF}
 \int_0^1
 |(\fF G^{\pm;\varsigma,\tilde\varsigma}_{\varepsilon;u})(p)|^2
 \,\rd u
 \leq
 \frac{C}{\omega(p)}.
\end{equation}
For $\tau>0$ the same argument gives
\begin{equation}
 \|\Delta_i\varPsi^\sigma_{\tau,\varepsilon}\|_\sC^2
 \leq
 C\,\tau^2
 \sum_{p\in(2\pi\tau\bZ)^2}
 \frac{|\chi_i(p)|^2}{\omega(p)}
 \leq
 C\,2^i.
\end{equation}
It follows that
\begin{equation}
 \sup_{i\in\{-1,0,1,\ldots\}}
 2^{-i/2}\|\Delta_i\varPsi^\sigma\|_\sC
 <\infty,
\end{equation}
and hence $\varPsi^\sigma\in\sC^{-1/2}$. We next prove convergence of every fixed Littlewood--Paley block. For $r\in(0,1]$ write
\begin{equation}
 \varPsi_{\tau,\varepsilon}
 =
 \varPsi_{\tau,\varepsilon;1,r}
 +
 \varPsi_{\tau,\varepsilon;r,0}.
\end{equation}
The same Fourier estimate as above yields, for every fixed $i$,
\begin{equation}
 \|\Delta_i\varPsi^\sigma_{\tau,\varepsilon;r,0}\|_\sC
 \leq
 C\,r^{1/2}2^i
\end{equation}
uniformly in $\tau,\varepsilon\in[0,1]$. Consequently, Item~(B) implies
\begin{equation}
 \lim_{\tau,\varepsilon\searrow0}
 \|\tilde w\Delta_i
 (\varPsi^\sigma-\varPsi^\sigma_{\tau,\varepsilon})\|_\sC
 =0
\end{equation}
for every fixed $i$: first let $\tau,\varepsilon\searrow0$ with $r>0$ fixed and then let $r\searrow0$. Finally, the uniform estimate above and $\tilde w\leq1$ imply
\begin{equation}
 2^{\alpha i}
 \|\tilde w\Delta_i
 (\varPsi^\sigma-\varPsi^\sigma_{\tau,\varepsilon})\|_\sC
 \leq
 C\,2^{(\alpha+1/2)i}.
\end{equation}
Since $\alpha<-1/2$, the right-hand side converges to zero as $i\to\infty$, uniformly in $\tau,\varepsilon$. Combining this high-frequency estimate with the convergence of each fixed block proves
\begin{equation}
 \lim_{\tau,\varepsilon\searrow0}
 \|\varPsi^\sigma-\varPsi^\sigma_{\tau,\varepsilon}\|_{\tilde\sC^\alpha}
 =0.
\end{equation}
This proves Item~(C).

To prove Item~(D), let $\chi_i$ denote the Fourier multiplier of $\Delta_i$. By Remark~\ref{rem:bound_noise},
\begin{equation}
 \|\Delta_i\varPsi^\sigma\|_\sC^2
 \geq
 \sum_{\tilde\varsigma\in\bF}
 \int_0^1\int_{\bR^2}
 |(\Delta_i
 G^{\pm;\varsigma,\tilde\varsigma}_{u})(y)|^2
 \,\rd y\,\rd u.
\end{equation}
By Plancherel's identity and Def.~\ref{dfn:covariance}, for all sufficiently large $i$ the right-hand side is bounded from below by
\begin{equation}
 c
 \int_{\bR^2}
 \frac{|\chi_i(p)|^2}{\omega(p)}
 \,\rd p.
\end{equation}
Indeed, for the minus component this follows directly from
\begin{equation}
 (\fF G^{-;\varsigma,\tilde\varsigma}_{u})(p)
 =
 \vartheta(u\omega(p)/2)\,1^{\varsigma,\tilde\varsigma},
\end{equation}
while for the plus component one uses
\begin{equation}
 (\fF G^{+;\varsigma,\tilde\varsigma}_{u})(p)
 =
 -\dot\vartheta(u\omega(p))
 \frac{\Xi^{\varsigma,\tilde\varsigma}(p)}{\omega(p)}
\end{equation}
and the fact that $\dot\vartheta$ does not vanish identically on $(1/2,1)$. Since $\chi_i$ is supported at frequencies of order $2^i$,
\begin{equation}
 \int_{\bR^2}
 \frac{|\chi_i(p)|^2}{\omega(p)}
 \,\rd p
 \geq
 c\,2^i
\end{equation}
for all sufficiently large $i$. Hence
\begin{equation}
 \limsup_{i\to\infty}
 2^{-i/2}\|\Delta_i\varPsi^\sigma\|_\sC
 >0,
\end{equation}
which proves Item~(D).
\end{proof}

\section{Effective potential and Polchinski equation}\label{sec:effective}

In this section we introduce the notion of an effective potential and prove that a solution of the mild form of the Polchinski equation~\eqref{eq:polchinski_integral} satisfies the equation for an effective potential. A solution $U_{\tau,\varepsilon;\Cdot}$ of the Polchinski equation will be constructed in Sec.~\ref{sec:relation_polchinski} in terms of the fixed point $X_{\tau,\varepsilon;\Cdot}$ of the map $\fX_{\tau,\varepsilon;\Cdot}$ constructed in Sec.~\ref{sec:fixed_point}. Then by the results of this section such $U_{\tau,\varepsilon;t}$ is an effective potential at the scale $t$. Using this fact in Sec.~\ref{sec:convergence} we will show that $U_{\tau,\varepsilon;t=1}$ is the generating functional of the connected amputated Schwinger functions.

\begin{lem}\label{lem:interacting_measure}
Let $\tau,\varepsilon\in(0,1]$. Suppose that $\fE\exp(U_{\tau,\varepsilon}(\varPsi_{\tau,\varepsilon}))
 \neq 0$. The interacting measure satisfies the equality
\begin{equation}
 \mu_{\tau,\varepsilon}(F) 
 = \frac{\fE(F(\varPsi_{\tau,\varepsilon})\exp(U_{\tau,\varepsilon}(\varPsi_{\tau,\varepsilon})))}{\fE\exp(U_{\tau,\varepsilon}(\varPsi_{\tau,\varepsilon}))}\in\bC
\end{equation}
for all functionals $F\in\sN(C^\infty(\bT_\tau^2)^\bFF)$.
\end{lem}
\begin{proof}
The lemma is a consequence of Def.~\ref{dfn:measures} of the measure $\mu_{\tau,\varepsilon}$ and Lemma~\ref{lem:free_measure}.
\end{proof}

\begin{dfn}\label{dfn:effective_potential}
Let $\tau,\varepsilon\in(0,1]$ and $U_{\tau,\varepsilon}\equiv U_{\tau}(g_{\tau,\varepsilon},r_{\tau,\varepsilon})\in\sN(C^\infty(\bT_\tau^2)^\bFF)$ be defined as specified in Def.~\ref{dfn:measures}. We call a functional $U_{\tau,\varepsilon;t}\in\sN(C^\infty(\bT_\tau^2)^\bFF)$ an effective potential at the spatial scale $t\in[0,1]$ (of the Gross--Neveu model with the parameters $g_{\tau,\varepsilon},r_{\tau,\varepsilon}\in\bR$) if it satisfies the following equality
\begin{equation}
 \exp(U_{\tau,\varepsilon;t}(\phi))
 =
 \fE\exp(U_{\tau,\varepsilon}(\varPsi_{\tau,\varepsilon;t,0}+\phi))\in\sG
\end{equation}
for all $\phi\in\mathcal C_{\tau,\varepsilon}\otimes_{\mathrm{alg}}\sG^-$,
where $\mathcal C_{\tau,\varepsilon}$ was introduced in Def.~\ref{dfn:GN_grassmann} and $\sG$ is an infinite-dimensional Grassmann algebra.
\end{dfn}

\begin{rem}
Note that for $s\leq t$ it holds
\begin{equation}
 \exp(U_{\tau,\varepsilon;s}(\varPsi_{\tau,\varepsilon;t,s}+\phi))
 =
 \fE_s\exp(U_{\tau,\varepsilon}(\varPsi_{\tau,\varepsilon;t,0}+\phi)).
\end{equation}
Consequently, for $s\leq t$  an effective potential fulfills the following identity
\begin{equation}
 \exp(U_{\tau,\varepsilon;t}(\phi))
 =
 \fE\exp(U_{\tau,\varepsilon;s}(\varPsi_{\tau,\varepsilon;t,s}+\phi)).
\end{equation} 
\end{rem}

\begin{rem}\label{rem:dot_G_finite}
For all $\tau,\varepsilon,t\in(0,1]$ we have
\begin{equation}
 \dot G_{\tau,\varepsilon;t}(\Cdot-\Cdot)\in \mathcal{C}_{\tau,\varepsilon}\otimes \mathcal{C}_{\tau,\varepsilon}\subset C^\infty(\bR^2\times\bR^2),
\end{equation}
where the function $G_{\tau,\varepsilon;t}(\Cdot-\Cdot)\in C(\bR^2\times\bR^2)$ coincides with the map $(x,y)\mapsto G_{\tau,\varepsilon;t}(x-y)$.
\end{rem}

\begin{lem}\label{lem:polchinski}
Let $\tau,\varepsilon\in(0,1]$ and $\sG$ be an infinite-dimensional Grassmann algebra. Suppose that the function $[0,1]\ni t\mapsto U_{\tau,\varepsilon;t}\in \sN(C^\infty(\bT_\tau^2)^\bFF)$ is such that $U_{\tau,\varepsilon;t}^m=0$ for all odd $m$ and all $t\in[0,1]$, and satisfies the Polchinski flow equation
\begin{equation}\label{eq:polchinski}
 \partial_t U_{\tau,\varepsilon;t}(\phi)
 =- \frac{1}{2} \langle\rD_\phi^2 U_{\tau,\varepsilon;t}(\phi),\dot G_{\tau,\varepsilon;t}(\Cdot-\Cdot)\rangle_\tau
 +\frac{1}{2} \langle \rD_\phi U_{\tau,\varepsilon;t}(\phi),\dot G_{\varepsilon;t} \ast \rD_\phi U_{\tau,\varepsilon;t}(\phi)\rangle_\tau
\end{equation}
with the initial condition $U_{\tau,\varepsilon;0}(\phi)=U_{\tau,\varepsilon}(\phi)$ for all $\phi\in C^\infty(\bT^2_\tau)^\bFF\otimes_{\mathrm{alg}}\sG^-$. Then for every $t\in[0,1]$ the functional \mbox{$U_{\tau,\varepsilon;t}\in \sN(C^\infty(\bT_\tau^2)^\bFF)$} is an effective potential at the spatial scale $t$.
\end{lem}
\begin{proof}
Let us first observe that for every $U\in \sN(C^\infty(\bT_\tau^2)^\bFF)$ such that $U(\phi)\in\sG^+$ for all $\phi\in C^\infty(\bT^2_\tau)^\bFF\otimes_{\mathrm{alg}}\sG^-$ it holds
\begin{equation}
\begin{gathered}
 \rD_\phi \exp(U(\phi)) = \exp(U(\phi)) \,\rD_\phi U(\phi) \in\sS'(\bT_\tau^2)^\bFF\otimes_{\mathrm{alg}}\sG^-,
 \\
 \rD_\phi^2 \exp(U(\phi)) = \exp(U(\phi)) \,\rD^2_\phi U(\phi)
 -
 \exp(U(\phi)) \,(\rD_\phi U(\phi)\otimes \rD_\phi U(\phi))\in\sS'(\bT_\tau^4)^{\bFF^2}\otimes_{\mathrm{alg}}\sG^+.
\end{gathered} 
\end{equation}
In consequence, the function $[0,1]\ni t\mapsto U_{\tau,\varepsilon;t}\in\sN(C^\infty(\bT_\tau^2)^\bFF)$ satisfies the Polchinski flow equation \eqref{eq:polchinski} for all $\phi\in C^\infty(\bT^2_\tau)^\bFF\otimes_{\mathrm{alg}}\sG^-$ if and only if it satisfies the equation
\begin{equation}
 \partial_t \exp(U_{\tau,\varepsilon;t}(\phi)) =
 - \frac{1}{2} \langle\rD_\phi^2 \exp(U_{\tau,\varepsilon;t}(\phi)),\dot G_{\tau,\varepsilon;t}(\Cdot-\Cdot)\rangle_\tau
\end{equation}
for all $t\in(0,1]$ and $\phi\in
C^\infty(\bT^2_\tau)^\bFF\otimes_{\mathrm{alg}}\sG^-$. Note that by Lemma~\ref{lem:expected_F_incremenet} the equation
\begin{equation}
 \partial_t \fE\exp(U_{\tau,\varepsilon}(\varPsi_{\tau,\varepsilon;t,0}+\phi)) =
 - \frac{1}{2} \langle\rD_\phi^2 \fE\exp(U_{\tau,\varepsilon}(\varPsi_{\tau,\varepsilon;t,0}+\phi)),\dot G_{\tau,\varepsilon;t}(\Cdot-\Cdot)\rangle_\tau
\end{equation}
holds true for all $t\in(0,1]$ and $\phi\in
C^\infty(\bT^2_\tau)^\bFF\otimes_{\mathrm{alg}}\sG^-$. The statement follows now from Lemma~\ref{lem:polchinski_Z}.
\end{proof}

\begin{lem}\label{lem:polchinski_Z}
Let $\tau,\varepsilon\in(0,1]$ and $\sG$ be an infinite-dimensional Grassmann algebra. Suppose that the function $[0,1]\ni t\mapsto Z_{\tau,\varepsilon;t}\in\sN(C^\infty(\bT_\tau^2)^\bFF)$ satisfies the equation
\begin{equation}\label{eq:Polchinski_Z}
 \partial_t Z_{\tau,\varepsilon;t}(\phi) =
 - \frac{1}{2} \langle\rD_\phi^2 Z_{\tau,\varepsilon;t}(\phi),\dot G_{\tau,\varepsilon;t}(\Cdot-\Cdot)\rangle_\tau
\end{equation}
with the initial condition $Z_{\tau,\varepsilon;0}(\phi)=0$ for all $\phi\in \mathcal{C}_{\tau,\varepsilon}\otimes_{\mathrm{alg}}\sG^-$. Then $Z_{\tau,\varepsilon;t}(\phi)=0$ for all $t\in[0,1]$ and $\phi\in \mathcal{C}_{\tau,\varepsilon}\otimes \sG^-$.
\end{lem}

\begin{proof}
Recall that a functional $Z_{\tau,\varepsilon;t}\in\sN(C^\infty(\bT_\tau^2)^\bFF)$ is a collection $Z_{\tau,\varepsilon;t}=(Z_{\tau,\varepsilon;t}^m)_{m\in\bN_0}$ such that $Z_{\tau,\varepsilon;t}^0\in\bC$ and $Z_{\tau,\varepsilon;t}^m\in \sS'(\bT_\tau^{2m})^{\bFF^m}$ is antisymmetric for $m\in\bN_+$. Next, set $m_{\tau,\varepsilon}:=|\Lambda_{\tau,\varepsilon}||\bFF|=\dim(\mathcal{C}_{\tau,\varepsilon})$ and let $\{e_1,\ldots,e_{m_{\tau,\varepsilon}}\}$ be a basis of $\mathcal{C}_{\tau,\varepsilon}$. For $m\in\bN_+$ and $i_1,\ldots,i_m\in\{1,\ldots,m_{\tau,\varepsilon}\}$ define
\begin{equation}
 Z^{(i_1,\ldots,i_m)}_{\tau,\varepsilon;t}
 :=
 \left\langle
 Z^m_{\tau,\varepsilon;t},
 e_{i_1}\otimes\ldots\otimes e_{i_m}
 \right\rangle.
\end{equation}
Observe that $Z^{(i_1,\ldots,i_m)}_{\tau,\varepsilon;t}$ vanishes identically if $m>m_{\tau,\varepsilon}$. Hence, by Remark~\ref{rem:dot_G_finite}, Eq.~\eqref{eq:Polchinski_Z} implies that the finite collection
\begin{equation}
 \{Z^0_{\tau,\varepsilon;\Cdot}\}\cup \{Z^{(i_1,\ldots,i_m)}_{\tau,\varepsilon;\Cdot}\,|\,m,i_1,\ldots,i_m\in\{1,\ldots,m_{\tau,\varepsilon}\}\}
\end{equation}
satisfies a first order linear ODE with a trivial boundary condition. This proves the claim.
\end{proof}

\begin{lem}\label{lem:polchinski2}
Let $\tau,\varepsilon\in(0,1]$ and $\sG$ be an infinite-dimensional Grassmann algebra. Suppose that a continuous function $[0,1]\ni t\mapsto U_{\tau,\varepsilon;t}\in\sN(C^\infty(\bT_\tau^2)^\bFF)$ is such that $U_{\tau,\varepsilon;t}^m=0$ for all odd $m$ and $t\in[0,1]$, and satisfies the integral form of the Polchinski equation
\begin{multline}\label{eq:polchinski_integral}
 U_{\tau,\varepsilon;t}(\phi) = \fE U_{\tau,\varepsilon}(\varPsi_{\tau,\varepsilon;t,0}+\phi) 
 \\
 + \frac{1}{2} \int_0^t \fE\langle \rD_\phi U_{\tau,\varepsilon;s}(\varPsi_{\tau,\varepsilon;t,s}+\phi),\dot G_{\varepsilon;s}\ast \rD_\phi U_{\tau,\varepsilon;s}(\varPsi_{\tau,\varepsilon;t,s}+\phi)\rangle_\tau\,\rd s
\end{multline} 
for all $\phi\in C^\infty(\bT^2_\tau)^\bFF\otimes_{\mathrm{alg}}\sG^-$. Then for every $t\in[0,1]$ the functional $U_{\tau,\varepsilon;t}\in \sN(C^\infty(\bT_\tau^2)^\bFF)$ is an effective potential at the spatial scale $t$.
\end{lem}
\begin{proof}
For $s\in[0,1]$ define
\begin{equation}
 K_s(\phi)
 :=
 \langle
 \rD_\phi U_{\tau,\varepsilon;s}(\phi),
 \dot G_{\varepsilon;s}\ast
 \rD_\phi U_{\tau,\varepsilon;s}(\phi)
 \rangle_\tau.
\end{equation}
Since $\dot G_{\varepsilon;s}=0$ for $s\leq\varepsilon$,  the integral in~\eqref{eq:polchinski_integral} is effectively over $s\in[\varepsilon,t]$. By continuity of the map $s\mapsto U_{\tau,\varepsilon;s}$ in $\sN$, the map $s\mapsto K_s$ is continuous.

For fixed $F\in\sN(C^\infty(\bT_\tau^2)^\bFF)$, Lemma~\ref{lem:expected_F_incremenet} implies that $t\mapsto\fE F(\varPsi_{\tau,\varepsilon;t,s}+\phi)$ is continuously differentiable for $t\geq s$, with
\begin{equation}
 \partial_t
 \fE F(\varPsi_{\tau,\varepsilon;t,s}+\phi)
 =
 -\frac12
 \langle
 \rD_\phi^2
 \fE F(\varPsi_{\tau,\varepsilon;t,s}+\phi),
 \dot G_{\tau,\varepsilon;t}(\Cdot-\Cdot)
 \rangle_\tau.
\end{equation}
At fixed $\tau,\varepsilon>0$, the covariance has bounded support in Fourier space. Hence, after restriction to the finite-dimensional space $\mathcal C_{\tau,\varepsilon}$, the expressions above involve only finitely many contractions. Consequently, for fixed $\phi$ the integrand and its $t$-derivative are continuous,  and hence uniformly bounded, on the compact set
\begin{equation}
 \{(s,t)\in[\varepsilon,1]^2:s\leq t\}.
\end{equation}
Consequently, we may differentiate~\eqref{eq:polchinski_integral} under the integral sign. In particular, $t\mapsto U_{\tau,\varepsilon;t}$ is continuously differentiable
on $(0,1]$, and it holds
\begin{multline}
 \partial_t U_{\tau,\varepsilon;t}(\phi) = \partial_t \fE U_{\tau,\varepsilon;0}(\varPsi_{\tau,\varepsilon;t,0}+\phi) 
 \\
 +
 \frac{1}{2} \int_0^t \partial_t\fE\langle \rD_\phi U_{\tau,\varepsilon;s}(\varPsi_{\tau,\varepsilon;t,s}+\phi),\dot G_{\varepsilon;s}\ast \rD_\phi U_{\tau,\varepsilon;s}(\varPsi_{\tau,\varepsilon;t,s}+\phi)\rangle_\tau\,\rd s
 \\
 + 
 \frac{1}{2} 
 \langle \rD_\phi U_{\tau,\varepsilon;t}(\phi),\dot G_{\varepsilon;t}\ast \rD_\phi U_{\tau,\varepsilon;t}(\phi)\rangle_\tau.
\end{multline} 
Lemma~\ref{lem:expected_F_incremenet} implies that
\begin{equation}
 \partial_t\fE F(\varPsi_{\tau,\varepsilon;t,s}+\phi) = -\frac{1}{2} \langle\rD_\phi^2 \fE F(\varPsi_{\tau,\varepsilon;t,s}+\phi),\dot G_{\tau,\varepsilon;t}(\Cdot-\Cdot)\rangle_\tau
\end{equation}
for all $F\in\sN(C^\infty(\bT^2_\tau)^\bFF)$. Hence, we obtain
\begin{multline}
 \partial_t U_{\tau,\varepsilon;t}(\phi) = 
 -\frac{1}{2}\langle\rD_\phi^2  \fE U_{\tau,\varepsilon;0}(\varPsi_{\tau,\varepsilon;t,0}+\phi),\dot G_{\tau,\varepsilon;t}(\Cdot-\Cdot)\rangle_\tau
 \\
 -
 \frac{1}{4} \int_0^t \langle\rD_\phi^2 \fE\langle \rD_\phi U_{\tau,\varepsilon;s}(\varPsi_{\tau,\varepsilon;t,s}+\phi),\dot G_{\varepsilon;s}\ast \rD_\phi U_{\tau,\varepsilon;s}(\varPsi_{\tau,\varepsilon;t,s}+\phi)\rangle_\tau,\dot G_{\tau,\varepsilon;t}(\Cdot-\Cdot)\rangle_\tau\,\rd s
 \\
 +
 \frac{1}{2} 
 \langle \rD_\phi U_{\tau,\varepsilon;t}(\phi),\dot G_{\varepsilon;t}\ast \rD_\phi U_{\tau,\varepsilon;t}(\phi)\rangle_\tau.
\end{multline} 
Using~\eqref{eq:polchinski_integral} the sum of the first two terms on the RHS of the above equation can be rewritten as
\begin{equation}
 -\frac{1}{2}\langle\rD_\phi^2 U_{\tau,\varepsilon;t}(\phi) ,\dot G_{\tau,\varepsilon;t}(\Cdot-\Cdot)\rangle_\tau.
\end{equation} 
This implies that the function $[0,1]\ni t\mapsto U_{\tau,\varepsilon;t}\in\sN(C^\infty(\bT_\tau^2)^\bFF)$ satisfies the Polchinski flow equation~\eqref{eq:polchinski} with the initial condition $U_{\tau,\varepsilon;0}(\phi)=U_{\tau,\varepsilon}(\phi)$. The statement follows now from Lemma~\ref{lem:polchinski}.
\end{proof}

\section{Symmetries}\label{sec:symmetries}

In this section we discuss the symmetries of the Gross--Neveu model that are instrumental in the decomposition of kernels of functionals into the local part and the remainder, which is defined in Sec.~\ref{sec:loc_ren}. As we mentioned in Sec.~\ref{sec:strategy} the form of the local terms that appear in this decomposition is restricted by the symmetries of the kernels.

\begin{dfn}\label{dfn:jet}
The jet prolongation of $\varphi\in C^\infty(\bR^2)^\bFF$ is defined by
\begin{equation}
 \fJ\,:\,C^\infty(\bR^2)^\bFF
 \to C^\infty(\bR^2)^{\bA\times\bFF},
 \qquad
 (\fJ\varphi)^{a,\sigma}
 \equiv
 \fJ\varphi^{a,\sigma} 
 =
 \partial^a\varphi^\sigma,
 \quad 
 a\in\bA,~\sigma\in\bFF.
\end{equation} 
\end{dfn}

\begin{dfn}[Symmetries of the plane]\label{dfn:symm_plane}
We fix once and for all a projective spinor representation $R\mapsto\gamma(R)\in U(2)$ of $O(2)$ obtained from the double cover $\operatorname{Pin}(2)\to O(2)$, such that
\begin{equation}
 \gamma(R)^{-1}\gamma_j\gamma(R)
 =
 \sum_{k\in\{1,2\}}R^{jk}\gamma_k,
 \qquad j\in\{1,2\}.
\end{equation}
We set $\Gamma(R)=\gamma(R)^{\oplus N}$. For $\varphi\equiv(\bar\varphi,\ubar\varphi)\in\sS(\bR^2)^\bFF$, $x\in\bR^2$ and $R\in O(2)$, we define
\begin{equation}
 \fT(x,R)\varphi\equiv(\fT(x,R) \bar\varphi,\fT(x,R) \ubar\varphi)\in\sS(\bR^2)^\bFF
\end{equation}
by the equalities
\begin{equation}
 (\fT(x,R) \bar\varphi)(y) := \Gamma(R)^{-\mathrm{t}}\bar\varphi(R^{-1}(y-x)),
 \qquad
 (\fT(x,R) \ubar\varphi)(y) := \Gamma(R)\ubar\varphi(R^{-1}(y-x))
\end{equation}
for all $y\in\bR^2$. Let $m\in\bN_+$. We say that a Schwartz distribution $V\in\sS'(\bR^{2m})^{\bFF^m}$ is invariant under the symmetries of the plane iff
\begin{equation}\label{eq:plane_symmetries}
 \langle V,\varphi_1\otimes\ldots\otimes\varphi_m\rangle = 
 \langle V,\fT(x,R)\varphi_1\otimes\ldots\fT(x,R)\varphi_m\rangle
\end{equation}
for all $\varphi_1,\ldots,\varphi_m\in \sS(\bR^2)^{\bFF}$, all $x\in\bR^2$ and all $R$ as above.
The action of $\fT(x,R)$ on $\sS(\bR^2)^{\bA\times\bFF}$ is the natural action induced on second-order jets: translations and rotations act on the spatial argument as above, the spinor indices transform according to the corresponding spin representation, and the derivative indices transform according to the chain rule. This action agrees with the jet prolongation introduced in Def.~\ref{dfn:jet}.
We say that a Schwartz distribution $V\in\sS'(\bR^{2m})^{\bA^m\times\bFF^m}$ is invariant under the symmetries of the plane iff the condition~\eqref{eq:plane_symmetries} is satisfied for all $\varphi_1,\ldots,\varphi_m\in \sS(\bR^2)^{\bA\times\bFF}$, all $x\in\bR^2$ and all $R$ as above.
\end{dfn}

\begin{rem}
With the fixed choice above, $(x,R)\mapsto\fT(x,R)$ is a projective
representation of the Euclidean group. This choice of projective lift
is used throughout the paper.
\end{rem}

\begin{dfn}[Symmetries of torus]
A Schwartz distribution $V\in\sS'(\bR^{2m})^{\bA^m\times\bFF^m}$ is invariant under the symmetries of the torus iff the condition~\eqref{eq:plane_symmetries} is satisfied for all $x\in\bR^2$ and all $R\,:\,\bR^2\to\bR^2$ that are a rotation by a multiple of $\pi/2$ or a reflection with respect to the line $x^1=0$ or $x^2=0$.
\end{dfn}

\begin{dfn}[Internal rotations]
Let
$\mathfrak u=(\mathfrak u_{nm})_{n,m\in\{1,\ldots,N\}}\in U(N)$
and
$\varphi\equiv(\bar\varphi,\ubar\varphi)
\in\sS(\bR^2)^{\bA\times\bFF}$.
We define
$\fT(\mathfrak u)\varphi
\in\sS(\bR^2)^{\bA\times\bFF}$ by
\begin{align}
 (\fT(\mathfrak u)\varphi)^{a,(-,n,\alpha)}(x)
 :=
 \sum_{m=1}^N
 \overline{\mathfrak u_{nm}}
 \varphi^{a,(-,m,\alpha)}(x),
 \quad
 (\fT(\mathfrak u)\varphi)^{a,(+,n,\alpha)}(x)
 :=
 \sum_{m=1}^N
 \mathfrak u_{nm}
 \varphi^{a,(+,m,\alpha)}(x).
\end{align}
Let $m\in\bN_+$. We say that $V\in\sS'(\bR^{2m})^{\bA^m\times\bFF^m}$ is invariant under the internal symmetries if
\begin{equation}
 \langle V,\varphi_1\otimes\ldots\otimes\varphi_m\rangle
 =\langle V,
 \fT(\mathfrak u)\varphi_1\otimes\ldots\otimes
 \fT(\mathfrak u)\varphi_m\rangle
\end{equation}
for all $\mathfrak u\in U(N)$ and $\varphi_1,\ldots,\varphi_m \in \sS(\bR^2)^{\bA\times\bFF}$.
\end{dfn}

\begin{dfn}[Charge conjugation]\label{dfn:charge}
For $\varphi\equiv(\bar\varphi,\ubar\varphi)\in\sS(\bR^2)^{\bA\times\bFF}$, we define
\begin{equation}
 \fC\varphi\equiv(\fC \bar\varphi,\fC \ubar\varphi)\in\sS(\bR^2)^{\bA\times\bFF}
\end{equation}
by the equalities
\begin{equation}\label{eq:charge_conjugation}
 (\fC\bar\varphi)^a(x) = \Gamma_2\ubar\varphi^a(x),
 \qquad
 (\fC \ubar\varphi)^a(x) = \Gamma_2\bar\varphi^a(x)
\end{equation}
for $a\in\bA$, $x\in\bR^2$. Let $m\in\bN_+$. We say that a distribution $V\in\sS'(\bR^{2m})^{\bA^m\times\bFF^m}$ is invariant under the charge conjugation symmetry iff
\begin{equation}
 \langle V,\varphi_1\otimes\ldots\otimes\varphi_{m}\rangle = 
 \langle V,\fC\varphi_1\otimes\ldots\otimes\fC\varphi_{m}\rangle
\end{equation} 
for all $\varphi_1,\ldots,\varphi_m\in\sS(\bR^2)^{\bA\times\bFF}$.
\end{dfn}
\begin{dfn}[Reality conjugation]\label{dfn:reality_conjugation}
For
$\varphi\equiv(\bar\varphi,\ubar\varphi)
\in\sS(\bR^2)^{\bA\times\bFF}$
we define the antilinear map
\begin{equation}
 \mathfrak K\varphi
 \equiv
 (\mathfrak K\bar\varphi,\mathfrak K\ubar\varphi)
 \in\sS(\bR^2)^{\bA\times\bFF}
\end{equation}
by
\begin{equation}
 (\mathfrak K\bar\varphi)^a(x)
 :=
 \Gamma_1\overline{\bar\varphi^a(x)},
 \qquad
 (\mathfrak K\ubar\varphi)^a(x)
 :=
 \Gamma_1\overline{\ubar\varphi^a(x)}
 \, ,
\end{equation}
for all $a\in\bA$ and $x\in\bR^2$, where the complex conjugation is understood componentwise.
Let $m\in\bN_+$ and $V\in\sS'(\bR^{2m})^{\bA^m\times\bFF^m}$. We define the reality conjugate $V^\sharp\in\sS'(\bR^{2m})^{\bA^m\times\bFF^m}$ by
\begin{equation}
 \langle V^\sharp,
 \varphi_1\otimes\ldots\otimes\varphi_m\rangle
 :=
 \overline{
 \langle V,
 \mathfrak K\varphi_1\otimes\ldots\otimes
 \mathfrak K\varphi_m\rangle}
 \, ,
\end{equation}
for all
$\varphi_1,\ldots,\varphi_m
\in\sS(\bR^2)^{\bA\times\bFF}$.
We say that $V$ is real if
\begin{equation}
 V^\sharp=V
 \, .
\end{equation}
\end{dfn}

\begin{rem}\label{rem:reality_conjugation}
Since $\Gamma_1$ is real and $\Gamma_1^2=1$, it holds
$\mathfrak K^2=\mathrm{Id}$ and consequently
\begin{equation}
 (V^\sharp)^\sharp=V
\end{equation}
for every
$V\in\sS'(\bR^{2m})^{\bA^m\times\bFF^m}$.
In particular, $V\mapsto V^\sharp$ is an antilinear involution.
Moreover, on scalar kernel spaces introduced below, reality
conjugation only conjugates and permutes the spinor components.
Consequently it preserves all componentwise kernel norms with real
weights.
\end{rem}

\begin{rem}\label{rem:symmetries}
The only local functionals of degree two in $\psi\in\cS(\bR^2)^\bFF\otimes_{\mathrm{alg}}\sG$ with up to one derivative compatible with all of the above symmetries are real linear combinations of
\begin{equation}\label{eq:inv_functionals_quadratic}
 \int_{\bR^2} 
 \bar\psi(x)\cdot\ubar\psi(x)\,\rd x,
 \qquad
 \int_{\bR^2} 
 \bar\psi(x)\cdot((\Gamma_1\partial_1+\Gamma_2\partial_2)\ubar\psi)(x)\,\rd x.
\end{equation}
The only local functionals of degree four in $\psi\in\cS(\bR^2)^\bFF\otimes_{\mathrm{alg}}\sG$ without derivatives compatible with all of the above symmetries are linear combinations of
\begin{equation}\label{eq:inv_functionals}
 \int_{\bR^2} 
 (\bar\psi(x)\cdot\ubar\psi(x))^2\,\rd x,
 \quad
 \sum_{j\in\{1,2\}}\int_{\bR^2} 
 (\bar\psi(x)\cdot\Gamma_j\ubar\psi(x))^2\,\rd x,
 \quad
 \int_{\bR^2} 
 (\bar\psi(x)\cdot\Gamma_1\Gamma_2\ubar\psi(x))^2\,\rd x.
\end{equation}
For the classification of these $U(N)$-invariant local spinor
structures see~\cite[Appendix~C]{DY24}; see also
\cite[Appendix]{MW73}.
\end{rem}

\section{Weights}\label{sec:weights}

In this section we define the weights that are used in the definitions of various norms and establish some of their properties. The use of weights growing stretched exponentially allows us to prove stretched exponential decay of truncated correlations. This choice of the weights will also play a role in the proof of the estimates for the map $\fR$ presented in Sec.~\ref{sec:loc_ren}.

\begin{dfn}\label{dfn:weights}
The diameter of the set of points $\{x_1,\ldots,x_m\}\subset\bR^2$ is defined by
\begin{equation}
 \rD(x_1,\ldots,x_m) := \max_{i,j\in\{1,\ldots,m\}} |x_i-x_j|.
\end{equation}
Let $\zeta:=4/5$. For $m\in\bN_+$, $\nu\in[0,1/2]$ and $t\in(0,1]$ we define the weights $w_\nu\in C(\bR^2)$ and $w^{m}_{t;\nu}\in C(\bR^{2m})$ by $w_\nu(x)=(1+|x|)^{-\nu}$ and
\begin{equation}
 w^{m}_{t;\nu}(x_1,\ldots,x_m) := 
 (1+|x_1|)^{-\nu}(1+\rD(x_1,\ldots,x_m))^{1/2-\nu}\exp(t^{-\zeta}\rD(x_1,\ldots,x_m)^\zeta).
\end{equation}
We set $\tilde w:=w_{1/2}$, $w^m_t:=w^{m}_{t;0}$, $\tilde w^m_t:=w^{m}_{t;1/2}$ and $\zeta_\star:=-2(1/\zeta-7/8)=-3/4$.
\end{dfn}
\begin{rem}
Note that the upper index $m\in\bN_+$ of $w^m_{t;\nu}$ does not denote the power. 
\end{rem}
\begin{rem}
We will frequently use the fact that $\tilde w^m_t(x_1,\ldots,x_m)\leq w^m_t(x_1,\ldots,x_m)$ for all $m\in\bN_+$, $t\in(0,1]$ and $x_1,\ldots,x_m\in\bR^2$. The weights $\tilde w^m_t$ will be used to study the infinite volume limit. Observe that the weight $\tilde w$ appeared already in Lemma~\ref{lem:Psi_bounds}.
\end{rem}
\begin{rem}\label{rem:w_tilde_young}
It holds $1/\tilde w(x)\leq 1/\tilde w(y) \, 1/\tilde w(x-y)$ for all $x,y\in\bR^2$. 
\end{rem}

\begin{lem}\label{lem:weights}
There exists $C\in(0,\infty)$ such that for all $l\in\{0,1\}$, $\nu\in[0,1/2]$, $t\in(0,1]$, $s\in(0,t)$, $m\in\bN_+$, $k\in\{1,\ldots,m\}$ and $x_1,\ldots,x_m,y_1,\ldots,y_m,y,z\in\bR^2$ it holds
\begin{enumerate}
\item[(a)]
$
 w^{m}_{t;\nu}(x_1,\ldots,x_m)\leq 
 w^{k+1}_{t;\nu}(x_1,\ldots,x_k,y)
 \,w^2_t(y,z)\,w^{m-k+1}_t(z,x_{k+1},\ldots,x_m),
$

\item[(b)]
$
 w^{m}_{t;\nu}(x_m,\ldots,x_1)\leq 
 w^{m-k+1}_{t}(z,x_m,\ldots,x_{k+1})
 \,w^2_t(y,z)\,w^{k+1}_{t;\nu}(y,x_k,\ldots,x_1),
$

\item[(c)]
$
 w^{m}_{t;\nu}(x_1,\ldots,x_m)\leq 
 w^{m}_{t;\nu}(y_1,\ldots,y_m)
 \,w^2_t(y_1,x_1)\ldots w^2_t(y_m,x_m),
$

\item[(d)]
$
 w^{k}_{t;\nu}(x_1,\ldots,x_k)
 \leq
 w^{m}_{t;\nu}(x_1,\ldots,x_k,x_{k+1},\ldots,x_m),
$

\item[(e)]
$
 w^{k}_{t;\nu}(x_1,\ldots,x_k)
 \leq
 w^{m}_t(x_1,\ldots,x_k,x_{k+1},\ldots,x_m)
 w_\nu(x_m),
$

\item[(f)]
$
\int_0^1 (1-u)^l\rD(x_1,\ldots,x_m)^{l+1} w^{m}_{t;\nu}(ux_1,\ldots,ux_m)\,\rd u
\leq
C\,s^{l+1}\,(1-s/t)^{\zeta_\star}\,w^{m}_{s;\nu}(x_1,\ldots,x_m)
$.
\end{enumerate}
\end{lem}
\begin{rem}
The above lemma will only be used with $\nu\in\{0,1/2\}$.
\end{rem}
\begin{proof}
Items~(a) and~(b) follow from the bounds
\begin{equation}
\begin{gathered}
 \rD(x_1,\ldots,x_m) \leq
 \rD(x_1,\ldots,x_k,y)+|y-z|+\rD(z,x_{k+1},\ldots,x_m),
 \\
 |y|\leq |y-z|+\rD(z,x_{k+1},\ldots,x_m)+|x_m|,
\end{gathered} 
\end{equation}
which are consequences of the triangle inequality, and the bounds 
\begin{equation}
 (a+b+c)^\mu\leq a^\mu + b^\mu + c^\mu,
 \qquad
 (1+a+b+c)^\mu\leq (1+a)^\mu (1+b)^\mu (1+c)^\mu
\end{equation}
valid for $\mu\in(0,1]$ and $a,b,c\in[0,\infty)$. Item~(c) is proved along the same lines with the use of the bounds
\begin{equation}
 \rD(x_1,\ldots,x_m) \leq
 \rD(y_1,\ldots,y_m)+|y_1-x_1|+\ldots+|y_m-x_m|,
 \qquad
 |y_1|\leq |y_1-x_1| + |x_1|.
\end{equation}
To prove Item~(d) one uses the fact that the function
\begin{equation}
 [0,\infty)\ni d\mapsto (1+d)^{1/2-\nu}\,\exp(t^{-\zeta} d^\zeta)\in\bR
\end{equation}
is monotonically increasing. Item~(e) follows from Item~(d), the bound
$|x_m|\leq |x_1| + \rD(x_1,\ldots,x_m)$ and the inequality $(1+a+b)^\mu\leq (1+a)^\mu (1+b)^\mu$ valid for $\mu\in(0,1]$ and $a,b\in[0,\infty)$. We proceed to the proof of Item~(f). Let $\rD(x_1,\ldots,x_m)=d$. Observe that it holds
\begin{multline}
 d^{l+1}\frac{w^{m}_{t;\nu}(ux_1,\ldots,ux_m)}{w^{m}_{s;\nu}(x_1,\ldots,x_m)}
 =
 \frac{(1+|x_1|)^\nu}{(1+u\,|x_1|)^\nu} \frac{(1+u\, d)^{1/2-\nu}}{(1+d)^{1/2-\nu}}d^{l+1}\exp((u^\zeta t^{-\zeta}-s^{-\zeta})d^\zeta)
 \\
 \leq
 u^{-\nu} d^{l+1}\exp((u^\zeta t^{-\zeta}-s^{-\zeta})d^\zeta)
 \leq
 \hat C\,u^{-\nu}\,(s^{-\zeta}-u^\zeta t^{-\zeta})^{-(l+1)/\zeta},
\end{multline}
where $\hat C\in(0,\infty)$ is such that 
$
 d^{l+1}\exp(-d^\zeta r)
 \leq
 \hat C\,r^{-(l+1)/\zeta}
$
for all $l\in\{0,1\}$ and $r\in(0,\infty)$. Consequently, it holds
\begin{equation}
 \int_0^1 (1-u)^l\rD(x_1,\ldots,x_m)^{l+1} \frac{w^{m}_{t;\nu}(ux_1,\ldots,ux_m)}{w^{m}_{s;\nu}(x_1,\ldots,x_m)}\,\rd u
 \leq
 \hat C \int_0^1 \frac{s^{l+1}\,u^{-\nu}\,(1-u)^l}{(1-u^\zeta (s/t)^{\zeta})^{(l+1)/\zeta}}\,\rd u
\end{equation}
Note that for $0<s<t\leq 1$ and $0<u\leq1$ it holds $1-u^\zeta\geq \zeta (1- u)$ and
\begin{equation}
 (1-u^\zeta (s/t)^{\zeta})^{(l+1)/\zeta}
 \geq
 \zeta^{(l+1)/\zeta}\, (1-u (s/t))^{(l+1)/\zeta}
 \\
 \geq 
 \zeta^{(l+1)/\zeta}\, (1-u)^{\zeta_\sharp(l)}
 \,
 (1-s/t)^{(l+1)(1/\zeta-7/8)},
\end{equation}
where $\zeta_\sharp(l):=(l+1)/\zeta-(l+1)(1/\zeta-7/8)>0$. Since $l\in\{0,1\}$ and $l-\zeta_\sharp(l)>-1$ the bound stated in Item~(f) holds true with
\begin{equation}
 C=\max_{l\in\{0,1\}}\zeta^{-2/\zeta}\, \hat C \int_0^1 \frac{u^{-1/2}\,(1-u)^l}{(1-u)^{\zeta_\sharp(l)}}\,\rd u <\infty.
\end{equation}
This finishes the proof.
\end{proof}

\section{Spaces of kernels and functionals}\label{sec:topology}

In this section we define the spaces of kernels and functionals in infinite volume that are used in Sec.~\ref{sec:fixed_point} to solve the system of equations~\eqref{eq:flow_g_intro_ren}, \eqref{eq:flow_r_intro_ren},
\eqref{eq:flow_2_intro_z}, \eqref{eq:flow_2_intro_W1}, 
\eqref{eq:flow_2_intro_W2} introduced in Sec.~\ref{sec:strategy} using the Banach fixed-point theorem. We also define the map $\fA_{\tau,\varepsilon;t,s}$ and analyze its properties. We refer the reader to Sec.~\ref{sec:strategy} for the motivation behind the definitions stated in this section.

\begin{dfn}\label{dfn:Banach_distributions}
Let $m,n\in\bN_+$ and $\sA$ be a unital Banach algebra. We denote by $\sS(\bR^{2m},\sA)$ the space of Schwartz test functions valued in $\sA$ equipped with the usual family of seminorms. The space of $\sA$-valued Schwartz distributions is denoted by $\sS'(\bR^{2m},\sA)$ and by definition coincides with the space of continuous linear maps $\sS(\bR^{2m})\to\sA$ equipped with the topology of pointwise convergence. We denote by 
$$\langle\Cdot,\Cdot\rangle\,:\,\sS'(\bR^{2m},\sA)\times \sS(\bR^{2m},\sA)\to\sA$$ 
the unique bilinear map such that $\langle V,\varphi A \rangle := \langle V,\varphi \rangle\, A$ for all $V\in\sS'(\bR^{2m},\sA)$, $\varphi\in \sS(\bR^{2m})$ and $A\in\sA$, where $\langle V,\varphi \rangle$ is the standard pairing between a distribution and a test function. We denote by 
$$\Cdot\otimes\Cdot\,:\,\sS(\bR^{2m},\sA)\times \sS(\bR^{2n},\sA)\to\sS(\bR^{2m+2n},\sA)$$ 
the unique bilinear map such that $\varphi A \otimes\psi B := (\varphi\otimes\psi)\, AB$ for all $\varphi\in\sS(\bR^{2m})$, $\psi\in \sS(\bR^{2n})$ and $A,B\in\sA$, where $\varphi\otimes\psi\in\sS(\bR^{2m+2n})$ is the standard tensor product of Schwartz functions. Let $\bK$ be a finite set. The maps
\begin{equation}
\begin{gathered}
\langle\Cdot,\Cdot\rangle\,:\,\sS'(\bR^{2m},\sA)^{\bK^m}\times \sS(\bR^{2m},\sA)^{\bK^m}\to\sA,
\\
\Cdot\otimes\Cdot\,:\,\sS(\bR^{2m},\sA)^{\bK^{m}}\times \sS(\bR^{2n},\sA)^{\bK^{n}}\to\sS(\bR^{2m+2n},\sA)^{\bK^{m+n}}
\end{gathered} 
\end{equation}
are defined in analogy to Def.~\ref{dfn:dist_K}.
\end{dfn}

\begin{rem}
In what follows, we assume that $\sG$ is an infinite-dimensional Grassmann algebra. It is important that $\sG$ be infinite-dimensional as we shall frequently use the fact that for every $n\in\bN_+$ there exist $g_1,\ldots,g_n\in\sG$ such that $g_1\ldots g_n\neq0$. We do not equip $\sG$ with any topology.
\end{rem}

\begin{rem}
Let $\sA=\sA^+\oplus\sA^-$ be a graded unital algebra. We define $\sA\otimes_{\mathrm{alg}}\sG$ to be the algebraic graded tensor product of graded algebras $\sA$ and $\sG$. In particular, for homogeneous elements the product in $\sA\otimes_{\mathrm{alg}}\sG$ satisfies the condition $(A\otimes g)(B\otimes h)=(-1)^{\deg(g)\deg(B)}\,AB\otimes gh$ for all $A,B\in\sA$ and $gh\in\sG$ of definite parity. We identify $A\in\sA$ and $g\in\sG$ with 
\begin{equation}
 A\otimes 1\in \sA\otimes_{\mathrm{alg}}\sG
 \qquad\mathrm{and}\qquad
 \mathds{1}\otimes g\in \sA\otimes_{\mathrm{alg}}\sG,
\end{equation}
respectively. We use these identifications to make sense of $A+g,Ag\in \sA\otimes_{\mathrm{alg}}\sG$. Let $m\in\bN_+$ and $\bK$ be a finite set. We define the pairing
\begin{equation}
\langle\Cdot,\Cdot\rangle\,:\,\sS'(\bR^{2m},\sA)^{\bK^m}\times (\sS(\bR^{2m},\sA)^{\bK^m}\otimes_{\mathrm{alg}}\sG)\to\sA\otimes_{\mathrm{alg}}\sG
\end{equation}
as the unique bilinear map such that $\langle V,\varphi \otimes g \rangle := \langle V,\varphi \rangle\, \otimes g$ for all $V\in\sS'(\bR^{2m},\sA)^{\bK^m}$, $\varphi\in \sS(\bR^{2m})^{\bK^m}$ and $g\in\sG$. We denote by
\begin{equation}
 \Cdot\otimes\Cdot\,:\,(\sS(\bR^{2m},\sA)^{\bK^{m}}\otimes_{\mathrm{alg}}\sG)\times (\sS(\bR^{2n},\sA)^{\bK^{n}}\otimes_{\mathrm{alg}}\sG)\to\sS(\bR^{2m+2n},\sA)^{\bK^{m+n}}\otimes_{\mathrm{alg}}\sG
\end{equation}
the unique bilinear map such that $(\varphi \otimes g) \otimes(\psi \otimes h) := (-1)^{\deg(g)\deg(\psi)}(\varphi\otimes\psi)\otimes gh$ for all $\varphi\in\sS(\bR^{2m},\sA)^{\bK^m}$, $\psi\in \sS(\bR^{2n},\sA)^{\bK^n}$ and $g,h\in\sG$ of definite parity.
\end{rem}

\begin{dfn}\label{dfn:permutations}
Let $m\in\bN_+$ and $\sA$ be a unital Banach algebra. We denote by $\mathcal{P}_m$ the group of permutations of the set $\{1,\ldots,m\}$. 
For a permutation $\pi\in\cP_m$ and a distribution~$V\in\sS'(\bR^{2m},\sA)^{\bA^m\times\bFF^m}$ we define \mbox{$\fS_\pi V\in\sS'(\bR^{2m},\sA)^{\bA^m\times\bFF^m}$} and~\mbox{$\fS V\in\sS'(\bR^{2m},\sA)^{\bA^m\times\bFF^m}$} by
\begin{equation}
 \langle \fS_\pi V,\varphi_1\otimes\ldots\otimes\varphi_m\rangle 
 :=
 \mathrm{sgn}(\pi)
 \langle V,\varphi_{\pi(1)}\otimes\ldots\otimes\varphi_{\pi(m)}\rangle,
 \qquad
 \fS V
 := \frac{1}{m!}\sum_{\pi\in\mathcal{P}_m}
 \fS_\pi V
\end{equation} 
for all $\varphi_1,\ldots,\varphi_m\in\sS(\bR^2)^{\bA\times\bFF}$. 
We call $\fS V$ the antisymmetric part of $V$. We say that a distribution $V$ is antisymmetric iff $V=\fS V$.
Equivalently, $V$ is antisymmetric iff $\fS_\pi V=V$ for every
$\pi\in\cP_m$.
\end{dfn}

\begin{rem}
Recall that $\sB=\sB(\mathscr{H})$ is the algebra of bounded operators acting on the Hilbert space $\mathscr{H}$ introduced in Def.~\ref{dfn:white_noise} 
\end{rem}

\begin{dfn}\label{dfn:sM}
Let $m\in\bN_+$ and let $\sA\subset\sB$ be a closed unital subalgebra,
endowed with the norm inherited from $\sB$. We denote by
$\sM^m(\sA)$ the vector space of Schwartz distributions
$V\in\sS'(\bR^{2m},\sA)$ such that there exists a kernel
$V_\rK\,:\,\bR^2 \times \mathrm{Borel}(\bR^{2(m-1)}) \to \sA$
satisfying:
\begin{itemize}
\item[(1)] for every $A\in\mathrm{Borel}(\bR^{2(m-1)})$ the map
$x\mapsto V_\rK(x,A)$ is strongly measurable,

\item[(2)] for every $x\in\bR^2$ the map
$A\mapsto V_\rK(x,A)$ is a countably additive $\sA$-valued
measure of finite variation,

\item[(3)] denoting its total variation by
$|V_\rK|(x,\Cdot)$, it holds
\begin{equation}
 \|V\|_{\sM^m}
 :=
 \operatorname*{ess\,sup}_{x_1\in\bR^2}
 |V_\rK|(x_1,\bR^{2(m-1)})<\infty.
\end{equation}
\item[(4)] for all $\varphi\in\sS(\bR^{2m})$ it holds
\begin{equation}\label{eq:V_rK}
 \langle V,\varphi\rangle = \int_{\bR^{2m}} V_\rK(x_1,\rd x_2\ldots\rd x_m)\,\varphi(x_1,\ldots,x_m)\,\rd x_1,
\end{equation}
where the integral on the RHS is understood in the
Bochner sense.
\end{itemize} 
We write $\sM^m:=\sM^m(\bC)$.
\end{dfn}
\begin{rem}
Note that it holds $\sM^m(\bC)\subset\sM^m(\sA)\subset\sM^m(\sB)$.
\end{rem}
\begin{rem}
Let $K\,:\,\bR^{2m}\to\sB$ be strongly measurable such that
\begin{equation}
 \|K\|_{\sK^m}:=\esssup_{x_1\in\bR^2}\int_{\bR^{2(m-1)}}\|K(x_1,\ldots,x_m)\|_\sB\,\rd x_2\ldots\rd x_m <\infty.
\end{equation}
Then $V\in\sS'(\bR^{2m},\sB)$ defined by~\eqref{eq:V_rK} with 
\begin{equation}
 V_\rK(x_1,\rd x_2\ldots\rd x_m)=K(x_1,\ldots,x_m)\,\rd x_2\ldots\rd x_m
\end{equation}
belongs to $\sM^m(\sB)$ and $\|V\|_{\sM^m}=\|K\|_{\sK^m}$. 
\end{rem}
\begin{rem}\label{rem:Dirac_diagonal}
Let $\delta^{(m)}\in\sS'(\bR^{2m})$ be the Dirac measure on the diagonal, i.e.
\begin{equation}
 \langle \delta^{(m)},\varphi_1\otimes\ldots\otimes\varphi_m\rangle:= \int_{\bR^2}
 \varphi_1(x)\ldots\varphi_m(x)\,\rd x
\end{equation}
for all $\varphi_1,\ldots,\varphi_m\in\sS(\bR^2)$. Then $\delta^{(m)}\in\sM^m$ and $\|\delta^{(m)}\|_{\sM^m}=1$. 
\end{rem}

\begin{lem}\label{lem:sM_banach}
Let $m\in\bN_+$ and let $\sA\subset\sB$ be a closed unital subalgebra.
The kernel representing an element $V\in\sM^m(\sA)$ is unique up to equality
for almost every value of its first variable. Consequently,
$\|\Cdot\|_{\sM^m}$ is well defined and $\sM^m(\sA)$ is a Banach space.
\end{lem}

\begin{proof}
Suppose that $V_\rK$ and $\tilde V_\rK$ represent the same distribution.
Then, for every $f\in C^\infty_c(\bR^2)$ and
$g\in C^\infty_c(\bR^{2(m-1)})$,
\begin{equation}
 \int_{\bR^2}
 f(x)\int_{\bR^{2(m-1)}}g(y)\,
 (V_\rK-\tilde V_\rK)(x,\rd y)\,\rd x=0.
\end{equation}
Hence, for every fixed $g$,
\begin{equation}
 \int_{\bR^{2(m-1)}}g(y)\,
 (V_\rK-\tilde V_\rK)(x,\rd y)=0
\end{equation}
for almost every $x$. Choosing a countable dense subset of
$C_0(\bR^{2(m-1)})$ and using finite variation, we conclude that
\begin{equation}
 V_\rK(x,\Cdot)=\tilde V_\rK(x,\Cdot)
\end{equation}
for almost every $x$. This proves that the norm in
Def.~\ref{dfn:sM} is independent of the representing kernel.

To prove completeness, let $(V_n)_{n\in\bN}$ be a Cauchy sequence in
$\sM^m(\sA)$. Passing to a subsequence, we may assume that
\begin{equation}
 \|V_{n+1}-V_n\|_{\sM^m}\leq 2^{-n}.
\end{equation}
Let $V_{n,\rK}$ be representing kernels. Outside a common null set,
$(V_{n,\rK}(x,\Cdot))_{n\in\bN}$ is then Cauchy in the Banach space of
finite $\sA$-valued measures equipped with the total variation norm.
Denote its limit by $V_\rK(x,\Cdot)$, and set $V_\rK(x,\Cdot)=0$ on the
exceptional set. Then
\begin{equation}
 \operatorname*{ess\,sup}_{x\in\bR^2}
 |V_\rK-V_{n,\rK}|(x,\bR^{2(m-1)})
 \leq
 \sum_{j\geq n}2^{-j}.
\end{equation}
Strong measurability of \(x\mapsto V_{\rm K}(x,A)\) and countable additivity of \(A\mapsto V_{\rm K}(x,A)\) follow from the corresponding properties of \(V_{n,\rm K}\) and convergence in total variation.
The kernel $V_\rK$ therefore defines an element $V\in\sM^m(\sA)$, and
$V_n\to V$ in $\sM^m(\sA)$. Indeed, for every
$\varphi\in\sS(\bR^{2m})$,
\begin{equation}
 \|\langle V,\varphi\rangle\|_\sA
 \leq
 \|V\|_{\sM^m}
 \int_{\bR^2}
 \sup_{y\in\bR^{2(m-1)}}|\varphi(x,y)|\,\rd x.
\end{equation}
This estimate also shows that $\|V\|_{\sM^m}=0$ implies $V=0$.
\end{proof}

\begin{rem}\label{rem:sM_periodic_pairing}
Let $V\in\sM^m$ be translationally invariant. The pairing $\langle V,\varphi_1\otimes\cdots\otimes\varphi_m\rangle$ extends canonically to the case where $\varphi_1$ is compactly supported and $\varphi_2,\ldots,\varphi_m$ are bounded smooth functions, by means of the measure kernel in Def.~\ref{dfn:sM}. The resulting pairing is absolutely convergent. In particular, this convention applies when $\varphi_2,\ldots,\varphi_m$ are periodic. The same convention is used componentwise for the kernels $V^{m,a,\sigma}\in\sM^m$ and, through~\eqref{eq:functionals_derivatives}, for the corresponding
pairings involving $V^{m,\sigma}$. In particular, this convention
is used in~\eqref{eq:U_V_intro} and~\eqref{eq:relation_U_V}.
\end{rem}

\begin{rem}\label{rem:fM_fE}
It follows from the definition of the total variation of a vector measure and the properties~(0) and~(4) of the conditional expectation introduced in Def.~\ref{dfn:conditional_exp} that
\begin{equation}
 \fE_t V\in \sM^m(\sF_t),
 \qquad
 \|\fE_t V\|_{\sM^m}\leq \|V\|_{\sM^m}
\end{equation}
for all $V\in\sM^m(\sF)$ and $t\in[0,1]$.
If the coefficients of $V$ have parity $p\in\{+,-\}$, then so do those of $\fE_tV$.
\end{rem}

\begin{dfn}\label{dfn:sN}
Let $m\in\bN_+$ and $\sA$ be a unital Banach subalgebra of $\sB$. Recall that $w_1^{m}$ and $\tilde w_1^{m}$ are weights introduced in Def.~\ref{dfn:weights}.
The vector space $\sNN^m(\sA)\subset\sM^m(\sA)^{\bA^m\times\bFF^m}$ consists of Schwartz distributions
\begin{equation}
 V
 \equiv
 (V^{a,\sigma})_{a\in\bA^m,\sigma\in\bFF^m}\in\sS'(\bR^{2m},\sA)^{\bA^m\times\bFF^m}
\end{equation}
such that the following norm
\begin{equation}
 \|V\|_{\sN^m}:=\sum_{a\in\bA^m}\sum_{\sigma\in\bFF^m}\|w_1^{m}V^{a,\sigma}\|_{\sM^m}
\end{equation}
is finite. We define the norm $\|V\|_{\tilde\sN^m}$ in an analogous way with $w_1^{m}$ replaced by $\tilde w_1^{m}$. 

Let $\sA=\sA^+\oplus\sA^-$ be a graded unital Banach subalgebra of $\sF$. We denote by $\sN^m(\sA)$ the subspace of $\sNN^m(\sA)$ consisting of antisymmetric Schwartz distributions valued in the component of $\sA$ of parity $m\mod 2$. The vector space $\sN(\sA):=\bigtimes_{m\in\bN_+}\sN^m(\sA)$ consists of tuples
\begin{equation}
 V\equiv(V^m)_{m\in\bN_+}
 \equiv
 (V^{m,a,\sigma})_{m\in\bN_+,a\in\bA^m,\sigma\in\bFF^m}.
\end{equation}
We endow $\sN(\sA)$ with the product topology. The subspace $\sN^{\mathrm{fin}}(\sA)\subset \sN(\sA)$ consists of tuples $V$ such that $V^m=0$ for all but finitely many $m\in\bN_+$. We set \mbox{$\Pi^mV:=V^m$}
and $\Pi^{m,a,\sigma} V:=V^{m,a,\sigma}$. For $V\in\sN(\sA)$ and $k\in\bN_+$ we define \mbox{$\Pi_k V\equiv((\Pi_k V)^m)_{m\in\bN_+}\in\sN(\sA)$} by the equalities $(\Pi_k V)^m:=0$ for all $m\in\bN_+\setminus\{k\}$ and $(\Pi_k V)^m:=V^m$ for $m=k$. We also define \mbox{$\Pi_{>k} V\equiv((\Pi_{>k} V)^m)_{m\in\bN_+}\in\sN(\sA)$} by the equalities $(\Pi_{>k} V)^m:=0$ for all \mbox{$m\in\{1,\ldots,k\}$} and $(\Pi_{>k} V)^m:=V^m$ for $m\in\bN_+\setminus\{1,\ldots,k\}$. For $V\in\sN^{\mathrm{fin}}(\sA)$ and $\varphi\in\sS(\bR^2,\sA)^{\bA\times\bFF}\otimes_{\mathrm{alg}}\sG$ of odd parity, we write
\begin{equation}
 V[\varphi]:=\sum_{m\in\bN_+} \langle V^m,\varphi^{\otimes m}\rangle\in\sA\otimes_{\mathrm{alg}}\sG.
\end{equation}
For $k\in\bN_0$ we define $\rD^k_\varphi V[\varphi]$ in analogy to Def.~\ref{dfn:functional}. Let $\bC=\bC^+\oplus\bC^-$ with $\bC^+:=\bC$ and $\bC^-:=\{0\}$. We set $\sNN^m:=\sNN^m(\bC)$, $\sN^m:=\sN^m(\bC)$, $\sN:=\sN(\bC)$ and $\sN^{\mathrm{fin}}:=\sN^{\mathrm{fin}}(\bC)$.
\end{dfn}
\begin{rem}
Observe that $\sN^m=\sN^m(\bC)=\{0\}$ for $m\in\bN_+\setminus2\bN_+$.
\end{rem}
\begin{rem}\label{eq:N_polynomial}
Let $V\in\sNN^m(\sA)$. Since the norm $\|V\|_{\sN^m}<\infty$ involves a weight of stretched exponential growth it holds  $\cX V^{a,\sigma}\in\sM$ for all $a\in\bA^m,\sigma\in\bFF^m$ and all translationally-invariant polynomials $\cX\in C^\infty(\bR^{2m})$.
\end{rem}

\begin{rem}\label{rem:sN_m}
Using the fact that $V\in\sN^m(\sA)$ is antisymmetric, one shows that it is uniquely determined by the map 
\begin{equation}
 \sS(\bR^2)^{\bA\times\bFF}\otimes_{\mathrm{alg}}\sG^-\ni\varphi \mapsto \langle V,\varphi^{\otimes m}\rangle\in\sA\otimes_{\mathrm{alg}}\sG.
\end{equation}
Actually, given $\psi_1,\ldots,\psi_m\in\sS(\bR^2)^{\bA\times\bFF}$ we choose $\varphi=\sum_{j=1}^m \psi_j g_j$, where $g_1,\ldots,g_m\in\sG$ are such that $g=g_1\ldots g_m\neq 0$. Then $m!\,\langle V,\psi_1\otimes\ldots\otimes\psi_m\rangle\otimes g =\langle V,\varphi^{\otimes m}\rangle$.
\end{rem}

\begin{dfn}\label{dfn:V_space}
Let $m\in\bN_+$, $\alpha\in(0,\infty)$, $\beta\in[0,\infty)$, $\gamma\in\bR$ and
$
 \rho_{\gamma,\kappa}(m) 
 :=
 \gamma+2\kappa m.
$
For continuous maps
\begin{equation}\label{eq:curve_sV}
 (0,1]\ni s\mapsto V^m_s\equiv(V^{m,a,\sigma}_s)_{a\in\bA^m,\sigma\in\bFF^m}\in \sN^m(\sF)
\end{equation}
we define
\begin{equation}\label{eq:norm_sV_m}
 \|V^m_\Cdot\|_{\sV^{m;\gamma}}:=\sum_{a\in\bA^m}\sum_{\sigma\in\bFF^m}\sup_{s\in(0,1]}
 \lambda_s^{-\rho_{\gamma,\kappa}(m)}\,
 s^{2-m/2-|a|}\,\|w^m_s V_s^{m,a,\sigma}\|_{\sM^m}.
\end{equation}
The vector space $\sV^{m;\gamma}$ consists of continuous maps
\begin{equation}
 (0,1]\ni s\mapsto V^m_s\equiv(V^{m,a,\sigma}_s)_{a\in\bA^m,\sigma\in\bFF^m}\in \sN^m\equiv\sN^m(\bC)
\end{equation}
such that $\|V^m_\Cdot\|_{\sV^{m;\gamma}}<\infty$. The vector space $\sV^{\gamma}$ consists of continuous maps
\begin{equation}
 (0,1]\ni s\mapsto V_s\in\sN\equiv\sN(\bC)
\end{equation}
such that $V_\Cdot^m\equiv\Pi^m V_\Cdot\in\sV^{m;\gamma}$ for all $m\in\bN_+$. 
We define
\begin{equation}
 \sV^{\mathrm{fin};\gamma}
 =
 \left\{
 V_\Cdot\in\sV^\gamma:
 \text{there exists }M\in\bN_+
 \text{ such that }
 V_s^m=0
 \text{ for all }m>M,\ s\in(0,1]
 \right\}
\end{equation}
We denote by $\sV^{\alpha,\beta;\gamma}$ the closure of \mbox{$\sV^{\mathrm{fin};\gamma}\subset \sV^{\gamma}$} with respect to the norm
\begin{equation}\label{eq:norm_sV_alpha_beta}
 \|V_\Cdot\|_{\sV^{\alpha,\beta;\gamma}}:=\sup_{m\in\bN_+} \alpha^m\,m^\beta\, \|V^m_\Cdot\|_{\sV^{m;\gamma}}.
\end{equation}
For $V^m_\Cdot\in\sV^{m;\gamma}$ we define $\|V^m_\Cdot\|_{\tilde\sV^{m;\gamma}}$ in an analogous way to $\|V^m_\Cdot\|_{\sV^{m;\gamma}}$ with $w^m_s$ replaced by $\tilde w^m_s$ in~\eqref{eq:norm_sV_m}. For $V_\Cdot\in\sV^{\alpha,\beta;\gamma}$ we define $\|V_\Cdot\|_{\tilde\sV^{\alpha,\beta;\gamma}}$ in an analogous way to $\|V^m_\Cdot\|_{\sV^{\alpha,\beta;\gamma}}$ with $\|\Cdot\|_{\sV^{m;\gamma}}$ replaced by $\|\Cdot\|_{\tilde\sV^{m;\gamma}}$ in~\eqref{eq:norm_sV_alpha_beta}.
\end{dfn}

\begin{rem}
We stress that the expressions $\|\cdot\|_{\sV^{\alpha,\beta;\gamma}}$ and
$\|\cdot\|_{\tilde{\sV}^{\alpha,\beta;\gamma}}$ will also be used for $\sF$-valued maps~\eqref{eq:curve_sV}, according to the formulas above, whenever the corresponding RHS are finite.
The spaces $\sV^{\alpha,\beta;\gamma}$ consist of $\bC$-valued maps.
\end{rem}

\begin{rem}\label{rem:C0_sV}
	By definition, $V_\Cdot\in\sV^{\alpha,\beta;\gamma}$ if and only if
	$V_\Cdot^m\in\sV^{m;\gamma}$ for every $m\in\bN_+$ and
	\begin{equation}
		\lim_{M\to\infty}
		\sup_{m>M}
		\alpha^m m^\beta
		\|V_\Cdot^m\|_{\sV^{m;\gamma}}
		=0.
	\end{equation}
	Equivalently,
	\begin{equation}
		\|\Pi_{>M}V_\Cdot\|_{\sV^{\alpha,\beta;\gamma}}
		\longrightarrow0
		\qquad\text{as}\quad M\to\infty.
	\end{equation}
\end{rem}

\begin{rem}\label{rem:sV_tilde}
Since $\tilde w^m_t\leq w^m_t$ the norms $\|\Cdot\|_{\tilde\sV^{m;\gamma}}$ and $\|\Cdot\|_{\tilde\sV^{\alpha,\beta;\gamma}}$ are weaker than the norms $\|\Cdot\|_{\sV^{m;\gamma}}$ and $\|\Cdot\|_{\sV^{\alpha,\beta;\gamma}}$, respectively. The norms $\|\Cdot\|_{\tilde\sV^{m;\gamma}}$ and $\|\Cdot\|_{\tilde\sV^{\alpha,\beta;\gamma}}$ will be used to study the infinite volume limit. Observe that we did not define the spaces $\tilde\sV^{m;\gamma}$ and $\tilde\sV^{\alpha,\beta;\gamma}$.
\end{rem}

\begin{rem}\label{rem:sV_tilde2}
Note that the weight $\tilde w^m_t$ used in the definition of the norms $\|\Cdot\|_{\tilde\sV^{m;\gamma}}$ and $\|\Cdot\|_{\tilde\sV^{\alpha,\beta;\gamma}}$ is not invariant under permutations of its arguments. 
If
\begin{equation}
	V^m\equiv  (V^{m,a,\sigma})_{a\in\bA^m,\sigma\in\bFF^m}
	\in\sN^m(\sF), 
\end{equation}
then $\fS_\pi V^m=V^m$ for every $\pi\in\cP_m$, and therefore
\begin{equation}
	\sum_{a\in\bA^m}\sum_{\sigma\in\bFF^m}
	\|\tilde w^m_s(\fS_\pi V^m)^{a,\sigma}\|_{\sM^m}
 	=
	\sum_{a\in\bA^m}\sum_{\sigma\in\bFF^m} 
	\|\tilde w^m_sV^{m,a,\sigma}\|_{\sM^m}.
\end{equation}
\end{rem}

\begin{rem}\label{rem:sV_norm_property}
For all $\alpha,\tilde \alpha\in[1,\infty)$, $\beta,\tilde\beta\in[0,\infty)$, $\gamma,\tilde\gamma\in\bR$ and $V_\Cdot\in \sV^{\alpha,\beta;\gamma}$ such that $\tilde\alpha\leq\alpha$, $\tilde\beta<\beta$, $\tilde\gamma\leq\gamma$ it holds
\begin{equation}
 \|V_\Cdot\|_{\sV^{\tilde\alpha,\tilde\beta;\tilde\gamma}}
 \leq
 \|V_\Cdot\|_{\sV^{\alpha,\beta;\gamma}},
 \qquad
 \|V_\Cdot\|_{\tilde\sV^{\tilde\alpha,\tilde\beta;\tilde\gamma}}
 \leq
 \|V_\Cdot\|_{\tilde\sV^{\alpha,\beta;\gamma}}.
\end{equation}
\end{rem}

\begin{rem}\label{rem:sV_norm_estimate}
For all $\alpha,\beta,\tilde \alpha,\tilde\beta\in[1,\infty)$, $k\in\bN_+$ there exists $C\in(0,\infty)$ such that for all $\gamma\in\bR$ and $V_\Cdot\in \sV^{\alpha,\beta;\gamma}$ it holds
\begin{equation}
 \|\Pi^k V_\Cdot\|_{\sV^{k;\gamma}}
 \leq
 C\,\|V_\Cdot\|_{\sV^{\alpha,\beta;\gamma}},
 \qquad
 \|\Pi_k V_\Cdot\|_{\sV^{\tilde\alpha,\tilde\beta;\gamma}}
 \leq
 C\,\|V_\Cdot\|_{\sV^{\alpha,\beta;\gamma}}.
\end{equation}
\end{rem}

\begin{lem}\label{lem:V_banach}
For all $m\in\bN_+$ and $\alpha\in(0,\infty)$, $\beta\in[0,\infty)$, $\gamma\in\bR$ the spaces $(\sV^{m;\gamma},\|\Cdot\|_{\sV^{m;\gamma}})$ and $(\sV^{\alpha,\beta;\gamma},\|\Cdot\|_{\sV^{\alpha,\beta;\gamma}})$ are Banach spaces.
\end{lem}
\begin{proof}
By Lemma~\ref{lem:sM_banach}, for every $m\in\bN_+$ the space $(\sM^m(\bC),\|\Cdot\|_{\sM^m})$ is a Banach space, and given a Banach space $\sT$ the space $C_{\mathrm b}((0,1],\sT)$ is also a Banach space. Using the fact that $\lambda_s^{-\rho_{\gamma,\kappa}(m)}s^{2-m/2-|a|}w_s^m\neq0$ for all $s\in(0,1]$ and $a\in\bA^m$, we conclude that $(\sV^{m;\gamma},\|\Cdot\|_{\sV^{m;\gamma}})$ is a Banach space.
By Remark~\ref{rem:C0_sV}, the space  $\sV^{\alpha,\beta;\gamma}$ is the weighted $c_0$-sum of the  Banach spaces $\sV^{m;\gamma}$,  and is therefore a Banach space.
\end{proof}

\begin{dfn}\label{dfn:A_map}
Let $\tau,\varepsilon\in[0,1]$, $t\in[0,1]$ and $s\in(0,1]$. The maps $$\fA_{\tau,\varepsilon;t,s},\fC_{\tau,\varepsilon;t,s}\,:\,\sN^{\mathrm{fin}}(\sF)\to\sN^{\mathrm{fin}}(\sF)$$  are defined by the equalities
\begin{equation}
 \fA_{\tau,\varepsilon;t,s}V[\varphi] :=  V[\varphi + \fJ\varPsi_{\tau,\varepsilon;t\vee s,s}]- V[\fJ\varPsi_{\tau,\varepsilon;t\vee s,s}],
 \qquad
 \varphi\in\sS(\bR^2)^{\bA\times\bFF}\otimes_{\mathrm{alg}}\sG^-,
\end{equation}
\begin{equation}
 \fC_{\tau,\varepsilon;t,s}V:=\fA_{\tau,\varepsilon;t,s}V-V,
\end{equation}
where $\varPsi_{\tau,\varepsilon;t,s}$ was introduced in Def.~\ref{dfn:Psi} and $\fJ$ is the jet prolongation introduced in Def.~\ref{dfn:jet}. We also set $\fA^m_{\tau,\varepsilon;t,s} :=\Pi^m\fA_{\tau,\varepsilon;t,s}$, $\fA^{m,a}_{\tau,\varepsilon;t,s} :=\Pi^{m,a}\fA_{\tau,\varepsilon;t,s}$ and $\fA^{m,a,\sigma}_{\tau,\varepsilon;t,s} :=\Pi^{m,a,\sigma}\fA_{\tau,\varepsilon;t,s}$ and analogously for the map $\fC_{\tau,\varepsilon;t,s}$. We omit $\tau$ and $\varepsilon$ if $\tau=0$ and $\varepsilon=0$.
\end{dfn}

\begin{rem}\label{rem:A_map_zero}
	If $\varepsilon>0$, we also define
	$\fA_{\tau,\varepsilon;t,0}$ and $\fC_{\tau,\varepsilon;t,0}$
	for $t\in(0,1]$ by the formulas of Def.~\ref{dfn:A_map},
	with $s=0$. The identities of
	Remarks~\ref{rem:fA_fE} and~\ref{rem:fA_fE2} remain valid with lower
	endpoint $s=0$.
\end{rem}

\begin{rem}
	The equation defining the map $\fA_{\tau,\varepsilon;t,s}$ should be interpreted as
	\begin{equation}\label{eq:fA_map_def_alt}
		\langle\fA^m_{\tau,\varepsilon;t,s}V,\varphi^{\otimes m}\rangle
		=
		\sum_{k\in\bN_0}
		\frac{(m+k)!}{m!k!}\,
		\langle
		V^{m+k},
		\varphi^{\otimes m}\otimes
		\fJ\varPsi_{\tau,\varepsilon;t\vee s,s}^{\otimes k}
		\rangle
	\end{equation}
	for all $m\in\bN_+$. 
	Since $V\in\sN^{\mathrm{fin}}(\sF)$, the sum on the RHS is finite. 
	Moreover, the bound
	\begin{equation}
		\|\fJ\varPsi^{a,\sigma}_{\tau,\varepsilon;t\vee s,s}\|_{\sC}<\infty
	\end{equation}
	for all $a\in\bA$ and $\sigma\in\bFF$, which follows from Lemma~\ref{lem:Psi_bounds}, shows that insertion of the last $k$ arguments maps
	$\sNN^{m+k}(\sF)$ continuously into $\sNN^m(\sF)$.
	Antisymmetry of $V^{m+k}$ implies antisymmetry in the remaining $m$ arguments.
	Furthermore, since $V^{m+k}$ has parity $m+k$ and
	$\fJ\varPsi_{\tau,\varepsilon;t\vee s,s}$ is odd, every term on the
	RHS of~\eqref{eq:fA_map_def_alt} has coefficient parity
	\begin{equation}
		(m+k)+k=m\mod 2.
	\end{equation}
	Hence $\fA^m_{\tau,\varepsilon;t,s}V\in\sN^m(\sF)$.
	Its uniqueness follows from Remark~\ref{rem:sN_m}.
	Thus $\fA_{\tau,\varepsilon;t,s}$ is well defined on
	$\sN^{\mathrm{fin}}(\sF)$.
\end{rem}
\begin{rem}\label{rem:fA_fE}
Let $t,s\in(0,1]$ be such that $t\geq s$ and let $V\in\sN^{\mathrm{fin}}(\sF_{s,0})$. Then we have
\begin{equation}
 \fA_{\tau,\varepsilon;t,s}V\in\sN^{\mathrm{fin}}(\sF_{t,0}),
 \qquad
 \fE_t\fA_{\tau,\varepsilon;t,s}V=\fE\fA_{\tau,\varepsilon;t,s}V,
 \qquad
 \fE_s\fA_{\tau,\varepsilon;t,s}V=\fA_{\tau,\varepsilon;t,s} \fE V.
\end{equation}
The first assertion follows from $\fJ\varPsi_{\tau,\varepsilon;t,s}\in\sF_{t,s}$.
The second follows from property~(1) of the conditional expectation,
while the third follows from properties~(1) and~(3) of Def.~\ref{dfn:conditional_exp}.
\end{rem}

\begin{rem}\label{rem:fA_fE2}
Let $t,u,s\in(0,1]$ be such that $t\geq u\geq s$ and let $V\in\sN^{\mathrm{fin}}(\bC)$. Then we have
\begin{itemize}
 \item[(A)] 
 $\fA_{\tau,\varepsilon;t,s}V
 =
 \fA_{\tau,\varepsilon;t,u}\fA_{\tau,\varepsilon;u,s}V$,
 
 \item[(B)] 
 $\fA_{\tau,\varepsilon;t,u}\fE\fA_{\tau,\varepsilon;u,s}V
%
 =
 \fE_u\fA_{\tau,\varepsilon;t,u}\fA_{\tau,\varepsilon;u,s}V
 =
 \fE_u\fA_{\tau,\varepsilon;t,s}V$,
 
 \item[(C)] 
 $\fE\fA_{\tau,\varepsilon;t,u}\fE\fA_{\tau,\varepsilon;u,s}V
 =
 \fE\fA_{\tau,\varepsilon;t,s}V$
\end{itemize}
by Remark~\ref{rem:fA_fE} and the tower property~(2) of the conditional expectation introduced in Def.~\ref{dfn:conditional_exp}.
\end{rem}

\begin{lem}\label{lem:fA_fE_one}
Let $\tau,\varepsilon\in[0,1]$, $t\in[0,1]$, $0<s\leq u\leq1$, $m\in\bN_+$, $a\in\bA^m$, $\sigma\in\bFF^m$, and $\nu\in[0,1/2]$. 
Then, for every $V\in\sN^{\mathrm{fin}}(\bC)$,
\begin{align}
\|w^m_{u;\nu}
\fA^{m,a,\sigma}_{\tau,\varepsilon;t\vee u,u}
\fE\fA_{\tau,\varepsilon;u,s}V\|_{\sM^m}
&\leq
\max_{r\in\{t,1\}}\|w^m_{s;\nu}
\fA^{m,a,\sigma}_{\tau,\varepsilon;r\vee s,s}V\|_{\sM^m},
\label{eq:fA_fE_one}
\\
\|w^m_{u;\nu}\fE(\fA^{m,a,\sigma}_{u,s}
-\fA^{m,a,\sigma}_{\tau,\varepsilon;u,s})V\|_{\sM^m}
&\leq
\|w^m_{s;\nu}(\fA^{m,a,\sigma}_{1,s}
-\fA^{m,a,\sigma}_{\tau,\varepsilon;1,s})V\|_{\sM^m}.
\label{eq:fA_fE_one_diff}
\end{align}
\end{lem}
\begin{proof}
On the unital algebra of finite Grassmann polynomials in the fermionic white noise, let $\mathfrak b$ denote extraction of the scalar coefficient. If $F=c\mathds{1}+n$, where $n$ has strictly positive Grassmann degree, then $n$ is nilpotent and hence
\begin{equation}
	|\mathfrak b F|=|c|\leq\|F\|_{\sB}.
\end{equation}
Thus $\mathfrak b$ extends uniquely to a contractive character on $\sF$, acting contractively on $\sF$-valued kernels.

If $u\leq t$, Remark~\ref{rem:fA_fE2}~(B) gives
\begin{equation}
\fA_{\tau,\varepsilon;t,u}
\fE\fA_{\tau,\varepsilon;u,s}V
=
\fE_u\fA_{\tau,\varepsilon;t,s}V.
\end{equation}
If $u>t$, then
$\fA_{\tau,\varepsilon;t\vee u,u}=\mathrm{Id}$, while $\fE\fA_{\tau,\varepsilon;u,s}V$ is scalar-valued.
Hence Remark~\ref{rem:fA_fE2}~(B) gives
\begin{equation}
	\fE\fA_{\tau,\varepsilon;u,s}V
	=
	\mathfrak b\,\fE_u\fA_{\tau,\varepsilon;1,s}V.
\end{equation}
The contraction properties of $\fE_u$ and $\mathfrak b$, together with $w^m_{u;\nu}\leq w^m_{s;\nu}$, prove~\eqref{eq:fA_fE_one}.
Applying the same scalar-part identity to the cutoff-free and cutoff shifts and subtracting yields
\begin{equation}
	\fE(\fA_{u,s}-\fA_{\tau,\varepsilon;u,s})V
	=\mathfrak b\,\fE_u
	(\fA_{1,s}-\fA_{\tau,\varepsilon;1,s})V.
\end{equation}
Contractivity and $w^m_{u;\nu}\leq w^m_{s;\nu}$ now give~\eqref{eq:fA_fE_one_diff}.
\end{proof}

\begin{lem}\label{lem:A_map_bound}
There exists $\lambda_\star\in(0,1]$ such that for all $\lambda\in(0,\lambda_\star]$ the following is true. Let $m\in\bN_+$, $\alpha,\beta\in[1,\infty)$, $\gamma\in\bR$, $\tau,\varepsilon\in[0,1]$ and $V_\Cdot\in\sV^{\mathrm{fin};\gamma}$. For all $t\in(0,1]$ the map
\begin{equation}\label{eq:A_m_map_lem}
 (0,1]\ni s\mapsto\fA^m_{\tau,\varepsilon;t,s} V_s\in\sN^m(\sF_{t\vee s,s})
\end{equation}
is continuous and it holds
\begin{itemize}
 \item[(A)] $\|s\mapsto\fA_{\tau,\varepsilon;t,s} V_s\|_{\sV^{\alpha/2,\beta;\gamma}} 
 \leq
 \|V_\Cdot\|_{\sV^{\alpha,\beta;\gamma}}$,
 \item[(B)] $\|s\mapsto\fC_{\tau,\varepsilon;t,s} V_s\|_{\sV^{\alpha/2,\beta-1;\gamma}} 
 \leq
 \lambda^\kappa\,
 \|V_\Cdot\|_{\sV^{\alpha,\beta;\gamma}}$,
 \item[(C)] $\|s\mapsto\fA_{\tau,\varepsilon;t,s} V_s\|_{\tilde\sV^{\alpha/2,\beta;\gamma}} 
 \leq
 \|V_\Cdot\|_{\tilde\sV^{\alpha,\beta;\gamma}}$,
 \item[(D)] $\|s\mapsto\fC_{\tau,\varepsilon;t,s} V_s\|_{\tilde\sV^{\alpha/2,\beta-1;\gamma}} 
 \leq
 \lambda^\kappa\,
 \|V_\Cdot\|_{\tilde\sV^{\alpha,\beta;\gamma}}$,
 \item[(E)] $\|s\mapsto(\fA_{t,s}-\fA_{\tau,\varepsilon;t,s}) V_s\|_{\tilde\sV^{\alpha/2,\beta-1;\gamma}}
 \leq
 \lambda_{\tau\vee\varepsilon}^\kappa\,\|V_\Cdot\|_{\sV^{\alpha,\beta;\gamma}}$.
\end{itemize}
The maps appearing in~(A)--(E), initially defined for $V_\Cdot\in\sV^{\mathrm{fin};\gamma}$, extend uniquely by continuity to all $V_\Cdot\in\sV^{\alpha,\beta;\gamma}$, and the bounds~(A)--(E) remain valid for these extensions.  
In particular, the map~\eqref{eq:A_m_map_lem}  is well defined and continuous for every~$V_\Cdot\in\sV^{\alpha,\beta;\gamma}$. 
\end{lem}
\begin{rem}
Note that the parameters $\alpha$ of the norms $\sV^{\alpha,\beta;\gamma}$, $\tilde\sV^{\alpha,\beta;\gamma}$ that appear on both sides of the bounds stated in the above lemma are different. As we discussed in Sec.~\ref{sec:strategy}, the bounds of this type appear to be false if the parameters $\alpha$ of the norms are the same on both sides of the bounds. This is the main reason why we work with the space $(\sW^{\alpha,\beta;\gamma}_{\tau,\varepsilon},\|\Cdot\|_{\sW^{\alpha,\beta;\gamma}_{\tau,\varepsilon}})$ introduced below. One of the consequences of the above lemma is the relation between the norms $\|\Cdot\|_{\sW^{\alpha,\beta;\gamma}_{\tau,\varepsilon}}$ and $\|\Cdot\|_{\sV^{2\alpha,\beta;\gamma}}$ stated in Remark~\ref{rem:sW_sV}.
\end{rem}

\begin{proof}
	First recall that $\lambda_s\leq \lambda$ by Def.~\ref{dfn:lambda} and
	\begin{equation}
		\|\fJ\varPsi^{a,\sigma}_{\tau,\varepsilon;t,s}\|_\sC
		\leq
		\lambda^{-\kappa}\,s^{-1/2-|a|},
		\qquad
		\|\tilde w(\fJ\varPsi^{a,\sigma}_{t,s}-
		\fJ\varPsi^{a,\sigma}_{\tau,\varepsilon;t,s})\|_\sC
		\leq
		\lambda^{-\kappa}\,
		\lambda_{\tau\vee\varepsilon}^\kappa\,
		\lambda_s^{-\kappa}\,
		s^{-1/2-|a|}
	\end{equation}
	by Lemma~\ref{lem:Psi_bounds} provided $\lambda\in(0,1]$ is small enough. Noting that
	$
	\rho_{\gamma,\kappa}(m)
	=
	\rho_{\gamma,\kappa}(m+k)-2\kappa k
	$
	and using Lemma~\ref{lem:weights}~(d), for every
	$a\in\bA^m$ and $\sigma\in\bFF^m$ we obtain
	\begin{multline}
		\lambda_s^{-\rho_{\gamma,\kappa}(m)}\,
		s^{2-m/2-|a|}\,
		\|w_s^m\fA_{\tau,\varepsilon;t,s}^{m,a,\sigma}V_s\|_{\sM^m}
		\\
		\leq
		\sum_{k\in\bN_0}
		\frac{(m+k)!}{m!k!}
		\sum_{b\in\bA^k}\sum_{\varsigma\in\bFF^k}
		\lambda_s^{-\rho_{\gamma,\kappa}(m+k)+\kappa k}\,
		s^{2-(m+k)/2-|a|-|b|}
		\|w_s^{m+k}
		V_s^{m+k,(a,b),(\sigma,\varsigma)}\|_{\sM^{m+k}}.
	\end{multline}
	Taking the supremum over $s\in(0,1]$, using $\lambda_s\leq\lambda$, and then summing over $a\in\bA^m$ and $\sigma\in\bFF^m$, we obtain
	\begin{equation}
		\|\fA_{\tau,\varepsilon;t,\Cdot}^mV_\Cdot\|_{\sV^{m;\gamma}}
		\leq
		\sum_{k\in\bN_0}
		\frac{(m+k)!}{m!k!}\,
		\lambda^{\kappa k}\,
		\|V_\Cdot^{m+k}\|_{\sV^{m+k;\gamma}}.
	\end{equation}
	Using Def.~\ref{dfn:V_space} of the norms
	$\|\Cdot\|_{\sV^{m+k;\gamma}}$ and
	$\|\Cdot\|_{\sV^{\alpha,\beta;\gamma}}$, we estimate the RHS of the above bound by
	\begin{multline}
		\sum_{k\in\bN_0}
		\frac{(m+k)!}{m!k!}\,
		\lambda^{\kappa k}\,
		\|V^{m+k}_\Cdot\|_{\sV^{m+k;\gamma}}
		\leq
		\sum_{k\in\bN_0}
		\frac{(m+k)!}{m!k!}\,
		\lambda^{\kappa k}\,
		\alpha^{-m-k}(m+k)^{-\beta}\,
		\|V_\Cdot\|_{\sV^{\alpha,\beta;\gamma}}
		\\
		\leq
		\alpha^{-m}\,m^{-\beta}\,
		(1-\lambda^\kappa)^{-m-1}\,
		\|V_\Cdot\|_{\sV^{\alpha,\beta;\gamma}}
		\leq
		2^m\,\alpha^{-m}m^{-\beta}\,
		\|V_\Cdot\|_{\sV^{\alpha,\beta;\gamma}}.
	\end{multline}
	The last bound above is true if $\lambda^\kappa\leq1/5$. This proves the bound~(A).
	
	Using analogous estimates with the sum over $k\in\bN_0$ replaced by the sum over $k\in\bN_+$, for every
	$a\in\bA^m$ and $\sigma\in\bFF^m$ we obtain
	\begin{multline}
		\lambda_s^{-\rho_{\gamma,\kappa}(m)}\,
		s^{2-m/2-|a|}\,
		\|w_s^m
		(\fA_{\tau,\varepsilon;t,s}^{m,a,\sigma}V_s
		-V_s^{m,a,\sigma})\|_{\sM^m}
		\\
		\leq
		\sum_{k\in\bN_+}
		\frac{(m+k)!}{m!k!}
		\sum_{b\in\bA^k}\sum_{\varsigma\in\bFF^k}
		\lambda_s^{-\rho_{\gamma,\kappa}(m+k)+\kappa k}\,
		s^{2-(m+k)/2-|a|-|b|}
		\|w_s^{m+k}
		V_s^{m+k,(a,b),(\sigma,\varsigma)}\|_{\sM^{m+k}}.
	\end{multline}
	Taking the supremum over $s\in(0,1]$ and then summing over $a$ and $\sigma$ gives
	\begin{equation}
		\|\fC_{\tau,\varepsilon;t,\Cdot}^mV_\Cdot\|_{\sV^{m;\gamma}}
		\leq
		\sum_{k\in\bN_+}
		\frac{(m+k)!}{m!k!}\,
		\lambda^{\kappa k}\,
		\|V_\Cdot^{m+k}\|_{\sV^{m+k;\gamma}}.
	\end{equation}
	Consequently,
	\begin{multline}
		\|\fC_{\tau,\varepsilon;t,\Cdot}^mV_\Cdot\|_{\sV^{m;\gamma}}
		\leq
		\alpha^{-m}\,(m+1)^{-\beta}\,
		((1-\lambda^\kappa)^{-m-1}-1)\,
		\|V_\Cdot\|_{\sV^{\alpha,\beta;\gamma}}
		\\
		\leq
		\lambda^\kappa\,2^m\,\alpha^{-m}m^{1-\beta}\,
		\|V_\Cdot\|_{\sV^{\alpha,\beta;\gamma}},
	\end{multline}
	where again the last bound is true if $\lambda^\kappa\leq1/5$. This implies the bound~(B).
	
	To prove the bounds~(C) and~(D), we use exactly the same argument with the weight $w_s^m$ replaced by $\tilde w_s^m$, taking the supremum over $s$ before summing over $a$ and $\sigma$.
	
	Let us turn to the proof of the bound~(E). By Lemma~\ref{lem:weights}~(e), for every
	$a\in\bA^m$ and $\sigma\in\bFF^m$ we have
	\begin{multline}
		\lambda^\kappa\,
		\lambda_{\tau\vee\varepsilon}^{-\kappa}\,
		\lambda_s^{-\rho_{\gamma,\kappa}(m)}\,
		s^{2-m/2-|a|}
		\|\tilde w_s^m
		(\fA_{t,s}^{m,a,\sigma}
		-\fA_{\tau,\varepsilon;t,s}^{m,a,\sigma})
		V_s\|_{\sM^m}
		\\
		\leq
		\sum_{k\in\bN_+}
		\frac{(m+k)!}{m!(k-1)!}
		\sum_{b\in\bA^k}\sum_{\varsigma\in\bFF^k}
		\lambda_s^{-\rho_{\gamma,\kappa}(m+k)+\kappa k}\,
		s^{2-(m+k)/2-|a|-|b|}
		\|w_s^{m+k}
		V_s^{m+k,(a,b),(\sigma,\varsigma)}\|_{\sM^{m+k}}.
	\end{multline}
	Taking the supremum over $s\in(0,1]$ and then summing over $a$ and $\sigma$ gives
	\begin{equation}
		\lambda^\kappa\,
		\lambda_{\tau\vee\varepsilon}^{-\kappa}
		\|(\fA_{t,\Cdot}-\fA_{\tau,\varepsilon;t,\Cdot})^m
		V_\Cdot\|_{\tilde\sV^{m;\gamma}}
		\leq
		\sum_{k\in\bN_+}
		\frac{(m+k)!}{m!(k-1)!}\,
		\lambda^{\kappa k}\,
		\|V_\Cdot^{m+k}\|_{\sV^{m+k;\gamma}}.
	\end{equation}
	Using again Def.~\ref{dfn:V_space} of the norms
	$\|\Cdot\|_{\sV^{m+k;\gamma}}$ and
	$\|\Cdot\|_{\sV^{\alpha,\beta;\gamma}}$, we estimate the RHS of the above bound by
	\begin{multline}
		\sum_{k\in\bN_+}
		\frac{(m+k)!}{m!(k-1)!}\,
		\lambda^{\kappa k}\,
		\|V^{m+k}_\Cdot\|_{\sV^{m+k;\gamma}}
		\leq
		\sum_{k\in\bN_+}
		\frac{(m+k)!}{m!(k-1)!}\,
		\lambda^{\kappa k}\,
		\alpha^{-m-k}(m+k)^{-\beta}\,
		\|V_\Cdot\|_{\sV^{\alpha,\beta;\gamma}}
		\\
		\leq
		\lambda^\kappa\,(m+1)\,
		\alpha^{-m}(m+1)^{-\beta}\,
		(1-\lambda^\kappa)^{-m-2}\,
		\|V_\Cdot\|_{\sV^{\alpha,\beta;\gamma}}
		\leq
		\lambda^\kappa\,2^m\,
		\alpha^{-m}m^{1-\beta}\,
		\|V_\Cdot\|_{\sV^{\alpha,\beta;\gamma}}.
	\end{multline}
	The last bound above is true if $\lambda^\kappa\leq1/5$. This proves the bound~(E).
	
	Finally, let $V_\Cdot\in\sV^{\alpha,\beta;\gamma}$ and set
	\begin{equation}
		V_\Cdot^{[M]}
		:=
		V_\Cdot-\Pi_{>M}V_\Cdot.
	\end{equation}
	By Remark~\ref{rem:C0_sV}, 
	\begin{equation}
		V_\Cdot^{[M]}
		\longrightarrow
		V_\Cdot
		\qquad
		\text{in}\quad \sV^{\alpha,\beta;\gamma}
		\quad\text{as}\quad M\to\infty,
	\end{equation}
	and $V_\Cdot^{[M]}\in\sV^{\mathrm{fin};\gamma}$.  Hence, by~(A),
	\begin{equation}
		\big\| \fA_{\tau,\varepsilon;t,\Cdot}
		(V_\Cdot^{[M]}-V_\Cdot^{[L]})
		\big\|_{\sV^{\alpha/2,\beta;\gamma}}
		\leq
		\|V_\Cdot^{[M]}-V_\Cdot^{[L]}\|_{\sV^{\alpha,\beta;\gamma}}.
	\end{equation}
	For every $m\in\bN_+$, the sequence $\fA_{\tau,\varepsilon;t,\Cdot}^mV_\Cdot^{[M]}$ therefore converges,  in the $\sF$-valued version of the norm $\|\Cdot\|_{\sV^{m;\gamma}}$, to an $\sF$-valued continuous map.  We define
	\begin{equation}
		\fA_{\tau,\varepsilon;t,\Cdot}V_\Cdot
		:=
		\lim_{M\to\infty}		\fA_{\tau,\varepsilon;t,\Cdot}V_\Cdot^{[M]}.
	\end{equation}
	The preceding bound implies that the convergence holds in the norm $\|\Cdot\|_{\sV^{\alpha/2,\beta;\gamma}}$ and that the extension is unique.  
	We then set
	\begin{equation}
		\fC_{\tau,\varepsilon;t,\Cdot}V_\Cdot
		:=
		\fA_{\tau,\varepsilon;t,\Cdot}V_\Cdot-V_\Cdot.
	\end{equation}
	Applying~(A)--(E) to $V_\Cdot^{[M]}$ and passing to the limit shows  that all the bounds~(A)--(E) remain valid for the extended maps.
\end{proof}

\begin{dfn}[Flow of charge and flavor]\label{dfn:flow_of_charge}
If $m$ is odd, then $\sN^m=\{0\}$, and we declare its unique element to be compatible with the flow of charge and flavor.
Let now $m=2k$ with $k\in\bN_+$. Write
\[
 \sigma=(q,n,\alpha)\in\bFF^m,
 \qquad
 q\in\{-,+\}^m,\quad
 n\in\{1,\ldots,N\}^m,\quad
 \alpha\in\{1,2\}^m.
\]
Let $\tilde\cP_k\subset \cP_{2k}$ denote the subgroup of permutations preserving the collection of pairs
\begin{equation}
	\big\{\{1,2\},\{3,4\},\ldots,\{2k-1,2k\}\big\}.
\end{equation}
We say that $V\in\sN^m$ is compatible with the flow of charge and flavor if there exist $\tilde V\in\sNN^m$ and a collection of kernels $(\hat V^{a,q,\alpha})_{a\in\bA^m,q\in\{-,+\}^m,\alpha\in\{1,2\}^m}$ such that
\[
 V=\fS \tilde V,
 \qquad
 \tilde V^{a,(q,n,\alpha)}
 =
 \left(
  \prod_{j=1}^k
  \mathbf 1_{\{q_{2j-1}\neq q_{2j}\}}
  \delta_{n_{2j-1},n_{2j}}
 \right)
 \hat V^{a,q,\alpha},
\]
and
\begin{equation}
 \fC_{2j-1}\fC_{2j} \tilde V= \tilde V,
 \quad
 j=1,\ldots,k,
 \qquad\textrm{and}\qquad
 \fS_\pi \tilde V=\tilde V,
 \quad
 \pi\in \tilde\cP_k.
\end{equation}
where $\fC_j$ denote charge conjugation $\fC$ acting on the
$j$-th argument of a kernel.
\end{dfn}

\begin{rem}\label{rem:flow_charge_internal}
Every kernel compatible with the flow of charge and flavor is invariant under the internal symmetries. Indeed, in an admissible representative the flavor dependence is a~product of the tensors $\delta_{n_{2j-1},n_{2j}}$, where the two members of each pair have opposite charges, and the remaining coefficient is independent of the flavor labels. These tensors are invariant under the $U(N)$ action. Since antisymmetrization commutes with the diagonal $U(N)$ action, the
claim follows.
\end{rem}

\begin{rem}\label{rem:symmetries_flow_charge}
The map $\fC$ was introduced in Def.~\ref{dfn:charge}. Note that if $V\in\sN^m$ is compatible with the flow of charge and flavor, then it is invariant under the charge conjugation. If $m\in\{4,6,\ldots\}$, then generically the reverse implication is false. Observe also that the second and third functionals~\eqref{eq:inv_functionals} are not compatible with the flow of charge and flavor. Hence, there are only three relevant or marginal local functionals invariant under all the symmetries listed in Sec.~\ref{sec:symmetries} and compatible with the flow of charge and flavor: the quadratic functionals~\eqref{eq:inv_functionals_quadratic} and the first of the quartic functionals~\eqref{eq:inv_functionals}. The reduction of the number of relevant or marginal local functionals that have to be investigated is the main reason behind introducing the notion of the compatibility with the flow of charge and flavor. See~\cite[Sec.~5.B]{GK85a} for a simple argument showing that perturbative corrections to an effective potential are compatible with the flow of charge and flavor.
\end{rem}

\begin{rem}
We say that $V\in\sN$ is invariant under a certain symmetry or compatible with the flow of charge and flavor iff $V^m \in\sN^m\subset\sS'(\bR^{2m})^{\bA^m\times\bFF^m}$ is invariant under this symmetry or compatible with the flow of charge and flavor for all $m\in\bN_+$.
\end{rem}

\begin{lem}\label{lem:flow_charge_closed}
For every $m\in\bN_+$, the subspace of $\sN^m$ consisting of kernels compatible with the flow of charge and flavor is closed.
\end{lem}
\begin{proof}
The statement is trivial for odd $m$. Let $m=2k$ and let $E_m=\bC^{\bFF^m}$. Denote by $R_m\subset E_m$ the linear subspace consisting of tensors $v=(v^{q,n,\alpha})$ of the form
\begin{equation}
 v^{q,n,\alpha}
 =
 \left(
  \prod_{j=1}^k
  \mathbf 1_{\{q_{2j-1}\neq q_{2j}\}}
  \delta_{n_{2j-1},n_{2j}}
 \right)
 \hat v^{q,\alpha}
\end{equation}
and satisfying
\begin{equation}
 \fC_{2j-1}\fC_{2j}v=v,
 \qquad j=1,\ldots,k.
\end{equation}
Let $\fp_m$ be the orthogonal projection onto $R_m$. We use the same notation for operators on $E_m$ and for their componentwise extensions to kernels, acting trivially on the spatial variables and jet indices.
Set
\begin{equation}
 \ft_m
 :=
 \frac1{m!}
 \sum_{\pi\in\cP_m}
 \fS_\pi \fp_m\fS_\pi^{-1}.
\end{equation}
Since $\fp_m$ acts only on the internal indices, the spatial and jet permutations cancel in each conjugation $\fS_\pi \fp_m\fS_\pi^{-1}$. Hence these operators act only on $E_m$ and are non-negative orthogonal projections. Therefore
\begin{equation}
 \operatorname{Ran}\ft_m
 =
 \sum_{\pi\in\cP_m}
 \operatorname{Ran}
 \bigl(\fS_\pi \fp_m\fS_\pi^{-1}\bigr).
\end{equation}
Moreover, $\ft_m$ commutes with every $\fS_\pi$.

Suppose first that $V=\fS\tilde V$ with $\fp_m\tilde V=\tilde V$. Then every $\fS_\pi\tilde V$ belongs to $\operatorname{Ran}(\fS_\pi \fp_m\fS_\pi^{-1})$, and therefore $V\in\operatorname{Ran}\ft_m$.

Conversely, let $V\in\sN^m\cap\operatorname{Ran}\ft_m$. Since $\ft_m$ is self-adjoint on the finite-dimensional internal-index space, we have
\begin{equation}
 E_m=\ker \ft_m\oplus(\ker \ft_m)^\perp,
\end{equation}
and the restriction
\begin{equation}
 \ft_m:(\ker \ft_m)^\perp\to\operatorname{Ran}\ft_m
\end{equation}
is an isomorphism. Let $\fs_m$ denote its inverse on $\operatorname{Ran}\ft_m$, extended by zero on $\ker \ft_m$.
Set
\begin{equation}
 Z:=\fs_mV.
\end{equation}
Then $\fs_m$ commutes with every $\fS_\pi$, and hence $Z$ is antisymmetric. Moreover, $\ft_m\fs_mV=V$, so 
\begin{equation}
 V
 =
 \ft_mZ
 =
 \frac1{m!}
 \sum_{\pi\in\cP_m}
 \fS_\pi \fp_m\fS_\pi^{-1}Z
 =
 \fS \,\fp_mZ.
\end{equation}
The subspace $R_m\subset E_m$ is invariant under the action induced by
every $\pi\in\tilde\cP_k$. Hence the orthogonal projection $\fp_m$
commutes with $\fS_\pi$ for every $\pi\in\tilde\cP_k$. Since $Z$ is
antisymmetric, we obtain
\begin{equation}
 \fS_\pi \fp_mZ
 =
 \fp_m\fS_\pi Z
 =
 \fp_mZ,
 \qquad
 \pi\in\tilde\cP_k.
\end{equation}
Thus $\fp_mZ$ satisfies the pair-preserving permutation condition in Def.~\ref{dfn:flow_of_charge}. By the definition of $R_m$, it also satisfies the charge, flavor and pairwise charge-conjugation
conditions. Since
\begin{equation}
 V=\fS\, \fp_mZ,
\end{equation}
the kernel $\fp_mZ$ is an admissible representative of $V$. Hence $V$
is compatible with the flow of charge and flavor.

We have therefore identified the flow-compatible subspace of $\sN^m$ with $\sN^m\cap\operatorname{Ran}\ft_m$. Since $\ft_m$ acts only on the finite-dimensional internal-index space, $\operatorname{Ran}\ft_m$ is closed, and hence so is the flow-compatible subspace of $\sN^m$.
\end{proof}

\begin{dfn}\label{dfn:sN_symmetries}
Let $m\in\bN_+$. We denote by $\sN^m_+$ the real linear subspace of $\sN^m$ consisting of Schwartz distributions which are invariant under the symmetries of the torus and the internal symmetries, compatible with the flow of charge and flavor, and real with respect to the reality conjugation. We denote by $\sN^m_0$ the real linear subspace of $\sN^m_+$ consisting of Schwartz distributions which are invariant under the symmetries of the plane.
Let $\mathrm{sgn}(\tau):=0$ if $\tau=0$ and $\mathrm{sgn}(\tau):=+$ if $\tau\in(0,1]$.
For $\tau\in[0,1]$ and $\gamma\in\bR$, we set
$\sN^m_\tau:=\sN^m_{\operatorname{sgn}(\tau)}$ and
\begin{equation}
\begin{aligned}
 \sN^{\mathrm{fin}}_\tau
 &:=
 \left\{
 V\in\sN^{\mathrm{fin}}
 \,\middle|\,
 \forall_{m\in\bN_+}\ V^m\in\sN^m_\tau
 \right\},
 \\
 \sV^{\mathrm{fin};\gamma}_\tau
 &:=
 \left\{
 V_\Cdot\in\sV^{\mathrm{fin};\gamma}
 \,\middle|\,
 \forall_{s\in(0,1]}\ V_s\in\sN^{\mathrm{fin}}_\tau
 \right\}.
\end{aligned}
\end{equation}
The spaces $\sN^m_\tau$, $\sN^{\mathrm{fin}}_\tau$ and
$\sV^{\mathrm{fin};\gamma}_\tau$ are understood as real vector spaces.
\end{dfn}

\begin{dfn}\label{dfn:sW}
Let $\alpha,\beta\in[1,\infty)$, $\gamma\in\bR$ and
$\tau,\varepsilon\in[0,1]$. We define the real vector space
$\sW^{\alpha,\beta;\gamma}_{\tau,\varepsilon}$ as the closure of the
real linear subspace
\[
 \sV^{\mathrm{fin};\gamma}_{\tau}
 \subset
 \sV^{\alpha,\beta;\gamma}
\]
with respect to the norm
\begin{equation}\label{eq:norm_sW_alpha_beta}
 \|V_\Cdot\|_{\sW^{\alpha,\beta;\gamma}_{\tau,\varepsilon}}
 :=
 \sup_{t\in[0,1]}
 \|s\mapsto
 \fA_{\tau,\varepsilon;t\vee s,s}V_s
 \|_{\sV^{\alpha,\beta;\gamma}}.
\end{equation}
The norm on the RHS of \eqref{eq:norm_sW_alpha_beta} is understood in the $\sF$-valued sense of Def.~\ref{dfn:V_space}.
For
$V_\Cdot\in\sW_{\tau,\varepsilon}^{\alpha,\beta;\gamma}$
we define
$\|V_\Cdot\|_{\tilde\sW^{\alpha,\beta;\gamma}_{\tau,\varepsilon}}$
in an analogous way to
$\|V_\Cdot\|_{\sW^{\alpha,\beta;\gamma}_{\tau,\varepsilon}}$,
with $\|\Cdot\|_{\sV^{\alpha,\beta;\gamma}}$ replaced by
$\|\Cdot\|_{\tilde\sV^{\alpha,\beta;\gamma}}$
in~\eqref{eq:norm_sW_alpha_beta}.
\end{dfn}
\begin{rem}\label{rem:sW_graph_norm}
We use the notation $\|V_\Cdot\|_{\sW_{\tau,\varepsilon}^{\alpha,\beta;\gamma}}$ and $\|V_\Cdot\|_{\tilde\sW_{\tau,\varepsilon}^{\alpha,\beta;\gamma}}$ whenever the expressions on the RHS of~\eqref{eq:norm_sW_alpha_beta} are well defined also for $V_\Cdot$, which do not necessarily belong to $\sW_{\tau,\varepsilon}^{\alpha,\beta;\gamma}$.
\end{rem}

\begin{rem}\label{rem:sW_tilde}
By Remark~\ref{rem:sV_tilde} the norm $\|\Cdot\|_{\tilde\sW^{\alpha,\beta;\gamma}_{\tau,\varepsilon}}$ is weaker than the norm $\|\Cdot\|_{\sW^{\alpha,\beta;\gamma}_{\tau,\varepsilon}}$. The former norm will be used to study the infinite volume limit. 
\end{rem}
\begin{rem}\label{rem:sV_sW}
Note that
\begin{equation}
 \|V_\Cdot\|_{\sV^{\alpha,\beta;\gamma}}
 \leq
 \|V_\Cdot\|_{\sW^{\alpha,\beta;\gamma}_{\tau,\varepsilon}},
 \qquad
 \|V_\Cdot\|_{\tilde\sV^{\alpha,\beta;\gamma}}
 \leq
 \|V_\Cdot\|_{\tilde\sW^{\alpha,\beta;\gamma}_{\tau,\varepsilon}}.
\end{equation}
Thus, in particular $\sW^{\alpha,\beta;\gamma}_{\tau,\varepsilon}\subset\sV^{\alpha,\beta;\gamma}$. 
\end{rem}
\begin{rem}\label{rem:sW_sV}
Observe that by Lemma~\ref{lem:A_map_bound}~(A),~(C) it holds
\begin{equation}
 \|V_\Cdot\|_{\sW^{\alpha,\beta;\gamma}_{\tau,\varepsilon}}
 \leq
 \|V_\Cdot\|_{\sV^{2\alpha,\beta;\gamma}},
 \qquad
 \|V_\Cdot\|_{\tilde\sW^{\alpha,\beta;\gamma}_{\tau,\varepsilon}}
 \leq
 \|V_\Cdot\|_{\tilde\sV^{2\alpha,\beta;\gamma}}.
\end{equation}
\end{rem}

\begin{rem}\label{rem:sW_symmetries}
By Lemma~\ref{lem:flow_charge_closed}, compatibility with the flow of charge and flavor is preserved under limits. The geometric and internal symmetries and the reality condition are preserved under limits as well. Consequently, $\sN_\tau^m$ is closed in $\sN^m$ for every $m\in\bN_+$ and $\tau\in[0,1]$. In particular, if
$V_\Cdot\in\sW_{\tau,\varepsilon}^{\alpha,\beta;\gamma}$, then
\begin{equation}
 V_s^m\in\sN_\tau^m,
 \qquad m\in\bN_+,
 \qquad s\in(0,1].
\end{equation}
\end{rem}

\begin{lem}\label{lem:W_banach}
For all $\alpha,\beta\in[1,\infty)$, $\gamma\in\bR$ and $\tau,\varepsilon\in[0,1]$ the space $(\sW^{\alpha,\beta;\gamma}_{\tau,\varepsilon},\|\Cdot\|_{\sW^{\alpha,\beta;\gamma}_{\tau,\varepsilon}})$ is a real Banach space.
\end{lem}
\begin{proof}
Let $({}^pV)_{p\in\bN_+}$ be a Cauchy sequence in $\sW^{\alpha,\beta;\gamma}_{\tau,\varepsilon}$. Then, $({}^pV)_{p\in\bN_+}$ is a Cauchy sequence in $\sV^{\alpha,\beta;\gamma}$. Consequently, there exists $V\in\sV^{\alpha,\beta;\gamma}$ such that 
\begin{equation}
\lim_{p\to\infty}\|V_\Cdot-({}^pV)_\Cdot\|_{\sV^{\alpha,\beta;\gamma}}=0.
\end{equation} 
By Lemma~\ref{lem:A_map_bound}~(A) it holds
\begin{equation}
	\|s\mapsto\fA^m_{\tau,\varepsilon;t,s} (({}^pV)_s-V_s)\|_{\sV^{m;\gamma}} \leq2^m\alpha^{-m}m^{-\beta}\,\|({}^pV)_\Cdot-V_\Cdot\|_{\sV^{\alpha,\beta;\gamma}}.
\end{equation}
Let $\delta\in(0,\infty)$. There exists $P\in\bN_+$ such that for all $p,q\in\bN_+$, $p,q\geq P$ it holds
\begin{equation}
 \|({}^pV)_\Cdot-({}^qV)_\Cdot\|_{\sW^{\alpha,\beta;\gamma}_{\tau,\varepsilon}}\leq \delta
\end{equation}
and for every $m\in\bN_+$ there exists $p(m)\in\bN_+$, $p(m)\geq P$, such that 
\begin{equation}
 \|s\mapsto\fA^m_{\tau,\varepsilon;t,s} (({}^{p(m)}V)_s-V_s)\|_{\sV^{m;\gamma}} \leq \alpha^{-m} m^{-\beta}\,\delta.
\end{equation}
As a result, for all $p\in\bN_+$, $p\geq P$ it holds
\begin{multline}
 \|s\mapsto\fA^m_{\tau,\varepsilon;t,s} (V_s-({}^p V)_s)\|_{\sV^{m;\gamma}} 
 \\
 \leq 
 \|s\mapsto\fA^m_{\tau,\varepsilon;t,s} (V_s-({}^{p(m)} V)_s)\|_{\sV^{m;\gamma}} 
 +
 \|s\mapsto\fA^m_{\tau,\varepsilon;t,s} (({}^{p(m)} V)_s-({}^{p} V)_s)\|_{\sV^{m;\gamma}}
 \\
 \leq 
 \alpha^{-m} m^{-\beta}\,\delta 
 + \alpha^{-m}m^{-\beta}\,
 \|({}^{p(m)}V)_\Cdot-({}^pV)_\Cdot\|_{\sW^{\alpha,\beta;\gamma}_{\tau,\varepsilon}}
 \leq 2\,\alpha^{-m} m^{-\beta}\,\delta.
\end{multline}
Taking the supremum over $m\in\bN_+$ and $t\in[0,1]$ gives
\begin{equation}
	\|V_\Cdot-({}^pV)_\Cdot\|_{\sW^{\alpha,\beta;\gamma}_{\tau,\varepsilon}}\leq 2\delta,
	\qquad p\geq P.
\end{equation}
Since every $({}^pV)_\Cdot$ belongs to the closure of  $\sV_\tau^{\mathrm{fin};\gamma}$ in the above norm,  a diagonal approximation shows that $V_\Cdot$ is in the closure. 
Hence, $V_\Cdot\in\sW^{\alpha,\beta;\gamma}_{\tau,\varepsilon}$ and $\lim_{p\to\infty}\|V_\Cdot-({}^pV)_\Cdot\|_{\sW^{\alpha,\beta;\gamma}_{\tau,\varepsilon}}=0$.
\end{proof}

\section{Local part and remainder}\label{sec:loc_ren}

In this section we introduce operators $\fL$, $\fL_\partial$ and $\fR$ and establish their most important properties. The operators $\fL$, $\fL_\partial$ and $\fR$ are used to decompose kernels of functionals into the local part and the remainder. As we discussed in Sec.~\ref{sec:strategy}, the above-mentioned decomposition is crucial for the solution of the renormalization problem. 

\begin{dfn}\label{dfn:U_functional}
The distributions
\begin{equation}
 U^{2}
 \equiv(U^{2,a,\sigma})_{a\in\bA^2,\sigma\in\bFF^2},\,
 U_\partial^{2}
 \equiv(U_\partial^{2,a,\sigma})_{a\in\bA^2,\sigma\in\bFF^2}\in\sN^2_0\subset \sS'(\bR^4)^{\bA^2\times\bFF^2},
\end{equation}
are defined by the equalities
\begin{equation}\label{eq:dfn_U_2}
\begin{gathered}
\langle U^{2},\psi^{\otimes 2}\rangle
 = 
 \int_{\bR^2} \bar\psi^{0}(x)\cdot\ubar\psi^{0}(x)\,\rd x,
 \\
 \langle U_\partial^{2},\psi^{\otimes 2}\rangle
 =
 \frac12\sum_{i=1}^2
 \int_{\mathbb R^2}
 \left(
 \bar\psi^0(x)\cdot\Gamma_i\ubar\psi^{e_i}(x)
 -
 \bar\psi^{e_i}(x)\cdot\Gamma_i\ubar\psi^0(x)
 \right)\,\rd x.
\end{gathered} 
\end{equation}
for all $\psi\in\sS(\bR^2)^{\bA\times\bFF}\otimes_{\mathrm{alg}}\sG^-$, where $e_1=(1,0)$, $e_2=(0,1)$. The distribution
\begin{equation}
 U^{4}
 \equiv(U^{4,a,\sigma})_{a\in\bA^4,\sigma\in\bFF^4}\in\sN^4_0\subset\sS'(\bR^8)^{\bA^4\times\bFF^4}
\end{equation}
is defined by the equalities
\begin{equation}\label{eq:dfn_U_4}
 \langle U^{4},\psi^{\otimes 4}\rangle
 = \int_{\bR^2}(\bar\psi^0(x)\cdot\ubar\psi^0(x))^2\,\rd x
\end{equation}
for all $\psi\in\sS(\bR^2)^{\bA\times\bFF}\otimes_{\mathrm{alg}}\sG^-$. For $g,r,z\in\bC$ we define $U(g,r,z)\in\sN$ by the equation
\begin{equation}
 U(g,r,z)[\psi]:= 
 g\,\langle U^{4},\psi^{\otimes4}\rangle 
 +
 r\, \langle U^{2},\psi^{\otimes2}\rangle
 +
 z\, \langle U_\partial^{2},\psi^{\otimes2}\rangle
\end{equation}
for all $\psi\in\sS(\bR^2)^{\bA\times\bFF}\otimes_{\mathrm{alg}}\sG^-$.
\end{dfn}

\begin{rem}
Note that $U(g,r,z)=g\,U(1,0,0) + r\,U(0,1,0) + z\, U(0,0,1)$.
\end{rem}

\begin{rem}\label{rem:U_reality}
Let the reality conjugation $V\mapsto V^\sharp$ be as in
Def.~\ref{dfn:reality_conjugation}. Then
\begin{equation}
 (U^2)^\sharp=U^2,
 \qquad
 (U_\partial^2)^\sharp=U_\partial^2,
 \qquad
 (U^4)^\sharp=U^4
 \, .
\end{equation}
Consequently, for all $g,r,z\in\bC$,
\begin{equation}
 U(g,r,z)^\sharp
 =
 U(\overline g,\overline r,\overline z)
 \, .
\end{equation}
Indeed, the first and the third identities follow from
$\Gamma_1^{\mathrm{t}}\Gamma_1=1$. For the derivative quadratic
functional one additionally uses
\begin{equation}
 \Gamma_1^{\mathrm{t}}\overline{\Gamma_j}\Gamma_1
 =
 \Gamma_j,
 \qquad
 j\in\{1,2\},
\end{equation}
which follows from Remark~\ref{rem:gamma}.
\end{rem}

\begin{rem}\label{rem:U_dfn}
We have
\begin{itemize}
 \item if $\sigma=((-,1,1),(+,1,1))\in\bFF^2$, then $U^{2,0,\sigma}=\delta^{(2)}/2$,
 \item if $a=((0,0),(1,0))\in\bA^2$, $\sigma=((-,1,1),(+,1,2))\in\bFF^2$, then $U_\partial^{2,a,\sigma}=\delta^{(2)}/4$,
 \item if $\sigma=((-,1,1),(+,1,1),(-,2,1),(+,2,1))\in\bFF^4$, then $U^{4,0,\sigma}=\delta^{(4)}/12$,
\end{itemize}  
where $\delta^{(m)}\in\sS'(\bR^{2m})$ is the Dirac measure on the diagonal introduced in Remark~\ref{rem:Dirac_diagonal}.
\end{rem}

\begin{rem}\label{rem:canonical_kernel}
Let $m\in\bN_+$ and let $V\in\sM^m$ be translationally invariant.
Then there exists a unique finite measure
$\mu_V$ on $\bR^{2(m-1)}$ such that
\begin{equation}\label{eq:canonical_kernel}
 \langle V,\varphi\rangle
 =
 \int_{\bR^2}\int_{\bR^{2(m-1)}}
 \varphi(x,x+x_2,\ldots,x+x_m)\,
 \mu_V(\rd x_2\ldots\rd x_m)\,\rd x
\end{equation}
for every $\varphi\in\sS(\bR^{2m})$.
Indeed, if \(V_{\rm K}\) is a representing kernel and
\(\mu_x(A):=V_{\rm K}(x,x+A)\), then translational invariance and Lemma~\ref{lem:sM_banach} imply that \(\mu_{x+h}=\mu_x\) for almost every \(x\). Hence \(x\mapsto\mu_x\) is almost everywhere constant.
In particular, $V$ admits the
canonical representative
\begin{equation}\label{eq:canonical_kernel_rep}
 V_\rK^{\rm can}(x,A)
 :=
 \int_{\bR^{2(m-1)}}
 1_A(x+x_2,\ldots,x+x_m)\,
 \mu_V(\rd x_2\ldots\rd x_m),
\end{equation}
which is translation covariant for every $x\in\bR^2$. Moreover,
\begin{equation}
 \|V\|_{\sM^m}
 =
 \sup_{x\in\bR^2}
 |V_\rK^{\rm can}|(x,\bR^{2(m-1)})
 =
 |\mu_V|(\bR^{2(m-1)}).
\end{equation}
More generally, the kernel obtained by multiplying
$V_\rK^{\rm can}$ by $w^m_{t;\nu}$ satisfies
\begin{equation}
 \|w^m_{t;\nu}V\|_{\sM^m}
 =
 \sup_{x\in\bR^2}
 \int_{\bR^{2(m-1)}}
 w^m_{t;\nu}(x,x+x_2,\ldots,x+x_m)\,
 |\mu_V|(\rd x_2\ldots\rd x_m),
\end{equation}
and the supremum on the RHS is attained at $x=0$.
Whenever a pointwise kernel formula is used for a translationally
invariant distribution, in particular in Def.~\ref{dfn:L}, we use the
canonical representative~\eqref{eq:canonical_kernel_rep}.
\end{rem}

\begin{dfn}\label{dfn:kernel_dilation}
Let $m\in\bN_+$, let $V\in\sM^m$ be translationally invariant, and
let $V_\rK^{\rm can}$ be its canonical representative introduced in
Remark~\ref{rem:canonical_kernel}. For $u\in(0,1]$ and
$x\in\bR^2$, define
\begin{equation}
 T_{x,u}(x_2,\ldots,x_m)
 :=
 \big(
 x+u(x_2-x),\ldots,x+u(x_m-x)
 \big).
\end{equation}
We define the relative-coordinate dilation $\mathsf D_u V_\rK^{\rm can}$
by
\begin{equation}\label{eq:kernel_dilation}
 \big(\mathsf D_u V_\rK^{\rm can}\big)(x,A)
 :=
 V_\rK^{\rm can}\big(x,T_{x,u}^{-1}(A)\big),
 \qquad
 A\in\operatorname{Borel}(\bR^{2(m-1)}).
\end{equation}
Equivalently, if $\mu_V$ is the relative-coordinate measure from
Remark~\ref{rem:canonical_kernel}, then
\begin{equation}\label{eq:kernel_dilation_relative}
 \big(\mathsf D_u V_\rK^{\rm can}\big)(x,A)
 =
 \int_{\bR^{2(m-1)}}
 1_A(x+ux_2,\ldots,x+ux_m)\,
 \mu_V(\rd x_2\ldots\rd x_m).
\end{equation}
Thus $\mathsf D_u$ leaves the distinguished variable $x$ fixed and
dilates only the relative coordinates $x_j-x$, $j=2,\ldots,m$.
\end{dfn}

\begin{dfn}\label{dfn:L}
Let $V^2=(V^{2,a,\sigma})_{a\in\bA^2,\sigma\in\bFF^2}\in\sN^2$ and $V^4=(V^{4,a,\sigma})_{a\in\bA^4,\sigma\in\bFF^4}\in\sN^4$ be translationally invariant. We define $\fL V^2\in\bC$ by the equality
\begin{equation}\label{eq:L_V_2_dfn}
 \fL V^2
 :=
 2\int_{\bR^2} V^{2,0,\sigma}(x_1,\rd x_2)
\end{equation}
with $\sigma=((-,1,1),(+,1,1))\in\bFF^2$. We define $\fL_\partial V^2:=\hat\fL_\partial V^2+\check\fL_\partial V^2\in\bC$ by the equalities
\begin{equation}\label{eq:L_V_2_partial_dfn}
 \hat\fL_\partial V^2
 :=
 2\int_{\bR^2} (x_2-x_1)^{a_1+a_2}\,V^{2,0,\sigma}(x_1,\rd x_2),
 \qquad
 \check\fL_\partial V^2
 := 
 4\int_{\bR^2} 
 V^{2,a,\sigma}(x_1,\rd x_2),
\end{equation}
with $a=(a_1,a_2)=((0,0),(1,0))\in\bA^2$ and $\sigma=((-,1,1),(+,1,2))\in\bFF^2$. We define $\fL V^4\in\bC$ by the equality
\begin{equation}\label{eq:L_V_4_dfn}
 \fL V^4
 :=
 12\int_{\bR^6} V^{4,0,\sigma}(x_1,\rd x_2,\rd x_3,\rd x_4)
\end{equation} 
with $\sigma=((-,1,1),(+,1,1),(-,2,1),(+,2,1))\in\bFF^4$.
\end{dfn}
\begin{rem}\label{rem:L_translational}
With the canonical representatives introduced above, the RHS of the definitions of $\fL V^2$, $\fL_\partial V^2$, and $\fL V^4$ are independent of $x_1\in\mathbb R^2$.
\end{rem}
\begin{rem}
Note that $\fL U^2 =1$, $\fL_\partial U_\partial^2=1$, $\fL U^4 = 1$.
\end{rem}

\begin{lem}\label{lem:fL_continuity}
Let $\nu\in[0,1/2]$. There exists $C\in(0,\infty)$ such that for all $t\in(0,1]$ and all translationally invariant \mbox{$V^2=(V^{2,a,\sigma})_{a\in\bA^2,\sigma\in\bFF^2}\in\sN^2$}, \mbox{$V^4=(V^{4,a,\sigma})_{a\in\bA^4,\sigma\in\bFF^4}\in\sN^4$} it holds:
\begin{itemize}
 \item[(A)] $|\fL V^2| \leq C\,\sup_{a\in\bA^2}\sup_{\sigma\in\bFF^2}\|w^2_{t;\nu} V^{2,a,\sigma}\|_{\sM^2}$,
 \item[(B)] $|\fL_\partial V^2| \leq C\, \sup_{a\in\bA^2}\sup_{\sigma\in\bFF^2}\|w^2_{t;\nu} V^{2,a,\sigma}\|_{\sM^2}$,
 \item[(C)] $|\fL V^4| \leq C\,\sup_{a\in\bA^4}\sup_{\sigma\in\bFF^4}\|w^4_{t;\nu} V^{4,a,\sigma}\|_{\sM^4}$.
\end{itemize}
\end{lem}
\begin{proof}
By Remark~\ref{rem:L_translational} we can set $x_1=0$ in the equations defining $\fL V^2,\fL_\partial V^2,\fL V^4$, which are given in Def.~\ref{dfn:L}. The bounds follow now immediately from Def.~\ref{dfn:sM} of the norm $\|\Cdot\|_{\sM^m}$, the fact that $w^m_{t;\nu}\geq1$ at $x_1=0$, and the fact that there exists $C\in(0,\infty)$ such that
\begin{equation}
 \frac{|x_2-x_1|}{w^2_{t;\nu}(x_1,x_2)}
 \leq C\,t
\end{equation}
for $x_1=0$, all $x_2\in\bR^2$ and all $t\in(0,1]$.
\end{proof}

\begin{lem}\label{lem:fL_symmetries}
Let $V^2=(V^{2,a,\sigma})_{a\in\bA^2,\sigma\in\bFF^2}\in\sN^2_+$, $V^4=(V^{4,a,\sigma})_{a\in\bA^4,\sigma\in\bFF^4}\in\sN^4_+$.
%
%
Then we have
\begin{equation}\label{eq:L_real}
 \fL V^2,  \hat\fL_\partial V^2, \check\fL_\partial V^2,  \fL V^4\in\bR
\end{equation}
and the following equalities
\begin{equation}
 \fL V^2\,\langle U^2,(\fJ\psi)^{\otimes 2}\rangle
 =
 \sum_{\sigma\in\bFF^2}\int_{\bR^4} V^{2,0,\sigma}(x_1,\rd x_2)\,
 \psi^{\sigma_1}(x_1)\,\psi^{\sigma_2}(x_1)\,\rd x_1,
\end{equation}
\begin{equation}
 \hat\fL_\partial V^2\,\langle U^2_\partial,(\fJ\psi)^{\otimes 2}\rangle
 =\sum_{\substack{a\in\bA\\|a|=1}}\sum_{\sigma\in\bFF^2}\int_{\bR^4} (x_2-x_1)^{a}\,V^{2,0,\sigma}(x_1,\rd x_2)\,
 \psi^{\sigma_1}(x_1)\,(\fJ\psi)^{a,\sigma_2}(x_1)\,\rd x_1,
\end{equation}
\begin{equation}
 \check\fL_\partial V^2\,\langle U^2_\partial,(\fJ\psi)^{\otimes 2}\rangle
 = 
 \sum_{\substack{a\in\bA^2\\|a|=1}}\sum_{\sigma\in\bFF^2}\int_{\bR^4} 
 V^{2,a,\sigma}(x_1,\rd x_2) 
 \,(\fJ\psi)^{a_1,\sigma_1}(x_1)\,(\fJ\psi)^{a_2,\sigma_2}(x_1)\,\rd x_1,
\end{equation}
\begin{equation}
 \fL V^4\,\langle U^{4},(\fJ\psi)^{\otimes 4}\rangle
 =
 \sum_{\sigma\in\bFF^4}\int_{\bR^8} V^{4,0,\sigma}(x_1,\rd x_2,\rd x_3,\rd x_4)\,\psi^{\sigma_1}(x_1)\,\psi^{\sigma_2}(x_1)\,\psi^{\sigma_3}(x_1)\,\psi^{\sigma_4}(x_1)\,\rd x_1
\end{equation} 
hold true for all $\psi\in\sS(\bR^2)^\bFF\otimes_{\mathrm{alg}}\sG^-$.
\end{lem}
\begin{rem}
Recall that $\fJ\psi^{a,\sigma}\equiv(\fJ\psi)^{a,\sigma}$ is the jet prolongation. In particular, $\psi^\sigma=\fJ\psi^{0,\sigma}$. The space $\sN^m_+$ was introduced in Def.~\ref{dfn:sN_symmetries}.
\end{rem}
\begin{proof}
 The four expressions on the RHS of the identities in the statement define local functionals supported on the diagonal. Moreover, the operations used in their  definition, namely integration in the  relative variables, multiplication by $x_2-x_1$, and insertion of the corresponding jets at $x_1$, preserve the symmetries of the torus, the internal symmetries, and compatibility with the flow of charge and flavor.

 At the orders considered here, the discrete geometric symmetries of the torus give the same low-degree local spinor structures as in Remark~\ref{rem:symmetries}. Together with the internal $U(N)$ symmetry, this leaves the two quadratic terms and the three quartic terms displayed there. Compatibility with the flow of charge and flavor excludes the vector and pseudoscalar quartic terms. Consequently, the four local functionals appearing on the RHS of the statement are respectively proportional to
 \begin{equation}
 \langle U^2,(\fJ\psi)^{\otimes2}\rangle,
 \qquad
 \langle U_\partial^2,(\fJ\psi)^{\otimes2}\rangle,
 \qquad
 \langle U_\partial^2,(\fJ\psi)^{\otimes2}\rangle,
 \qquad
 \langle U^4,(\fJ\psi)^{\otimes4}\rangle.
 \end{equation}
 Comparing the components, respectively the one-sided first moment, singled out in Def.~\ref{dfn:L}, and using Remark~\ref{rem:U_dfn},  shows that the corresponding proportionality constants are
 \begin{equation}
 \fL V^2,\qquad
 \hat\fL_\partial V^2,\qquad
 \check\fL_\partial V^2,\qquad
 \fL V^4.
 \end{equation}
 This proves the four identities in the statement.
 
 Since $V^2$ and $V^4$ are real with respect to the reality conjugation, the four local functionals constructed above are real as well. Indeed, the operations  entering their definition commute with the reality  conjugation. On the other hand, by Remark~\ref{rem:U_reality},
 \begin{equation}
 (U^2)^\sharp=U^2,\qquad
 (U_\partial^2)^\sharp=U_\partial^2,\qquad
 (U^4)^\sharp=U^4.
 \end{equation}
 Applying the reality conjugation to the four identities already proved
 therefore shows that
 \begin{equation}
 \fL V^2=\overline{\fL V^2},\qquad
 \hat\fL_\partial V^2=\overline{\hat\fL_\partial V^2},\qquad
 \check\fL_\partial V^2=\overline{\check\fL_\partial V^2},\qquad
 \fL V^4=\overline{\fL V^4},
 \end{equation}
 which proves~\eqref{eq:L_real}.
\end{proof}

\begin{dfn}\label{dfn:R_map}
Let $V^2=(V^{2,a,\sigma})_{a\in\bA^2,\sigma\in\bFF^2}\in\sN^2 \subset \sS'(\bR^4)^{\bA^2\times\bFF^2}$ be translationally invariant
and such that \mbox{$\cX V^2\in\sNN^2$} for all translationally invariant polynomials $\cX\in C^\infty(\bR^4)$ of degree two. We define $\hat W^2\in\sNN^2\subset \sS'(\bR^4)^{\bA^2\times\bFF^2}$ and $\check W^2\in\sS'(\bR^4)^{\bA^2\times\bA^2\times\bFF^2}$ by~\eqref{eq:V_rK} with 
\begin{equation}
 \hat W_\rK^{2,a,\sigma}(x_1,\rd x_2)
 :=
 \frac{2!}{a!}
 \int_0^1
 (1-u)\,u^{-2}\,
 (x_1-x_1)^{a_1}(x_2-x_1)^{a_2}
 \mathsf D_u(V_\rK^{2,0,\sigma})(x_1,\rd x_2)
 \,\rd u
\end{equation}
and
\begin{equation}
 \check W_\rK^{2,b,c,\sigma}(x_1,\rd x_2)
 :=
 \int_0^1
 u^{-1}\,
 (x_1-x_1)^{b_1}(x_2-x_1)^{b_2}
 \mathsf D_u(V_\rK^{2,c,\sigma})(x_1,\rd x_2)
 \,\rd u
\end{equation}
if $|a|=2$, $|b|=|c|=1$ and $\hat W_\rK^{2,a,\sigma}=0$, $\check W_\rK^{2,b,c,\sigma}=0$ otherwise. We define $W^2\in\sNN^2$ by the equality
\begin{equation}
 W^{2,a,\sigma}=\hat W^{2,a,\sigma}+
 \sum_{\substack{b,c\in\bA^2\\b+c=a}}
 \check W^{2,b,c,\sigma}.
\end{equation}
We define $\fR V^2 \equiv ((\fR V^{2})^{a,\sigma})_{a\in\bA^2,\sigma\in\bFF^2}\in\sN^2$ by the equalities
\begin{equation}
\begin{aligned}
 (\fR V^{2})^{a,\sigma}&=0,
 &
 \qquad |a|&\leq 1,
 \\
 (\fR V^{2})^{a,\sigma}&=V^{2,a,\sigma}+(\fS W^{2})^{a,\sigma}, 
 &
 \qquad|a|&=2, 
 \\
 (\fR V^{2})^{a,\sigma}&=V^{2,a,\sigma},
 &
 \qquad|a|&\geq 3.
\end{aligned} 
\end{equation}
Let $$V^4=(V^{4,a,\sigma})_{a\in\bA^4,\sigma\in\bFF^4}\in\sN^4\subset \sS'(\bR^8)^{\bA^4\times\bFF^4}$$ be translationally invariant and such that $\cX V^4\in\sNN^4$ for all translationally invariant polynomials $\cX\in C^\infty(\bR^8)$ of degree one. We define $W^4\in\sNN^4\subset \sS'(\bR^8)^{\bA^4\times\bFF^4}$ by~\eqref{eq:V_rK} with
\begin{equation}\label{eq:dfn_W_4}
 W_\rK^{4,a,\sigma}(x_1,\rd x_2,\rd x_3,\rd x_4)
 :=
 \int_0^1
 u^{-1}\,
 \prod_{j=1}^4(x_j-x_1)^{a_j}
 \mathsf D_u(V_\rK^{4,0,\sigma})
 (x_1,\rd x_2,\rd x_3,\rd x_4)
 \,\rd u
\end{equation}
if $|a|=1$ and $W_\rK^{4,a,\sigma}=0$ otherwise. We define $\fR V^4 \equiv ((\fR V^{4})^{a,\sigma})_{a\in\bA^4,\sigma\in\bFF^4}\in\sN^4$ by the equalities 
\begin{equation}
\begin{aligned}
 (\fR V^{4})^{a,\sigma}&=0,
 &
 \qquad |a|&= 0,
 \\
 (\fR V^{4})^{a,\sigma}&=V^{4,a,\sigma}+(\fS W^{4})^{a,\sigma}, 
 &
 \qquad|a|&=1, 
 \\
 (\fR V^{4})^{a,\sigma}&=V^{4,a,\sigma},
 &
 \qquad|a|&\geq 2.
\end{aligned} 
\end{equation}
\end{dfn}
\begin{rem}\label{rem:R_kernel_independence}
In Def.~\ref{dfn:R_map}, all kernel formulas are understood using the canonical representatives from Remark~\ref{rem:canonical_kernel}. Hence the definition of $\fR$ is independent of the choice of a
representing kernel.  The integrals appearing in the above formulas for $\hat W_\rK^{2,a,\sigma}$, $\check W_\rK^{2,b,c,\sigma}$ and  $W_\rK^{4,a,\sigma}$ are defined setwise for every Borel set.
\end{rem}
\begin{rem}
Recall that $\fS V$, introduced in Def.~\ref{dfn:permutations}, denotes the antisymmetric part of~$V$. Note that $\hat W_\rK^{2,a,\sigma}=0$ unless $a_1=0$, $\check W_\rK^{2,b,c,\sigma}=0$ unless $b_1=0$ and $W_\rK^{4,a,\sigma}=0$ unless $a_1=0$.
\end{rem}
\begin{rem}
It follows from Def.~\ref{dfn:weights} of the weight $w^m_s$ that if $V^2_\Cdot\in\sV^{2;\gamma}$, then \mbox{$\cX V^2_s\in\sNN^2$} for all translationally invariant polynomials $\cX\in C^\infty(\bR^4)$ and all $s\in(0,1)$ and if $V^4_\Cdot\in\sV^{4;\gamma}$, then \mbox{$\cX V^4_s\in\sNN^4$} for all translationally invariant polynomials $\cX\in C^\infty(\bR^8)$ and all $s\in(0,1)$.
\end{rem}
\begin{rem}\label{rem:R_symmetries}
The map $\fR$ preserves the symmetries on its natural domain. 
More precisely, let $\tau\in\{0,+\}$ and suppose that $V^2\in\sN^2_\tau$ and $V^4\in\sN^4_\tau$ satisfy the moment assumptions of Def.~\ref{dfn:R_map}.  Then
\begin{equation}
\fR V^2\in\sN^2_\tau,
\qquad
\fR V^4\in\sN^4_\tau.
\end{equation}
Indeed, the operations entering Def.~\ref{dfn:R_map}, namely dilation in the relative coordinates, multiplication by a monomial in relative coordinates, and antisymmetrization,  commute with the geometric, charge-conjugation and reality symmetries. They act trivially on  the flavor indices, while antisymmetrization commutes with the diagonal $U(N)$ action. Hence they also preserve the internal symmetries.
If $V=\fS\tilde V$ with $\tilde V$ satisfying the conditions of Def.~\ref{dfn:flow_of_charge}, the same conditions are preserved when the  operations entering $\fR$ are applied to $\tilde V$, since they act only on  the spatial variables and jet indices and do not mix the charge or flavor labels. Antisymmetrizing the resulting kernel therefore shows that $\fR V$ is again compatible with the flow of charge and flavor.
\end{rem}

\begin{lem}\label{lem:R_bounds}
There exists $C\in(0,\infty)$ such that
\begin{equation}
\begin{aligned}
 \|w^{2}_{t;\nu}\,(\fR V^{2})^{a,\sigma}\|_{\sM^2} 
 &\leq
 C\,
 (1-s/t)^{\zeta_\star}
 \sup_{b\in\bA^2} s^{|a|-|b|}\,\|w^2_{s;\nu}\,V^{2,b,\sigma}\|_{\sM^2},
 \\
 \|w^{4}_{t;\nu}\,(\fR V^{4})^{a,\sigma}\|_{\sM^4} 
 &\leq C\,
 (1-s/t)^{\zeta_\star}
 \sup_{b\in\bA^4} s^{|a|-|b|}\,\|w^4_{s;\nu}\,V^{4,b,\sigma}\|_{\sM^4}.
\end{aligned}
\end{equation}
for all $\nu\in[0,1/2]$, $t\in(0,1]$, $s\in(0,t)$ and all $V^2,V^4$ in the domain of $\fR$ from Def.~\ref{dfn:R_map} such that the RHS of the above bounds are finite.
\end{lem}
\begin{rem}
Recall that $\zeta_\star=-3/4$ was introduced in Def.~\ref{dfn:weights}. See also Lemma~\ref{lem:weights}~(f).
\end{rem}

\begin{proof}
Since by assumption $V^2,V^4$ are antisymmetric it suffices to prove that for all $a\in\bA^2$, $|a|=2$, $\sigma\in\bFF^2$ it holds
\begin{equation}
 \|w^{2}_{t;\nu}\,\hat W^{2,a,\sigma}\|_{\sM^2} 
 \leq
 C\,
 (1-s/t)^{\zeta_\star}\, s^{|a|}
 \sup_{\sigma\in\bFF^2}\|w^2_{s;\nu}\,V^{2,0,\sigma}\|_{\sM^2},
\end{equation}
for all $b,c\in\bA^2$, $|b|=|c|=1$, $\sigma\in\bFF^2$ it holds
\begin{equation}
 \|w^{2}_{t;\nu}\,\check W^{2,b,c,\sigma}\|_{\sM^2} 
 \leq
 C\,
 (1-s/t)^{\zeta_\star}\, s^{|b|}
 \sup_{\sigma\in\bFF^2}\|w^2_{s;\nu}\,V^{2,c,\sigma}\|_{\sM^2}
\end{equation}
and for all $a\in\bA^4$, $|a|=1$, $\sigma\in\bFF^4$ it holds
\begin{equation}
 \|w^{4}_{t;\nu}\,W^{4,a,\sigma}\|_{\sM^4} 
 \leq C\,
 (1-s/t)^{\zeta_\star}\,s^{|a|}
 \sup_{\sigma\in\bFF^4}\|w^4_{s;\nu}\,V^{4,0,\sigma}\|_{\sM^4}.
\end{equation}
For every homogeneous monomial $\cX$ in relative coordinates of degree $r$,
\begin{equation}
	\cX\,\mathsf D_u\mu
	=
	u^r\mathsf D_u(\cX\mu).
\end{equation}
Thus the factors $u^{-2}$ and $u^{-1}$ in Def.~\ref{dfn:R_map} are cancelled by $r=2$ and $r=1$, respectively, so that the setwise integrals define finite-variation measures. The stated estimates then follow from Lemma~\ref{lem:weights}~(f).
\end{proof}

\begin{lem}\label{lem:R_L_identity}
Let $V^2\in\sN^2_+\subset\sN^2$ be such that \mbox{$\cX V^2\in\sNN^2$} for all translationally invariant polynomials $\cX\in C^\infty(\bR^4)$ of degree two and \mbox{$V^4\in\sN^4_+\subset\sN^4$} be such that \mbox{$\cX V^4\in\sNN^4$} for all translationally invariant polynomials $\cX\in C^\infty(\bR^8)$ of degree one. Then, for every \(\psi\in\sS(\mathbb R^2)^{\bFF}\otimes_{\mathrm{alg}}\sG^-\), one has
\begin{equation}
\begin{gathered}
 \langle(\fR+U^{2}\fL+U_\partial^2\fL_\partial)V^2,(\fJ\psi)^{\otimes2}\rangle
 =
 \langle V^2,(\fJ\psi)^{\otimes2}\rangle,
 \\
 \langle(\fR+U^{4}\fL)V^4,(\fJ\psi)^{\otimes4}\rangle
 =
 \langle V^4,(\fJ\psi)^{\otimes4}\rangle.
\end{gathered} 
\end{equation}
\end{lem}
\begin{proof}
The proof is an application of the Taylor theorem. The moment assumptions in the statement guarantee that all the expressions below are well-defined.

We first note that, for $m\in\{2,4\}$ and every family $W^m=(W^{m,a,\sigma})_{a\in\bA^m,\sigma\in\bFF^m}$,
\begin{multline}\label{eq:S_pairing_J}
 \sum_{a\in\bA^m}\sum_{\sigma\in\bFF^m}
 \langle (\fS W^m)^{a,\sigma},
 (\fJ\psi)^{a_1,\sigma_1}\otimes\ldots\otimes
 (\fJ\psi)^{a_m,\sigma_m}\rangle
 \\
 =
 \sum_{a\in\bA^m}\sum_{\sigma\in\bFF^m}
 \langle W^{m,a,\sigma},
 (\fJ\psi)^{a_1,\sigma_1}\otimes\ldots\otimes
 (\fJ\psi)^{a_m,\sigma_m}\rangle.
\end{multline}
Indeed, this follows from Def.~\ref{dfn:permutations} and the fact that all components of $\fJ\psi$ are Grassmann odd. By the Taylor theorem the following identity
\begin{multline}
 \psi^{\sigma_1}(x_1)\,\psi^{\sigma_2}(x_2)
 =
 \psi^{\sigma_1}(x_1)\,\psi^{\sigma_2}(x_1)
 +
 \sum_{\substack{a\in\bA^2\\|a|=1}}
 (x_1-x_1)^{a_1} (x_2-x_1)^{a_2}
 (\fJ\psi)^{a_1,\sigma_1}(x_1)
 (\fJ\psi)^{a_2,\sigma_2}(x_1)
 \\
 +\sum_{\substack{a\in\bA^2\\|a|=2}}
 \frac{2!}{a!}
 (x_1-x_1)^{a_1}(x_2-x_1)^{a_2}
 \int_0^1 (1-u)\,
 (\fJ\psi)^{a_1,\sigma_1}(x_1)\,
 (\fJ\psi)^{a_2,\sigma_2}(x_1+u(x_2-x_1))
 \,\rd u
\end{multline}
is true for $\sigma\in\bFF^2$. As a result, by Lemma~\ref{lem:fL_symmetries} and Def.~\ref{dfn:R_map} of $\hat W^{2,a,\sigma}$ we obtain
\begin{multline}\label{eq:fL_fR_lemma_2_hat}
 \sum_{\sigma\in\bFF^2}
 \langle V^{2,0,\sigma},
 \psi^{\sigma_1}\otimes\psi^{\sigma_2}\rangle
 \\
 =
 \sum_{a\in\bA^2}\sum_{\sigma\in\bFF^2}
 \langle
 U^{2,a,\sigma}\fL V^2
 +U_\partial^{2,a,\sigma}\hat\fL_\partial V^2
 +\hat W^{2,a,\sigma},
 (\fJ\psi)^{a_1,\sigma_1}\otimes
 (\fJ\psi)^{a_2,\sigma_2}
 \rangle.
\end{multline}
Next, we note that
\begin{multline}
 (\fJ\psi)^{c_1,\sigma_1}(x_1)\,
 (\fJ\psi)^{c_2,\sigma_2}(x_2)
 =
 (\fJ\psi)^{c_1,\sigma_1}(x_1)
 (\fJ\psi)^{c_2,\sigma_2}(x_1)
 \\
 +
 \sum_{\substack{b\in\bA^2\\|b|=1}}
 (x_1-x_1)^{b_1}(x_2-x_1)^{b_2}
 \int_0^1
 (\fJ\psi)^{b_1+c_1,\sigma_1}(x_1)\,
 (\fJ\psi)^{b_2+c_2,\sigma_2}
 (x_1+u(x_2-x_1))
 \,\rd u
\end{multline}
for $c\in\bA^2$, $|c|=1$ and $\sigma\in\bFF^2$. Consequently, by
Lemma~\ref{lem:fL_symmetries} and Def.~\ref{dfn:R_map} of
$\check W^{2,b,c,\sigma}$ we have
\begin{multline}\label{eq:fL_fR_lemma_2_check}
 \sum_{\substack{c\in\bA^2\\|c|=1}}
 \sum_{\sigma\in\bFF^2}
 \langle
 V^{2,c,\sigma},
 (\fJ\psi)^{c_1,\sigma_1}\otimes
 (\fJ\psi)^{c_2,\sigma_2}
 \rangle
 \\
 =
 \sum_{a\in\bA^2}\sum_{\sigma\in\bFF^2}
 \langle
 U_\partial^{2,a,\sigma}\check\fL_\partial V^2,
 (\fJ\psi)^{a_1,\sigma_1}\otimes
 (\fJ\psi)^{a_2,\sigma_2}
 \rangle
 \\
 +
 \sum_{b,c\in\bA^2}\sum_{\sigma\in\bFF^2}
 \langle
 \check W^{2,b,c,\sigma},
 (\fJ\psi)^{b_1+c_1,\sigma_1}\otimes
 (\fJ\psi)^{b_2+c_2,\sigma_2}
 \rangle.
\end{multline}
Adding Eqs.~\eqref{eq:fL_fR_lemma_2_hat} and
\eqref{eq:fL_fR_lemma_2_check}, using
\begin{equation}
 \fL_\partial
 =
 \hat\fL_\partial+\check\fL_\partial
\end{equation}
and
\begin{equation}
 W^{2,a,\sigma}
 =
 \hat W^{2,a,\sigma}
 +
 \sum_{\substack{b,c\in\bA^2\\b+c=a}}
 \check W^{2,b,c,\sigma},
\end{equation}
and then adding the components of $V^2$ with $|a|\geq2$, we obtain
\begin{multline}
 \langle V^2,(\fJ\psi)^{\otimes2}\rangle
 =
 \langle
 (U^{2}\fL+U_\partial^2\fL_\partial)V^2,
 (\fJ\psi)^{\otimes2}
 \rangle
 \\
 +
 \sum_{\substack{a\in\bA^2\\|a|=2}}
 \sum_{\sigma\in\bFF^2}
 \langle
 W^{2,a,\sigma},
 (\fJ\psi)^{a_1,\sigma_1}\otimes
 (\fJ\psi)^{a_2,\sigma_2}
 \rangle
 \\
 +
 \sum_{\substack{a\in\bA^2\\|a|\geq2}}
 \sum_{\sigma\in\bFF^2}
 \langle
 V^{2,a,\sigma},
 (\fJ\psi)^{a_1,\sigma_1}\otimes
 (\fJ\psi)^{a_2,\sigma_2}
 \rangle.
\end{multline}
By~\eqref{eq:S_pairing_J} and Def.~\ref{dfn:R_map}, the last two
lines are equal to
\begin{equation}
 \langle \fR V^2,(\fJ\psi)^{\otimes2}\rangle.
\end{equation}
Hence
\begin{equation}
 \langle
 (\fR+U^{2}\fL+U_\partial^2\fL_\partial)V^2,
 (\fJ\psi)^{\otimes2}
 \rangle
 =
 \langle V^2,(\fJ\psi)^{\otimes2}\rangle.
\end{equation}

Finally, using the identity
\begin{equation}
 \prod_{j=1}^4 \psi^{\sigma_j}(x_j)
 =
 \prod_{j=1}^4 \psi^{\sigma_j}(x_1)
 +
 \sum_{\substack{a\in\bA^4\\|a|=1}}
 \prod_{j=1}^4 (x_j-x_1)^{a_j}
 \int_0^1
 \prod_{j=1}^4
 (\fJ\psi)^{a_j,\sigma_j}
 (x_1+u(x_j-x_1))
 \,\rd u,
\end{equation}
Lemma~\ref{lem:fL_symmetries} and Def.~\ref{dfn:R_map} of
$W^{4,a,\sigma}$ we get
\begin{multline}\label{eq:fL_fR_lemma_4}
 \sum_{\sigma\in\bFF^4}
 \langle
 V^{4,0,\sigma},
 \psi^{\sigma_1}\otimes\ldots\otimes\psi^{\sigma_4}
 \rangle
 \\
 =
 \sum_{a\in\bA^4}\sum_{\sigma\in\bFF^4}
 \langle
 U^{4,a,\sigma}\fL V^4+W^{4,a,\sigma},
 (\fJ\psi)^{a_1,\sigma_1}\otimes\ldots\otimes
 (\fJ\psi)^{a_4,\sigma_4}
 \rangle.
\end{multline}
Adding the components of $V^4$ with $|a|\geq1$ to~\eqref{eq:fL_fR_lemma_4}, we obtain
\begin{multline}
 \langle V^4,(\fJ\psi)^{\otimes4}\rangle
 =
 \langle
 U^4\fL V^4,
 (\fJ\psi)^{\otimes4}
 \rangle
 \\
 +
 \sum_{\substack{a\in\bA^4\\|a|=1}}
 \sum_{\sigma\in\bFF^4}
 \langle
 W^{4,a,\sigma},
 (\fJ\psi)^{a_1,\sigma_1}\otimes\ldots\otimes
 (\fJ\psi)^{a_4,\sigma_4}
 \rangle
 \\
 +
 \sum_{\substack{a\in\bA^4\\|a|\geq1}}
 \sum_{\sigma\in\bFF^4}
 \langle
 V^{4,a,\sigma},
 (\fJ\psi)^{a_1,\sigma_1}\otimes\ldots\otimes
 (\fJ\psi)^{a_4,\sigma_4}
 \rangle.
\end{multline}
By~\eqref{eq:S_pairing_J} and Def.~\ref{dfn:R_map}, the last two
lines are equal to
\[
 \langle\fR V^4,(\fJ\psi)^{\otimes4}\rangle.
\]
Therefore
\begin{equation}
 \langle
 (\fR+U^{4}\fL)V^4,
 (\fJ\psi)^{\otimes4}
 \rangle
 =
 \langle V^4,(\fJ\psi)^{\otimes4}\rangle.
\end{equation}
This finishes the proof.
\end{proof}

\section{Useful maps}\label{sec:maps}

In this section we study properties of the operators that appear on the RHS of the equations~\eqref{eq:flow_g_intro_ren}, \eqref{eq:flow_r_intro_ren},
\eqref{eq:flow_2_intro_z}, \eqref{eq:flow_2_intro_W1}, 
\eqref{eq:flow_2_intro_W2} introduced in Sec.~\ref{sec:strategy}. Recall that our goal is to solve this system of equations by rewriting it as a fixed point equation of a certain map $\fX_{\tau,\varepsilon;\Cdot}$, which is defined in Sec.~\ref{sec:fixed_point}. The estimates we establish below play a crucial role in the proof that the map $\fX_{\tau,\varepsilon;\Cdot}$ is well-defined and is a contraction.

We shall repeatedly use the following elementary estimates for the running coupling $\lambda_t$ introduced in Def.~\ref{dfn:lambda}.
\begin{lem}\label{lem:bounds_relevant_irrelevant}
For all $\rho,\varrho\in\bR$ and $\eta\in(0,1)$ there exists $C\in(0,\infty)$ such that the following statements are true for all $\lambda\in(0,1]$ and $t\in(0,1]$.
\begin{itemize}
 \item[(A)] If $\rho\geq 0$ and $\varrho> 0$, then $\int_0^t \lambda_s^\rho\, s^{\varrho-1}\,\rd s 
 \leq
 \lambda_t^\rho\, t^\varrho/\varrho$.
 \item[(B)] If $\rho\geq 0$ and $\varrho> 0$, then $\int_0^t \lambda_s^\rho\, s^{\varrho-1}\,(1-s/t)^{-\eta}\,\rd s
 \leq
 C\,\lambda_t^\rho\,t^\varrho$.
 \item[(C)] If $\rho\geq 0$ and $\varrho< -2\beta_2\,\rho\,\lambda$, then $\int_t^1 \lambda_s^\rho\, s^{\varrho-1}\,\rd s 
 \leq C\,\lambda_t^\rho\, t^\varrho$.
 \item[(D)] If $\rho> 0$, then $\int_0^t \lambda_s^{\rho+1}\, s^{-1}\,\rd s
 \leq
 C\,\lambda_t^\rho$.
 \item[(E)] If $\rho<0$, then $\int_t^1 \lambda_s^{\rho+1}\, s^{-1}\,\rd s 
 \leq C\,\lambda_t^\rho$.
\end{itemize}
\end{lem}
\begin{proof}
For $\rho\in[0,\infty)$ and $\varrho\in(0,\infty)$ we have
\begin{equation}
 \int_0^t \lambda_s^\rho\, s^{\varrho-1}\,\rd s 
 \leq
 \lambda_t^\rho\int_0^t s^{\varrho-1}\,\rd s
 =
 \lambda_t^\rho\, t^\varrho/\varrho
\end{equation}
as well as
\begin{equation}
 \int_0^t \lambda_s^\rho\, s^{\varrho-1}\,(1-s/t)^{-\eta}\,\rd s
 \leq
 \lambda_t^\rho \int_0^t s^{\varrho-1}\,(1-s/t)^{-\eta}\,\rd s
 \leq 
 \lambda_t^\rho\,t^\varrho \int_0^1 s^{\varrho-1}\,(1-s)^{-\eta}\,\rd s.
\end{equation}
This proves the bounds~(A) and~(B). To prove the bound~(C) we use the fact that for $\rho\in[0,\infty)$ and $\varrho\in(-\infty,-2\beta_2\,\rho\,\lambda)$ the following estimates
\begin{multline}
 \int_t^1 \lambda_s^\rho\, s^{\varrho-1}\,\rd s 
 \leq
 (-\varrho-\beta_2\,\rho\, \lambda)^{-1}\int_t^1  (-\varrho-\beta_2\,\rho\, \lambda_s)\,\lambda_s^\rho\,s^{\varrho-1} 
 \,\rd s 
 \\
 =
 (-\varrho-\beta_2\,\rho\, \lambda)^{-1}\int_t^1  \partial_s (-\lambda_s^\rho\,s^\varrho)
 \,\rd s 
 \leq \lambda_t^\rho\, t^\varrho/(-\varrho-\beta_2\,\rho\, \lambda)
 \leq
 \frac{2}{-\varrho}\lambda_t^\rho t^\varrho.
\end{multline}
hold true. The bound~(D) is a consequence of the equalities
\begin{equation}
 \int_0^t \lambda_s^{\rho+1}\, s^{-1}\,\rd s
 =
 \beta_2^{-1}\int_0^t \lambda_s^{\rho-1}\, \rd \lambda_s
 = \lambda_t^\rho/(\beta_2\rho)
\end{equation}
valid for $\rho\in(0,\infty)$ and the bound~(E) is a consequence of the estimates
\begin{equation}
 \int_t^1 \lambda_s^{\rho+1}\, s^{-1}\,\rd s 
 =
 \beta_2^{-1}\int_t^1 \lambda_s^{\rho-1}\, \rd \lambda_s
 \leq \lambda_t^\rho/(-\beta_2\rho)
\end{equation} 
valid for $\rho\in(-\infty,0)$. This finishes the proof.
\end{proof}

\begin{dfn}\label{dfn:covariance_H}
For $\tau,\varepsilon\in[0,1]$ and $t\in(0,1]$ we define 
$$
\dot H_{\tau,\varepsilon;t}\equiv(\dot H^{a,\sigma}_{\tau,\varepsilon;t})_{a\in\bA^2,\sigma\in\bFF^2}\in C^\infty(\bR^2\times\bR^2)^{\bA^2\times\bFF^2}
$$ by the formula
\begin{equation}
 \dot H^{a,\sigma}_{\tau,\varepsilon;t}(x,y):=\partial^{a_1}_x\partial^{a_2}_y\dot G^\sigma_{\tau,\varepsilon;t}(x-y)
\end{equation}
for all $a=(a_1,a_2)\in\bA^2$, $\sigma\in\bFF^2$ and $x,y\in\bR^2$. We omit $\tau$ if $\tau=0$ and we omit~$\varepsilon$ if $\varepsilon=0$.
\end{dfn}

\begin{rem}\label{rem:H_symmetries}
The function
\begin{equation}
 \dot H_{\tau,\varepsilon;t}
 \in
 C^\infty(\bR^2\times\bR^2)^{\bA^2\times\bFF^2}
\end{equation}
is antisymmetric and transforms covariantly under the symmetries of the torus, with the action on the derivative indices induced by the chain rule. If $\tau=0$, the same holds for the symmetries of the plane.
The invariance under the internal symmetries follows from the fact that the propagators are diagonal and independent of the flavor index.
Moreover, $\dot H_{\tau,\varepsilon;t}$ is invariant under the charge conjugation symmetry of Def.~\ref{dfn:charge} and is real with respect to the reality conjugation of Def.~\ref{dfn:reality_conjugation}, that is,
\begin{equation}
 \dot H_{\tau,\varepsilon;t}^\sharp
 =
 \dot H_{\tau,\varepsilon;t}
 \, .
\end{equation}
The analogous properties hold for $G_{\tau,\varepsilon;t}$ and $\dot G_{\tau,\varepsilon;t}$. The charge conjugation and reality properties follow from Remark~\ref{rem:gamma} and Def.~\ref{dfn:covariance}: the cutoff multipliers are real and even functions of the momentum, and these properties are preserved by periodization and differentiation.
\end{rem}

\begin{lem}\label{lem:H_dot_estimates}
There exists $\lambda_\star\in(0,1]$ such that for all $\lambda\in(0,\lambda_\star]$, $\varepsilon\in[0,1]$, $t\in(0,1]$, and $a\in\bA^2$, $\sigma\in\bFF^2$ it holds
\begin{equation}
 \|w_t^2 \dot H^{a,\sigma}_{\varepsilon;t}\|_{\sM^2} \leq \lambda^{-\kappa}
 \,t^{-|a|},
 \qquad
 \|w_t^2 (\dot H^{a,\sigma}_{t}-\dot H^{a,\sigma}_{\varepsilon;t})\|_{\sM^2} \leq \lambda^{-\kappa}\,\lambda_\varepsilon^\kappa\,\lambda_t^{-\kappa}\,t^{-|a|},
\end{equation}
where the weight $w^2_t$ was introduced in Def.~\ref{dfn:weights}. Moreover, the map $(0,1]\ni t\mapsto \dot H_{\varepsilon;t}\in\sN^2$ is continuous for all $\varepsilon\in[0,1]$.
\end{lem}
\begin{proof}
Observe that $w^{2}_{t}(x_1,x_2) = 
(1+|x_1-x_2|)^{1/2}\exp(t^{-\zeta}|x_1-x_2|^\zeta)$ and $\zeta=4/5<\zeta_\flat$. As a result, by Lemma~\ref{lem:dot_G_infty} and Remark~\ref{rem:G}~(E) there exists a universal constant $C\in(0,\infty)$ such that $\|w_t^2 \dot H^{a,\sigma}_{\varepsilon;t}\|_{\sM^2} \leq C\,t^{-|a|}$. This proves the first of the estimates stated in the lemma. The second of the estimates follows from the above bound and Remark~\ref{rem:G}~(C).
\end{proof}

\begin{dfn}\label{dfn:B_map}
	Let $m\in\bN_+$, $\varepsilon\in[0,1]$ and $s\in(0,1]$. For
	$V,W\in\sN(\sF)$ we define
	$\fB^m_{\varepsilon;s}(V,W)\in\sN^m(\sF)$ by
	\begin{multline}\label{eq:B_map_components}
		\langle
		\fB^m_{\varepsilon;s}(V,W),
		\varphi^{\otimes m} \rangle
		:=
		\frac{1}{2}
		\sum_{k=0}^m
		(-1)^{m-k+k(m-k+1)}
		(k+1)(m-k+1) 
		\\
		{}\times
		\langle  V^{k+1}\otimes W^{m-k+1},
		\varphi^{\otimes k} \otimes\dot H_{\varepsilon;s}
		\otimes\varphi^{\otimes(m-k)}
		\rangle
	\end{multline}
	for all
	$\varphi\in\sS(\bR^2)^{\bA\times\bFF}
	\otimes_{\mathrm{alg}}\sG^-$.
	We set
	\begin{equation}
		\fB_{\varepsilon;s}(V,W)
		:=
		\big(
		\fB^m_{\varepsilon;s}(V,W)
		\big)_{m\in\bN_+},
		\qquad
		\fB_{\varepsilon;s}(V)
		:=\fB_{\varepsilon;s}(V,V).
	\end{equation}
	We also set
	$\fB^{m,a}_{\varepsilon;s} :=\Pi^{m,a}\fB_{\varepsilon;s}$ and
	$\fB^{m,a,\sigma}_{\varepsilon;s} :=
	\Pi^{m,a,\sigma}\fB_{\varepsilon;s}$. We omit $\varepsilon$ if $\varepsilon=0$.
\end{dfn}
\begin{rem}\label{rem:fB}
	Let $V,W\in\sN^{\mathrm{fin}}(\sF)$ and $\varphi\in\sS(\bR^2)^{\bA\times\bFF}
	\otimes_{\mathrm{alg}}\sG^-$. 
	Using the graded tensor product and the antisymmetry of the kernels, one obtains
	\begin{equation}\label{eq:fB_functional}
		\fB_{\varepsilon;s}(V,W)[\varphi]
		=
		\frac{1}{2}\langle
		\rD_\varphi V[\varphi]\otimes
		\rD_\varphi  W[\varphi],
		\dot H_{\varepsilon;s}\rangle 
		-
		\frac{1}{2}
		\langle		\rD_\varphi  V[0]\otimes 
		\rD_\varphi W[0],
		\dot H_{\varepsilon;s}\rangle.
	\end{equation}
	Indeed, in the contribution containing $V^{k+1}$ and $W^{m-k+1}$, antisymmetry produces the factor $(-1)^{m-k}$, whereas the graded tensor product produces the additional factor
	$(-1)^{k(m-k+1)}$, since the coefficient $W^{m-k+1}$ has parity $m-k+1$. 
	This gives the  sign appearing in Def.~\ref{dfn:B_map}.
	
	Since $V$ and $W$ have total even parity, their first functional derivatives have total odd parity. Using also the antisymmetry  of~$\dot H_{\varepsilon;s}$ and~\eqref{eq:fB_functional}, we obtain
	\begin{equation}
		\fB_{\varepsilon;s}(V,W)
		=
		\fB_{\varepsilon;s}(W,V).
	\end{equation}
	Consequently,
	\begin{equation}\label{eq:fB_polarization}
		\fB_{\varepsilon;s}(V,W)
		=
		\frac{1}{4}
		\left(
		\fB_{\varepsilon;s}(V+W)
		-
		\fB_{\varepsilon;s}(V-W)
		\right).
	\end{equation}
	
	Let $\tau\in[0,1]$, $t\in(0,1]$ and set
	$\eta:=\fJ\varPsi_{\tau,\varepsilon;t\vee s,s}$.
	By Def.~\ref{dfn:A_map},
	\begin{equation}
		\rD_\varphi
		\big(
		\fA_{\tau,\varepsilon;t,s}V
		\big)[\varphi]
		=
		\rD_\varphi V[\varphi+\eta],
		\qquad
		\rD_\varphi
		\big(
		\fA_{\tau,\varepsilon;t,s}V
		\big)[0]
		=
		\rD_\varphi V[\eta].
	\end{equation}
	Therefore, \eqref{eq:fB_functional} gives
	\begin{multline}
		\fB_{\varepsilon;s}
		\big(
		\fA_{\tau,\varepsilon;t,s}V, 
		\fA_{\tau,\varepsilon;t,s}W\big)[\varphi]
		\\
		=
		\frac{1}{2}
		\langle 
		\rD_\varphi V[\varphi+\eta]\otimes
		\rD_\varphi W[\varphi+\eta],
		\dot H_{\varepsilon;s}
		\rangle
		-
		\frac{1}{2}\langle
		\rD_\varphi V[\eta]\otimes
		\rD_\varphi W[\eta], 
		\dot H_{\varepsilon;s}
		\rangle
		\\
		=
		\fB_{\varepsilon;s}(V,W)[\varphi+\eta]
		-\fB_{\varepsilon;s}(V,W)[\eta].
	\end{multline}
	Hence
	\begin{equation}\label{eq:fB_fA_identity}
		\fA_{\tau,\varepsilon;t,s}
		\fB_{\varepsilon;s}(V,W)
		=
		\fB_{\varepsilon;s}
		\big(
		\fA_{\tau,\varepsilon;t,s}V,
		\fA_{\tau,\varepsilon;t,s}W
		\big).
	\end{equation}
\end{rem}

\begin{lem}\label{lem:B_map_bound}
Let \mbox{$\alpha\in[2,\infty)$}, $\beta\in(2,\infty)$, $\gamma_1,\gamma_2\in\bR$. There exists $\lambda_\star\in(0,1]$ such that for all $\lambda\in(0,\lambda_\star]$ and $\tau,\varepsilon\in[0,1]$ the following holds.
For $V_\Cdot\in\sW^{\alpha,\beta;\gamma_1}_{\tau,\varepsilon}$ and $W_\Cdot\in\sW^{\alpha,\beta;\gamma_2}_{\tau,\varepsilon}$, the map $s\mapsto s\,\fB_{\varepsilon;s}(V_s,W_s)$ belongs to $\sW^{\alpha,\beta-1;\gamma_1+\gamma_2}_{\tau,\varepsilon}$ and the bounds~(A),~(B) below hold. Moreover, for $V_\Cdot\in\sV^{\alpha,\beta;\gamma_1}$ and $W_\Cdot\in\sV^{\alpha,\beta;\gamma_2}$ the bound~(C) holds:
\begin{itemize}
\item[(A)]
$\|s\mapsto s\,\fB_{\varepsilon;s}(V_s,W_s)\|_{\sW^{\alpha,\beta-1;\gamma_1+\gamma_2}_{\tau,\varepsilon}}
 \leq
 \lambda^\kappa\,
 \|V_\Cdot\|_{\sW^{\alpha,\beta;\gamma_1}_{\tau,\varepsilon}}\,
 \|W_\Cdot\|_{\sW^{\alpha,\beta;\gamma_2}_{\tau,\varepsilon}}$,
 
\item[(B)]
$\|s\mapsto s\,\fB_{\varepsilon;s}(V_s,W_s)\|_{\tilde\sW^{\alpha,\beta-1;\gamma_1+\gamma_2}_{\tau,\varepsilon}}
 \leq
 \lambda^\kappa\,
 \|V_\Cdot\|_{\tilde\sW^{\alpha,\beta;\gamma_1}_{\tau,\varepsilon}}\,
 \|W_\Cdot\|_{\sW^{\alpha,\beta;\gamma_2}_{\tau,\varepsilon}}$,
  
\item[(C)]
$\|s\mapsto s\,(\fB_{s}(V_s,W_s)-\fB_{\varepsilon;s}(V_s,W_s))\|_{\tilde\sW^{\alpha/2,\beta-1;\gamma_1+\gamma_2}_{\tau,\varepsilon}}
 \leq
 \lambda_\varepsilon^\kappa\,
 \|V_\Cdot\|_{\sV^{\alpha,\beta;\gamma_1}}\,
 \|W_\Cdot\|_{\sV^{\alpha,\beta;\gamma_2}}$.
\end{itemize}
\end{lem}
\begin{proof}
Let $\nu\in\{0,1/2\}$, $s\in(0,1]$ and $V_s,W_s\in\sN(\sF)$. Fix $m\in\bN_+$, $a=(a_1,\ldots,a_m)\in\bA^m$ and $\sigma=(\sigma_1,\ldots,\sigma_m)\in\bFF^m$. In the contribution to $\fB^m_{\varepsilon;s}(V_s,W_s)$ containing $V^{k+1}_s$ and $W^{m-k+1}_s$, let $c=(c_1,c_2)\in\bA^2$ and $\hat\sigma=(\hat\sigma_1,\hat\sigma_2)\in\bFF^2$ be the indices contracted with $\dot H_{\varepsilon;s}$, and set
\begin{align}
 b&:=(a_1,\ldots,a_k,c_1),
 &\eta&:=(\sigma_1,\ldots,\sigma_k,\hat\sigma_1),
 \\
 d&:=(c_2,a_{k+1},\ldots,a_m),
 &\theta&:=(\hat\sigma_2,\sigma_{k+1},\ldots,\sigma_m).
\end{align}
Then
$
 |b|+|d|=|a|+|c|
$.
After antisymmetrization the output variables may be permuted. By Remark~\ref{rem:sV_tilde2}, these permutations do not affect the component sums below. Hence, by Def.~\ref{dfn:B_map}, Def.~\ref{dfn:sM}, Lemma~\ref{lem:weights}~(a),~(b), Lemma~\ref{lem:H_dot_estimates} and the preceding identity, there exists a universal constant $C\in(0,\infty)$ such that
\begin{multline}\label{eq:B_map_bound_proof}
 s^{-|a|}\|w^m_{s;\nu}\,\fB^{m,a,\sigma}_{\varepsilon;s}(V_s,W_s)\|_{\sM^m}
 \\
 \leq
 C\lambda^{-\kappa}
 \sum_{k=0}^m(k+1)(m-k+1)
 \sum_{c\in\bA^2}\sum_{\hat\sigma\in\bFF^2}
 s^{-|b|}\|w^{k+1}_{s;\nu}V^{k+1,b,\eta}_s\|_{\sM^{k+1}}
 \\
 \times
 s^{-|d|}\|w^{m-k+1}_sW^{m-k+1,d,\theta}_s\|_{\sM^{m-k+1}}
\end{multline}
and
\begin{multline}\label{eq:B_map_bound_proof_diff}
 s^{-|a|}\|w^m_{s;\nu}\,(\fB^{m,a,\sigma}_{s}(V_s,W_s)-\fB^{m,a,\sigma}_{\varepsilon;s}(V_s,W_s))\|_{\sM^m}
 \\
 \leq
 C\lambda^{-\kappa}\lambda_\varepsilon^\kappa\lambda_s^{-\kappa}
 \sum_{k=0}^m(k+1)(m-k+1)
 \sum_{c\in\bA^2}\sum_{\hat\sigma\in\bFF^2}
 s^{-|b|}\|w^{k+1}_{s;\nu}V^{k+1,b,\eta}_s\|_{\sM^{k+1}}
 \\
 \times
 s^{-|d|}\|w^{m-k+1}_sW^{m-k+1,d,\theta}_s\|_{\sM^{m-k+1}}.
\end{multline}
Note that
\begin{equation}
 \rho_{\gamma_1+\gamma_2,\kappa}(m)
 =
 \rho_{\gamma_1,\kappa}(k+1)
 +\rho_{\gamma_2,\kappa}(m-k+1)
 -4\kappa.
\end{equation}
Multiplying~\eqref{eq:B_map_bound_proof} by
$\lambda_s^{-\rho_{\gamma_1+\gamma_2,\kappa}(m)}s^{3-m/2}$,
taking the supremum in $s$ for each fixed output component and only then summing over $a$ and $\sigma$, we obtain
\begin{multline}\label{eq:B_map_bound_degree}
 \|s\mapsto s\,\fB^m_{\varepsilon;s}(V_s,W_s)\|_{\sV^{m;\gamma_1+\gamma_2}}
 \\
 \leq
 C\lambda^{3\kappa}
 \sum_{k=0}^m(k+1)(m-k+1)
 \|V^{k+1}_\Cdot\|_{\sV^{k+1;\gamma_1}}
 \|W^{m-k+1}_\Cdot\|_{\sV^{m-k+1;\gamma_2}}.
\end{multline}
The same argument with $\nu=1/2$ gives
\begin{multline}\label{eq:B_map_bound_degree_tilde}
 \|s\mapsto s\,\fB^m_{\varepsilon;s}(V_s,W_s)\|_{\tilde\sV^{m;\gamma_1+\gamma_2}}
 \\
 \leq
 C\lambda^{3\kappa}
 \sum_{k=0}^m(k+1)(m-k+1)
 \|V^{k+1}_\Cdot\|_{\tilde\sV^{k+1;\gamma_1}}
 \|W^{m-k+1}_\Cdot\|_{\sV^{m-k+1;\gamma_2}}.
\end{multline}
Similarly,~\eqref{eq:B_map_bound_proof_diff} implies
\begin{multline}\label{eq:B_map_bound_degree_diff}
 \|s\mapsto s\,(\fB^m_s(V_s,W_s)-\fB^m_{\varepsilon;s}(V_s,W_s))\|_{\tilde\sV^{m;\gamma_1+\gamma_2}}
 \\
 \leq
 C\lambda_\varepsilon^\kappa\lambda^{2\kappa}
 \sum_{k=0}^m(k+1)(m-k+1)
 \|V^{k+1}_\Cdot\|_{\sV^{k+1;\gamma_1}}
 \|W^{m-k+1}_\Cdot\|_{\sV^{m-k+1;\gamma_2}}.
\end{multline}
Moreover, for every $\beta\in(2,\infty)$ there exists $\lambda_\star\in(0,1]$ such that for all $\alpha\in[1,\infty)$, $m\in\bN_+$ and $\lambda\in(0,\lambda_\star]$ it holds
\begin{equation}
 \sum_{k=0}^m
 \frac{\alpha^{-(k+1)}}{(k+1)^{\beta-1}}
 \frac{\alpha^{-(m-k+1)}}{(m-k+1)^{\beta-1}}
 \leq
 \lambda^{-\kappa}\,\frac{\alpha^{-m}}{m^{\beta-1}}.
\end{equation}
Shrinking $\lambda_\star$ if necessary, the bounds~\eqref{eq:B_map_bound_degree}--\eqref{eq:B_map_bound_degree_diff} and Def.~\ref{dfn:V_space} imply that
\begin{itemize}
\item[(A${}_0$)]
$\|s\mapsto s\,\fB_{\varepsilon;s}(V_s,W_s)\|_{\sV^{\alpha,\beta-1;\gamma_1+\gamma_2}}
 \leq
 \lambda^\kappa\,
 \|V_\Cdot\|_{\sV^{\alpha,\beta;\gamma_1}}\,
 \|W_\Cdot\|_{\sV^{\alpha,\beta;\gamma_2}}$,
 
\item[(B${}_0$)]
$\|s\mapsto s\,\fB_{\varepsilon;s}(V_s,W_s)\|_{\tilde\sV^{\alpha,\beta-1;\gamma_1+\gamma_2}}
 \leq
 \lambda^\kappa\,
 \|V_\Cdot\|_{\tilde\sV^{\alpha,\beta;\gamma_1}}\,
 \|W_\Cdot\|_{\sV^{\alpha,\beta;\gamma_2}}$,
  
\item[(C${}_0$)]
$\|s\mapsto s\,(\fB_{s}(V_s,W_s)-\fB_{\varepsilon;s}(V_s,W_s))\|_{\tilde\sV^{\alpha,\beta-1;\gamma_1+\gamma_2}}
 \leq
 \lambda_\varepsilon^\kappa\,
 \|V_\Cdot\|_{\sV^{\alpha,\beta;\gamma_1}}\,
 \|W_\Cdot\|_{\sV^{\alpha,\beta;\gamma_2}}$.
\end{itemize}
The bounds~(A) and~(B) follow first for $V_\Cdot\in\sV^{\mathrm{fin};\gamma_1}_\tau$ and
$W_\Cdot\in\sV^{\mathrm{fin};\gamma_2}_\tau$ from Def.~\ref{dfn:sW} of the norm  $\sW^{\alpha,\beta;\gamma}_{\tau,\varepsilon}$,  the bounds~(A${}_0$) and~(B${}_0$) and the identity~\eqref{eq:fB_fA_identity}.
We next verify that in this case the map $s\mapsto s\,\fB_{\varepsilon;s}(V_s,W_s)$ belongs to
$\sW^{\alpha,\beta-1;\gamma_1+\gamma_2}_{\tau,\varepsilon}$. It follows from Lemma~\ref{lem:symmetries} that $\fB_{\varepsilon;s}(V_s,W_s)\in\sN^{\mathrm{fin}}_\tau$ for all
$s\in(0,1]$. Recall that the map $(0,1]\ni s\mapsto\dot H_{\varepsilon;s}\in\sN^2$ is continuous by Lemma~\ref{lem:H_dot_estimates}. Hence, by a bound similar to
\eqref{eq:B_map_bound_proof}, the map $(0,1]\ni s\mapsto\fB^m_{\varepsilon;s}(V_s,W_s)\in\sN^m$ is
continuous. Together with the bounds just proved, this shows that
$s\mapsto s\,\fB_{\varepsilon;s}(V_s,W_s)$ belongs to $\sV^{\mathrm{fin};\gamma_1+\gamma_2}_\tau
\subset\sW^{\alpha,\beta-1;\gamma_1+\gamma_2}_{\tau,\varepsilon}$.
The general case follows by density. Indeed, applying the bounds just
proved to differences of finite-degree approximations and using
bilinearity shows that the corresponding outputs are Cauchy in the
target $\sW$-norm, while the bounds~(A${}_0$),~(B${}_0$) identify the
limit with $s\mapsto s\,\fB_{\varepsilon;s}(V_s,W_s)$.
The bound~(C) is a consequence of the bound~(C${}_0$) and
Lemma~\ref{lem:A_map_bound}~(C).
\end{proof}

\begin{lem}\label{lem:I_bound}
Let $\alpha\in[2,\infty)$, $\beta\in[1,\infty)$, $\gamma\in[0,\infty)$. For $\tau,\varepsilon\in[0,1]$, $u\in(0,1]$ and $V_\Cdot\in\sV^{\alpha,\beta;\gamma}$ we define $\fI_{\tau,\varepsilon;u} V_\Cdot\equiv(\fI^m_{\tau,\varepsilon;u} V_\Cdot)_{m\in\bN_+}\in\sN$ by the equalities $\fI^m_{\tau,\varepsilon;u} V_\Cdot=0$ for $m\in\bN_+\setminus2\bN_+$,
\begin{equation}
 \fI^m_{\tau,\varepsilon;u} V_\Cdot
 :=
 \int_0^u\fE\fA^m_{\tau,\varepsilon;u,s}V_s\frac{\rd s}{s} \in \sN^m
\end{equation}
for $m\in\{6,8,\ldots\}$ and
\begin{equation}
 \fI^m_{\tau,\varepsilon;u} V_\Cdot:=\int_0^u\fR\fE\fA^m_{\tau,\varepsilon;1,s}V_s\frac{\rd s}{s} \in \sN^m
\end{equation}
for $m\in\{2,4\}$ provided the integrands appearing above are absolutely convergent. We set $\fI^{m,a,\sigma}_{\tau,\varepsilon;u}V=\Pi^{m,a,\sigma}\fI_{\tau,\varepsilon;u}V$. Moreover, we omit $\tau$ and $\varepsilon$ if $\tau=0$ and $\varepsilon=0$. 
There exists a constant $C\in(0,\infty)$ such that for all
$\tau,\varepsilon\in[0,1]$ and $V_\Cdot\in\sW^{\alpha,\beta;\gamma}_{\tau,\varepsilon}$ the map
$u\mapsto\fI_{\tau,\varepsilon;u}V_\Cdot$ belongs to $\sW^{\alpha,\beta+1;\gamma}_{\tau,\varepsilon}$ and
\begin{itemize}
 \item[(A)] 
 $\|u\mapsto \fI_{\tau,\varepsilon;u} V_\Cdot\|_{\sW^{\alpha,\beta+1;\gamma}_{\tau,\varepsilon}}
 \leq
 C\,\|V_\Cdot\|_{\sW^{\alpha,\beta;\gamma}_{\tau,\varepsilon}}$,
 \item[(B)] 
 $\|u\mapsto \fI_{\tau,\varepsilon;u} V_\Cdot\|_{\tilde\sW^{\alpha,\beta+1;\gamma}_{\tau,\varepsilon}}
 \leq
 C\,\|V_\Cdot\|_{\tilde\sW^{\alpha,\beta;\gamma}_{\tau,\varepsilon}}$,
\end{itemize}
If, in addition, $\alpha\in[4,\infty)$, then for all $\tau,\varepsilon\in[0,1]$ and $V_\Cdot\in\sV^{\alpha,\beta;\gamma}$,
\begin{itemize}
 \item[(C)] 
 $\|u\mapsto (\fI_u-\fI_{\tau,\varepsilon;u}) V_\Cdot\|_{\tilde\sW^{\alpha/4,\beta;\gamma}_{\tau,\varepsilon}}
 \leq
 C\,\lambda_{\tau\vee\varepsilon}^\kappa\,\|V_\Cdot\|_{\sV^{\alpha,\beta;\gamma}}$.
\end{itemize}
\end{lem}
\begin{rem}
 The restriction \(\gamma\ge0\) is used below when applying Lemma~\ref{lem:bounds_relevant_irrelevant} (A) and (B), since it ensures \(\rho_{\gamma,\kappa}(m)\ge0\).
\end{rem}

\begin{proof}
We shall prove the following bounds
\begin{itemize}
 \item[(A${}_1$)] 
 $\|u\mapsto \Pi_k\fI_{\tau,\varepsilon;u} V_\Cdot\|_{\sV^{\alpha,\beta;\gamma}}
 \leq
 C\,\|V_\Cdot\|_{\sW^{\alpha,\beta;\gamma}_{\tau,\varepsilon}}$,
 \item[(B${}_1$)] 
 $\|u\mapsto \Pi_k\fI_{\tau,\varepsilon;u} V_\Cdot\|_{\tilde\sV^{\alpha,\beta;\gamma}}
 \leq
 C\,\|V_\Cdot\|_{\tilde\sW^{\alpha,\beta;\gamma}_{\tau,\varepsilon}}$,
 \item[(C${}_1$)] 
 $\|u\mapsto \Pi_k(\fI_u-\fI_{\tau,\varepsilon;u}) V_\Cdot\|_{\tilde\sV^{\alpha/2,\beta-1;\gamma}}
 \leq
 C\,\lambda_{\tau\vee\varepsilon}^\kappa\,\|V_\Cdot\|_{\sV^{\alpha,\beta;\gamma}}$.
\end{itemize}
with $k\in\{2,4\}$ and
\begin{itemize}
 \item[(A${}_2$)] 
 $\|u\mapsto \Pi_{>4}\fI_{\tau,\varepsilon;u} V_\Cdot\|_{\sW^{\alpha,\beta+1;\gamma}_{\tau,\varepsilon}}
 \leq
 C\,\|V_\Cdot\|_{\sW^{\alpha,\beta;\gamma}_{\tau,\varepsilon}}$,
 \item[(B${}_2$)] 
 $\|u\mapsto \Pi_{>4}\fI_{\tau,\varepsilon;u} V_\Cdot\|_{\tilde\sW^{\alpha,\beta+1;\gamma}_{\tau,\varepsilon}}
 \leq
 C\,\|V_\Cdot\|_{\tilde\sW^{\alpha,\beta;\gamma}_{\tau,\varepsilon}}$,
 \item[(C${}_2$)] 
 $\|u\mapsto \Pi_{>4}(\fI_u-\fI_{\tau,\varepsilon;u}) V_\Cdot\|_{\tilde\sW^{\alpha/4,\beta;\gamma}_{\tau,\varepsilon}}
 \leq
 C\,\lambda_{\tau\vee\varepsilon}^\kappa\,\|V_\Cdot\|_{\sV^{\alpha,\beta;\gamma}}$.
\end{itemize}
Recall that the operators $\Pi_k$ and $\Pi_{>k}$ were introduced in Def.~\ref{dfn:sN}. Note that by Remarks~\ref{rem:sW_sV} and~\ref{rem:sV_norm_estimate} the bounds~(A${}_1$),~(B${}_1$),~(C${}_1$) imply analogous bounds with the norms as in the bounds~(A),~(B),~(C) and some universal constant $C\in(0,\infty)$. Hence, to prove the lemma it is enough to establish the bounds~(A${}_1$),~(B${}_1$),~(C${}_1$) and~(A${}_2$),~(B${}_2$),~(C${}_2$). In order to verify the bound~(A${}_1$) note that by Lemma~\ref{lem:R_bounds} and Remark~\ref{rem:fM_fE} for $m\in\{2,4\}$ it holds
\begin{multline}
 \|w_{u}^m\fI^{m,a,\sigma}_{\tau,\varepsilon;u} V_\Cdot\|_{\sM^m} 
 \leq 
 \int_0^u s^{-1}\,\|w_u^m
 (\fR\fE\fA^m_{\tau,\varepsilon;1,s}V_s)^{a,\sigma}
 \|_{\sM^m}\,\rd s
 \\
 \leq
 \sup_{b\in\bA^m}
 \int_0^u s^{|a|-|b|-1}\,(1-s/u)^{\zeta_\star}\,\|w_s^m\fA^{m,b,\sigma}_{\tau,\varepsilon;1,s} V_s\|_{\sM^m}\,\rd s.
\end{multline}
By Def.~\ref{dfn:R_map} of the map $\fR$ it holds $\|w_{u}^m\fI^{m,a,\sigma}_{\tau,\varepsilon;u} V_\Cdot\|_{\sM^m} =0$ if $|a|\leq 1$ and $m=2$, or $|a|=0$ and $m=4$. Using Def.~\ref{dfn:sW} and~\ref{dfn:V_space} as well as Lemma~\ref{lem:bounds_relevant_irrelevant}~(B) we obtain
\begin{multline}
 \|w_{u}^m\fI^{m,a,\sigma}_{\tau,\varepsilon;u} V_\Cdot\|_{\sM^m} 
 \leq 
 \alpha^{-m} m^{-\beta}\,
 \|V_\Cdot\|_{\sW^{\alpha,\beta;\gamma}_{\tau,\varepsilon}}
 \int_0^u \lambda_s^{\rho_{\gamma,\kappa}(m)}\, s^{|a|+m/2-3}\,(1-s/u)^{\zeta_\star}
 \,\rd s
 \\
 \leq 
 C\,\alpha^{-m} m^{-\beta}\,
 \lambda_u^{\rho_{\gamma,\kappa}(m)}\, u^{|a|+m/2-2}\,
 \|V_\Cdot\|_{\sW^{\alpha,\beta;\gamma}_{\tau,\varepsilon}}
\end{multline}
with some universal constant $C\in(0,\infty)$ for $m=2$ and $|a|>1$, or $m=4$ and $|a|>0$. This proves the bound~(A${}_1$). The proof of bound~(B${}_1$) is the same, except that the weight $w_s^m$ is replaced with the weight $\tilde w_s^m$. In order to verify the bound~(C${}_1$) note that by Lemma~\ref{lem:R_bounds} and Remark~\ref{rem:fM_fE} for $m\in\{2,4\}$ it holds
\begin{multline}
 \|\tilde w_{u}^m(\fI^{m,a,\sigma}_{\tau,\varepsilon;u}-\fI^{m,a,\sigma}_{u}) V_\Cdot\|_{\sM^m} 
  \\
 \leq
 \sup_{b\in\bA^m}
 \int_0^u s^{|a|-|b|-1}\,(1-s/u)^{\zeta_\star}\,\|\tilde w_s^m(\fA^{m,b,\sigma}_{\tau,\varepsilon;1,s}-\fA^{m,b,\sigma}_{1,s}) V_s\|_{\sM^m}\,\rd s.
\end{multline}
By Def.~\ref{dfn:R_map} of the map $\fR$ it holds $ \|\tilde w_{u}^m(\fI^{m,a,\sigma}_{\tau,\varepsilon;u}-\fI^{m,a,\sigma}_{u}) V_\Cdot\|_{\sM^m}  =0$ if $|a|\leq 1$ and $m=2$, or $|a|=0$ and $m=4$.
Using Lemma~\ref{lem:A_map_bound}~(E), Def.~\ref{dfn:V_space} and Lemma~\ref{lem:bounds_relevant_irrelevant}~(B), we obtain
\begin{multline}
	\|\tilde w_{u}^m
	(\fI^{m,a,\sigma}_{\tau,\varepsilon;u}
	-\fI^{m,a,\sigma}_{u})V_\Cdot\|_{\sM^m}
	\\
	\leq
	(\alpha/2)^{-m}m^{1-\beta}
	\lambda_{\tau\vee\varepsilon}^{\kappa}
	\|V_\Cdot\|_{\sV^{\alpha,\beta;\gamma}}
	\int_0^u
	\lambda_s^{\rho_{\gamma,\kappa}(m)}
	s^{|a|+m/2-3}(1-s/u)^{\zeta_\star}
	\,\rd s
	\\
	\leq
	C\,(\alpha/2)^{-m}m^{1-\beta}
	\lambda_{\tau\vee\varepsilon}^{\kappa}
	\lambda_u^{\rho_{\gamma,\kappa}(m)}
	u^{|a|+m/2-2}
	\|V_\Cdot\|_{\sV^{\alpha,\beta;\gamma}}.
\end{multline}
with some universal constant $C\in(0,\infty)$ for $m=2$ and $|a|>1$, or $m=4$ and $|a|>0$.  This proves the bound~(C${}_1$). Let us turn to the proof of the bounds~(A${}_2$),~(B${}_2$),~(C${}_2$). By Def.~\ref{dfn:sW} of the norm $\|\Cdot\|_{\sW^{\alpha,\beta;\gamma}_{\tau,\varepsilon}}$ and the triangle inequality it suffices to show that
\begin{itemize}
 \item[(A${}_3$)] 
 $\|u\mapsto \Pi_{k}\fA_{\tau,\varepsilon;t,u} \Pi_{>4}\fI_{\tau,\varepsilon;u} V_\Cdot\|_{\sV^{\alpha,\beta+1;\gamma}}
 \leq
 C\,\|V_\Cdot\|_{\sW^{\alpha,\beta;\gamma}_{\tau,\varepsilon}}$,
 \item[(B${}_3$)] 
 $\|u\mapsto \Pi_{k}\fA_{\tau,\varepsilon;t,u}\Pi_{>4}\fI_{\tau,\varepsilon;u} V_\Cdot\|_{\tilde\sV^{\alpha,\beta+1;\gamma}}
 \leq
 C\,\|V_\Cdot\|_{\tilde\sW^{\alpha,\beta;\gamma}_{\tau,\varepsilon}}$,
 \item[(C${}_3$)] 
 $\|u\mapsto \Pi_{k}\fA_{\tau,\varepsilon;t,u}\Pi_{>4}(\fI_u-\fI_{\tau,\varepsilon;u}) V_\Cdot\|_{\tilde\sV^{\alpha/4,\beta;\gamma}}
 \leq
 C\,\lambda_{\tau\vee\varepsilon}^\kappa\,\|V_\Cdot\|_{\sV^{\alpha,\beta;\gamma}}$
\end{itemize}
for $k\in\{1,2,3,4\}$ and
\begin{itemize}
 \item[(A${}_4$)] 
 $\|u\mapsto \Pi_{>4}\fA_{\tau,\varepsilon;t,u} \Pi_{>4}\fI_{\tau,\varepsilon;u} V_\Cdot\|_{\sV^{\alpha,\beta+1;\gamma}}
 \leq
 C\,\|V_\Cdot\|_{\sW^{\alpha,\beta;\gamma}_{\tau,\varepsilon}}$,
 \item[(B${}_4$)] 
 $\|u\mapsto \Pi_{>4}\fA_{\tau,\varepsilon;t,u}\Pi_{>4}\fI_{\tau,\varepsilon;u} V_\Cdot\|_{\tilde\sV^{\alpha,\beta+1;\gamma}}
 \leq
 C\,\|V_\Cdot\|_{\tilde\sW^{\alpha,\beta;\gamma}_{\tau,\varepsilon}}$,
 \item[(C${}_4$)] 
 $\|u\mapsto \Pi_{>4}\fA_{\tau,\varepsilon;t,u}\Pi_{>4}(\fI_u-\fI_{\tau,\varepsilon;u}) V_\Cdot\|_{\tilde\sV^{\alpha/4,\beta;\gamma}}
 \leq
 C\,\lambda_{\tau\vee\varepsilon}^\kappa\,\|V_\Cdot\|_{\sV^{\alpha,\beta;\gamma}}$
\end{itemize}
with some constant $C\in(0,\infty)$ independent of $\tau,\varepsilon,t\in[0,1]$. Actually, the bounds~(A${}_4$), (B${}_4$) with $t=0$ together with Lemma~\ref{lem:A_map_bound}~(A),~(C) and Remark~\ref{rem:sV_norm_estimate} imply the bounds~(A${}_3$), (B${}_3$), since $\fA_{\tau,\varepsilon;0,u}=\mathrm{Id}$ on positive-degree functionals. The bounds~(C${}_3$),~(C${}_4$) will follow below from the stronger bound
\begin{equation}\label{eq:I_cutoff_bare}
	\|u\mapsto \Pi_{>4}(\fI_u-\fI_{\tau,\varepsilon;u}) V_\Cdot\|_{\tilde\sV^{\alpha/2,\beta;\gamma}}
	\leq
	C\,\lambda_{\tau\vee\varepsilon}^{\kappa}
	\|V_\Cdot\|_{\sV^{\alpha,\beta;\gamma}}
\end{equation}
together with Lemma~\ref{lem:A_map_bound}~(C) and Remark~\ref{rem:sV_norm_estimate}. Hence, it remains to prove the bounds~(A${}_4$), (B${}_4$) and~\eqref{eq:I_cutoff_bare}.

Let $\nu\in\{0,1/2\}$. For $m\in\{5,6,\ldots\}$, $a\in\bA^m$, $\sigma\in\bFF^m$ and $r\in[0,1]$ set
\begin{equation}
	M^\nu_{m,a,\sigma}(r)
	:=
	\sup_{s\in(0,1]}
	\lambda_s^{-\rho_{\gamma,\kappa}(m)}s^{2-m/2-|a|}
	\|w^m_{s;\nu}\fA^{m,a,\sigma}_{\tau,\varepsilon;r,s}V_s\|_{\sM^m}.
\end{equation}
By Def.~\ref{dfn:sW} and Def.~\ref{dfn:V_space}, for every $r\in[0,1]$,
\begin{equation}
	\begin{aligned}
		\sum_{a\in\bA^m}\sum_{\sigma\in\bFF^m}M^0_{m,a,\sigma}(r)
		&\leq
		\alpha^{-m}m^{-\beta}
		\|V_\Cdot\|_{\sW^{\alpha,\beta;\gamma}_{\tau,\varepsilon}},
		\\
		\sum_{a\in\bA^m}\sum_{\sigma\in\bFF^m}M^{1/2}_{m,a,\sigma}(r)
		&\leq
		\alpha^{-m}m^{-\beta}
		\|V_\Cdot\|_{\tilde\sW^{\alpha,\beta;\gamma}_{\tau,\varepsilon}}.
	\end{aligned}
\end{equation}
By Lemma~\ref{lem:fA_fE_one}, Remark~\ref{rem:fM_fE} and Lemma~\ref{lem:bounds_relevant_irrelevant}~(A), for every $t\in[0,1]$ and $u\in(0,1]$,
\begin{multline}
	\|w^m_{u;\nu}\fA^{m,a,\sigma}_{\tau,\varepsilon;t,u}
	\fI_{\tau,\varepsilon;u}V_\Cdot\|_{\sM^m}
	\leq
	\max_{r\in\{t,1\}}M^\nu_{m,a,\sigma}(r)
	\int_0^u
	\lambda_s^{\rho_{\gamma,\kappa}(m)}s^{|a|+m/2-3}\,\rd s
	\\
	\leq
	10m^{-1}\lambda_u^{\rho_{\gamma,\kappa}(m)}
	u^{|a|+m/2-2}
	\max_{r\in\{t,1\}}M^\nu_{m,a,\sigma}(r).
\end{multline}
Multiplying by the factors in Def.~\ref{dfn:V_space}, summing over $a,\sigma$ and taking the supremum over $m$ proves the bounds~(A${}_4$) and~(B${}_4$), with $\nu=0$ and $\nu=1/2$, respectively.

For $m\in\{5,6,\ldots\}$, $a\in\bA^m$ and $\sigma\in\bFF^m$ set
\begin{equation}
D_{m,a,\sigma}
:=
\sup_{s\in(0,1]}
\lambda_s^{-\rho_{\gamma,\kappa}(m)}s^{2-m/2-|a|}
\|\tilde w^m_s(\fA^{m,a,\sigma}_{1,s}
-\fA^{m,a,\sigma}_{\tau,\varepsilon;1,s})V_s\|_{\sM^m}.
\end{equation}
By Lemma~\ref{lem:A_map_bound}~(E) and Def.~\ref{dfn:V_space},
\begin{equation}
\sum_{a\in\bA^m}\sum_{\sigma\in\bFF^m}D_{m,a,\sigma}
\leq
(\alpha/2)^{-m}m^{1-\beta}
\lambda_{\tau\vee\varepsilon}^{\kappa}
\|V_\Cdot\|_{\sV^{\alpha,\beta;\gamma}}.
\end{equation}
Using Lemma~\ref{lem:fA_fE_one}, Remark~\ref{rem:fM_fE} and Lemma~\ref{lem:bounds_relevant_irrelevant}~(A), we obtain
\begin{multline}
\|\tilde w_u^m (\fI^{m,a,\sigma}_u-\fI^{m,a,\sigma}_{\tau,\varepsilon;u})V_\Cdot\|_{\sM^m}
\leq
D_{m,a,\sigma}
\int_0^u\lambda_s^{\rho_{\gamma,\kappa}(m)}s^{|a|+m/2-3}\,\rd s
\\
\leq 10m^{-1}D_{m,a,\sigma}
\lambda_u^{\rho_{\gamma,\kappa}(m)}u^{|a|+m/2-2}.
\end{multline}
Summing over $a,\sigma$ and taking the supremum over $m$ gives~\eqref{eq:I_cutoff_bare}. By Lemma~\ref{lem:A_map_bound}~(C), uniformly in $t\in[0,1]$,
\begin{equation}
	\|u\mapsto \fA_{\tau,\varepsilon;t,u}\Pi_{>4}(\fI_u-\fI_{\tau,\varepsilon;u}) V_\Cdot\|_{\tilde\sV^{\alpha/4,\beta;\gamma}}
	\leq
	C\,\lambda_{\tau\vee\varepsilon}^{\kappa}
	\|V_\Cdot\|_{\sV^{\alpha,\beta;\gamma}}.
\end{equation}
This proves the bound~(C${}_4$), while the bound~(C${}_3$) follows from Remark~\ref{rem:sV_norm_estimate}.

It remains to prove that the map $u\mapsto \fI_{\tau,\varepsilon;u} V_\Cdot$ belongs to $\sW^{\alpha,\beta+1;\gamma}_{\tau,\varepsilon}$. By the bound~(A) and the fact that $\sV^{\mathrm{fin};\gamma}_{\tau}$ is a dense subset of $\sW^{\alpha,\beta;\gamma}_{\tau,\varepsilon}$ it suffices to establish the above claim for $V_\Cdot\in \sV^{\mathrm{fin};\gamma}_{\tau}$. It follows from Lemma~\ref{lem:symmetries} and Remark~\ref{rem:R_symmetries} that $\fI_{\tau,\varepsilon;u} V_\Cdot\in \sN^{\mathrm{fin}}_{\tau}$ for all $u\in(0,1]$. Moreover, Lemmas~\ref{lem:A_map_bound},~\ref{lem:R_bounds} and dominated convergence imply that the map $u\mapsto \fI_{\tau,\varepsilon;u} V_\Cdot$ is continuous componentwise. Since $V_\Cdot$ has a uniform degree cutoff, the same is true for $u\mapsto \fI_{\tau,\varepsilon;u} V_\Cdot$. Hence, by the bound~(A), the map $u\mapsto \fI_{\tau,\varepsilon;u} V_\Cdot$ belongs to $\sV^{\mathrm{fin};\gamma}_{\tau}\subset \sW^{\alpha,\beta+1;\gamma}_{\tau,\varepsilon}$. This finishes the proof.
\end{proof}

\begin{lem}\label{lem:symmetries}
Let $\tau,\varepsilon\in[0,1]$ and $t,s\in(0,1]$. If $V,W\in\sN^{\mathrm{fin}}_\tau$, then $\fB_{\varepsilon;t}(V,W)\in\sN^{\mathrm{fin}}_\tau$ and $\fE \fA_{\tau,\varepsilon;t,s} V\in\sN^{\mathrm{fin}}_\tau$. 
\end{lem}

\begin{proof}
Invariance under the symmetries of the torus or plane, as well as reality, follows directly from Def.~\ref{dfn:A_map}, Def.~\ref{dfn:B_map}, Remark~\ref{rem:H_symmetries} and Remark~\ref{rem:Wick}. Moreover, both maps preserve finite degree. It remains to verify compatibility with the flow of charge and flavor.

For $j\in\bN_+$ and
$W\in\sS'(\bR^{4j})^{\bA^{2j}\times\bFF^{2j}}$, set
\begin{equation}
 \tilde\fS W
 :=
 \frac{1}{|\tilde\cP_j|}
 \sum_{\pi\in\tilde\cP_j}\fS_\pi W.
\end{equation}
Given $V\in\sN^{\mathrm{fin}}_\tau$, for every even $p$ choose a representative kernel $\tilde V^p\in\sNN^p$ as in Def.~\ref{dfn:flow_of_charge}, so that $V^p=\fS \tilde V^p$. Let $p=2r$, $q=2\ell$ and define $\fK_{\varepsilon;t}(\tilde V^p,\tilde V^q)\in\sNN^{p+q-2}$ by
\begin{multline}
	\langle
	\fK_{\varepsilon;t}(\tilde V^p,\tilde V^q),
	\varphi_1\otimes\ldots\otimes\varphi_{p+q-2}\rangle
	\\
	:=
	\langle
	\tilde V^p\otimes \tilde V^q,
	\varphi_1\otimes\ldots\otimes\varphi_{p-1}
	\otimes\dot H_{\varepsilon;t}
	\otimes\varphi_p\otimes\ldots\otimes\varphi_{p+q-2}\rangle.
\end{multline}
By Def.~\ref{dfn:covariance} and the block diagonal form of the matrices $\Gamma_j$, the kernel $\dot H_{\varepsilon;t}$ contracts only opposite charges of the same flavor, and its spinor part is independent of that flavor. 
Consequently, the uncontracted members of the two contracted pairs form a new pair satisfying the conditions in  Def.~\ref{dfn:flow_of_charge}. The pairwise charge conjugation property follows from Remark~\ref{rem:H_symmetries}. Moreover, for every
$\pi\in\tilde\cP_{r+\ell-1}$,
\begin{equation}
 \fS_\pi
 \fK_{\varepsilon;t}(\tilde V^p,\tilde V^q)
 \in\sNN^{p+q-2}.
\end{equation}
Indeed, using the invariance of $\tilde V^p$ and $\tilde V^q$ under $\tilde\cP_r$ and $\tilde\cP_\ell$, respectively, each such kernel can be represented as a contraction of one argument of $\tilde V^p$ with one argument of $\tilde V^q$ of the same form as the contraction defining $\fK_{\varepsilon;t}(\tilde V^p,\tilde V^q)$. Consequently,
\begin{equation}
 \tilde\fS
 \fK_{\varepsilon;t}(\tilde V^p,\tilde V^q)
 \in\sNN^{p+q-2}.
\end{equation}
The averaging makes this kernel invariant under $\tilde\cP_{r+\ell-1}$ and preserves the charge, flavor and pairwise charge-conjugation conditions. Hence it is an admissible representative.

Since $\tilde\cP_r$ and $\tilde\cP_\ell$ act transitively on the arguments of $\tilde V^p$ and $\tilde V^q$,  respectively, all choices of the two contracted arguments give the same contribution after antisymmetrization.  Hence, using Def.~\ref{dfn:B_map}, for even
$m$,
\begin{equation}
\fB^m_{\varepsilon;t}(V)
=
\fS U^m,
\qquad
U^m:=
-\frac12
\sum_{\substack{p,q\in2\bN_+\\p+q=m+2}}
pq\, \tilde\fS
\fK_{\varepsilon;t}(\tilde V^p,\tilde V^q).
\end{equation}
For odd $m$ the LHS vanishes. Thus $\fB_{\varepsilon;t}(V)$ is compatible with the flow of charge and flavor. The same is true for $\fB_{\varepsilon;t}(V,W)$ by polarization.

We turn to $\fE\fA_{\tau,\varepsilon;t,s}V$. Set
\begin{equation}
	\eta:=\fJ\varPsi_{\tau,\varepsilon;t\vee s,s}.
\end{equation}
For $n=2k\geq4$ and an admissible representative
$\tilde V\in\sNN^n$, define the Schwartz distributions
$\fD_{\rm in}\tilde V$ and $\fD_{\rm out}\tilde V$ by
\begin{equation}
 \begin{aligned}
 \langle\fD_{\rm in}\tilde V,
 \varphi_1\otimes\ldots\otimes\varphi_{n-2}\rangle
 &:=
 \fE\langle
 \tilde V,\eta\otimes\eta\otimes
 \varphi_1\otimes\ldots\otimes\varphi_{n-2}\rangle,
 \\
 \langle\fD_{\rm out}\tilde V,
 \varphi_1\otimes\ldots\otimes\varphi_{n-2}\rangle
 &:=
 \fE\langle
 \tilde V,\varphi_1\otimes\eta\otimes\eta\otimes
 \varphi_2\otimes\ldots\otimes\varphi_{n-2}\rangle.
 \end{aligned}
\end{equation}
For every $\pi\in\tilde\cP_{k-1}$,
\begin{equation}
 \fS_\pi\fD_{\rm in}\tilde V\in\sNN^{n-2},
 \qquad
 \fS_\pi\fD_{\rm out}\tilde V\in\sNN^{n-2}.
\end{equation}
Indeed, using the invariance of $\tilde V$ under $\tilde\cP_k$, the above kernels can be represented as a contraction of two arguments of $\tilde V$ of the same type as the corresponding contraction above, with an uncontracted argument in the first position. Consequently,
\begin{equation}
 \tilde\fS\fD_{\rm in}\tilde V\in\sNN^{n-2},
 \qquad
 \tilde\fS\fD_{\rm out}\tilde V\in\sNN^{n-2}.
\end{equation}
We set
\begin{equation}
 \tilde\fD \tilde V
 :=
 n\,\tilde\fS\fD_{\rm in}\tilde V
 +n(n-2)\,\tilde\fS\fD_{\rm out}\tilde V.
\end{equation}
The first contraction removes one pair, whereas the second contracts members of two different pairs and joins their uncontracted members into a new pair. 
By Def.~\ref{dfn:covariance}, Remarks~\ref{rem:H_symmetries} and~\ref{rem:Wick}, and the argument above, $\tilde\fD\tilde V$ is an admissible representative whenever $\tilde V$ is one.
The ordered pairs of contracted arguments form two orbits under $\tilde\cP_k$: contractions within one pair and contractions between two distinct pairs. 
Their cardinalities are $n$ and $n(n-2)$, respectively.
Hence Wick's formula gives, for even $m$,
\begin{equation}
	\fE\fA^m_{\tau,\varepsilon;t,s}V
	=
	\fS
	\sum_{j=0}^\infty
	\frac1{2^j j!}
	\tilde\fD^{\,j}\tilde V^{m+2j}.
\end{equation}
The sum is finite, and every summand on the RHS is an admissible representative. 
Hence  $\fE\fA_{\tau,\varepsilon;t,s}V$ is compatible with the flow of charge and flavor.
For odd $m$ the LHS vanishes.
Finally, Remark~\ref{rem:flow_charge_internal} implies invariance under
the internal symmetries.
\end{proof}

\section{Fixed point problem}\label{sec:fixed_point}

In this section we construct a solution $X_{\tau,\varepsilon;\Cdot}\equiv(g_{\tau,\varepsilon;\Cdot},r_{\tau,\varepsilon;\Cdot},z_{\tau,\varepsilon;\Cdot},W_{\tau,\varepsilon;\Cdot})$ of the system of equations~\eqref{eq:flow_g_intro_ren}, \eqref{eq:flow_r_intro_ren},
\eqref{eq:flow_2_intro_z}, \eqref{eq:flow_2_intro_W1}, 
\eqref{eq:flow_2_intro_W2} introduced in Sec.~\ref{sec:strategy}. To this end, we rewrite this system of equations as a fixed point equation of a certain map $\fX_{\tau,\varepsilon;\Cdot}$ acting on a complete metric space $\sY_{\tau,\varepsilon}$ and prove that the map $\fX_{\tau,\varepsilon;\Cdot}$ is a contraction provided the parameter $\lambda\in(0,1]$ is sufficiently small. As argued in Sec.~\ref{sec:strategy} and proved in Sec.~\ref{sec:relation_polchinski} using $X_{\tau,\varepsilon;\Cdot}$ it is possible to construct a solution $U_{\tau,\varepsilon;\Cdot}$ of the flow equation~\eqref{eq:polchinski_integral}.

\begin{dfn}\label{dfn:sX}
For $\tau,\varepsilon\in[0,1]$ and
\begin{equation}
 X_\Cdot\equiv(g_\Cdot,r_\Cdot,z_\Cdot,W_\Cdot)\in C((0,1],\bR)\times C((0,1],\bR)\times C((0,1],\bR)\times \sW^{8,4;2-80\kappa}_{\tau,\varepsilon}
\end{equation}
we define
\begin{equation}
\begin{gathered}
 \|X_\Cdot\|_{\sP}
 := 
 \sup_{t\in(0,1]}\lambda_t^{1+10\kappa}\, |g_t|
 +
 \sup_{t\in(0,1]}\lambda_t^{36\kappa-1}\, t\,|r_t|
 +
 \sup_{t\in(0,1]}\lambda_t^{36\kappa-1}\,|z_t|,
 \\
 \|X_\Cdot\|_{\sX_{\tau,\varepsilon}}
 := 
 \|X_\Cdot\|_{\sP}
 +
 \|W_\Cdot\|_{\sW^{8,4;2-80\kappa}_{\tau,\varepsilon}},
 \qquad
 \|X_\Cdot\|_{\tilde\sX_{\tau,\varepsilon}}
 := 
 \|X_\Cdot\|_{\sP}
 +
 \|W_\Cdot\|_{\tilde\sW^{2,3;2-80\kappa}_{\tau,\varepsilon}}.
\end{gathered} 
\end{equation}
By definition the real vector space $\sX_{\tau,\varepsilon}$ consists of maps $X_\Cdot$ of the above form such that $\|X_\Cdot\|_{\sX_{\tau,\varepsilon}}<\infty$. We also define the set
\begin{equation}
 \sY_{\tau,\varepsilon}:=\{X_\Cdot\in \sX_{\tau,\varepsilon}\,|\,\|X_\Cdot\|_{\sX_{\tau,\varepsilon}}\leq 1,~\sP_\varepsilon(X_\Cdot)\geq\lambda^\kappa\},
\end{equation}
where $\sP_\varepsilon(X_\Cdot):=\inf_{t\in(0,1]} \lambda_{\varepsilon\vee t} \, g_t$.
We omit $\tau,\varepsilon$ if $\tau=0$, $\varepsilon=0$.
\end{dfn}

\begin{dfn}\label{dfn:theta_ep_t}
For $\varepsilon\in[0,1]$ and $t\in(0,1]$ we define $\theta_{\varepsilon;t}:=\lambda_t/\lambda_{\varepsilon\vee t}\in[0,1]$.
\end{dfn}

\begin{lem}\label{lem:complete}
For all $\tau,\varepsilon\in[0,1]$ the space $(\sX_{\tau,\varepsilon},\|\Cdot\|_{\sX_{\tau,\varepsilon}})$ is a Banach space. Moreover, the set $\sY_{\tau,\varepsilon}$ is a closed subset of $\sX_{\tau,\varepsilon}$.
\end{lem}
\begin{proof}
The completeness of $\sX_{\tau,\varepsilon}$ follows from Lemma~\ref{lem:W_banach} and the completeness of the weighted spaces of real-valued continuous functions defining the first three
components. The set $\sY_{\tau,\varepsilon}$ is closed since the norm
bound and the inequalities $\lambda_{\varepsilon\vee t}g_t\geq\lambda^\kappa$, $t\in(0,1]$, are preserved under convergence in $\sX_{\tau,\varepsilon}$.
\end{proof}

\begin{lem}\label{lem:X_norm}
Let $\tau,\varepsilon\in[0,1]$. 
\begin{itemize}
 \item[(A)] For all $X_\Cdot\in \sX_{\tau,\varepsilon}$ it holds $\|X_\Cdot\|_{\tilde\sX_{\tau,\varepsilon}}\leq \|X_\Cdot\|_{\sX_{\tau,\varepsilon}}$.
 
 \item[(B)] For all $X_\Cdot\in \sX$ it holds $\|X_\Cdot\|_{\tilde\sX_{\tau,\varepsilon}}\leq \|X_\Cdot\|_{\sX}$.
%
\end{itemize}
\end{lem}
\begin{proof}
Item~(A) follows from Def.~\ref{dfn:sW} and Remark~\ref{rem:sV_norm_property}. To prove Item~(B) observe that
\begin{equation}
 \|W_\Cdot\|_{\tilde\sW^{2,3;2-80\kappa}_{\tau,\varepsilon}}
 \leq
 \|W_\Cdot\|_{\sW^{2,3;2-80\kappa}_{\tau,\varepsilon}}
 \leq
 \|W_\Cdot\|_{\sV^{4,3;2-80\kappa}}
 \leq
 \|W_\Cdot\|_{\sV^{8,4;2-80\kappa}}
 \leq
 \|W_\Cdot\|_{\sW^{8,4;2-80\kappa}}
\end{equation}
by Remark~\ref{rem:sW_tilde},~\ref{rem:sW_sV},~\ref{rem:sV_norm_property} and~\ref{rem:sV_sW}, respectively. 
\end{proof}

\begin{dfn}\label{dfn:X_map}
Let $\tau,\varepsilon\in[0,1]$. For $X_\Cdot\equiv(g_\Cdot,r_\Cdot,z_\Cdot,W_\Cdot)
\in\sY_{\tau,\varepsilon}\cup\sY$
and $t\in(0,1]$ define
\begin{equation}
\begin{aligned}
 \fW^m_{\tau,\varepsilon;t}(X_\Cdot) 
 &:=
 \int_0^t
 \fE\fA^m_{\tau,\varepsilon;t,s}\fB_{\varepsilon;s}(V_s)\,\rd s\in\sN^m,
 &&
 m\in\bN_+\setminus\{2,4\},
 \\
 \fW^m_{\tau,\varepsilon;t}(X_\Cdot) 
 &:=
 \int_0^t
 \fR\fE\fA^m_{\tau,\varepsilon;1,s}\fB_{\varepsilon;s}(V_s)\,\rd s
 -
 \fE\fC_{\tau,\varepsilon;1,t}^mW_t \in\sN^m,
 \qquad &&
 m\in\{2,4\},
\end{aligned} 
\end{equation}
as well as
\begin{equation}
\begin{aligned}
 \mathbf{g}_{\tau,\varepsilon;t}(X_\Cdot)
 :=&
 \lambda^{-1}+\int_t^1 (g_s)^2\,\fL\fE\fA^4_{\tau,\varepsilon;1,s}\fB_{\varepsilon;s}(V_s)\,\rd s\in\bR,
 \\
 \mathbf{r}_{\tau,\varepsilon;t}(X_\Cdot) 
 :=&
 -\fL\fE\fA^2_{\tau,\varepsilon;1,t} U(1/g_t,0,0)
%
 -\int_t^1
 \fL\fE\fA^2_{\tau,\varepsilon;1,s}\fB_{\varepsilon;s}(V_s)\,\rd s\in\bR,
 \\
 \mathbf{z}_{\tau,\varepsilon;t}(X_\Cdot)
 :=&
 \int_0^t
 \fL_\partial\fE\fA^2_{\tau,\varepsilon;1,s}\fB_{\varepsilon;s}(V_s)\,\rd s\in\bR,
\end{aligned} 
\end{equation}
where
\begin{equation}
 V_s\equiv V_s(X_\Cdot)= U(1/g_s,r_s,z_s)+W_s\in \sN,\qquad s\in(0,1].
\end{equation}
For $\lambda\in(0,1]$ small enough the map $\fX_{\tau,\varepsilon;\Cdot}\,:\, \sY_{\tau,\varepsilon}\to \sX_{\tau,\varepsilon}$ is defined by
\begin{equation}
 \fX_{\tau,\varepsilon;\Cdot} 
 :=
 \mathbf{g}_{\tau,\varepsilon;\Cdot}\times\mathbf{r}_{\tau,\varepsilon;\Cdot}\times\mathbf{z}_{\tau,\varepsilon;\Cdot}\times\fW_{\tau,\varepsilon;\Cdot},
 \qquad
 \fW_{\tau,\varepsilon;\Cdot} 
 :=
 \bigtimes_{m\in\bN_+}\fW^m_{\tau,\varepsilon;\Cdot}.
\end{equation}
For $Z_\Cdot\in\sY$ we use the same notation
$\fX_{\tau,\varepsilon;\Cdot}(Z_\Cdot)$ for the tuple defined by the
above formulas.
\end{dfn}
\begin{rem}
The fact that the map $\fX_{\tau,\varepsilon;\Cdot}$ is well-defined is non-trivial and is a consequence of the estimates established below. We also point out that $\fW^m_{\tau,\varepsilon;t}(X_\Cdot) =0$ for $m\in\bN_+\setminus2\bN_+$.
\end{rem}

\begin{rem}
We call the parameter $g_{\tau,\varepsilon;t}$ the inverse of the effective coupling constant. Note that we fixed $g_{\tau,\varepsilon;t}$ at unit scale $t=1$ to be $g_{\tau,\varepsilon;t=1}:=\lambda^{-1}$, where $\lambda\in(0,\lambda_\star]$ and $\lambda_\star\in(0,1]$ is a small constant.
\end{rem}
\begin{rem}
The parameter $r_{\tau,\varepsilon;t}$ is the correction to the effective mass due to the interaction.  Note that we fixed $r_{\tau,\varepsilon;t}$ at unit scale $t=1$ to be $r_{\tau,\varepsilon;t=1}=0$. Recall that the free part of the action contains the unit mass term. Hence, the above condition implies that the effective mass at unit scale equals one.
\end{rem}

\begin{rem}\label{rem:g_r_z_real}
The fact that $\mathbf{g}_{\tau,\varepsilon;t}(X_\Cdot),\mathbf{r}_{\tau,\varepsilon;t}(X_\Cdot),\mathbf{z}_{\tau,\varepsilon;t}(X_\Cdot)\in\bR$ follows from Lemma~\ref{lem:fL_symmetries} and the symmetry properties of functionals $U(g,r,z)$ and $W\in\sW^{\alpha,\beta;\gamma}_{\tau,\varepsilon}$.
\end{rem}

\begin{thm}\label{thm:contraction}
There exists $\lambda_\star\in(0,1]$ and $C\in(0,\infty)$ such that for all $\lambda\in(0,\lambda_\star]$, all $\tau,\varepsilon\in[0,1]$ and all $X_\Cdot,Y_\Cdot\in \sY_{\tau,\varepsilon}$, $Z_\Cdot\in\sY$ it holds:
\begin{itemize}
 \item[(A)] $\|s\mapsto\fX_{\tau,\varepsilon;s}(X_\Cdot)\|_{\sX_{\tau,\varepsilon}} 
 \leq
 C\,\lambda^\kappa$ and $\sP_\varepsilon(s\mapsto\fX_{\tau,\varepsilon;s}(X_\Cdot))\geq 1/C$,
 \item[(B)] $\|s\mapsto(\fX_{\tau,\varepsilon;s}(X_\Cdot)-\fX_{\tau,\varepsilon;s}(Y_\Cdot))\|_{\sX_{\tau,\varepsilon}}
 \leq
 C\,\lambda^\kappa\, \|X_\Cdot-Y_\Cdot\|_{\sX_{\tau,\varepsilon}}$,
 \item[(C)] $\|s\mapsto(\fX_{\tau,\varepsilon;s}(X_\Cdot)
 -
 \fX_{\tau,\varepsilon;s}(Z_\Cdot))\|_{\tilde\sX_{\tau,\varepsilon}}
 \leq
 C\,\lambda^\kappa\,
 \|X_\Cdot-Z_\Cdot\|_{\tilde\sX_{\tau,\varepsilon}}$,
 
 \item[(D)] $\|s\mapsto(\fX_s(Z_{\Cdot}) 
 -
 \fX_{\tau,\varepsilon;s}(Z_{\Cdot}))\|_{\tilde\sX_{\tau,\varepsilon}}
 \leq 
 C\,\lambda_{\tau\vee\varepsilon}^\kappa$.
\end{itemize}
\end{thm}
\begin{proof}
The theorem is an immediate consequence of Lemmas~\ref{lem:contraction_W},~\ref{lem:contraction_g},~\ref{lem:contraction_r},~\ref{lem:contraction_z}.
\end{proof}

\begin{cor}\label{cor:contraction}
There exists $\lambda_\star\in(0,1]$ and $C\in(0,\infty)$ such that for all $\lambda\in(0,\lambda_\star]$ and all $\tau,\varepsilon\in[0,1]$ the map $\fX_{\tau,\varepsilon;\Cdot}\,:\,\sY_{\tau,\varepsilon}\to \sY_{\tau,\varepsilon}$ is well-defined and has the unique fixed point denoted by $X_{\tau,\varepsilon;\Cdot}\equiv (g_{\tau,\varepsilon;\Cdot},r_{\tau,\varepsilon;\Cdot},z_{\tau,\varepsilon;\Cdot},W_{\tau,\varepsilon;\Cdot})$ such that
\begin{equation}\label{eq:cor_bound}
 \|X_{\Cdot}-X_{\tau,\varepsilon;\Cdot}\|_{\tilde\sX_{\tau,\varepsilon}}
 \leq 
 C\,\lambda_{\tau\vee\varepsilon}^\kappa\,,
\end{equation}
where $X_{\Cdot}:=X_{\tau,\varepsilon;\Cdot}$ with $\tau=0$, $\varepsilon=0$, and
\begin{equation}\label{eq:cor_bound_2}
 \|V_{\tau,\varepsilon;\Cdot}\|_{\sV^{8,4;1-40\kappa}} \leq C\,,
 \qquad
 \|V_{\Cdot}-V_{\tau,\varepsilon;\Cdot}\|_{\tilde\sV^{2,3;1-40\kappa}} \leq C\,\lambda_{\tau\vee\varepsilon}^\kappa\,
\end{equation}
where
\begin{equation}
 V_{\tau,\varepsilon;t}\equiv (V^m_{\tau,\varepsilon;t})_{m\in\bN_+}:= U(\theta^2_{\varepsilon;t}/g_{\tau,\varepsilon;t},r_{\tau,\varepsilon;t},z_{\tau,\varepsilon;t})+W_{\tau,\varepsilon;t}\in \sN
\end{equation}
for all $\tau,\varepsilon\in[0,1]$ and $t\in(0,1]$, we set $V_t:=V_{\tau,\varepsilon;t}$ with $\tau=0$, $\varepsilon=0$, and the function $\theta_{\varepsilon;t}$ was introduced in Def.~\ref{dfn:theta_ep_t}.
\end{cor}
\begin{proof}
The fact that for all sufficiently small $\lambda\in(0,1]$ the map $\fX_{\tau,\varepsilon;\Cdot}\,:\,\sY_{\tau,\varepsilon}\to \sY_{\tau,\varepsilon}$ is well-defined and is a contraction follows from Theorem~\ref{thm:contraction} and Remark~\ref{rem:g_r_z_real}. The existence and uniqueness of the fixed point $X_{\tau,\varepsilon;\Cdot}$ is then a consequence of Lemma~\ref{lem:complete} and the Banach fixed-point theorem. In order to prove the bound~\eqref{eq:cor_bound} note that
\begin{equation}
 X_\Cdot-X_{\tau,\varepsilon;\Cdot}
 =
 \fX_\Cdot(X_\Cdot) 
 -
 \fX_{\tau,\varepsilon;\Cdot}(X_{\tau,\varepsilon;\Cdot})\,.
\end{equation}
Consequently, by the triangle inequality we obtain
\begin{equation}
 \|X_\Cdot-X_{\tau,\varepsilon;\Cdot}\|_{\tilde\sX_{\tau,\varepsilon}}
 \leq
 \|\fX_\Cdot(X_{\Cdot}) 
 -
 \fX_{\tau,\varepsilon;\Cdot}(X_{\Cdot})\|_{\tilde\sX_{\tau,\varepsilon}}
 +
 \|\fX_{\tau,\varepsilon;\Cdot}(X_{\Cdot})
 -
 \fX_{\tau,\varepsilon;\Cdot}(X_{\tau,\varepsilon;\Cdot})\|_{\tilde\sX_{\tau,\varepsilon}}\,.
\end{equation}
Since $X_\Cdot\in \sY$ and $X_{\tau,\varepsilon;\Cdot}\in \sY_{\tau,\varepsilon}$, bounds~(C) and~(D) stated in Theorem~\ref{thm:contraction} imply that
\begin{equation}
 \|\fX_{\tau,\varepsilon;\Cdot}(X_\Cdot)
 -
 \fX_{\tau,\varepsilon;\Cdot}(X_{\tau,\varepsilon;\Cdot})\|_{\tilde\sX_{\tau,\varepsilon}}
 \leq
 C\,\lambda^\kappa\,
 \|X_\Cdot-X_{\tau,\varepsilon;\Cdot}\|_{\tilde\sX_{\tau,\varepsilon}}
\end{equation}
and
\begin{equation}
 \|\fX_\Cdot(X_{\Cdot}) 
 -
 \fX_{\tau,\varepsilon;\Cdot}(X_{\Cdot})\|_{\tilde\sX_{\tau,\varepsilon}}
 \leq 
 C\,\lambda_{\tau\vee\varepsilon}^\kappa\,.
\end{equation}
Consequently, we have
\begin{equation}
 \|X_\Cdot-X_{\tau,\varepsilon;\Cdot}\|_{\tilde\sX_{\tau,\varepsilon}}
 \leq
 C\,\lambda_{\tau\vee\varepsilon}^\kappa
 +
 C\,\lambda^\kappa\,
 \|X_\Cdot-X_{\tau,\varepsilon;\Cdot}\|_{\tilde\sX_{\tau,\varepsilon}}.
\end{equation}
This proves the bound~\eqref{eq:cor_bound}. Next, note that
\begin{equation}
 V_{\tau,\varepsilon;t}=\fV^{(0)}_{\varepsilon;t}(X_{\tau,\varepsilon;\Cdot}),
 \qquad
 V_{t} - V_{\tau,\varepsilon;t}
 =
 \fV^{(0)}_{t}(X_{\Cdot})
 -
 \fV^{(0)}_{\varepsilon;t}(X_{\Cdot})
 +
 \fV^{(0)}_{\varepsilon;t}(X_{\Cdot})
 -
 \fV^{(0)}_{\varepsilon;t}(X_{\tau,\varepsilon;\Cdot})
\end{equation}
where the map $\fV^{(0)}_{\varepsilon;t}$ was defined in Lemma~\ref{lem:contraction_V}. Hence, the bounds~\eqref{eq:cor_bound_2} are immediate consequences of Lemma~\ref{lem:contraction_V}. This finishes the proof.
\end{proof}

\begin{lem}\label{lem:contraction_g_aux}
For all $\tau,\varepsilon\in[0,1]$ and
\begin{equation}
 X_\Cdot\equiv(g_\Cdot,r_\Cdot,z_\Cdot,W_\Cdot)\in \sY_{\tau,\varepsilon},
 \qquad
 Y_\Cdot\equiv(\tilde g_\Cdot,\tilde r_\Cdot,\tilde z_\Cdot,\tilde W_\Cdot)\in \sY_{\tau,\varepsilon},
 \qquad
 Z_\Cdot\equiv(\hat g_\Cdot,\hat r_\Cdot,\hat z_\Cdot,\hat W_\Cdot)\in\sY
\end{equation}
it holds
\begin{itemize}
 \item[(A)] $\sup_{s\in(0,1]}\lambda^{12\kappa-1}_{s}\,|\theta^2_{\varepsilon;s}/g_{s}|
 \leq 1$,
 \item[(B)] $\sup_{s\in(0,1]}\lambda^{12\kappa-1}_{s}\,|\theta^2_{\varepsilon;s}/g_{s}-\theta^2_{\varepsilon;s}/\tilde g_{s}|
 \leq \|X_\Cdot-Y_\Cdot\|_{\sX_{\tau,\varepsilon}}$,
 \item[(C)] $\sup_{s\in(0,1]}\lambda^{12\kappa-1}_{s}\,|\theta^2_{\varepsilon;s}/g_{s}-\theta^2_{\varepsilon;s}/\hat g_{s}|
 \leq \|X_\Cdot-Z_\Cdot\|_{\tilde\sX_{\tau,\varepsilon}}$,
 \item[(D)] $\sup_{s\in(0,1]}\lambda^{12\kappa-1}_{s}\,|1/\hat g_{s}-\theta^2_{\varepsilon;s}/\hat g_{s}|
 \leq \lambda_\varepsilon^\kappa$,
\end{itemize}
where $\theta_{\varepsilon;t}$ was introduced in Def.~\ref{dfn:theta_ep_t}.
\end{lem}
\begin{proof}
First observe that
$\lambda_{\varepsilon\vee s} \theta^2_{\varepsilon;s}\leq \lambda_{s}$,
$\lambda^{2}_{\varepsilon\vee s} \theta^2_{\varepsilon;s}\leq \lambda^{2}_{s}$ and $\lambda_s\leq\lambda$.
The above bounds and Def.~\ref{dfn:sX} imply that
\begin{itemize}
 \item[(A${}_1$)] $\lambda^{\kappa-1}_{s}\,|\theta^2_{\varepsilon;s}/g_{s}|
 \leq \lambda_s^\kappa/\sP_\varepsilon(X_\Cdot)$,
 \item[(B${}_1$)] $\lambda^{12\kappa-1}_{s}\,|\theta^2_{\varepsilon;s}/g_{s}-\theta^2_{\varepsilon;s}/\tilde g_{s}|
 \leq \lambda_s^\kappa/\sP_\varepsilon(X_\Cdot)~\lambda_s^\kappa/\sP_\varepsilon(Y_\Cdot)~\|X_\Cdot-Y_\Cdot\|_{\sP}$,
 \item[(C${}_1$)] $\lambda^{12\kappa-1}_{s}\,|\theta^2_{\varepsilon;s}/g_{s}-\theta^2_{\varepsilon;s}/\hat g_{s}|
 \leq \lambda_s^\kappa/\sP_\varepsilon(X_\Cdot)~\lambda_s^\kappa/\sP(Z_\Cdot)~\|X_\Cdot-Z_\Cdot\|_{\sP}$,
 \item[(D${}_1$)] $\lambda^{2\kappa-1}_{s}\,|1/\hat g_{s}-\theta^2_{\varepsilon;s}/\hat g_{s}|
 \leq \lambda_\varepsilon^\kappa \lambda_s^\kappa/\sP(Z_\Cdot)$,
\end{itemize}
for all $s\in(0,1]$.
The lemma follows now from Def.~\ref{dfn:sX}.
\end{proof}

\begin{lem}\label{lem:contraction_U}
Let $\alpha,\beta\in[1,\infty)$, $\gamma\in\bR$. There exists $C\in(0,\infty)$ such that for all $\tau,\varepsilon\in[0,1]$ it holds:
\begin{itemize}
\item[(A)] 
 $\|s\mapsto \lambda_s^{\gamma+8\kappa}\, U(1,0,0)\|_{\sW^{\alpha,\beta;\gamma}_{\tau,\varepsilon}}\leq C$,
 
\item[(B)] 

 $\|s\mapsto \lambda_s^{\gamma+4\kappa}\,s^{-1}\, U(0,1,0)\|_{\sW^{\alpha,\beta;\gamma}_{\tau,\varepsilon}}\leq C$,

\item[(C)]
 
 $\|s\mapsto \lambda_s^{\gamma+4\kappa}\, U(0,0,1)\|_{\sW^{\alpha,\beta;\gamma}_{\tau,\varepsilon}}
 \leq C$.
\end{itemize}
Moreover, analogous bounds with $\sW^{\alpha,\beta;\gamma}_{\tau,\varepsilon}$ replaced by $\tilde\sW^{\alpha,\beta;\gamma}_{\tau,\varepsilon}$ are true.
\end{lem}
\begin{proof}
From Def.~\ref{dfn:U_functional} of the functional $U(g,r,z)$ and Def.~\ref{dfn:V_space} of the norms~$\|\Cdot\|_{\sV^{\alpha,\beta;\gamma}}$ and~$\|\Cdot\|_{\tilde\sV^{\alpha,\beta;\gamma}}$ it follows that the bounds
\begin{itemize}
\item[(A${}_1$)] 
 $\|s\mapsto \lambda_s^{\gamma+8\kappa}\, U(1,0,0)\|_{\sV^{2\alpha,\beta;\gamma}}\leq C$,
 
\item[(B${}_1$)] 

 $\|s\mapsto \lambda_s^{\gamma+4\kappa}\,s^{-1}\, U(0,1,0)\|_{\sV^{2\alpha,\beta;\gamma}}\leq C$,

\item[(C${}_1$)]
 
 $\|s\mapsto \lambda_s^{\gamma+4\kappa}\, U(0,0,1)\|_{\sV^{2\alpha,\beta;\gamma}}
 \leq C$
\end{itemize}
as well as the bounds with $\sV^{2\alpha,\beta;\gamma}$ replaced by $\tilde\sV^{2\alpha,\beta;\gamma}$ hold true. The bounds~(A),~(B),~(C) are consequences of the bounds~(A${}_1$),~(B${}_1$),~(C${}_1$) and Remark~\ref{rem:sW_sV}.
\end{proof}

\begin{rem}\label{rem:sV_sW_function}
It follows immediately from Def.~\ref{dfn:V_space} of the norm~$\|\Cdot\|_{\sV^{\alpha,\beta;\gamma}}$ and Def.~\ref{dfn:sW} of the norm~$\|\Cdot\|_{\sW^{\alpha,\beta;\gamma}_{\tau,\varepsilon}}$ that for all continuous functions $h_\Cdot\,:\,(0,1]\to\bR$ it holds
\begin{equation}
 \|s\mapsto h_s\,V_s\|_{\sV^{\alpha,\beta;\gamma}} \leq\|h_\Cdot\|_\infty\, \|V_\Cdot\|_{\sV^{\alpha,\beta;\gamma}},
 \qquad
 \|s\mapsto h_s\,V_s\|_{\sW^{\alpha,\beta;\gamma}_{\tau,\varepsilon}} \leq \|h_\Cdot\|_\infty\,\|V_\Cdot\|_{\sW^{\alpha,\beta;\gamma}_{\tau,\varepsilon}},
\end{equation}
where $\|h_\Cdot\|_\infty:=\sup_{s\in(0,1]} |h_s|$. The analogous bounds with $\sV^{\alpha,\beta;\gamma}$ replaced by $\tilde\sV^{\alpha,\beta;\gamma}$ and $\sW^{\alpha,\beta;\gamma}_{\tau,\varepsilon}$ replaced by $\tilde\sW^{\alpha,\beta;\gamma}_{\tau,\varepsilon}$ hold true as well.
\end{rem}

\begin{lem}\label{lem:contraction_V}
For $\tau,\varepsilon\in[0,1]$, $X_\Cdot\equiv(g_\Cdot,r_\Cdot,z_\Cdot,W_\Cdot)\in\sY_{\tau,\varepsilon}\cup\sY$ and $s\in(0,1]$ define
\begin{equation}
 \fV^{(1)}_{\varepsilon;s}(X_\Cdot):=U(0,r_s,z_s),
 \qquad
 \fV^{(2)}_{\varepsilon;s}(X_\Cdot):=U(\theta^2_{\varepsilon;s}/g_{s},0,0),
 \qquad
 \fV^{(3)}_{\varepsilon;s}(X_\Cdot):=W_s
\end{equation}
and
\begin{equation}
 \fV^{(0)}_{\varepsilon;s}(X_\Cdot):=
 \fV^{(1)}_{\varepsilon;s}(X_\Cdot)+
 \fV^{(2)}_{\varepsilon;s}(X_\Cdot)+
 \fV^{(3)}_{\varepsilon;s}(X_\Cdot)
 =U(\theta^2_{\varepsilon;s}/g_{s},r_s,z_s)+W_s.
\end{equation}
We omit $\varepsilon$ if $\varepsilon=0$. There exists $\lambda_\star\in(0,1]$ and $C\in(0,\infty)$ such that for all $\lambda\in(0,\lambda_\star]$, all $\tau,\varepsilon\in[0,1]$, all $X_\Cdot,Y_\Cdot\in \sY_{\tau,\varepsilon}$, $Z_\Cdot\in\sY$ and all $i\in\{0,1,2,3\}$ the following hold:
\begin{itemize}
 \item[(A)] 
 $
 \|s\mapsto \fV^{(i)}_{\varepsilon;s}(X_\Cdot)\|_{\sW^{8,4;\gamma(i)}_{\tau,\varepsilon}}
 \leq
 C$,
 \item[(B)] $\|s\mapsto (\fV^{(i)}_{\varepsilon;s}(X_\Cdot)-\fV^{(i)}_{\varepsilon;s}(Y_\Cdot))\|_{\sW^{8,4;\gamma(i)}_{\tau,\varepsilon}}
 \leq 
 C\,\|X_\Cdot-Y_\Cdot\|_{\sX_{\tau,\varepsilon}}$,

 \item[(C)] $\|s\mapsto (\fV^{(i)}_{\varepsilon;s}(X_\Cdot)-\fV^{(i)}_{\varepsilon;s}(Z_\Cdot))\|_{\tilde\sW^{2,3;\gamma(i)}_{\tau,\varepsilon}}
 \leq
 C\, \|X_\Cdot-Z_\Cdot\|_{\tilde\sX_{\tau,\varepsilon}}$,

 \item[(D)] $\|s\mapsto (\fV^{(i)}_{s}(Z_\Cdot)-\fV^{(i)}_{\varepsilon;s}(Z_\Cdot))\|_{\tilde\sW^{2,3;\gamma(i)}_{\tau,\varepsilon}}
 \leq
 C\,\lambda_\varepsilon^\kappa$,
\end{itemize}
where $\gamma(0)=\gamma(1)=1-40\kappa$, $\gamma(2)=1-20\kappa$ and $\gamma(3)=2-80\kappa$.
\end{lem}
\begin{proof}
For $i=1$ the bounds~(A),~(B),~(C),~(D) follow from Lemma~\ref{lem:contraction_U} applied with \mbox{$\gamma=\gamma(1)$}, Remark~\ref{rem:sV_sW_function} and Def.~\ref{dfn:sX}. For $i=2$ the bounds~(A),~(B),~(C),~(D) follow from Lemma~\ref{lem:contraction_U} applied with $\gamma=\gamma(2)$, Lemma~\ref{lem:contraction_g_aux}, Remark~\ref{rem:sV_sW_function} and Def.~\ref{dfn:sX}. For $i=3$ the bounds~(A),~(B),~(C) follow immediately from Def.~\ref{dfn:sX}, and (D) is trivial. For $i=0$ the bounds~(A),~(B),~(C),~(D) are consequences of the bounds~(A),~(B),~(C),~(D) with $i\in\{1,2,3\}$ and the triangle inequality.
\end{proof}

\begin{lem}\label{lem:contraction_B}
For $\tau,\varepsilon\in[0,1]$, $X_\Cdot\in\sY_{\tau,\varepsilon}\cup\sY$, $i,j\in\{0,1,2,3\}$ and $s\in(0,1]$ define
\begin{equation}
 \fG^{(i,j)}_{\varepsilon;s}(X_\Cdot)
 :=
 \fB_{\varepsilon;s}(\fV^{(i)}_{\varepsilon;s}(X_\Cdot),\fV^{(j)}_{\varepsilon;s}(X_\Cdot)),
 \qquad
 \fG_{\varepsilon;s}(X_\Cdot):=\fG^{(0,0)}_{\varepsilon;s}(X_\Cdot),
\end{equation}
where the maps $\fV^{(i)}_{\varepsilon;s}$, $i\in\{0,1,2,3\}$, were introduced in Lemma~\ref{lem:contraction_V}. We omit $\varepsilon$ if $\varepsilon=0$.  There exists $\lambda_\star\in(0,1]$ and $C\in(0,\infty)$ such that for all $\lambda\in(0,\lambda_\star]$, all $\tau,\varepsilon\in[0,1]$, all $X_\Cdot,Y_\Cdot\in \sY_{\tau,\varepsilon}$, $Z_\Cdot\in\sY$ and all $i,j\in\{0,1,2,3\}$ it holds:
\begin{itemize}
 \item[(A)] $\|s\mapsto s\,\fG^{(i,j)}_{\varepsilon;s}(X_\Cdot)\|_{\sW^{8,3;\gamma(i)+\gamma(j)}_{\tau,\varepsilon}}
 \leq
 C\,\lambda^\kappa$,
 \item[(B)] $\|s\mapsto s\,(\fG^{(i,j)}_{\varepsilon;s}(X_\Cdot)-\fG^{(i,j)}_{\varepsilon;s}(Y_\Cdot))\|_{\sW^{8,3;\gamma(i)+\gamma(j)}_{\tau,\varepsilon}}
 \leq 
 C\,\lambda^\kappa\, \|X_\Cdot-Y_\Cdot\|_{\sX_{\tau,\varepsilon}}$,

 \item[(C)] $\|s\mapsto s\,(\fG^{(i,j)}_{\varepsilon;s}(X_\Cdot)-\fG^{(i,j)}_{\varepsilon;s}(Z_\Cdot))\|_{\tilde\sW^{2,2;\gamma(i)+\gamma(j)}_{\tau,\varepsilon}}
 \leq
 C\,\lambda^\kappa\, \|X_\Cdot-Z_\Cdot\|_{\tilde\sX_{\tau,\varepsilon}}$,
 
 \item[(D)] $\|s\mapsto s\,(\fG^{(i,j)}_s(Z_\Cdot)-\fG^{(i,j)}_{\varepsilon;s}(Z_\Cdot))\|_{\tilde\sW^{2,2;\gamma(i)+\gamma(j)}_{\tau,\varepsilon}}
 \leq
 C\,\lambda_{\varepsilon}^\kappa$,
\end{itemize}
where $\gamma(0)=\gamma(1)=1-40\kappa$, $\gamma(2)=1-20\kappa$ and $\gamma(3)=2-80\kappa$.
\end{lem}
\begin{rem}
Observe that for $X_\Cdot\equiv(g_\Cdot,r_\Cdot,z_\Cdot,W_\Cdot)\in\sY_{\tau,\varepsilon}\cup\sY$ it holds
\begin{multline}
 \fB_{\varepsilon;s}(U(1/g_s,r_s,z_s)+W_s)=\fG_{\varepsilon;s}(X_\Cdot)=\fG^{(0,0)}_{\varepsilon;s}(X_\Cdot)
 \\
 =\fG^{(1,1)}_{\varepsilon;s}(X_\Cdot)+2\,\fG^{(1,2)}_{\varepsilon;s}(X_\Cdot)+2\,\fG^{(1,3)}_{\varepsilon;s}(X_\Cdot)+\fG^{(2,2)}_{\varepsilon;s}(X_\Cdot)+2\,\fG^{(2,3)}_{\varepsilon;s}(X_\Cdot)+\fG^{(3,3)}_{\varepsilon;s}(X_\Cdot).
\end{multline} 
\end{rem}
\begin{proof}
Note that $\theta_{\varepsilon;s}=1$ if $s>\varepsilon$, while $\dot H_{\varepsilon;s}=0$ and hence $\fB_{\varepsilon;s}=0$ if $s\in(0,\varepsilon]$. Therefore, inside $\fB_{\varepsilon;s}$ one may
replace $\fV^{(i)}_{\varepsilon;s}(X_\Cdot)$ by $\fV^{(i)}_{s}(X_\Cdot)$. Consequently,
\begin{equation}
 \fG^{(i,j)}_{\varepsilon;s}(X_\Cdot)=\fB_{\varepsilon;s}(\fV^{(i)}_{s}(X_\Cdot),\fV^{(j)}_{s}(X_\Cdot))
\end{equation}
and
\begin{equation}
 \fG^{(i,j)}_{s}(Z_\Cdot)-\fG^{(i,j)}_{\varepsilon;s}(Z_\Cdot)
 =\fB_{s}(\fV^{(i)}_{s}(Z_\Cdot),\fV^{(j)}_{s}(Z_\Cdot))-\fB_{\varepsilon;s}(\fV^{(i)}_{s}(Z_\Cdot),\fV^{(j)}_{s}(Z_\Cdot)).
\end{equation}
Moreover, it holds
\begin{multline}
 \fG^{(i,j)}_{\varepsilon;s}(X_\Cdot)-\fG^{(i,j)}_{\varepsilon;s}(Y_\Cdot)
 =\fB_{\varepsilon;s}(\fV^{(i)}_{\varepsilon;s}(X_\Cdot),\fV^{(j)}_{\varepsilon;s}(X_\Cdot)-\fV^{(j)}_{\varepsilon;s}(Y_\Cdot))
 \\
 +
 \fB_{\varepsilon;s}(\fV^{(j)}_{\varepsilon;s}(Y_\Cdot),\fV^{(i)}_{\varepsilon;s}(X_\Cdot)-\fV^{(i)}_{\varepsilon;s}(Y_\Cdot))
\end{multline}
and
\begin{multline}
  \fG^{(i,j)}_{\varepsilon;s}(X_\Cdot)-\fG^{(i,j)}_{\varepsilon;s}(Z_\Cdot)
 =\fB_{\varepsilon;s}(\fV^{(i)}_{\varepsilon;s}(X_\Cdot),\fV^{(j)}_{\varepsilon;s}(X_\Cdot)-\fV^{(j)}_{\varepsilon;s}(Z_\Cdot))
 \\
 +
 \fB_{\varepsilon;s}(\fV^{(j)}_{\varepsilon;s}(Z_\Cdot),\fV^{(i)}_{\varepsilon;s}(X_\Cdot)-\fV^{(i)}_{\varepsilon;s}(Z_\Cdot)).
\end{multline}
Using also the symmetry of $\fB_{\varepsilon;s}$ stated in Remark~\ref{rem:fB}, the application of Lemma~\ref{lem:B_map_bound} with the appropriate choices of $\gamma_1,\gamma_2$ yields
\begin{itemize}
 \item[(A)] $\|s\mapsto s\,\fG^{(i,j)}_{\varepsilon;s}(X_\Cdot)\|_{\sW^{8,3;\gamma(i)+\gamma(j)}_{\tau,\varepsilon}}
 \leq
 \lambda^\kappa
 \,
 \|s\mapsto \fV^{(i)}_{\varepsilon;s}(X_\Cdot)\|_{\sW^{8,4;\gamma(i)}_{\tau,\varepsilon}}\,
 \|s\mapsto \fV^{(j)}_{\varepsilon;s}(X_\Cdot)\|_{\sW^{8,4;\gamma(j)}_{\tau,\varepsilon}}$,
 \item[(B)] $\|s\mapsto s\,(\fG^{(i,j)}_{\varepsilon;s}(X_\Cdot)-\fG^{(i,j)}_{\varepsilon;s}(Y_\Cdot))\|_{\sW^{8,3;\gamma(i)+\gamma(j)}_{\tau,\varepsilon}}
 \\\leq 
 \lambda^\kappa
 \,
 \|s\mapsto \fV^{(i)}_{\varepsilon;s}(X_\Cdot)\|_{\sW^{8,4;\gamma(i)}_{\tau,\varepsilon}}
 \,
 \|s\mapsto (\fV^{(j)}_{\varepsilon;s}(X_\Cdot)-\fV^{(j)}_{\varepsilon;s}(Y_\Cdot))\|_{\sW^{8,4;\gamma(j)}_{\tau,\varepsilon}}
 \\
 +\lambda^\kappa
 \,
 \|s\mapsto \fV^{(j)}_{\varepsilon;s}(Y_\Cdot)\|_{\sW^{8,4;\gamma(j)}_{\tau,\varepsilon}}
 \,
 \|s\mapsto (\fV^{(i)}_{\varepsilon;s}(X_\Cdot)-\fV^{(i)}_{\varepsilon;s}(Y_\Cdot))\|_{\sW^{8,4;\gamma(i)}_{\tau,\varepsilon}}$,

 \item[(C)] $\|s\mapsto s\,(\fG^{(i,j)}_{\varepsilon;s}(X_\Cdot)-\fG^{(i,j)}_{\varepsilon;s}(Z_\Cdot))\|_{\tilde\sW^{2,2;\gamma(i)+\gamma(j)}_{\tau,\varepsilon}}
 \\
 \leq 
 \lambda^\kappa
 \, 
 \|s\mapsto \fV^{(i)}_{\varepsilon;s}(X_\Cdot)\|_{\sW^{2,3;\gamma(i)}_{\tau,\varepsilon}}
 \,
 \|s\mapsto (\fV^{(j)}_{\varepsilon;s}(X_\Cdot)-\fV^{(j)}_{\varepsilon;s}(Z_\Cdot))\|_{\tilde\sW^{2,3;\gamma(j)}_{\tau,\varepsilon}}
 \\
 +
 \lambda^\kappa
 \, 
 \|s\mapsto \fV^{(j)}_{\varepsilon;s}(Z_\Cdot)\|_{\sW^{2,3;\gamma(j)}_{\tau,\varepsilon}}
 \,
 \|s\mapsto (\fV^{(i)}_{\varepsilon;s}(X_\Cdot)-\fV^{(i)}_{\varepsilon;s}(Z_\Cdot))\|_{\tilde\sW^{2,3;\gamma(i)}_{\tau,\varepsilon}}$,
 
 \item[(D)] $\|s\mapsto s\,(\fG^{(i,j)}_s(Z_\Cdot)-\fG^{(i,j)}_{\varepsilon;s}(Z_\Cdot))\|_{\tilde\sW^{2,2;\gamma(i)+\gamma(j)}_{\tau,\varepsilon}}
 \\
 \leq
 \lambda_{\varepsilon}^\kappa
 \,
 \|s\mapsto \fV^{(i)}_{s}(Z_\Cdot)\|_{\sV^{4,3;\gamma(i)}}
 \,
 \|s\mapsto \fV^{(j)}_{s}(Z_\Cdot)\|_{\sV^{4,3;\gamma(j)}}$.
\end{itemize}
Moreover, for $Z_\Cdot\in\sY$ and $i\in\{0,1,2,3\}$ it holds
\begin{equation}
	\|s\mapsto \fV^{(i)}_{\varepsilon;s}(Z_\Cdot)\|_{\sW^{8,4;\gamma(i)}}
	\leq C.
\end{equation}
Indeed, for $i\in\{1,3\}$ this follows from Lemma~\ref{lem:contraction_V}~(A)
applied with $\tau=\varepsilon=0$. For $i=2$, using
$\theta_{\varepsilon;s}\leq1$, Lemma~\ref{lem:contraction_g_aux}~(A)
applied with $\tau=\varepsilon=0$ gives
\begin{equation}
	\sup_{s\in(0,1]}
	\lambda_s^{12\kappa-1}
	|\theta_{\varepsilon;s}^2/\hat g_s|
	\leq 1,
\end{equation}
and the required bound follows from Lemma~\ref{lem:contraction_U}~(A)
with $\gamma=\gamma(2)$. The case $i=0$ follows by the triangle inequality.
Next, observe that
\begin{equation}
\begin{gathered}
 \|s\mapsto \fV^{(i)}_{\varepsilon;s}(X_\Cdot)\|_{\sW^{2,3;\gamma(i)}_{\tau,\varepsilon}}
 \leq
 \|s\mapsto \fV^{(i)}_{\varepsilon;s}(X_\Cdot)\|_{\sV^{4,3;\gamma(i)}}
 \leq
 \|s\mapsto \fV^{(i)}_{\varepsilon;s}(X_\Cdot)\|_{\sW^{8,4;\gamma(i)}_{\tau,\varepsilon}},
 \\
 \|s\mapsto \fV^{(i)}_{\varepsilon;s}(Z_\Cdot)\|_{\sW^{2,3;\gamma(i)}_{\tau,\varepsilon}}
 \leq
 \|s\mapsto \fV^{(i)}_{\varepsilon;s}(Z_\Cdot)\|_{\sV^{4,3;\gamma(i)}}
 \leq
 \|s\mapsto \fV^{(i)}_{\varepsilon;s}(Z_\Cdot)\|_{\sW^{8,4;\gamma(i)}},
 \\
 \|s\mapsto \fV^{(i)}_{\varepsilon;s}(Z_\Cdot)\|_{\sV^{4,3;\gamma(i)}}
 \leq
 \|s\mapsto \fV^{(i)}_{\varepsilon;s}(Z_\Cdot)\|_{\sV^{8,4;\gamma(i)}}
 \leq
 \|s\mapsto \fV^{(i)}_{\varepsilon;s}(Z_\Cdot)\|_{\sW^{8,4;\gamma(i)}}
\end{gathered}
\end{equation}
by Remarks~\ref{rem:sV_norm_property},~\ref{rem:sV_sW} and~\ref{rem:sW_sV}. The statement of the lemma is now an immediate consequence of Lemma~\ref{lem:contraction_V}.
\end{proof}

\begin{lem}\label{lem:contraction_W}
There exists $\lambda_\star\in(0,1]$ and $C\in(0,\infty)$ such that for all $\lambda\in(0,\lambda_\star]$, all $\tau,\varepsilon\in[0,1]$ and all $X_\Cdot,Y_\Cdot\in \sY_{\tau,\varepsilon}$, $Z_\Cdot\in\sY$ it holds:
\begin{itemize}
 \item[(A)] $\|s\mapsto \fW_{\tau,\varepsilon;s}(X_\Cdot)\|_{\sW^{8,4;2-80\kappa}_{\tau,\varepsilon}}
 \leq
 C\,\lambda^\kappa$,
 \item[(B)] $\|s\mapsto (\fW_{\tau,\varepsilon;s}(X_\Cdot)-\fW_{\tau,\varepsilon;s}(Y_\Cdot))\|_{\sW^{8,4;2-80\kappa}_{\tau,\varepsilon}}
 \leq 
 C\,\lambda^\kappa\, \|X_\Cdot-Y_\Cdot\|_{\sX_{\tau,\varepsilon}}$,

 \item[(C)] $\|s\mapsto (\fW_{\tau,\varepsilon;s}(X_\Cdot)-\fW_{\tau,\varepsilon;s}(Z_\Cdot))\|_{\tilde\sW^{2,3;2-80\kappa}_{\tau,\varepsilon}}
 \leq
 C\,\lambda^\kappa\, \|X_\Cdot-Z_\Cdot\|_{\tilde\sX_{\tau,\varepsilon}}$,
 
 \item[(D)] $\|s\mapsto (\fW_s(Z_\Cdot)-\fW_{\tau,\varepsilon;s}(Z_\Cdot))\|_{\tilde\sW^{2,3;2-80\kappa}_{\tau,\varepsilon}}
 \leq
 C\,\lambda_{\tau\vee\varepsilon}^\kappa$.
\end{itemize}
\end{lem}
\begin{proof}
For $\tau,\varepsilon\in[0,1]$, $X_\Cdot\equiv(g_\Cdot,r_\Cdot,z_\Cdot,W_\Cdot)\in\sX_{\tau,\varepsilon}\cup\sX$ and $s\in(0,1]$ we define $$\fH_{\tau,\varepsilon;s}(X_\Cdot):=(\fH^m_{\tau,\varepsilon;s}(X_\Cdot))_{m\in\bN_+}\in\sN$$ by the equality
\begin{equation}
 \fH_{\tau,\varepsilon;s}(X_\Cdot):=W_s-\fE\fA_{\tau,\varepsilon;1,s}W_s=-\fE\fC_{\tau,\varepsilon;1,s}W_s.
\end{equation}
Recalling that the maps $\fI_{\tau,\varepsilon;s}$ and $\fG_{\varepsilon;s}$ were defined in Lemmas~\ref{lem:I_bound} and~\ref{lem:contraction_B}, respectively, and using the notation introduced in Def.~\ref{dfn:sN} we obtain
\begin{equation}
 \fW_{\tau,\varepsilon;s}(X_\Cdot)
 :=
 \fW^{(1)}_{\tau,\varepsilon;s}(X_\Cdot)
 +
 \fW^{(2)}_{\tau,\varepsilon;s}(X_\Cdot),
\end{equation}
where
\begin{equation}
 \fW^{(1)}_{\tau,\varepsilon;s}(X_\Cdot):=\fI_{\tau,\varepsilon;s}(u\mapsto u\,\fG_{\varepsilon;u}(X_\Cdot)),
 \qquad
 \fW^{(2)}_{\tau,\varepsilon;s}(X_\Cdot):=(\Pi_2+\Pi_4)\fH_{\tau,\varepsilon;s}(X_\Cdot).
\end{equation}
Observe that
\begin{equation}
\begin{gathered}
 \fW^{(1)}_{\tau,\varepsilon;s}(X_\Cdot)-\fW^{(1)}_{\tau,\varepsilon;s}(Y_\Cdot)
 =
 \fI_{\tau,\varepsilon;s}(u\mapsto u\,(\fG_{\varepsilon;u}(X_\Cdot)-\fG_{\varepsilon;u}(Y_\Cdot))),
 \\
 \fW^{(1)}_{\tau,\varepsilon;s}(X_\Cdot)-\fW^{(1)}_{\tau,\varepsilon;s}(Z_\Cdot)
 =
 \fI_{\tau,\varepsilon;s}(u\mapsto u\,(\fG_{\varepsilon;u}(X_\Cdot)-\fG_{\varepsilon;u}(Z_\Cdot))),
 \\
 \fW^{(1)}_s(Z_\Cdot)-\fW^{(1)}_{\tau,\varepsilon;s}(Z_\Cdot)
 =
 (\fI_{s}-\fI_{\tau,\varepsilon;s})(u\mapsto u\,\fG_{u}(Z_\Cdot))
 +
 \fI_{\tau,\varepsilon;s}(u\mapsto u\,(\fG_{u}(Z_\Cdot)-\fG_{\varepsilon;u}(Z_\Cdot))).
\end{gathered}
\end{equation}
It follows now from Lemma~\ref{lem:I_bound} and Lemma~\ref{lem:contraction_B} that the bounds~(A),~(B),~(C),~(D) with $\fW_{\tau,\varepsilon;s}(X_\Cdot)$ replaced by $\fW^{(1)}_{\tau,\varepsilon;s}(X_\Cdot)$ hold true. Consequently, by the triangle inequality, to complete the proof of the lemma it suffices to establish the bounds~(A),~(B),~(C),~(D) with $\fW_{\tau,\varepsilon;s}(X_\Cdot)$ replaced by $\fW^{(2)}_{\tau,\varepsilon;s}(X_\Cdot)$. By Remarks~\ref{rem:sW_sV} and~\ref{rem:sV_norm_estimate} the latter bounds are implied by the following bounds:
\begin{itemize}
 \item[(A${}_1$)] $\|s\mapsto \fH_{\tau,\varepsilon;s}(X_\Cdot)\|_{\sV^{4,3;2-80\kappa}}
 \leq
 C\,\lambda^\kappa$,
 \item[(B${}_1$)] $\|s\mapsto (\fH_{\tau,\varepsilon;s}(X_\Cdot)-\fH_{\tau,\varepsilon;s}(Y_\Cdot))\|_{\sV^{4,3;2-80\kappa}}
 \leq 
 C\,\lambda^\kappa\, \|X_\Cdot-Y_\Cdot\|_{\sX_{\tau,\varepsilon}}$,

 \item[(C${}_1$)] $\|s\mapsto (\fH_{\tau,\varepsilon;s}(X_\Cdot)-\fH_{\tau,\varepsilon;s}(Z_\Cdot))\|_{\tilde\sV^{1,2;2-80\kappa}}
 \leq
 C\,\lambda^\kappa\, \|X_\Cdot-Z_\Cdot\|_{\tilde\sX_{\tau,\varepsilon}}$,
 
 \item[(D${}_1$)] $\|s\mapsto (\fH_s(Z_\Cdot)-\fH_{\tau,\varepsilon;s}(Z_\Cdot))\|_{\tilde\sV^{1,2;2-80\kappa}}
 \leq
 C\,\lambda_{\tau\vee\varepsilon}^\kappa$.
\end{itemize}
The above bounds follow from Lemma~\ref{lem:A_map_bound}~(B),~(D),~(E), Remark~\ref{rem:fM_fE} and Def.~\ref{dfn:sX}. This finishes the proof.
\end{proof}

\begin{lem}\label{lem:contraction_c}
Let $\gamma\in\bR$. For $\tau,\varepsilon\in[0,1]$ and a map $W_\Cdot$ such that $s\mapsto sW_s\in\sV^{8,3;\gamma}$ define, for $s\in(0,1]$,
\begin{equation}
\begin{aligned}
 \mathbf{c}^{(m,0)}_{\tau,\varepsilon;s}(W_\Cdot)
 &:=
 \fL\fE\fA^m_{\tau,\varepsilon;1,s}W_s,\qquad m\in\{2,4\},
 \\
 \mathbf{c}^{(2,1)}_{\tau,\varepsilon;s}(W_\Cdot)
 &:=
 \fL_\partial\fE\fA^2_{\tau,\varepsilon;1,s}W_s. 
\end{aligned}
\end{equation}
There exists $\lambda_\star\in(0,1]$ and $C\in(0,\infty)$ such that for all $\lambda\in(0,\lambda_\star]$, all $\tau,\varepsilon\in[0,1]$, $t\in(0,1]$, and all $(m,i)\in\{(2,0),(4,0),(2,1)\}$ it holds:
\begin{itemize}
 \item[(A)] $\lambda_t^{-\gamma-2\kappa m}\,t^{3-m/2-i}\,|\mathbf{c}^{(m,i)}_{\tau,\varepsilon;t}(W_\Cdot)|
 \leq C\, \|s\mapsto s\,W_s\|_{\sW_{\tau,\varepsilon}^{8,3;\gamma}}$,
 \item[(B)] $\lambda_t^{-\gamma-2\kappa m}\,t^{3-m/2-i}\,|\mathbf{c}^{(m,i)}_{\tau,\varepsilon;t}(W_\Cdot)-\mathbf{c}^{(m,i)}_{\tau,\varepsilon;t}(\tilde W_\Cdot)|
 \leq 
 C\, \|s\mapsto s\,(W_s-\tilde W_s)\|_{\sW_{\tau,\varepsilon}^{8,3;\gamma}}$.

 \item[(C)] $\lambda_t^{-\gamma-2\kappa m}\,t^{3-m/2-i}\,|\mathbf{c}^{(m,i)}_{\tau,\varepsilon;t}(W_\Cdot)-\mathbf{c}^{(m,i)}_{\tau,\varepsilon;t}(\hat W_\Cdot)|
 \leq 
 C\, \|s\mapsto s\,(W_s-\hat W_s)\|_{\tilde\sW_{\tau,\varepsilon}^{2,2;\gamma}}$,
 
 \item[(D)] $\lambda_t^{-\gamma-2\kappa m}\,t^{3-m/2-i}\,|\mathbf{c}^{(m,i)}_{t}(\hat W_\Cdot)-\mathbf{c}^{(m,i)}_{\tau,\varepsilon;t}(\hat W_\Cdot)|
 \leq
 C\,\lambda_{\tau\vee\varepsilon}^\kappa\, \|s\mapsto s\,\hat W_s\|_{\sW^{8,3;\gamma}}$,
\end{itemize}
for all $W_\Cdot,\tilde W_\cdot,\hat W_\Cdot$ such that the RHS of the above bounds are finite.
\end{lem}
\begin{rem}
Actually, since the map $\mathbf{c}^{(m,i)}_{\tau,\varepsilon;t}$ is linear, Item~(B) of the above lemma follows immediately from Item~(A).
\end{rem}

\begin{proof}
By Remark~\ref{rem:sV_norm_estimate} and Lemma~\ref{lem:A_map_bound}~(A),~(C),~(E) there exists $C\in(0,\infty)$ such that
\begin{itemize}
 \item[(A${}_1$)] $\|s\mapsto s\,\fA^m_{\tau,\varepsilon;1,s}W_s\|_{\sV^{m;\gamma}}
 \leq 
 C\, \|s\mapsto s\,W_s\|_{\sW_{\tau,\varepsilon}^{8,3;\gamma}}$,
 \item[(B${}_1$)] $\|s\mapsto s\,(\fA^m_{\tau,\varepsilon;1,s}W_s-\fA^m_{\tau,\varepsilon;1,s}\tilde W_s)\|_{\sV^{m;\gamma}}
 \leq 
 C\, \|s\mapsto s\,(W_s-\tilde W_s)\|_{\sW_{\tau,\varepsilon}^{8,3;\gamma}}$,

 \item[(C${}_1$)] $\|s\mapsto s\,(\fA^m_{\tau,\varepsilon;1,s}W_s-\fA^m_{\tau,\varepsilon;1,s}\hat W_s)\|_{\tilde\sV^{m;\gamma}}
 \leq
 C\, \|s\mapsto s\,(W_s-\hat W_s)\|_{\tilde\sW_{\tau,\varepsilon}^{2,2;\gamma}}$,
 
 \item[(D${}_1$)] $\|s\mapsto s\,(\fA^m_{1,s}\hat W_s-\fA^m_{\tau,\varepsilon;1,s}\hat W_s)\|_{\tilde\sV^{m;\gamma}}
 \leq
 C\,\lambda_{\tau\vee\varepsilon}^\kappa\, \|s\mapsto s\,\hat W_s\|_{\sW^{8,3;\gamma}}$
\end{itemize}
for all $\tau,\varepsilon\in[0,1]$ and $m\in\{2,4\}$. The bounds~(A), (B), (C), (D) follow now from Def.~\ref{dfn:V_space} of the norm $\|\Cdot\|_{\sV^{m;\gamma}}$, Remark~\ref{rem:fM_fE} and Lemma~\ref{lem:fL_continuity}. This completes the proof.
\end{proof}

\begin{lem}\label{lem:contraction_a}
For $\tau,\varepsilon\in[0,1]$, $X_\Cdot\in\sY_{\tau,\varepsilon}\cup\sY$ and $s\in(0,1]$ define
\begin{equation}
\begin{aligned}
 \mathbf{a}^{(1)}_{\tau,\varepsilon;s}(X_\Cdot)
 &:=
 \fL\fE\fA^4_{\tau,\varepsilon;1,s}\fB_{\varepsilon;s}(U(1/g_s,0,0)),
 \\
 \mathbf{a}^{(2)}_{\tau,\varepsilon;s}(X_\Cdot)
 &:=
 \fL\fE\fA^4_{\tau,\varepsilon;1,s}\fB_{\varepsilon;s}(U(2/g_s,r_s,z_s),U(0,r_s,z_s)),
 \\
 \mathbf{a}^{(3)}_{\tau,\varepsilon;s}(X_\Cdot)
 &:=
 \fL\fE\fA^4_{\tau,\varepsilon;1,s}\fB_{\varepsilon;s}(2U(1/g_s,r_s,z_s)+W_s,W_s)
\end{aligned}
\end{equation}
and
\begin{equation}
\begin{aligned}
 \mathbf{a}^{(4)}_{\tau,\varepsilon;s}(X_\Cdot)
 &:=
 \fL_\partial\fE\fA^2_{\tau,\varepsilon;1,s}\fB_{\varepsilon;s}(U(1/g_s,0,0)),
 \\
 \mathbf{a}^{(5)}_{\tau,\varepsilon;s}(X_\Cdot)
 &:=
 \fL_\partial\fE\fA^2_{\tau,\varepsilon;1,s}\fB_{\varepsilon;s}(U(2/g_s,r_s,z_s),U(0,r_s,z_s)),
 \\
 \mathbf{a}^{(6)}_{\tau,\varepsilon;s}(X_\Cdot)
 &:=
 \fL_\partial\fE\fA^2_{\tau,\varepsilon;1,s}\fB_{\varepsilon;s}(2U(1/g_s,r_s,z_s)+W_s,W_s).
\end{aligned}
\end{equation}
There exists $\lambda_\star\in(0,1]$ and $C\in(0,\infty)$ such that for all $\lambda\in(0,\lambda_\star]$, all $\tau,\varepsilon\in[0,1]$, all $t\in(0,1]$, all \mbox{$X_\Cdot,Y_\Cdot\in \sY_{\tau,\varepsilon}$}, $Z_\Cdot\in\sY$ and all \mbox{$i\in\{1,\ldots,6\}$} it holds:
\begin{itemize}
 \item[(A)] $\lambda_t^{-\gamma(i)}\,t\,|\mathbf{a}^{(i)}_{\tau,\varepsilon;t}(X_\Cdot)|
 \leq C\, \lambda^\kappa$,
 \item[(B)] $\lambda_t^{-\gamma(i)}\,t\,|\mathbf{a}^{(i)}_{\tau,\varepsilon;t}(X_\Cdot)-\mathbf{a}^{(i)}_{\tau,\varepsilon;t}(Y_\Cdot)|
 \leq
 C\,\lambda^\kappa\, \|X_\Cdot-Y_\Cdot\|_{\sX_{\tau,\varepsilon}}$,

 \item[(C)] $\lambda_t^{-\gamma(i)}\,t\,|\mathbf{a}^{(i)}_{\tau,\varepsilon;t}(X_\Cdot)-\mathbf{a}^{(i)}_{\tau,\varepsilon;t}(Z_\Cdot)|
 \leq
 C\,\lambda^\kappa\, \|X_\Cdot-Z_\Cdot\|_{\tilde\sX_{\tau,\varepsilon}}$,

 \item[(D)] $\lambda_t^{-\gamma(i)}\,t\,|\mathbf{a}^{(i)}_{t}(Z_\Cdot)-\mathbf{a}^{(i)}_{\tau,\varepsilon;t}(Z_\Cdot)|
 \leq
 C\,\lambda_{\tau\vee\varepsilon}^\kappa$,

 \item[(E)] $\mathbf{a}^{(2)}_{\tau,\varepsilon;t}(X_\Cdot)
 =
 \mathbf{a}^{(5)}_{\tau,\varepsilon;t}(X_\Cdot)
 =0$ for all $t\in(0,1/2)$.
\end{itemize}
Here $\gamma(1)=2-32\kappa$, $\gamma(2)=2-72\kappa$, $\gamma(3)=3-112\kappa$, $\gamma(4)=2-36\kappa$, $\gamma(5)=2-76\kappa$, $\gamma(6)=3-116\kappa$.
\end{lem}

\begin{rem}
Observe that it holds
\begin{equation}\label{eq:contraction_a_decomposition}
\begin{aligned}
 (g_s)^2\fL\fE\fA^4_{\tau,\varepsilon;1,s}\fB_{\varepsilon;s}(V_s)
 &=
 (g_s)^2\mathbf{a}^{(1)}_{\tau,\varepsilon;s}(X_\Cdot)
 +
 (g_s)^2\mathbf{a}^{(2)}_{\tau,\varepsilon;s}(X_\Cdot)
 +
 (g_s)^2\mathbf{a}^{(3)}_{\tau,\varepsilon;s}(X_\Cdot),
 \\
 \fL_\partial\fE\fA^2_{\tau,\varepsilon;1,s}\fB_{\varepsilon;s}(V_s)
 &=
 \mathbf{a}^{(4)}_{\tau,\varepsilon;s}(X_\Cdot)
 +
 \mathbf{a}^{(5)}_{\tau,\varepsilon;s}(X_\Cdot)
 +
 \mathbf{a}^{(6)}_{\tau,\varepsilon;s}(X_\Cdot),
\end{aligned}
\end{equation}
where $V_s=U(1/g_s,r_s,z_s)+W_s$ and $X_\Cdot\equiv(g_\Cdot,r_\Cdot,z_\Cdot,W_\Cdot)$. Note that the expressions
\begin{equation}
 (g_s)^2\fL\fE\fA^4_{\tau,\varepsilon;1,s}\fB_{\varepsilon;s}(V_s),
 \qquad
 \fL\fE\fA^2_{\tau,\varepsilon;1,s}\fB_{\varepsilon;s}(V_s),
 \qquad
 \fL_\partial\fE\fA^2_{\tau,\varepsilon;1,s}\fB_{\varepsilon;s}(V_s)
\end{equation}
appear in Def.~\ref{dfn:X_map} of the maps $\mathbf{g}_{\tau,\varepsilon;\Cdot}$, $\mathbf{r}_{\tau,\varepsilon;\Cdot}$ and $\mathbf{z}_{\tau,\varepsilon;\Cdot}$. The second of the above expressions will be estimated directly in Lemma~\ref{lem:contraction_r} using Lemmas~\ref{lem:contraction_c} and~\ref{lem:contraction_B}. One could also estimate directly the remaining two expressions without using the decompositions~\eqref{eq:contraction_a_decomposition}. However, estimates obtained in that way are not strong enough for our purposes, since they do not exploit the support property stated in Lemma~\ref{lem:contraction_a}(E). This property will be essential when integrating the terms $\mathbf{a}^{(2)}_{\tau,\varepsilon;s}(X_\Cdot)$ and $\mathbf{a}^{(5)}_{\tau,\varepsilon;s}(X_\Cdot)$ below. Let us also mention that we will not use at all the bounds for $\mathbf{a}^{(1)}_{\tau,\varepsilon;s}(X_\Cdot)$ stated in the above lemma. Instead we will estimate directly $(g_s)^2\mathbf{a}^{(1)}_{\tau,\varepsilon;s}(X_\Cdot)=-\partial_s f_{\tau,\varepsilon;s}$, where $f_{\tau,\varepsilon;s}$ is introduced in Def.~\ref{dfn:f_function}. Note that $f_{\tau,\varepsilon;s}$ does not depend on $X_\Cdot$.
\end{rem}
\begin{proof}
First observe that for $X_\Cdot\equiv(g_\Cdot,r_\Cdot,z_\Cdot,W_\Cdot)\in\sY_{\tau,\varepsilon}\cup\sY$ it holds
\begin{equation}
\begin{aligned}
 \fB_{\varepsilon;s}(U(1/g_s,0,0))
 &=
 \fG^{(2,2)}_{\varepsilon;s}(X_\Cdot),
 \\
 \fB_{\varepsilon;s}(U(2/g_s,r_s,z_s),U(0,r_s,z_s))
 &=
 \fG^{(1,1)}_{\varepsilon;s}(X_\Cdot)+2\,\fG^{(1,2)}_{\varepsilon;s}(X_\Cdot),
 \\
 \fB_{\varepsilon;s}(2U(1/g_s,r_s,z_s)+W_s,W_s)
 &=
 2\,\fG^{(1,3)}_{\varepsilon;s}(X_\Cdot)
 +
 2\,\fG^{(2,3)}_{\varepsilon;s}(X_\Cdot)
 +
 \fG^{(3,3)}_{\varepsilon;s}(X_\Cdot),
\end{aligned}
\end{equation}
where the maps $\fG^{(i,j)}_{\varepsilon;s}$ were introduced in Lemma~\ref{lem:contraction_B}. Using the fact that $\int_{\bR^2}\dot G_{\varepsilon;s}(x)\,\rd x=0$ for $s\in(0,1/2)$ by Remark~\ref{rem:G}~(D) as well as Def.~\ref{dfn:A_map},~\ref{dfn:B_map} and~\ref{dfn:U_functional} one shows that
\begin{equation}
\begin{gathered}
 \fL\fE\fA^4_{\tau,\varepsilon;1,s}\fB_{\varepsilon;s}(U(2/g_s,r_s,z_s),U(0,r_s,z_s))=0,
 \\
 \fL_\partial\fE\fA^2_{\tau,\varepsilon;1,s}\fB_{\varepsilon;s}(U(2/g_s,r_s,z_s),U(0,r_s,z_s))=0
\end{gathered} 
\end{equation}
for $s\in(0,1/2)$. Recall that $\theta_{1/2;t}=\lambda_t/\lambda_{1/2\vee t}$. Since $\theta_{1/2;s}=1$ for $s\in[1/2,1]$ we obtain
\begin{equation}\label{eq:a2_identity}
 \mathbf{a}^{(2)}_{\tau,\varepsilon;s}(X_\Cdot)
 =
 \theta_{1/2;s}^{1-40\kappa}\,\fL\fE\fA^4_{\tau,\varepsilon;1,s}\fB_{\varepsilon;s}(U(2/g_s,r_s,z_s),U(0,r_s,z_s))
\end{equation} 
and
\begin{equation}\label{eq:a5_identity}
 \mathbf{a}^{(5)}_{\tau,\varepsilon;s}(X_\Cdot)
 =
 \theta_{1/2;s}^{1-40\kappa}\,\fL_\partial\fE\fA^2_{\tau,\varepsilon;1,s}\fB_{\varepsilon;s}(U(2/g_s,r_s,z_s),U(0,r_s,z_s)).
\end{equation}
To prove the bounds~(A), (B), (C), (D) in the cases $i\in\{1,\ldots,6\}$ we use the following arguments, respectively.
\begin{itemize}
 \item[(1)] We apply Lemma~\ref{lem:contraction_c} with $m=4$, $i=0$ and $\gamma=2-40\kappa$ and Lemma~\ref{lem:contraction_B} with $(i,j)=(2,2)$. 
 
 \item[(2)] We apply Lemma~\ref{lem:contraction_c} with $m=4$, $i=0$ and $\gamma=2-80\kappa$ and Lemma~\ref{lem:contraction_B} with $(i,j)\in\{(1,1),(1,2)\}$ and use the identity~\eqref{eq:a2_identity}.
 
 \item[(3)] We apply Lemma~\ref{lem:contraction_c} with $m=4$, $i=0$ and $\gamma=3-120\kappa$ and Lemma~\ref{lem:contraction_B} with $(i,j)\in\{(1,3),(2,3),(3,3)\}$.
 
 \item[(4)] We apply Lemma~\ref{lem:contraction_c} with $m=2$, $i=1$ and $\gamma=2-40\kappa$ and Lemma~\ref{lem:contraction_B} with $(i,j)=(2,2)$. 
 
 \item[(5)] We apply Lemma~\ref{lem:contraction_c} with $m=2$, $i=1$ and $\gamma=2-80\kappa$ and Lemma~\ref{lem:contraction_B} with $(i,j)\in\{(1,1),(1,2)\}$ and use the identity~\eqref{eq:a5_identity}.
 
 \item[(6)] We apply Lemma~\ref{lem:contraction_c} with $m=2$, $i=1$ and $\gamma=3-120\kappa$ and Lemma~\ref{lem:contraction_B} with $(i,j)\in\{(1,3),(2,3),(3,3)\}$.
\end{itemize}
For the bound~(D), in each of the above expressions we add and subtract the corresponding term with $\mathbf c_{\tau,\varepsilon;t}^{(4,0)}$ or $\mathbf c_{\tau,\varepsilon;t}^{(2,1)}$ applied to the
unregularized combination of the maps $\fG_t^{(i,j)}(Z_\Cdot)$. We then use Lemma~\ref{lem:contraction_c}~(D),~(C) and Lemma~\ref{lem:contraction_B}~(A),~(D), respectively. This finishes the proof.
\end{proof}

\begin{dfn}\label{dfn:f_function}
For $\tau,\varepsilon\in[0,1]$ and $t\in(0,1]$ we define
\begin{equation}
 f_{\tau,\varepsilon;t}:=
 \lambda^{-1}+
 \int_t^1\fL\fE\fA^4_{\tau,\varepsilon;1,s}\fB_{\varepsilon;s}(U(1,0,0))\,\rd s,
 \qquad
 h_{\tau,\varepsilon;t}:=\fL\fE\fA^2_{\tau,\varepsilon;1,t} U(1,0,0).
\end{equation}
\end{dfn}
\begin{rem}
For $X_\Cdot \equiv (g_\Cdot,r_\Cdot,z_\Cdot,W_\Cdot)\in\sY_{\tau,\varepsilon}$ it holds
\begin{equation}
 \mathbf{g}_{\tau,\varepsilon;t}(X_\Cdot)
 =
 f_{\tau,\varepsilon;t}+\int_t^1 (g_s)^2\,(\fL\fE\fA^4_{\tau,\varepsilon;1,s}\fB_{\varepsilon;s}(V_s)-\fL\fE\fA^4_{\tau,\varepsilon;1,s}\fB_{\varepsilon;s}(U(1/g_s,0,0)))\,\rd s
\end{equation}
and
\begin{equation}
 \mathbf{r}_{\tau,\varepsilon;t}(X_\Cdot) 
 =
 -h_{\tau,\varepsilon;t}/g_t
 -
 \int_t^1
 \fL\fE\fA^2_{\tau,\varepsilon;1,s}\fB_{\varepsilon;s}(V_s(X_\Cdot))\,\rd s,
\end{equation}
where $V_s(X_\Cdot)= U(1/g_s,r_s,z_s)+W_s$.
\end{rem}

\begin{lem}\label{lem:contraction_f_h}
There exists $\lambda_\star\in(0,1]$ and $C\in(0,\infty)$ such that for all $\lambda\in(0,\lambda_\star]$ it holds
\begin{itemize}
 \item[(A)] $\lambda_{\varepsilon\vee t}^{-1}/C\leq f_{\tau,\varepsilon;t}\leq C\,\lambda_t^{-1}$,
 \item[(B)] $|f_t-f_{\tau,\varepsilon;t}|\leq C\, \lambda_{\varepsilon\vee\tau}^\kappa\,\lambda_t^{-1-8\kappa}$,
 \item[(C)] $|h_{\tau,\varepsilon;t}|\leq C\,\lambda_t^{-4\kappa}\,t^{-1}$,
 \item[(D)] $|h_t-h_{\tau,\varepsilon;t}|\leq C\, \lambda_{\varepsilon\vee\tau}^\kappa\,\lambda_t^{-4\kappa}\,t^{-1}$
\end{itemize}
for all $\tau,\varepsilon\in[0,1]$ and $t\in(0,1]$.
\end{lem}

\begin{proof}
Let us first note that $\fA^4_{\tau,\varepsilon;1,s}\fB_{\varepsilon;s}(U(1,0,0))$ coincides with the RHS of~\eqref{eq:fA_map_def_alt} with $V=\fB_{\varepsilon;s}(U(1,0,0))$. Since $\fB^{4+k}_{\varepsilon;s}(U(1,0,0))=0$ for all $k\in\bN_0$ except $k=2$ in the case at hand only the term $k=2$ contributes to the sum over $k\in\bN_0$ on the RHS of~\eqref{eq:fA_map_def_alt}. Consequently, taking into account~\eqref{eq:covariance_Psi_H} and Def.~\ref{dfn:B_map} we conclude that $\fL\fE\fA^4_{\tau,\varepsilon;1,s}\fB_{\varepsilon;s}(U(1,0,0))$ can be represented in terms of the propagators $G_{\tau,\varepsilon;s}$ and $\dot G_{\varepsilon;s}$. More specifically, $\fL\fE\fA^4_{\tau,\varepsilon;1,s}\fB_{\varepsilon;s}(U(1,0,0))$ is a~linear combination of one-loop Feynman diagrams with two quartic vertices, one of which is integrated over $\bR^2$. The following types of diagrams appear: (1) diagrams with the vertices connected by two edges representing the propagators $\dot G_{\varepsilon;s}$ and $G_{\tau,\varepsilon;s}$, respectively, and (2) diagrams with the vertices connected by a single edge representing the propagator $\dot G_{\varepsilon;s}$ and with a self-contraction of one of the vertices by the propagator $G_{\tau,\varepsilon;s}$. Note that for $s\in(0,1/2)$ the diagrams of the type~(2) vanish identically since $\int_{\bR^2}\dot G_{\varepsilon;s}(x)\,\rd x=0$ for $s\in(0,1/2)$ by Remark~\ref{rem:G}~(D). For $s\in[1/2,1]$ both of the above-mentioned diagrams contribute to $\fL\fE\fA^4_{\tau,\varepsilon;1,s}\fB_{\varepsilon;s}(U(1,0,0))$. However, in this regime both propagators $\dot G_{\varepsilon;s}$ and $G_{\tau,\varepsilon;s}$ have UV cutoffs of order one and by elementary estimates $\fL\fE\fA^4_{\tau,\varepsilon;1,s}\fB_{\varepsilon;s}(U(1,0,0))$ is uniformly bounded for all $\tau,\varepsilon\in[0,1]$ and $s\in[1/2,1]$. Thus, we can restrict attention to the diagrams of the type~(1) and $s\in(0,1/2]$.

Set $c_N:=(N-2)/(N-1)$ and
\begin{equation}
 F_{\varepsilon,s}(p):=
 \frac{\vartheta(2\varepsilon\omega(p))^2
 \vartheta(s\omega(p))\dot\vartheta(s\omega(p))
 \omega(p)(c_N-|p|^2)}
 {(1+|p|^2)^2}.
\end{equation}
Taking into account all possible contractions of the fields represented by edges of the diagrams of the type~(1) we obtain that for $s\in(0,1/2]$ and $\tau\in(0,1]$ it holds
\begin{equation}
 \fL\fE\fA^4_{\tau,\varepsilon;1,s}\fB_{\varepsilon;s}(U(1,0,0))
 =
 8(N-1)\, \tau^2\!
 \sum_{p\in(2\pi\tau\bZ)^2}\!
 F_{\varepsilon,s}(p)
\end{equation}
whereas for $s\in(0,1/2]$ and $\tau=0$ (recall that we omit $\tau$ in the notation if $\tau=0$) it holds
\begin{equation}
 \fL\fE\fA^4_{\varepsilon;1,s}\fB_{\varepsilon;s}(U(1,0,0))
 =
 \frac{8(N-1)}{(2\pi)^2}
 \int_{\bR^2}
 F_{\varepsilon,s}(p)\,\rd p.
\end{equation}
Since $\dot\vartheta(s\omega(p))\neq0$ only if $\omega(p)\asymp s^{-1}$,
and on the support of $F_{\varepsilon,s}$ one also has
$\varepsilon\lesssim s$, direct differentiation gives
\begin{equation}
 \|\partial_jF_{\varepsilon,s}\|_{L^1(\bR^2)}\leq C,
 \qquad
 \|\partial_1\partial_2F_{\varepsilon,s}\|_{L^1(\bR^2)}\leq Cs.
\end{equation}
Applying the Riemann-sum estimate yields
\begin{equation}
 \left|
 \tau^2\sum_{p\in(2\pi\tau\bZ)^2}F_{\varepsilon,s}(p)
 -
 \frac{1}{(2\pi)^2}\int_{\bR^2}F_{\varepsilon,s}(p)\,\rd p
 \right|
 \leq C(\tau+\tau^2s)\leq C.
\end{equation}
Consequently,
\begin{equation}
 |\fL\fE\fA^4_{\varepsilon;1,s}\fB_{\varepsilon;s}(U(1,0,0))
 -
 \fL\fE\fA^4_{\tau,\varepsilon;1,s}\fB_{\varepsilon;s}(U(1,0,0))|
 \leq C
\end{equation}
for all $\tau,\varepsilon\in[0,1]$ and $s\in(0,1/2]$. This and the first paragraph of the proof imply that there exists $C\in(0,\infty)$ such that
\begin{equation}\label{eq:f_proof_bound_tau}
 |f_{\tau,\varepsilon;t}-f_{\varepsilon;t}|\leq C
\end{equation}
for all $\tau,\varepsilon\in[0,1]$ and $t\in(0,1]$. By the Lebesgue dominated convergence theorem it holds
\begin{equation}
 \lim_{s\searrow0} s\,\fL\fE\fA^4_{1,s}\fB_{s}(U(1,0,0))
 =
 -\frac{\beta_2}{\pi}
 \int_{\bR^2}
 \frac{\vartheta(|p|)\,\dot\vartheta(|p|)}{|p|}\,\rd p
 =
 -2\beta_2
 \int_0^\infty
 \vartheta(|p|)\,\dot\vartheta(|p|)\,\rd |p|=\beta_2.
\end{equation}
Next, observe that for $s\in(4\varepsilon,1/2]$,
\begin{multline}
 \fL\fE\fA^4_{1,s}\fB_{\varepsilon;s}(U(1,0,0))-\beta_2/s
 \\
 =
 \frac{\beta_2}{\pi s} \int_{\bR^2}
 \bigg(
 \frac{\vartheta((s^2+|p|^2)^{1/2})\,\dot\vartheta((s^2+|p|^2)^{1/2})\,(c_Ns^2-|p|^2)}{(s^2+|p|^2)^{3/2}}
 +
 \frac{\vartheta(|p|)\,\dot\vartheta(|p|)}{|p|}
 \bigg)\,\rd p.
\end{multline}
Hence, we conclude that there exists $C\in(0,\infty)$ such that
\begin{equation}
 |\fL\fE\fA^4_{\varepsilon;1,s}\fB_{\varepsilon;s}(U(1,0,0))-\beta_2/s|\leq C
\end{equation}
for all $\varepsilon\in[0,1]$ and $s\in(4\varepsilon,1]$. Consequently,
\begin{equation}\label{eq:f_proof_bound_t}
 |f_{\varepsilon;t}-\lambda^{-1}_t|\leq
 \int_t^1|\fL\fE\fA^4_{\varepsilon;1,s}\fB_{\varepsilon;s}(U(1,0,0))-\beta_2/s|\leq C
\end{equation}
for all $\varepsilon\in[0,1]$ and $t\in[4\varepsilon,1]$. Using the bound
\begin{multline}
 |\fL\fE\fA^4_{\varepsilon;1,s}\fB_{\varepsilon;s}(U(1,0,0))|
 \\
 \leq
 \frac{\beta_2}{\pi s} \int_{\bR^2}
 \frac{\vartheta(2\varepsilon(1+|p|^2/s^2)^{1/2})^2\,\vartheta((s^2+|p|^2)^{1/2})\,|\dot\vartheta((s^2+|p|^2)^{1/2})|\,|c_Ns^2-|p|^2|}{(s^2+|p|^2)^{3/2}}\,\rd p
\end{multline}
we show that there exists $C\in(0,\infty)$ such that
\begin{equation}
 |\fL\fE\fA^4_{\varepsilon;1,s}\fB_{\varepsilon;s}(U(1,0,0))|\leq
 C/(\varepsilon\vee s)
\end{equation}
for all $\varepsilon\in[0,1]$ and $s\in(0,1]$. Hence,
\begin{equation}\label{eq:f_proof_bound_epsilon}
 |f_{\varepsilon;t}-f_{\varepsilon;4\varepsilon\wedge 1}|\leq
 \int_t^{4\varepsilon \wedge 1} |\fL\fE\fA^4_{\varepsilon;1,s}\fB_{\varepsilon;s}(U(1,0,0))|\,\rd s\leq 4\,C
\end{equation}
for all $\varepsilon\in[0,1]$ and $t\in(0,4\varepsilon\wedge 1]$. Combining the bounds~\eqref{eq:f_proof_bound_tau},~\eqref{eq:f_proof_bound_t} and~\eqref{eq:f_proof_bound_epsilon} we obtain that there exists $\lambda_\star\in(0,1]$ such that
\begin{equation}
 |f_{\tau,\varepsilon;t}-\lambda^{-1}_{\varepsilon\vee t}|\leq \lambda_\star^{-1}/2\leq \lambda^{-1}/2 \leq \lambda^{-1}_{\varepsilon\vee t}/2
\end{equation}
for all $\lambda\in(0,\lambda_\star]$, $\tau,\varepsilon\in[0,1]$ and $t\in(0,1]$. This proves the bound~(A).

To prove the bound~(B) first note that
\begin{multline}
 |\fL\fE\fA^4_{1,s}\fB_{s}(U(1,0,0))-\fL\fE\fA^4_{\tau,\varepsilon;1,s}\fB_{\varepsilon;s}(U(1,0,0))|
 \\
 \leq
 |\mathbf{c}^{4,0}_{s}(\fB_{s}(U(1,0,0)))
 -
 \mathbf{c}^{4,0}_{\tau,\varepsilon;s}(\fB_{s}(U(1,0,0)))|
 \\
 +
 |\mathbf{c}^{4,0}_{\tau,\varepsilon;s}(\fB_{s}(U(1,0,0)))
 -
 \mathbf{c}^{4,0}_{\tau,\varepsilon;s}(\fB_{\varepsilon;s}(U(1,0,0)))|
\end{multline}
where the map $\mathbf{c}^{4,0}_{\tau,\varepsilon;t}$ was introduced in Lemma~\ref{lem:contraction_c}. By Lemma~\ref{lem:B_map_bound}~(A),~(C) applied with $\gamma_1=\gamma_2=-8\kappa$ and Lemma~\ref{lem:contraction_U} applied with $\gamma=-8\kappa$ there exists $C\in(0,\infty)$ such that
\begin{equation}
\begin{gathered}
 \|s\mapsto s\fB_{s}(U(1,0,0))\|_{\sW^{8,3;-16\kappa}}\leq C\,,
 \\
 \|s\mapsto s(\fB_{s}(U(1,0,0))-\fB_{\varepsilon;s}(U(1,0,0)))\|_{\tilde\sW^{2,2;-16\kappa}_{\tau,\varepsilon}}\leq C\,\lambda_\varepsilon^\kappa
\end{gathered}
\end{equation}
for all $\tau,\varepsilon\in[0,1]$. Consequently, by Lemma~\ref{lem:contraction_c}~(C),~(D) applied with $m=4$, $i=0$ and $\gamma=-16\kappa$ there exists $C\in(0,\infty)$ such that
\begin{equation}
 |\fL\fE\fA^4_{1,s}\fB_{s}(U(1,0,0))-\fL\fE\fA^4_{\tau,\varepsilon;1,s}\fB_{\varepsilon;s}(U(1,0,0))|\leq C\,\lambda_{\tau\vee\varepsilon}^\kappa\, \lambda_s^{-8\kappa}\,s^{-1}
\end{equation}
for all $\tau,\varepsilon\in[0,1]$ and $s\in(0,1]$. The bound~(B) follows now from Lemma~\ref{lem:bounds_relevant_irrelevant}~(E) applied with $\rho=-1-8\kappa$.

Finally, let us turn to the proof of the bounds~(C),~(D). Observe that
\begin{equation}
 h_{\tau,\varepsilon;t}
 =
 \fL\fE\fA^2_{\tau,\varepsilon;1,t} U(1,0,0)
 =
 t\,\mathbf{c}^{2,0}_{\tau,\varepsilon;t}(s\mapsto s^{-1}\,U(1,0,0)),
\end{equation}
where the map $\mathbf{c}^{2,0}_{\tau,\varepsilon;t}$ was introduced in Lemma~\ref{lem:contraction_c}. The bounds~(C),~(D)  follow now from Lemma~\ref{lem:contraction_c}~(A),~(D) applied with $m=2$, $i=0$ and $\gamma=-8\kappa$ as well as Lemma~\ref{lem:contraction_U}~(A) applied with $\gamma=-8\kappa$. This finishes the proof.
\end{proof}

\begin{lem}\label{lem:contraction_g}
There exists $\lambda_\star\in(0,1]$ and $C\in(0,\infty)$ such that for all $\lambda\in(0,\lambda_\star]$, all $\tau,\varepsilon\in[0,1]$, all $t\in(0,1]$ and all $X_\Cdot,Y_\Cdot\in \sY_{\tau,\varepsilon}$, $Z_\Cdot\in\sY$ it holds:
\begin{itemize}
 \item[(A)] $\lambda_t^{1+10\kappa}\,|\mathbf{g}_{\tau,\varepsilon;t}(X_\Cdot)|
 \leq C\, \lambda^\kappa$,
 \item[(B)] $\lambda_t^{1+10\kappa}\,|\mathbf{g}_{\tau,\varepsilon;t}(X_\Cdot)-\mathbf{g}_{\tau,\varepsilon;t}(Y_\Cdot)|
 \leq
 C\,\lambda^\kappa\, \|X_\Cdot-Y_\Cdot\|_{\sX_{\tau,\varepsilon}}$.

 \item[(C)] $\lambda_t^{1+10\kappa}\,|\mathbf{g}_{\tau,\varepsilon;t}(X_\Cdot)-\mathbf{g}_{\tau,\varepsilon;t}(Z_\Cdot)|
 \leq
 C\,\lambda^\kappa\, \|X_\Cdot-Z_\Cdot\|_{\tilde\sX_{\tau,\varepsilon}}$,

 \item[(D)] $\lambda_t^{1+10\kappa}\,|\mathbf{g}_{t}(Z_\Cdot)-\mathbf{g}_{\tau,\varepsilon;t}(Z_\Cdot)|
 \leq
 C\,\lambda_{\tau\vee\varepsilon}^\kappa$,
 \item[(E)] $\lambda_{\varepsilon\vee t}\,\mathbf{g}_{\tau,\varepsilon;t}(X_\Cdot)
 \geq 1/C$.
\end{itemize}
\end{lem}
\begin{proof}
For $\tau,\varepsilon\in[0,1]$, $X_\Cdot\equiv(g_\Cdot,r_\Cdot,z_\Cdot,W_\Cdot)
\in\sY_{\tau,\varepsilon}\cup\sY$ and $t\in(0,1]$ define
\begin{equation}
 \mathbf{g}^{(2)}_{\tau,\varepsilon;t}(X_\Cdot)
 :=
 \int_t^1
 (g_s)^2\,
 \mathbf{a}^{(2)}_{\tau,\varepsilon;s}(X_\Cdot)
 \,\rd s,
 \qquad
 \mathbf{g}^{(3)}_{\tau,\varepsilon;t}(X_\Cdot)
 :=
 \int_t^1
 (g_s)^2\,
 \mathbf{a}^{(3)}_{\tau,\varepsilon;s}(X_\Cdot)
 \,\rd s,
\end{equation}
where the maps $\mathbf{a}^{(2)}_{\tau,\varepsilon;s}$ and $\mathbf{a}^{(3)}_{\tau,\varepsilon;s}$ were introduced in Lemma~\ref{lem:contraction_a}. Note that it holds
\begin{equation}
 \mathbf{g}_{\tau,\varepsilon;t}(X_\Cdot)
 =
 f_{\tau,\varepsilon;t}
 +
 \mathbf{g}^{(2)}_{\tau,\varepsilon;t}(X_\Cdot)
 +
 \mathbf{g}^{(3)}_{\tau,\varepsilon;t}(X_\Cdot).
\end{equation}
We shall prove that for $j\in\{2,3\}$ it holds
\begin{itemize}
 \item[(A${}_1$)] $\lambda_t^{132\kappa}\,|\mathbf{g}^{(j)}_{\tau,\varepsilon;t}(X_\Cdot)|
 \leq C\, \lambda^\kappa$,
 \item[(B${}_1$)] $\lambda_t^{132\kappa}\,|\mathbf{g}^{(j)}_{\tau,\varepsilon;t}(X_\Cdot)-\mathbf{g}^{(j)}_{\tau,\varepsilon;t}(Y_\Cdot)|
 \leq
 C\,\lambda^\kappa\, \|X_\Cdot-Y_\Cdot\|_{\sX_{\tau,\varepsilon}}$,

 \item[(C${}_1$)] $\lambda_t^{132\kappa}\,|\mathbf{g}^{(j)}_{\tau,\varepsilon;t}(X_\Cdot)-\mathbf{g}^{(j)}_{\tau,\varepsilon;t}(Z_\Cdot)|
 \leq
 C\,\lambda^\kappa\, \|X_\Cdot-Z_\Cdot\|_{\tilde\sX_{\tau,\varepsilon}}$,

 \item[(D${}_1$)] $\lambda_t^{132\kappa}\,|\mathbf{g}^{(j)}_{t}(Z_\Cdot)-\mathbf{g}^{(j)}_{\tau,\varepsilon;t}(Z_\Cdot)|
 \leq
 C\,\lambda_{\tau\vee\varepsilon}^\kappa$.
\end{itemize}

We first consider $j=3$. By Lemma~\ref{lem:contraction_a} and Def.~\ref{dfn:sX},
\begin{equation}
 \lambda_s^{132\kappa-1}\,s\,
 \big|
 (g_s)^2\mathbf{a}^{(3)}_{\tau,\varepsilon;s}(X_\Cdot)
 \big|
 \leq
 C\,\lambda^\kappa.
\end{equation}
The analogous Lipschitz and cutoff-difference estimates follow in the same way from Lemma~\ref{lem:contraction_a}~(B)--(D) and Def.~\ref{dfn:sX}. Therefore the bounds~(A${}_1$)--(D${}_1$) for $j=3$ follow from Lemma~\ref{lem:bounds_relevant_irrelevant}~(E) applied with $\rho=-132\kappa$.

We next consider $j=2$. By Lemma~\ref{lem:contraction_a}~(E),
\begin{equation}
 \mathbf{g}^{(2)}_{\tau,\varepsilon;t}(X_\Cdot)
 =
 \int_{t\vee(1/2)}^1
 (g_s)^2\,
 \mathbf{a}^{(2)}_{\tau,\varepsilon;s}(X_\Cdot)
 \,\rd s.
\end{equation}
Moreover, by Lemma~\ref{lem:contraction_a}~(A) and Def.~\ref{dfn:sX},
\begin{equation}
 s\,
 \big|
 (g_s)^2\mathbf{a}^{(2)}_{\tau,\varepsilon;s}(X_\Cdot)
 \big|
 \leq
 C\,\lambda^\kappa\,\lambda_s^{-92\kappa}.
\end{equation}
Since $\lambda_s\asymp\lambda$ for $s\in[1/2,1]$, it follows that
\begin{equation}
 \lambda_t^{132\kappa}
 |\mathbf{g}^{(2)}_{\tau,\varepsilon;t}(X_\Cdot)|
 \leq
 C\,\lambda^\kappa\lambda_t^{132\kappa}
 \int_{t\vee(1/2)}^1
 \lambda_s^{-92\kappa}\,\frac{\rd s}{s}
 \leq
 C\,\lambda^{41\kappa}
 \leq
 C\,\lambda^\kappa.
\end{equation}
This proves~(A${}_1$) for $j=2$. Using the bounds on $g_s$ and its differences from Def.~\ref{dfn:sX}, together with Lemma~\ref{lem:contraction_a}~(A)--(C), the same calculation gives
\begin{equation}
 \lambda_t^{132\kappa}
 |\mathbf{g}^{(2)}_{\tau,\varepsilon;t}(X_\Cdot)
 -
 \mathbf{g}^{(2)}_{\tau,\varepsilon;t}(Y_\Cdot)|
 \leq
 C\,\lambda^\kappa\,
 \|X_\Cdot-Y_\Cdot\|_{\sX_{\tau,\varepsilon}}
\end{equation}
and
\begin{equation}
 \lambda_t^{132\kappa}
 |\mathbf{g}^{(2)}_{\tau,\varepsilon;t}(X_\Cdot)
 -
 \mathbf{g}^{(2)}_{\tau,\varepsilon;t}(Z_\Cdot)|
 \leq
 C\,\lambda^\kappa\,
 \|X_\Cdot-Z_\Cdot\|_{\tilde\sX_{\tau,\varepsilon}}.
\end{equation}
Finally, using Lemma~\ref{lem:contraction_a}~(D) and the support property in Lemma~\ref{lem:contraction_a}~(E), we obtain
\begin{equation}
\begin{aligned}
 &\lambda_t^{132\kappa}
 |\mathbf{g}^{(2)}_{t}(Z_\Cdot)
 -
 \mathbf{g}^{(2)}_{\tau,\varepsilon;t}(Z_\Cdot)|
 \\
 &\qquad\leq
 C\,\lambda_{\tau\vee\varepsilon}^\kappa
 \lambda_t^{132\kappa}
 \int_{t\vee(1/2)}^1
 \lambda_s^{-92\kappa}\,\frac{\rd s}{s}
 \leq
 C\,\lambda_{\tau\vee\varepsilon}^\kappa\,\lambda^{40\kappa}
 \leq
 C\,\lambda_{\tau\vee\varepsilon}^\kappa.
\end{aligned}
\end{equation}
This proves~(A${}_1$)--(D${}_1$) for $j=2$.

Note that by Lemma~\ref{lem:contraction_f_h}~(A),~(B) it holds
\begin{equation}
 \lambda_t^{1+10\kappa}\,|f_{\tau,\varepsilon;t}|\leq C\,\lambda_t^{10\kappa}\leq C\,\lambda^\kappa,
 \qquad
 \lambda_t^{1+10\kappa}\,|f_t-f_{\tau,\varepsilon;t}|\leq C\, \lambda_{\varepsilon\vee\tau}^\kappa\,\lambda_t^{2\kappa}
 \leq C\,\lambda_{\varepsilon\vee\tau}^\kappa.
\end{equation}
Since $1-122\kappa>0$, the bounds~(A), (B), (C),(D) follow from the above bounds and the bounds~(A${}_1$), (B${}_1$), (C${}_1$), (D${}_1$).

To prove the bound~(E) first note that since $\dot H_{\varepsilon;s}=0$ if $s<\varepsilon$ it holds $\fB_{\varepsilon;s}=0$. Consequently, for $j\in\{2,3\}$,
\begin{equation}
 \mathbf{g}^{(j)}_{\tau,\varepsilon;t}(X_\Cdot)
 =
 \mathbf{g}^{(j)}_{\tau,\varepsilon;\varepsilon}(X_\Cdot)
 \qquad\text{if }t<\varepsilon.
\end{equation}
Recall also that by Lemma~\ref{lem:contraction_f_h}~(A) it holds $f_{\tau,\varepsilon;t}\geq \lambda_{\varepsilon\vee t}^{-1}/C$. Consequently, by~(A${}_1$), there exists $C\in(0,\infty)$ such that
\begin{multline}
 \mathbf{g}_{\tau,\varepsilon;t}(X_\Cdot)
 \geq
 \lambda_{\varepsilon\vee t}^{-1}/C
 -
 C\,\lambda^\kappa \, \lambda_{\varepsilon\vee t}^{-132\kappa}
 =
 \lambda_{\varepsilon\vee t}^{-1}\,
 (1/C- C\,\lambda^\kappa \, \lambda_{\varepsilon\vee t}^{1-132\kappa})
 \\
 \geq
 \lambda_{\varepsilon\vee t}^{-1}\,
 (1/C- C \, \lambda^{1-131\kappa})
 \geq
 \lambda_{\varepsilon\vee t}^{-1}/(2C)
\end{multline}
for all $\lambda\in(0,\lambda_\star]$, $\tau,\varepsilon\in[0,1]$ and $t\in(0,1]$ provided $\lambda_\star\in(0,1]$ is small enough. This proves the bound~(E) and completes the proof of the lemma.
\end{proof}

\begin{lem}\label{lem:contraction_r}
There exists $\lambda_\star\in(0,1]$ and $C\in(0,\infty)$ such that for all $\lambda\in(0,\lambda_\star]$, all $\tau,\varepsilon\in[0,1]$, all $t\in(0,1]$ and all $X_\Cdot,Y_\Cdot\in \sY_{\tau,\varepsilon}$, $Z_\Cdot\in\sY$ it holds:
\begin{itemize}
 \item[(A)] $\lambda_t^{36\kappa-1}\, t\,|\mathbf{r}_{\tau,\varepsilon;t}(X_\Cdot)|
 \leq C\, \lambda^\kappa$,
 \item[(B)] $\lambda_t^{36\kappa-1}\, t\,|\mathbf{r}_{\tau,\varepsilon;t}(X_\Cdot)-\mathbf{r}_{\tau,\varepsilon;t}(Y_\Cdot)|
 \leq 
 C\,\lambda^\kappa\, \|X_\Cdot-Y_\Cdot\|_{\sX_{\tau,\varepsilon}}$,

 \item[(C)] $\lambda_t^{36\kappa-1}\, t\,|\mathbf{r}_{\tau,\varepsilon;t}(X_\Cdot)-\mathbf{r}_{\tau,\varepsilon;t}(Z_\Cdot)|
 \leq 
 C\,\lambda^\kappa\, \|X_\Cdot-Z_\Cdot\|_{\tilde\sX_{\tau,\varepsilon}}$,
 
 \item[(D)] $\lambda_t^{36\kappa-1}\, t\,|\mathbf{r}_{t}(Z_\Cdot)-\mathbf{r}_{\tau,\varepsilon;t}(Z_\Cdot)|
 \leq
 C\,\lambda_{\tau\vee\varepsilon}^\kappa$.
\end{itemize}
\end{lem}
\begin{proof}
We first observe that for all $\tau,\varepsilon\in[0,1]$, $X_\Cdot\equiv(g_\Cdot,r_\Cdot,z_\Cdot,W_\Cdot)\in\sY_{\tau,\varepsilon}\cup\sY$
and $t\in(0,1]$ it holds
\begin{equation}
 \mathbf{r}_{\tau,\varepsilon;t}(X_\Cdot)=
 \mathbf{r}^{(1)}_{\tau,\varepsilon;t}(X_\Cdot) 
 +
 \mathbf{r}^{(2)}_{\tau,\varepsilon;t}(X_\Cdot),
\end{equation}
where
\begin{equation}
 \mathbf{r}^{(1)}_{\tau,\varepsilon;t}(X_\Cdot) :=
 -
 \int_t^1
 \mathbf{r}^{(3)}_{\tau,\varepsilon;s}(X_\Cdot)\,\rd s,
 \qquad
 \mathbf{r}^{(2)}_{\tau,\varepsilon;t}(X_\Cdot) :=
 -h_{\tau,\varepsilon;t}/g_t
\end{equation}
and
\begin{equation}
 \mathbf{r}^{(3)}_{\tau,\varepsilon;t}(X_\Cdot) :=
 \fL\fE\fA^2_{\tau,\varepsilon;1,t}\fG_{\varepsilon;t}(X_\Cdot)=
 \mathbf{c}^{2,0}_{\tau,\varepsilon;t}(\fG_{\varepsilon;t}(X_\Cdot)).
\end{equation}
Recall that the map $\fG_{\varepsilon;s}$ was introduced in Lemma~\ref{lem:contraction_B}, the map $\mathbf{c}^{2,0}_{\tau,\varepsilon;s}$ was introduced in Lemma~\ref{lem:contraction_c} and $h_{\tau,\varepsilon;t}$ was introduced in Def.~\ref{dfn:f_function}. We shall prove that for $i\in\{1,2,3\}$ it holds
\begin{itemize}
 \item[(A${}_1$)] $\lambda_t^{-\gamma(i)}\, t^{-\varrho(i)}\,|\mathbf{r}^{(i)}_{\tau,\varepsilon;t}(X_\Cdot)|
 \leq C\, \lambda^\kappa$,
 \item[(B${}_1$)] $\lambda_t^{-\gamma(i)}\, t^{-\varrho(i)}\,|\mathbf{r}^{(i)}_{\tau,\varepsilon;t}(X_\Cdot)-\mathbf{r}^{(i)}_{\tau,\varepsilon;t}(Y_\Cdot)|
 \leq 
 C\,\lambda^\kappa\, \|X_\Cdot-Y_\Cdot\|_{\sX_{\tau,\varepsilon}}$,

 \item[(C${}_1$)] $\lambda_t^{-\gamma(i)}\, t^{-\varrho(i)}\,|\mathbf{r}^{(i)}_{\tau,\varepsilon;t}(X_\Cdot)-\mathbf{r}^{(i)}_{\tau,\varepsilon;t}(Z_\Cdot)|
 \leq 
 C\,\lambda^\kappa\, \|X_\Cdot-Z_\Cdot\|_{\tilde\sX_{\tau,\varepsilon}}$,
 
 \item[(D${}_1$)] $\lambda_t^{-\gamma(i)}\, t^{-\varrho(i)}\,|\mathbf{r}^{(i)}_{t}(Z_\Cdot)-\mathbf{r}^{(i)}_{\tau,\varepsilon;t}(Z_\Cdot)|
 \leq
 C\,\lambda_{\tau\vee\varepsilon}^\kappa$,
\end{itemize}
where $\gamma(1)=2-76\kappa$, $\varrho(1)=-1$, $\gamma(2)=1-20\kappa$, $\varrho(2)=-1$ and $\gamma(3)=2-76\kappa$, $\varrho(3)=-2$. Note that the above bounds with $i\in\{1,2\}$ imply immediately the bounds~(A), (B), (C), (D).

It remains to prove the bounds~(A${}_1$), (B${}_1$), (C${}_1$), (D${}_1$). Application of Lemma~\ref{lem:contraction_c} with $m=2$, $i=0$ and $\gamma=2-80\kappa$ and Lemma~\ref{lem:contraction_B} with $(i,j)=(0,0)$ yields the bounds~(A${}_1$), (B${}_1$), (C${}_1$) with $i=3$. For~(D${}_1$), we add and subtract $\mathbf c^{2,0}_{\tau,\varepsilon;t}(\fG_t(Z_\Cdot))$ and use Lemma~\ref{lem:contraction_c}~(D),~(C) together with Lemma~\ref{lem:contraction_B}~(A),~(D), respectively. Consequently, the bounds~(A${}_1$), (B${}_1$), (C${}_1$), (D${}_1$) with $i=1$ follow from Lemma~\ref{lem:bounds_relevant_irrelevant}~(C) applied with $\varrho=-1$ and $\rho=2-76\kappa$.

To prove these bounds with $i=2$, first suppose that $t\geq\varepsilon$. Then $\theta_{\varepsilon;t}=1$, and Lemma~\ref{lem:contraction_f_h}~(C),~(D) and Lemma~\ref{lem:contraction_g_aux} give the required bounds, with an additional factor $\lambda_t^{4\kappa}\leq\lambda^\kappa$ in~(A${}_1$)--(C${}_1$). If $t < \varepsilon$, then
$h_{\tau,\varepsilon;t}=h_{\tau,\varepsilon;\varepsilon}$ and $1/g_t = \theta_{\varepsilon;t}^{-2}\theta_{\varepsilon;t}^2/g_t$. Thus the bounds~(A${}_1$)--(C${}_1$) follow from the same lemmas and
\begin{equation}
	\frac{t}{\varepsilon} 
	\left( \frac{\lambda_\varepsilon}{\lambda_t} \right)^{2-8\kappa}
	\lambda_\varepsilon^{4\kappa}
	\leq  C\,\lambda^\kappa.
\end{equation}
The bound~(D${}_1$) follows from  Lemma~\ref{lem:contraction_f_h}~(D) and
Lemma~\ref{lem:contraction_g_aux}~(A) applied with $\varepsilon=0$.
\end{proof}

\begin{lem}\label{lem:contraction_z}
There exists $\lambda_\star\in(0,1]$ and $C\in(0,\infty)$ such that for all $\lambda\in(0,\lambda_\star]$, all $\tau,\varepsilon\in[0,1]$, all $t\in(0,1]$ and all $X_\Cdot,Y_\Cdot\in \sY_{\tau,\varepsilon}$, $Z_\Cdot\in\sY$ it holds:
\begin{itemize}
 \item[(A)] $\lambda_t^{36\kappa-1}\,|\mathbf{z}_{\tau,\varepsilon;t}(X_\Cdot)|
 \leq C\, \lambda^\kappa$,
 \item[(B)] $\lambda_t^{36\kappa-1}\,|\mathbf{z}_{\tau,\varepsilon;t}(X_\Cdot)-\mathbf{z}_{\tau,\varepsilon;t}(Y_\Cdot)|
 \leq
 C\,\lambda^\kappa\, \|X_\Cdot-Y_\Cdot\|_{\sX_{\tau,\varepsilon}}$,

 \item[(C)] $\lambda_t^{36\kappa-1}\,|\mathbf{z}_{\tau,\varepsilon;t}(X_\Cdot)-\mathbf{z}_{\tau,\varepsilon;t}(Z_\Cdot)|
 \leq
 C\,\lambda^\kappa\, \|X_\Cdot-Z_\Cdot\|_{\tilde\sX_{\tau,\varepsilon}}$,

 \item[(D)] $\lambda_t^{36\kappa-1}\,|\mathbf{z}_{t}(Z_\Cdot)-\mathbf{z}_{\tau,\varepsilon;t}(Z_\Cdot)|
 \leq
 C\,\lambda_{\tau\vee\varepsilon}^\kappa$.
\end{itemize}
\end{lem}
\begin{proof}
For $i\in\{4,5,6\}$ define
\begin{equation}
 \mathbf{z}^{(i)}_{\tau,\varepsilon;t}(X_\Cdot)
 :=
 \int_0^t
 \mathbf{a}^{(i)}_{\tau,\varepsilon;s}(X_\Cdot)\,\rd s.
\end{equation}
Then
\begin{equation}
 \mathbf{z}_{\tau,\varepsilon;t}(X_\Cdot)
 =
 \mathbf{z}^{(4)}_{\tau,\varepsilon;t}(X_\Cdot)
 +
 \mathbf{z}^{(5)}_{\tau,\varepsilon;t}(X_\Cdot)
 +
 \mathbf{z}^{(6)}_{\tau,\varepsilon;t}(X_\Cdot),
\end{equation}
where the maps $\mathbf{a}^{(4)}_{\tau,\varepsilon;s}$, $\mathbf{a}^{(5)}_{\tau,\varepsilon;s}$ and $\mathbf{a}^{(6)}_{\tau,\varepsilon;s}$ were introduced in Lemma~\ref{lem:contraction_a}.

We first consider $i\in\{4,6\}$. By Lemma~\ref{lem:contraction_a} it holds
\begin{equation}
 s\,|\mathbf{a}^{(i)}_{\tau,\varepsilon;s}(X_\Cdot)|
 \leq
 C\,\lambda^\kappa\,\lambda_s^{2-36\kappa},
\end{equation}
where for $i=6$ we used $3-116\kappa\geq2-36\kappa$ and $\lambda_s\leq1$. The analogous bounds hold for the differences appearing in~(B), (C) and~(D). Consequently, the required bounds for $\mathbf{z}^{(4)}_{\tau,\varepsilon;t}$ and $\mathbf{z}^{(6)}_{\tau,\varepsilon;t}$ follow from Lemma~\ref{lem:bounds_relevant_irrelevant}~(D) applied with $\rho=1-36\kappa$.

It remains to consider $i=5$. By Lemma~\ref{lem:contraction_a}~(E),
\begin{equation}
 \mathbf{z}^{(5)}_{\tau,\varepsilon;t}(X_\Cdot)
 =
 1_{\{t>1/2\}}
 \int_{1/2}^t
 \mathbf{a}^{(5)}_{\tau,\varepsilon;s}(X_\Cdot)\,\rd s.
\end{equation}
Using Lemma~\ref{lem:contraction_a}~(A) and the fact that $\lambda_s\asymp\lambda$ for $s\in[1/2,1]$, we obtain
\begin{equation}
\begin{aligned}
 \lambda_t^{36\kappa-1}
 |\mathbf{z}^{(5)}_{\tau,\varepsilon;t}(X_\Cdot)|
 &\leq
 C\,\lambda^\kappa\,\lambda_t^{36\kappa-1}
 1_{\{t>1/2\}}
 \int_{1/2}^t
 \lambda_s^{2-76\kappa}\,\frac{\rd s}{s}
 \\
 &\leq
 C\,\lambda^{1-39\kappa}
 \leq
 C\,\lambda^\kappa.
\end{aligned}
\end{equation}
This proves the required bound~(A) for $\mathbf{z}^{(5)}_{\tau,\varepsilon;t}$. The same argument using Lemma~\ref{lem:contraction_a}~(B) and~(C) gives
\begin{equation}
 \lambda_t^{36\kappa-1}
 |\mathbf{z}^{(5)}_{\tau,\varepsilon;t}(X_\Cdot)
 -
 \mathbf{z}^{(5)}_{\tau,\varepsilon;t}(Y_\Cdot)|
 \leq
 C\,\lambda^\kappa\,\|X_\Cdot-Y_\Cdot\|_{\sX_{\tau,\varepsilon}}
\end{equation}
and
\begin{equation}
 \lambda_t^{36\kappa-1}
 |\mathbf{z}^{(5)}_{\tau,\varepsilon;t}(X_\Cdot)
 -
 \mathbf{z}^{(5)}_{\tau,\varepsilon;t}(Z_\Cdot)|
 \leq
 C\,\lambda^\kappa\,\|X_\Cdot-Z_\Cdot\|_{\tilde\sX_{\tau,\varepsilon}}.
\end{equation}
Finally, Lemma~\ref{lem:contraction_a}~(D) and~(E) imply
\begin{equation}
\begin{aligned}
 &\lambda_t^{36\kappa-1}
 |\mathbf{z}^{(5)}_{t}(Z_\Cdot)
 -
 \mathbf{z}^{(5)}_{\tau,\varepsilon;t}(Z_\Cdot)|
 \\
 &\qquad\leq
 C\,\lambda_{\tau\vee\varepsilon}^\kappa\,
 \lambda_t^{36\kappa-1}
 1_{\{t>1/2\}}
 \int_{1/2}^t
 \lambda_s^{2-76\kappa}\,\frac{\rd s}{s}
 \\
 &\qquad\leq
 C\,\lambda_{\tau\vee\varepsilon}^\kappa\,
 \lambda^{1-40\kappa}
 \leq
 C\,\lambda_{\tau\vee\varepsilon}^\kappa.
\end{aligned}
\end{equation}
Combining the bounds for $i\in\{4,5,6\}$ proves~(A)--(D).
\end{proof}

\section{Relation to Polchinski equation}\label{sec:relation_polchinski}

In this section we construct a solution $U_{\tau,\varepsilon;\Cdot}$ of the flow equation~\eqref{eq:polchinski_integral} using the fixed point $X_{\tau,\varepsilon;\Cdot}$ of the map $\fX_{\tau,\varepsilon;\Cdot}$ defined in Sec.~\ref{sec:fixed_point}. The strategy of the construction was discussed in Sec.~\ref{sec:strategy}. We first construct a solution $V_{\tau,\varepsilon;\Cdot}$ of the flow equation~\eqref{eq:polchinski_aux} and subsequently use it to construct a solution $U_{\tau,\varepsilon;\Cdot}$ of the flow equation~\eqref{eq:polchinski_integral}.

\begin{thm}\label{thm:polchinski}
Let $\tau,\varepsilon\in(0,1]$ and $\lambda\in(0,\lambda_\star]$, where $\lambda_\star$ is as in Corollary~\ref{cor:contraction}. Let 
\begin{equation}
 (0,1]\ni t\mapsto X_{\tau,\varepsilon;t}\equiv (g_{\tau,\varepsilon;t},r_{\tau,\varepsilon;t},z_{\tau,\varepsilon;t},W_{\tau,\varepsilon;t})\in \bR\times\bR\times\bR\times\sN
\end{equation}
be the fixed point of the map $\fX_{\tau,\varepsilon;\Cdot}\,:\,\sY_{\tau,\varepsilon}\to \sY_{\tau,\varepsilon}$, introduced in Def.~\ref{dfn:X_map}. Moreover, set
\begin{equation}
 g_{\tau,\varepsilon;0}:=g_{\tau,\varepsilon;\varepsilon},
 \quad
 r_{\tau,\varepsilon;0}:=r_{\tau,\varepsilon;\varepsilon},
 \quad
 z_{\tau,\varepsilon;0}:=0,
 \quad
 W_{\tau,\varepsilon;0}:=0.
\end{equation} 
For $t\in[0,1]$ define
\begin{equation}\label{eq:V_def_polchinski}
 V_{\tau,\varepsilon;t}\equiv(V_{\tau,\varepsilon;t}^m)_{m\in\bN_+}:= U(1/g_{\tau,\varepsilon;t},r_{\tau,\varepsilon;t},z_{\tau,\varepsilon;t})+W_{\tau,\varepsilon;t}\in \sN.
\end{equation}
Then
\begin{equation}\label{eq:polchinski_aux}
 \langle V^m_{\tau,\varepsilon;t} ,(\fJ\varphi)^{\otimes m}\rangle
 =
 \langle\fE\fA^m_{\tau,\varepsilon;t,0}V_{\tau,\varepsilon;0},(\fJ\varphi)^{\otimes m}\rangle
 +
 \int_0^t
 \langle\fE\fA^m_{\tau,\varepsilon;t,s}\fB_{\varepsilon;s}(V_{\tau,\varepsilon;s}),(\fJ\varphi)^{\otimes m}\rangle\,\rd s
\end{equation} 
for all $t\in(0,1]$, $m\in\bN_+$ and $\varphi\in\sS(\bR^2)^\bFF\otimes_{\mathrm{alg}}\sG^-$.
\end{thm}
\begin{rem}\label{rem:fixed_point_epsilon}
Since $\dot G_{\tau,\varepsilon;t}=0$ for $t\in(0,\varepsilon]$ by Remark~\ref{rem:G}~(A), it follows from Def.~\ref{dfn:B_map} that $\fB_{\varepsilon;t}(\Cdot)=0$ for all
$t\in(0,\varepsilon]$. Hence, by Def.~\ref{dfn:X_map},$g_{\tau,\varepsilon;t}$ and $z_{\tau,\varepsilon;t}$ are constant on $(0,\varepsilon]$, with $z_{\tau,\varepsilon;t}=0$. 
Although the operator-valued field $\varPsi_{\tau,\varepsilon;1,t}$ need only be constant for $t\in(0,\varepsilon/2]$, its covariance is constant on the whole
interval $(0,\varepsilon]$, since $\dot G_{\tau,\varepsilon;t}=0$ there. Hence
Remark~\ref{rem:Wick} implies that
\begin{equation}
 \fL\fE\fA^2_{\tau,\varepsilon;1,t}
 U(1/g_{\tau,\varepsilon;t},0,0)
\end{equation}
is constant for $t\in(0,\varepsilon]$, since $g_{\tau,\varepsilon;t}$ is already constant on this interval. Moreover, the integral appearing in the definition of $r_{\tau,\varepsilon;t}$ is independent of $t$ on $(0,\varepsilon]$, because $\fB_{\varepsilon;s}(\Cdot)=0$ for $s\in(0,\varepsilon]$. Thus $r_{\tau,\varepsilon;t}$ is constant on $(0,\varepsilon]$.

Finally, the fixed point equation for $W_{\tau,\varepsilon;t}$ and the triangular dependence of $\fC^m_{\tau,\varepsilon;1,t}W_{\tau,\varepsilon;t}$ on components of degree strictly larger than $m$ imply successively that all components of $W_{\tau,\varepsilon;t}$ vanish for $t\in(0,\varepsilon]$. Consequently,
\begin{equation}
 V_{\tau,\varepsilon;0}=V_{\tau,\varepsilon;t},
 \quad
 g_{\tau,\varepsilon;0}=g_{\tau,\varepsilon;t},
 \quad
 r_{\tau,\varepsilon;0}=r_{\tau,\varepsilon;t},
 \quad
 z_{\tau,\varepsilon;0}=z_{\tau,\varepsilon;t},
 \quad
 W_{\tau,\varepsilon;0}=W_{\tau,\varepsilon;t}
\end{equation}
for all $t\in(0,\varepsilon]$, where $(g_{\tau,\varepsilon;\Cdot},r_{\tau,\varepsilon;\Cdot},z_{\tau,\varepsilon;\Cdot},W_{\tau,\varepsilon;\Cdot})$ is the fixed point of $\fX_{\tau,\varepsilon;\Cdot}$.
\end{rem}
\begin{proof}
First note that~\eqref{eq:V_def_polchinski} is equivalent to the system of equations
\begin{equation}
\begin{aligned}
 V_{\tau,\varepsilon;t}^m&= W^m_{\tau,\varepsilon;t},
 \qquad m\in\bN_+\setminus\{2,4\},
 \\
 V_{\tau,\varepsilon;t}^4&= U^4\,1/g_{\tau,\varepsilon;t}+W^4_{\tau,\varepsilon;t},
 \\
 V_{\tau,\varepsilon;t}^2&= U^2\,r_{\tau,\varepsilon;t} + U^2_\partial\,z_{\tau,\varepsilon;t}+W^2_{\tau,\varepsilon;t},
\end{aligned} 
\end{equation}
where $U^4,U^2,U^2_\partial$ were introduced in Def.~\ref{dfn:U_functional}. Observe that the fixed point of the map $\fX_{\tau,\varepsilon;\Cdot}\,:\,\sY_{\tau,\varepsilon}\to \sY_{\tau,\varepsilon}$, introduced in Def.~\ref{dfn:X_map}, is equivalent to the system of Eqs.~\eqref{eq:polchinski_proof_m},~\eqref{eq:polchinski_proof_4} and~\eqref{eq:polchinski_proof_2} presented below. These equations imply in particular that
\begin{equation}
\begin{aligned}
 V_{\tau,\varepsilon;0}^m&= 0,
 \qquad m\in\bN_+\setminus\{2,4\},
 \\
 V_{\tau,\varepsilon;0}^4&= U^4\,1/g_{\tau,\varepsilon;0},
 \\
 V_{\tau,\varepsilon;0}^2&= U^2\,r_{\tau,\varepsilon;0}.
\end{aligned} 
\end{equation}
As a result, it holds
\begin{equation}
\begin{aligned}
 \fE\fA^m_{\tau,\varepsilon;t,0}V_{\tau,\varepsilon;0}&= 0,
 \qquad m\in\bN_+\setminus\{2,4\},
 \\
 \fE\fA^4_{\tau,\varepsilon;t,0}V_{\tau,\varepsilon;0}&= U^4\,1/g_{\tau,\varepsilon;0},
 \\
 \fE\fA^2_{\tau,\varepsilon;t,0}V_{\tau,\varepsilon;0}&= U^2\,r_{\tau,\varepsilon;0} +1/g_{\tau,\varepsilon;0}\, \fE\fA^2_{\tau,\varepsilon;t,0}U(1,0,0).
\end{aligned} 
\end{equation}
If $m\in\bN_+\setminus\{2,4\}$, then the first term on the RHS of~\eqref{eq:polchinski_aux} vanishes identically and $V^m_{\tau,\varepsilon;t}= W^m_{\tau,\varepsilon;t}$. Consequently, for $m\in\bN_+\setminus\{2,4\}$ Eq.~\eqref{eq:polchinski_aux} follows immediately from the fixed point equation 
\begin{equation}\label{eq:polchinski_proof_m}
 W^m_{\tau,\varepsilon;t}
 =
 \int_0^t
 \fE\fA^m_{\tau,\varepsilon;t,s}\fB_{\varepsilon;s}(V_{\tau,\varepsilon;s})\,\rd s.
\end{equation}
Let us turn to the proof of~\eqref{eq:polchinski_aux} for $m=4$. By the fixed point equation we obtain
\begin{equation}\label{eq:polchinski_proof_4}
\begin{aligned}
 W^4_{\tau,\varepsilon;t}
 &=
 \int_0^t
 \fR\fE\fA^4_{\tau,\varepsilon;1,s}\fB_{\varepsilon;s}(V_{\tau,\varepsilon;s})\,\rd s
 -
 \fE\fC_{\tau,\varepsilon;1,t}^4W_{\tau,\varepsilon;t},
 \\
 g_{\tau,\varepsilon;t}
 &=
 \lambda^{-1}+\int_t^1 (g_{\tau,\varepsilon;s})^2\, \fL\fE\fA^4_{\tau,\varepsilon;1,s}\fB_{\varepsilon;s}(V_{\tau,\varepsilon;s})\,\rd s.
\end{aligned} 
\end{equation}
In particular, the function $(0,1]\ni t\mapsto g_{\tau,\varepsilon;t}\in\bR$ is continuously differentiable and
\begin{equation}
 \partial_t(1/g_{\tau,\varepsilon;t})=-(\partial_t g_{\tau,\varepsilon;t})/(g_{\tau,\varepsilon;t})^2=\fL\fE\fA^4_{\tau,\varepsilon;1,t}\fB_{\varepsilon;t}(V_{\tau,\varepsilon;t}).
\end{equation}
Consequently, we obtain
\begin{equation}\label{eq:Polchinski_proof_g}
 1/g_{\tau,\varepsilon;t}=
 1/g_{\tau,\varepsilon;0}+\int_0^t \fL\fE\fA^4_{\tau,\varepsilon;1,s}\fB_{\varepsilon;s}(V_{\tau,\varepsilon;s})\,\rd s.
\end{equation}
Next, we observe that by Lemma~\ref{lem:R_L_identity} it holds
\begin{multline}
 \langle U^4\,\fL\fE\fA^4_{\tau,\varepsilon;1,s}\fB_{\varepsilon;s}(V_{\tau,\varepsilon;s})+\fR\fE\fA^4_{\tau,\varepsilon;1,s}\fB_{\varepsilon;s}(V_{\tau,\varepsilon;s}),(\fJ\varphi)^{\otimes 4}\rangle
 \\
 =
 \langle \fE\fA^4_{\tau,\varepsilon;1,s}\fB_{\varepsilon;s}(V_{\tau,\varepsilon;s}),(\fJ\varphi)^{\otimes 4}\rangle.
\end{multline}
We also note that it follows from Def.~\ref{dfn:A_map} of the map $\fC_{\tau,\varepsilon;1,t}$ that $\fE\fC_{\tau,\varepsilon;1,t}^4W_{\tau,\varepsilon;t}$ depends only on $W_{\tau,\varepsilon;t}^m$ with $m>4$. Hence, taking into account~\eqref{eq:polchinski_proof_m} we obtain
\begin{equation}
 \fE\fC_{\tau,\varepsilon;1,t}^4W_{\tau,\varepsilon;t}
 =
 \fE\fC_{\tau,\varepsilon;1,t}^4  \int_0^t
 \fE\fA_{\tau,\varepsilon;t,s}\fB_{\varepsilon;s}(V_{\tau,\varepsilon;s})\,\rd s.
\end{equation}
The identities of Remark~\ref{rem:fA_fE2} extend to the maps appearing below by finite-degree approximation and Lemma~\ref{lem:A_map_bound}, and we use these extensions without further comment.
Consequently, by Remark~\ref{rem:fA_fE2}~(C) and the equality $\fC_{\tau,\varepsilon;1,t}V=\fA_{\tau,\varepsilon;1,t}V-V$ we get
\begin{equation}
 \fE\fC_{\tau,\varepsilon;1,t}^4W_{\tau,\varepsilon;t}
 =
 \int_0^t
 \fE\fA^4_{\tau,\varepsilon;1,s}\fB_{\varepsilon;s}(V_{\tau,\varepsilon;s})\,\rd s
 -
 \int_0^t
 \fE\fA^4_{\tau,\varepsilon;t,s}\fB_{\varepsilon;s}(V_{\tau,\varepsilon;s})\,\rd s.
\end{equation}
Using the above identities we arrive at
\begin{multline}
 \langle V^4_{\tau,\varepsilon;t} ,(\fJ\varphi)^{\otimes 4}\rangle
 =
 \langle U^4\,1/g_{\tau,\varepsilon;t} ,(\fJ\varphi)^{\otimes 4}\rangle
 +
 \langle W^4_{\tau,\varepsilon;t} ,(\fJ\varphi)^{\otimes 4}\rangle
 \\
 =
 \langle U^4\,1/g_{\tau,\varepsilon;0},(\fJ\varphi)^{\otimes 4}\rangle
 +
 \int_0^t
 \langle\fE\fA^4_{\tau,\varepsilon;1,s}\fB_{\varepsilon;s}(V_{\tau,\varepsilon;s}),(\fJ\varphi)^{\otimes 4}\rangle\,\rd s
 -
 \langle\fE\fC_{\tau,\varepsilon;1,t}^4W_{\tau,\varepsilon;t},(\fJ\varphi)^{\otimes 4}\rangle
 \\
 =
 \langle\fE\fA^4_{\tau,\varepsilon;t,0}V_{\tau,\varepsilon;0},(\fJ\varphi)^{\otimes 4}\rangle
 +
 \int_0^t
 \langle\fE\fA^4_{\tau,\varepsilon;t,s}\fB_{\varepsilon;s}(V_{\tau,\varepsilon;s}),(\fJ\varphi)^{\otimes 4}\rangle\,\rd s.
\end{multline}
This proves~\eqref{eq:polchinski_aux} for $m=4$. It remains to show that~\eqref{eq:polchinski_aux} holds true for $m=2$. Observe that by the fixed point equation we obtain
\begin{equation}\label{eq:polchinski_proof_2}
\begin{aligned}
 W^2_{\tau,\varepsilon;t}
 =&
 \int_0^t
 \fR\fE\fA^2_{\tau,\varepsilon;1,s}\fB_{\varepsilon;s}(V_{\tau,\varepsilon;s})\,\rd s
 -
 \fE\fC_{\tau,\varepsilon;1,t}^2W_{\tau,\varepsilon;t},
 \\
 r_{\tau,\varepsilon;t}
 =&
 -\fL\fE\fA^2_{\tau,\varepsilon;1,t} U(1/g_{\tau,\varepsilon;t},0,0)
 -
 \int_t^1
 \fL\fE\fA^2_{\tau,\varepsilon;1,s}\fB_{\varepsilon;s}(V_{\tau,\varepsilon;s})\,\rd s,
 \\
 z_{\tau,\varepsilon;t}
 =&
 \int_0^t
 \fL_\partial\fE\fA^2_{\tau,\varepsilon;1,s}\fB_{\varepsilon;s}(V_{\tau,\varepsilon;s})\,\rd s.
\end{aligned} 
\end{equation}
The second of the above equations implies
\begin{equation}
 r_{\tau,\varepsilon;0}
 =
 -
 \fL\fE\fA^2_{\tau,\varepsilon;1,0} U(1/g_{\tau,\varepsilon;0},0,0)
 -
 \int_0^1
 \fL\fE\fA^2_{\tau,\varepsilon;1,s}\fB_{\varepsilon;s}(V_{\tau,\varepsilon;s})\,\rd s,
\end{equation}
and
\begin{multline}\label{eq:r_relation_pol}
 r_{\tau,\varepsilon;t}
 =
 r_{\tau,\varepsilon;0}
 +
 \fL\fE\fA^2_{\tau,\varepsilon;t,0} U(1/g_{\tau,\varepsilon;0},0,0)
 +
 \fL\fE\fA^2_{\tau,\varepsilon;1,t} U(1/g_{\tau,\varepsilon;0}-1/g_{\tau,\varepsilon;t},0,0)
 \\
 +
 \int_0^t
 \fL\fE\fA^2_{\tau,\varepsilon;1,s}\fB_{\varepsilon;s}(V_{\tau,\varepsilon;s})\,\rd s,
\end{multline}
where we used the identity
\begin{equation}
 \fE\fA^2_{\tau,\varepsilon;1,0} U(1/g_{\tau,\varepsilon;0},0,0)
 =
 \fE\fA^2_{\tau,\varepsilon;1,t} U(1/g_{\tau,\varepsilon;0},0,0)
 +
 \fE\fA^2_{\tau,\varepsilon;t,0} U(1/g_{\tau,\varepsilon;0},0,0).
\end{equation}
Using Def.~\ref{dfn:A_map} of the maps and Remark~\ref{rem:fA_fE2}~(C) we show that
\begin{multline}\label{eq:polchinski_C_2}
 \fE\fC_{\tau,\varepsilon;1,t}^2 W_{\tau,\varepsilon;t}
 =
 \fE\fC_{\tau,\varepsilon;1,t}^2 V_{\tau,\varepsilon;t}
 -
 \fE\fC_{\tau,\varepsilon;1,t}^2 U(1/g_{\tau,\varepsilon;t},r_{\tau,\varepsilon;t},z_{\tau,\varepsilon;t})
 \\
 =
 \fE\fC_{\tau,\varepsilon;1,t}^2 V_{\tau,\varepsilon;t}
 -
 \fE\fA_{\tau,\varepsilon;1,t}^2 U(1/g_{\tau,\varepsilon;t},0,0)
\end{multline}
and
\begin{multline}\label{eq:polchinski_C_2_init}
 \fE\fC_{\tau,\varepsilon;1,t}^2\fE\fA_{\tau,\varepsilon;t,0}V_{\tau,\varepsilon;0}
 =
 \fE\fC_{\tau,\varepsilon;1,t}^2\fE\fA_{\tau,\varepsilon;t,0} U(1/g_{\tau,\varepsilon;0},0,0)
 \\
 =
 (\fE\fA^2_{\tau,\varepsilon;1,0}-\fE\fA^2_{\tau,\varepsilon;t,0})U(1/g_{\tau,\varepsilon;0},0,0)
 =
 \fE\fA^2_{\tau,\varepsilon;1,t}U(1/g_{\tau,\varepsilon;0},0,0).
\end{multline}
Moreover, we observe that $\fE\fC_{\tau,\varepsilon;1,t}^2 V_{\tau,\varepsilon;t}$ depends only on $V_{\tau,\varepsilon;t}^m$ with $m>2$. Hence, by~\eqref{eq:polchinski_aux} with $m>2$, which has already been proved to be true, and~\eqref{eq:polchinski_C_2_init} we obtain
\begin{multline}
 \langle \fE\fC_{\tau,\varepsilon;1,t}^2 V_{\tau,\varepsilon;t},(\fJ\varphi)^{\otimes 2}\rangle
 =
 \langle\fE\fA^2_{\tau,\varepsilon;1,t}U(1/g_{\tau,\varepsilon;0},0,0),(\fJ\varphi)^{\otimes 2}\rangle
 \\+
 \int_0^t\langle \fE\fC_{\tau,\varepsilon;1,t}^2
 \fE\fA_{\tau,\varepsilon;t,s}\fB_{\varepsilon;s}(V_{\tau,\varepsilon;s}),(\fJ\varphi)^{\otimes 2}\rangle\,\rd s.
\end{multline}
Consequently, by Remark~\ref{rem:fA_fE2}~(C) and~\eqref{eq:polchinski_C_2} it holds
\begin{multline}
 \langle \fE\fC_{\tau,\varepsilon;1,t}^2 W_{\tau,\varepsilon;t},(\fJ\varphi)^{\otimes 2}\rangle
 =
 \langle\fE\fA^2_{\tau,\varepsilon;1,t}U(1/g_{\tau,\varepsilon;0}-1/g_{\tau,\varepsilon;t},0,0),(\fJ\varphi)^{\otimes 2}\rangle
 \\
 +
 \int_0^t\langle
 (\fE\fA^2_{\tau,\varepsilon;1,s}-\fE\fA^2_{\tau,\varepsilon;t,s})\fB_{\varepsilon;s}(V_{\tau,\varepsilon;s}),(\fJ\varphi)^{\otimes 2}\rangle\,\rd s.
\end{multline}
Using Lemma~\ref{lem:R_L_identity} and the fact that $\fR\fE\fA^2_{\tau,\varepsilon;1,t}U(1/g_{\tau,\varepsilon;0}-1/g_{\tau,\varepsilon;t},0,0)=0$ as well as $\fL_\partial\fE\fA^2_{\tau,\varepsilon;1,t}U(1/g_{\tau,\varepsilon;0}-1/g_{\tau,\varepsilon;t},0,0)=0$ we obtain
\begin{multline}\label{eq:polchinski_fC_W}
 \langle \fE\fC_{\tau,\varepsilon;1,t}^2 W_{\tau,\varepsilon;t},(\fJ\varphi)^{\otimes 2}\rangle
 =
 \langle U^2\fL\fE\fA^2_{\tau,\varepsilon;1,t}U(1/g_{\tau,\varepsilon;0}-1/g_{\tau,\varepsilon;t},0,0),(\fJ\varphi)^{\otimes 2}\rangle
 \\
 +
 \int_0^t\langle
 (\fE\fA^2_{\tau,\varepsilon;1,s}-\fE\fA^2_{\tau,\varepsilon;t,s})\fB_{\varepsilon;s}(V_{\tau,\varepsilon;s}),(\fJ\varphi)^{\otimes 2}\rangle\,\rd s.
\end{multline}
Next, we observe that by Lemma~\ref{lem:R_L_identity} it holds
\begin{multline}
 \langle U^2\,\fL\fE\fA^2_{\tau,\varepsilon;1,s}\fB_{\varepsilon;s}(V_{\tau,\varepsilon;s})
 +
 U^2_\partial\,\fL_\partial\fE\fA^2_{\tau,\varepsilon;1,s}\fB_{\varepsilon;s}(V_{\tau,\varepsilon;s})
 +\fR\fE\fA^2_{\tau,\varepsilon;1,s}\fB_{\varepsilon;s}(V_{\tau,\varepsilon;s}),(\fJ\varphi)^{\otimes 2}\rangle
 \\
 =
 \langle \fE\fA^2_{\tau,\varepsilon;1,s}\fB_{\varepsilon;s}(V_{\tau,\varepsilon;s}),(\fJ\varphi)^{\otimes 2}\rangle.
\end{multline}
Hence, by the first and third of Eqs.~\eqref{eq:polchinski_proof_2} and~\eqref{eq:r_relation_pol} we obtain
\begin{multline}
 \langle V^2_{\tau,\varepsilon;t} ,(\fJ\varphi)^{\otimes 2}\rangle
 =
 \langle U^2\,r_{\tau,\varepsilon;t} + U^2_\partial\,z_{\tau,\varepsilon;t} 
 +
 W^2_{\tau,\varepsilon;t} ,(\fJ\varphi)^{\otimes 2}\rangle
 \\[5pt]
 =
 \langle U^2\,r_{\tau,\varepsilon;0}
 +
 U^2\,
 \fL\fE\fA^2_{\tau,\varepsilon;t,0} U(1/g_{\tau,\varepsilon;0},0,0)
 +
 U^2\,
 \fL\fE\fA^2_{\tau,\varepsilon;1,t} U(1/g_{\tau,\varepsilon;0}-1/g_{\tau,\varepsilon;t},0,0),(\fJ\varphi)^{\otimes 2}\rangle
 \\[2pt]
 +
 \int_0^t
 \langle\fE\fA^2_{\tau,\varepsilon;1,s}\fB_{\varepsilon;s}(V_{\tau,\varepsilon;s}),(\fJ\varphi)^{\otimes 2}\rangle\,\rd s
 -
 \langle \fE\fC_{\tau,\varepsilon;1,t}^2 W_{\tau,\varepsilon;t},(\fJ\varphi)^{\otimes 2}\rangle.
\end{multline}
Consequently, by~\eqref{eq:polchinski_fC_W} we have
\begin{multline}
 \langle V^2_{\tau,\varepsilon;t} ,(\fJ\varphi)^{\otimes 2}\rangle
 =
 \langle U^2\,r_{\tau,\varepsilon;0}
 +
 U^2\,
 \fL\fE\fA^2_{\tau,\varepsilon;t,0} U(1/g_{\tau,\varepsilon;0},0,0)
 ,(\fJ\varphi)^{\otimes 2}\rangle
 \\
 +
 \int_0^t
 \langle\fE\fA^2_{\tau,\varepsilon;t,s}\fB_{\varepsilon;s}(V_{\tau,\varepsilon;s}),(\fJ\varphi)^{\otimes 2}\rangle\,\rd s.
\end{multline}
Hence, Eq.~\eqref{eq:polchinski_aux} is true for $m=2$. This finishes the proof.
\end{proof}

\begin{cor}\label{cor:polchinski}
Let $\tau,\varepsilon\in(0,1]$ and $\lambda\in(0,\lambda_\star]$, where $\lambda_\star$ is as in Corollary~\ref{cor:contraction}. For $t\in[0,1]$ define $V_{\tau,\varepsilon;t}=(V_{\tau,\varepsilon;t}^m)_{m\in\bN_+}\in\sN$ as in the statement of Theorem~\ref{thm:polchinski}. For $m\in\bN_+$ and $t\in[0,1]$ define the antisymmetric distribution $U^m_{\tau,\varepsilon;t}\in\sS'(\bT_\tau^{2m})^{\bFF^m}$ by the equality
\begin{equation}\label{eq:relation_U_V}
 \langle U^m_{\tau,\varepsilon;t},\phi^{\otimes m}\rangle_\tau 
 =
 \langle V^m_{\tau,\varepsilon;t},\fJ(\chi_\tau\phi)\otimes (\fJ\phi)^{\otimes (m-1)}\rangle
\end{equation}
for all $\phi\in C^\infty(\bT_\tau^2)^\bFF\otimes_{\mathrm{alg}}\sG^-$. For $t\in[0,1]$ define the functional $\tilde U_{\tau,\varepsilon;t}\in\sN(C^\infty(\bT_\tau^2)^\bFF)$ by the equality
\begin{equation}
 \tilde U_{\tau,\varepsilon;t}(\phi):= \sum_{m\in\bN_+} \langle U^m_{\tau,\varepsilon;t},\phi^{\otimes m}\rangle_\tau \in \sG,
 \qquad
 \phi\in C^\infty(\bT_\tau^2)^\bFF\otimes_{\mathrm{alg}}\sG^-.
\end{equation}
For $t\in[0,1]$ define $U^0\in\bC$ by the equality 
\begin{equation}\label{eq:U_0}
 U^0_{\tau,\varepsilon;t} = 
 \fE U_{\tau,\varepsilon}(\varPsi_{\tau,\varepsilon;t,0}) 
 + 
 \frac{1}{2} \int_0^t \fE\langle \rD_\phi  \tilde U_{\tau,\varepsilon;s}(\varPsi_{\tau,\varepsilon;t,s}),\dot G_{\varepsilon;s}\ast \rD_\phi  \tilde U_{\tau,\varepsilon;s}(\varPsi_{\tau,\varepsilon;t,s})\rangle_\tau\,\rd s. 
\end{equation}
For $t\in[0,1]$ define the functional $U_{\tau,\varepsilon;t}\in\sN(C^\infty(\bT_\tau^2)^\bFF)$ by the equality
\begin{equation}
 U_{\tau,\varepsilon;t}(\phi):= \sum_{m\in\bN_0} \langle U^m_{\tau,\varepsilon;t},\phi^{\otimes m}\rangle_\tau \in \sG,
 \qquad
 \phi\in C^\infty(\bT_\tau^2)^\bFF\otimes_{\mathrm{alg}}\sG^-.
\end{equation}
Finally, define the functional $U_{\tau,\varepsilon}\in\sN(C^\infty(\bT_\tau^2)^\bFF)$
as specified in Def.~\ref{dfn:measures} with $g_{\tau,\varepsilon}:=g_{\tau,\varepsilon;0}$ and $r_{\tau,\varepsilon}:=r_{\tau,\varepsilon;0}$, where $g_{\tau,\varepsilon;0}$ and $r_{\tau,\varepsilon;0}$ are as in the statement of Theorem~\ref{thm:polchinski}.
Then the map
$
 [0,1] \ni t\mapsto  U_{\tau,\varepsilon;t}\in \sN(C^\infty(\bT_\tau^2)^\bFF)
$
is continuous and is a solution of the integral form of the Polchinski equation
\begin{multline}\label{eq:polchinski_corollary}
  U_{\tau,\varepsilon;t}(\phi) = \fE  U_{\tau,\varepsilon}(\varPsi_{\tau,\varepsilon;t,0}+\phi) 
 \\
 + \frac{1}{2}  \int_0^t \fE\langle \rD_\phi  U_{\tau,\varepsilon;s}(\varPsi_{\tau,\varepsilon;t,s}+\phi),\dot G_{\varepsilon;s}\ast \rD_\phi  U_{\tau,\varepsilon;s}(\varPsi_{\tau,\varepsilon;t,s}+\phi)\rangle_\tau\,\rd s
\end{multline}
for all $\phi\in C^\infty(\bT_\tau^2)^\bFF \otimes_{\mathrm{alg}}\sG^-$.
Moreover, $U_{\tau,\varepsilon;t}^m=0$ for every odd $m$, since $V_{\tau,\varepsilon;t}$ is compatible with the flow of charge and flavor.
\end{cor}

\begin{proof}
Since the distributions involved are antisymmetric, Theorem~\ref{thm:polchinski}, the graded polarization argument of Remark~\ref{rem:graded_uniqueness}, and the density of finite sums of tensor products in the Schwartz space imply
\begin{equation}
 \langle V^m_{\tau,\varepsilon;t} ,\fJ^{\otimes m}\psi\rangle
 =
 \langle\fE\fA^m_{\tau,\varepsilon;t,0}V_{\tau,\varepsilon;0},\fJ^{\otimes m}\psi\rangle
 +
 \int_0^t
 \langle\fE\fA^m_{\tau,\varepsilon;t,s}\fB_{\varepsilon;s}(V_{\tau,\varepsilon;s}),\fJ^{\otimes m}\psi\rangle\,\rd s
\end{equation}
for all $t\in(0,1]$, $m\in\bN_+$ and $\psi\in\sS(\bR^{2m})^{\bFF^m}$. Hence, by Def.~\ref{dfn:A_map} and~\ref{dfn:B_map} of the maps $\fA^m_{\tau,\varepsilon;t,s}$ and $\fB_{\varepsilon;s}$ we have
\begin{multline}
 U_{\tau,\varepsilon;t}(\phi) = \fE U_{\tau,\varepsilon}(\varPsi_{\tau,\varepsilon;t,0}+\phi) 
 \\
 + \frac{1}{2} \int_0^t \fE\langle \rD_\phi U_{\tau,\varepsilon;s}(\varPsi_{\tau,\varepsilon;t,s}+\phi),\dot G_{\varepsilon;s}\ast \rD_\phi U_{\tau,\varepsilon;s}(\varPsi_{\tau,\varepsilon;t,s}+\phi)\rangle_\tau\,\rd s 
\end{multline}
for all $t\in(0,1]$ and $\phi\in C^\infty(\bT_\tau^2)^\bFF\otimes_{\mathrm{alg}}\sG^-$. Finally, to prove the continuity of the map $
 [0,1] \ni t\mapsto  U_{\tau,\varepsilon;t}\in \sN(C^\infty(\bT_\tau^2)^\bFF)
$ we observe that the map is constant on $[0,\varepsilon]$ and use its definition in terms of the fixed point of the map $\fX_{\tau,\varepsilon;\Cdot}\,:\,\sY_{\tau,\varepsilon}\to \sY_{\tau,\varepsilon}$ as well as Def.~\ref{dfn:sX} of the set $\sY_{\tau,\varepsilon}$ and Def.~\ref{dfn:sW} of the space $\sW_{\tau,\varepsilon}^{\alpha,\beta;\gamma}$. This finishes the proof.
\end{proof}

\begin{rem}
For $\tau=\varepsilon=0$, the formula from Theorem~\ref{thm:polchinski} still defines
$V_t=(V_t^m)_{m\in\bN_+}$, but Corollary~\ref{cor:polchinski} is only stated in the presence of the cutoffs because the scalar component $U^0_{\tau,\varepsilon;t}$ is not needed, and is not defined by~\eqref{eq:U_0}, in the cutoff-free theory.
\end{rem}

\begin{rem}\label{rem:Z_not_zero}
In the presence of the  UV and IR cutoffs $\tau,\varepsilon\in(0,1]$ it is straightforward to show that $U^0_{\tau,\varepsilon;t}\in\bC$ given by~\eqref{eq:U_0} is well-defined for all $t\in[0,1]$. Consequently, for the choice of parameters $g_{\tau,\varepsilon}$ and $r_{\tau,\varepsilon}$ as in the above corollary it holds
\begin{equation}
 \fE\exp(U_{\tau,\varepsilon}(\varPsi_{\tau,\varepsilon;t,0}))
 =
 \exp(U_{\tau,\varepsilon;t}(0))
 =
 \exp(U^0_{\tau,\varepsilon;t})\neq 0.
\end{equation}
In particular,
\begin{equation}
 \fE\exp(U_{\tau,\varepsilon}(\varPsi_{\tau,\varepsilon}))
 =
 \exp(U_{\tau,\varepsilon;t=1}(0))
 =
 \exp(U^0_{\tau,\varepsilon;t=1})\neq 0.
\end{equation}
\end{rem}

\section{Convergence of Schwinger functions}\label{sec:convergence}

In this section we establish a relation between an effective potential $U_{\tau,\varepsilon;t}$ at the scale $t=1$ and the generating functional of the Schwinger functions and complete the proof of Theorem~\ref{thm:main}. Recall that the effective potential $U_{\tau,\varepsilon;t}$ was constructed in Sec.~\ref{sec:relation_polchinski} with the use of the fixed point $X_{\tau,\varepsilon;\Cdot}$ of the map $\fX_{\tau,\varepsilon;\Cdot}$ constructed in Sec.~\ref{sec:fixed_point}. 

\begin{lem}\label{lem:generating_functional}
Let $\tau,\varepsilon\in(0,1]$. Suppose that $\fE\exp(U_{\tau,\varepsilon}(\varPsi_{\tau,\varepsilon}))
 \neq 0$. The following equality
\begin{equation}
 \mu_{\tau,\varepsilon}(\exp(\langle\Cdot,\phi\rangle_\tau))
 =
 \exp(\langle -\phi,G_{\tau,\varepsilon}\ast_\tau\phi\rangle_\tau/2+U_{\tau,\varepsilon;1}(G_{\tau,\varepsilon}\ast_\tau\phi)-U_{\tau,\varepsilon;1}(0))
\end{equation}
holds for all $\phi\in C^\infty(\bT^2_\tau)^\bFF\otimes_{\mathrm{alg}}\sG^-$, where the interacting measure $\mu_{\tau,\varepsilon}$ was introduced in Def.~\ref{dfn:measures} and $U_{\tau,\varepsilon;1}$ is an effective potential at unit scale introduced in Def.~\ref{dfn:effective_potential}.
\end{lem}

\begin{proof}
First observe that by Def.~\ref{dfn:measures} of the measure $\mu_{\tau,\varepsilon}$ it holds
\begin{equation}
 \mu_{\tau,\varepsilon}(\exp(\langle\Cdot,\phi\rangle_\tau)) 
 = 
 \frac{\int \exp(-A_{\tau}(\psi_{\tau,\varepsilon})+U_{\tau,\varepsilon}(\vartheta_\varepsilon\ast\psi_{\tau,\varepsilon})+\langle\psi_{\tau,\varepsilon},\vartheta_\varepsilon\ast\phi\rangle_\tau)\,\rd\psi_{\tau,\varepsilon}}{\int \exp(-A_{\tau}(\psi_{\tau,\varepsilon})+U_{\tau,\varepsilon}(\vartheta_\varepsilon\ast\psi_{\tau,\varepsilon}))\,\rd\psi_{\tau,\varepsilon}}
\end{equation}
for all $\phi\in C^\infty(\bT^2_\tau)^\bFF\otimes_{\mathrm{alg}}\sG^-$.
Using Def.~\ref{dfn:measures} of the free action $A_\tau$ one proves the following identity
\begin{equation}
 \langle\psi,\phi\rangle_\tau
 -
 A_{\tau}(\psi) 
 =
 \frac{1}{2} \langle
 G\ast\phi,\phi\rangle_\tau
 -
 A_{\tau}(\psi-G\ast\phi) 
\end{equation}
for all $\psi,\phi\in C^\infty(\bT^2_\tau)^\bFF\otimes_{\mathrm{alg}}\sG^-$. Consequently, using the equalities 
\begin{equation}
\vartheta_\varepsilon \ast G \ast \vartheta_\varepsilon 
 =
 G_\varepsilon,
 \qquad
 G_\varepsilon\ast\phi = G_{\tau,\varepsilon}\ast_\tau\phi
\end{equation}
we obtain
\begin{multline}
 \mu_{\tau,\varepsilon}(\exp(\langle\Cdot,\phi\rangle_\tau)) 
 \\
 =
 \exp(-\langle \phi,G_{\tau,\varepsilon}\ast_\tau\phi\rangle_\tau/2)
 \,
 \frac{\int \exp(-A_{\tau}(\psi_{\tau,\varepsilon}-G\ast \vartheta_{\varepsilon}\ast\phi)+U_{\tau,\varepsilon}(\vartheta_\varepsilon\ast\psi_{\tau,\varepsilon}))\,\rd\psi_{\tau,\varepsilon}}{\int \exp(-A_{\tau}(\psi_{\tau,\varepsilon})+U_{\tau,\varepsilon}(\vartheta_\varepsilon\ast\psi_{\tau,\varepsilon}))\,\rd\psi_{\tau,\varepsilon}}
\end{multline}
for all $\phi\in C^\infty(\bT^2_\tau)^\bFF\otimes_{\mathrm{alg}}\sG^-$. Note that since $\fF_\tau\fP_\tau\vartheta_\varepsilon\subset\Lambda_{\tau,\varepsilon}$ it holds
\begin{equation}
 G\ast \vartheta_{\varepsilon}\ast\phi(x) = \tau^2\sum_{p\in\Lambda_{\tau,\varepsilon}} (\fF_\tau G\ast \vartheta_{\varepsilon}\ast\phi)(p)\,\re^{\ri p\cdot x}
\end{equation}
and $G\ast \vartheta_{\varepsilon}\ast\phi\in \mathcal{C}_{\tau,\varepsilon}$ for all $\phi\in C^\infty(\bT^2_\tau)^\bFF\otimes_{\mathrm{alg}}\sG^-$, where $\mathcal{C}_{\tau,\varepsilon}$ was introduced in Def.~\ref{dfn:GN_grassmann}. Taking advantage of the invariance under translation of the Berezin integral, which was stated in Lemma~\ref{lem:berezin_translation}, we show the following equality
\begin{multline}
 \mu_{\tau,\varepsilon}(\exp(\langle\Cdot,\phi\rangle_\tau)) 
 \\
 =
 \exp(-\langle \phi,G_{\tau,\varepsilon}\ast_\tau\phi\rangle_\tau/2)
 \,
 \frac{\int \exp(-A_{\tau}(\psi_{\tau,\varepsilon})+U_{\tau,\varepsilon}(\vartheta_\varepsilon\ast(\psi_{\tau,\varepsilon}+G\ast \vartheta_{\varepsilon}\ast\phi)))\,\rd\psi_{\tau,\varepsilon}}{\int \exp(-A_{\tau}(\psi_{\tau,\varepsilon})+U_{\tau,\varepsilon}(\vartheta_\varepsilon\ast\psi_{\tau,\varepsilon}))\,\rd\psi_{\tau,\varepsilon}}.
\end{multline}
The above equality can be rewritten as
\begin{multline}
 \mu_{\tau,\varepsilon}(\exp(\langle\Cdot,\phi\rangle_\tau)) 
 \\
 =
 \exp(-\langle \phi,G_{\tau,\varepsilon}\ast_\tau\phi\rangle_\tau/2)
 \,
 \frac{\int \exp(-A_{\tau}(\psi_{\tau,\varepsilon})+U_{\tau,\varepsilon}(\vartheta_\varepsilon\ast \psi_{\tau,\varepsilon}+G_{\tau,\varepsilon}\ast_\tau\phi))\,\rd\psi_{\tau,\varepsilon}}{\int \exp(-A_{\tau}(\psi_{\tau,\varepsilon})+U_{\tau,\varepsilon}(\vartheta_\varepsilon\ast\psi_{\tau,\varepsilon}))\,\rd\psi_{\tau,\varepsilon}}
\end{multline}
for all $\phi\in C^\infty(\bT^2_\tau)^\bFF\otimes_{\mathrm{alg}}\sG^-$. Using Lemma~\ref{lem:free_measure} we arrive at
\begin{equation}
 \mu_{\tau,\varepsilon}(\exp(\langle\Cdot,\phi\rangle_\tau))
 =
 \exp(-\langle \phi,G_{\tau,\varepsilon}\ast_\tau\phi\rangle_\tau/2)
 \,\frac{\fE(\exp(U_{\tau,\varepsilon}(\varPsi_{\tau,\varepsilon}+G_{\tau,\varepsilon}\ast_\tau\phi)))}{\fE(\exp(U_{\tau,\varepsilon}(\varPsi_{\tau,\varepsilon})))}.
\end{equation}
The statement follows now from the fact that $G_{\tau,\varepsilon}\ast_\tau\phi\in \mathcal{C}_{\tau,\varepsilon}$ for all $\phi\in C^\infty(\bT^2_\tau)^\bFF\otimes_{\mathrm{alg}}\sG^-$, $\varPsi_{\tau,\varepsilon}=\varPsi_{\tau,\varepsilon;1,0}$ and Def.~\ref{dfn:effective_potential} of an effective potential.
\end{proof}

\begin{dfn}\label{dfn:correlations}
Let $\tau,\varepsilon\in(0,1]$. Suppose that $\fE\exp(U_{\tau,\varepsilon}(\varPsi_{\tau,\varepsilon}))
 \neq 0$. We call
\begin{equation}
 \sS(\bR^2)^\bFF\otimes_{\mathrm{alg}}\sG^-\ni \varphi\mapsto Z_{\tau,\varepsilon}[\varphi]
 :=
 \mu_{\tau,\varepsilon}(\exp(\langle\Cdot,\varphi\rangle))
 \in\sG
\end{equation}
the generating functional. The Schwinger function of order $m\in\bN_+$ is defined by the equation
\begin{equation}
 S^m_{\tau,\varepsilon}:=
 (-1)^{m(m-1)/2} \rD_\varphi^m Z_{\tau,\varepsilon}[\varphi]\big|_{\varphi=0}
 \in\sS'(\bR^{2m})^{\bFF^m}.
\end{equation}
The truncated Schwinger function of order $m\in\bN_+$ is defined by the equation
\begin{equation}
 T^m_{\tau,\varepsilon}:=(-1)^{m(m-1)/2}\rD_\varphi^m \log(Z_{\tau,\varepsilon}[\varphi])\big|_{\varphi=0}
 \in\sS'(\bR^{2m})^{\bFF^m},
\end{equation}
where the Grassmann logarithm is obtained by a finite nilpotent expansion.
\end{dfn}
\begin{rem}
Note that by the periodicity convention of
Def.~\ref{dfn:distributions_torus}, we have $\langle\psi,\varphi\rangle=\langle\psi,P_\tau\varphi\rangle_\tau$.
\end{rem}

\begin{rem}\label{rem:S_T}
The Schwinger function can be expressed in terms of the truncated Schwinger functions using the formula
\begin{equation}
\langle S^m_{\tau,\varepsilon},\varphi^{\otimes m}\rangle 
=
\sum_{\pi\in\Pi_m}\prod_{S\in\pi}
\langle T^{|S|}_{\tau,\varepsilon},\varphi^{\otimes |S|}\rangle
\end{equation}
valid for all $\varphi\in\sS(\bR^2)^\bFF\otimes_{\mathrm{alg}}\sG^-$, where $\Pi_m$ denotes the set of partitions of the set $\{1,\ldots,m\}$ into disjoint subsets whose union coincides with $\{1,\ldots,m\}$ and $|S|$ denotes the number of elements of a set $S$.
\end{rem}

\begin{thm}\label{thm:convergence}
Fix $\lambda\in(0,\lambda_\star]$, where $\lambda_\star$ is as in Corollary~\ref{cor:contraction}. For all $\tau,\varepsilon\in[0,1]$ let 
\begin{equation}
 (0,1]\ni t\mapsto X_{\tau,\varepsilon;t}\equiv (g_{\tau,\varepsilon;t},r_{\tau,\varepsilon;t},z_{\tau,\varepsilon;t},W_{\tau,\varepsilon;t})\in \bR\times\bR\times\bR\times\sN
\end{equation}
be the fixed point of the map $\fX_{\tau,\varepsilon;\Cdot}\,:\,\sY_{\tau,\varepsilon}\to \sY_{\tau,\varepsilon}$, which was introduced in Def.~\ref{dfn:X_map}. Moreover, for $\tau,\varepsilon\in(0,1]$ set
\begin{equation}
 g_{\tau,\varepsilon}:=g_{\tau,\varepsilon;\varepsilon},
 \quad
 r_{\tau,\varepsilon}:=r_{\tau,\varepsilon;\varepsilon},
\end{equation} 
Suppose that for all $\tau,\varepsilon\in(0,1]$ the interacting measure $\mu_{\tau,\varepsilon}$ is defined as specified in Def.~\ref{dfn:measures} with the above choice of the parameters $g_{\tau,\varepsilon}$ and $r_{\tau,\varepsilon}$. Then for all $m\in\bN_+$ there exist $S^m,T^m\in \sS'(\bR^{2m})^{\bFF^m}$ such that: 
\begin{itemize}
\item[(A)] $\lim_{\tau,\varepsilon\searrow0}\langle T^m_{\tau,\varepsilon},\varphi_1\otimes\ldots\otimes\varphi_m\rangle 
 =
 \langle T^m,\varphi_1\otimes\ldots\otimes\varphi_m\rangle$ for all $\varphi_1,\ldots,\varphi_m\in\sS(\bR^2)^{\bFF}$,
 \item[(B)] $\lim_{\tau,\varepsilon\searrow0}\langle S^m_{\tau,\varepsilon},\varphi_1\otimes\ldots\otimes\varphi_m\rangle 
 =
 \langle S^m,\varphi_1\otimes\ldots\otimes\varphi_m\rangle$ for all $\varphi_1,\ldots,\varphi_m\in\sS(\bR^2)^{\bFF}$,

 \item[(C)] $S^m,T^m$ are invariant under the dual action associated with the action of the symmetries of the plane from Def.~\ref{dfn:symm_plane},
 
 \item[(D)] $T^m$ has the properties stated in Items~(C) and~(D) of Theorem~\ref{thm:main}.
 
 \item[(E)] $\lim_{|x|\to\infty}\exp(|x|^{1/2})\,|\langle T^m,\varphi_x\otimes\psi \rangle|=0$ for all $n\in\{1,\ldots,m-1\}$, $\varphi\in C^\infty_\rc(\bR^{2n})^{\bFF^n}$, \mbox{$\psi\in C^\infty_\rc(\bR^{2(m-n)})^{\bFF^{m-n}}$}, where for $x\in\bR^2$ we define $\varphi_x\in C^\infty_\rc(\bR^{2n})^{\bFF^n}$ by the equality $\varphi_x(y_1,\ldots,y_n):=\varphi(y_1-x,\ldots,y_n-x)$ for all $y_1,\ldots,y_n\in\bR^2$.
\end{itemize}
\end{thm}
\begin{rem}
Note that $g_{\tau,\varepsilon}=g_{\tau,\varepsilon;t}$, $r_{\tau,\varepsilon}=r_{\tau,\varepsilon;t}$ for all $t\in(0,\varepsilon]$ by Remark~\ref{rem:fixed_point_epsilon}. 
\end{rem}
\begin{rem}
The exponent $1/2$ in Item~(E) of the above theorem has no significance. With some extra work, it is possible to prove exponential decay of the truncated correlations. To this end, one has to use a decomposition $(G_{\varepsilon,t})_{t\in[0,1]}$ of the covariance $G_\varepsilon$ at $\varepsilon=0$, which in contrast to the decomposition introduced in Sec.~\ref{sec:covariance_decomposition}, has exponential decay.
\end{rem}

\begin{proof}
First, note that by Corollary~\ref{cor:contraction} for all $\tau,\varepsilon\in[0,1]$ the map $\fX_{\tau,\varepsilon;\Cdot}\,:\,\sY_{\tau,\varepsilon}\to \sY_{\tau,\varepsilon}$ is well-defined and has a unique fixed point $X_{\tau,\varepsilon;\Cdot}\equiv (g_{\tau,\varepsilon;\Cdot},r_{\tau,\varepsilon;\Cdot},z_{\tau,\varepsilon;\Cdot},W_{\tau,\varepsilon;\Cdot})\in \sY_{\tau,\varepsilon}$. For all $\tau,\varepsilon\in[0,1]$ and $t\in(0,1]$ we define $V_{\tau,\varepsilon;t}=(V_{\tau,\varepsilon;t}^m)_{m\in\bN_+}\in\sN$ by the equation
\begin{equation}\label{eq:schwinger_V_def}
 V_{\tau,\varepsilon;t}[\varphi]:= 
 U(1/g_{\tau,\varepsilon;t},r_{\tau,\varepsilon;t},z_{\tau,\varepsilon;t})[\varphi]+W_{\tau,\varepsilon;t}[\varphi]
\end{equation}
for all $\varphi\in\sS(\bR^2)^\bFF\otimes_{\mathrm{alg}}\sG^-$. Recall that we omit $\tau$ if $\tau=0$ and we omit~$\varepsilon$ if $\varepsilon=0$. Moreover, for all $\tau,\varepsilon\in(0,1]$ and $t\in[0,1]$ we define $U_{\tau,\varepsilon;t}\in\sN(C^\infty(\bT_\tau^2)^\bFF)$ as in Corollary~\ref{cor:polchinski}. Then by Corollary~\ref{cor:polchinski} and Lemma~\ref{lem:polchinski2} for all $t\in[0,1]$ the functional $U_{\tau,\varepsilon;t}\in\sN(C^\infty(\bT_\tau^2)^\bFF)$ is an effective potential at scale $t$.  Since by Remark~\ref{rem:Z_not_zero} we have $\fE\exp(U_{\tau,\varepsilon}(\varPsi_{\tau,\varepsilon}))\neq 0$ using Lemma~\ref{lem:generating_functional} we conclude that
\begin{equation}\label{eq:generating_proof}
 \mu_{\tau,\varepsilon}(\exp(\langle\Cdot,\phi\rangle_\tau))
 =
 \exp(-\langle \phi,G_{\tau,\varepsilon}\ast_\tau\phi\rangle_\tau/2+U_{\tau,\varepsilon;1}(G_{\tau,\varepsilon}\ast_\tau\phi)-U_{\tau,\varepsilon;1}(0))
\end{equation}
for all $\phi\in C^\infty(\bT^2_\tau)^\bFF\otimes_{\mathrm{alg}}\sG^-$ and $\tau,\varepsilon\in(0,1]$.

Let us prove Items~(A) and~(B). We observe that it suffices to show that for all $m\in\bN_+$ there exist $S^m,T^m\in \sS'(\bR^{2m})^{\bFF^m}$ such that:
\begin{itemize}
\item[(A')]
$\lim_{\tau,\varepsilon\searrow0}\langle T^m_{\tau,\varepsilon},\varphi_1\otimes\ldots\otimes\varphi_m\rangle 
 =
 \langle T^m,\varphi_1\otimes\ldots\otimes\varphi_m\rangle$ for all $\varphi_1,\ldots,\varphi_m\in\sS(\bR^2)^\bFF$,
\item[(B')]
 $\lim_{\tau,\varepsilon\searrow0}\langle S^m_{\tau,\varepsilon},\varphi_1\otimes\ldots\otimes\varphi_m\rangle 
 =
 \langle S^m,\varphi_1\otimes\ldots\otimes\varphi_m\rangle$ for all $\varphi_1,\ldots,\varphi_m\in\sS(\bR^2)^\bFF$.
\end{itemize}
By Remark~\ref{rem:S_T} Item~(B') follows from Item~(A'). We proceed to the proof of Item~(A'). For $\varphi\in\sS(\bR^2)^\bFF\otimes_{\mathrm{alg}}\sG^-$ we apply~\eqref{eq:generating_proof} with $\phi=\fP_\tau\varphi\in C^\infty(\bT_\tau^2)^\bFF\otimes_{\mathrm{alg}}\sG^-$ to obtain
\begin{equation}
 \log(Z_{\tau,\varepsilon}[\varphi])
 =
 \log(\mu_{\tau,\varepsilon}(\exp(\langle\Cdot,\varphi\rangle)))
 =
 -\langle \varphi,G_{\tau,\varepsilon}\ast\varphi\rangle/2+ U_{\tau,\varepsilon;1}(G_{\tau,\varepsilon}\ast\varphi)
 -U_{\tau,\varepsilon;1}(0)
\end{equation}
for $\tau,\varepsilon\in(0,1]$. Note that by the translational invariance of $V^m_{\tau,\varepsilon;t}$ Eq.~\eqref{eq:relation_U_V} implies that
\begin{equation}
 \langle U^m_{\tau,\varepsilon;t},(G_{\tau,\varepsilon}\ast\varphi)^{\otimes m}\rangle_\tau 
 =
 \langle V^m_{\tau,\varepsilon;t},\fJ(G_{\varepsilon}\ast\varphi)\otimes (\fJ(G_{\tau,\varepsilon}\ast\varphi))^{\otimes (m-1)}\rangle.
\end{equation}
for all $m\in\bN_+$, $\tau,\varepsilon,t\in(0,1]$ and $\varphi\in\sS(\bR^2)^\bFF\otimes_{\mathrm{alg}}\sG^-$.
Consequently, using Def.~\ref{dfn:correlations} we arrive at
\begin{equation}
\begin{aligned}
 \langle T_{\tau,\varepsilon}^2 ,\varphi^{\otimes 2}\rangle&:=
 \langle \varphi,G_{\tau,\varepsilon}\ast\varphi\rangle
 -
 2\,\langle V^2_{\tau,\varepsilon;1},\fJ(G_{\varepsilon}\ast\varphi)\otimes (\fJ(G_{\tau,\varepsilon}\ast\varphi))\rangle,
 \\
 \langle T_{\tau,\varepsilon}^m,\varphi^{\otimes m}\rangle&:=(-1)^{m(m-1)/2}m!\, \langle V^m_{\tau,\varepsilon;1},\fJ(G_{\varepsilon}\ast\varphi)\otimes (\fJ(G_{\tau,\varepsilon}\ast\varphi))^{\otimes (m-1)}\rangle,\qquad m\in\bN_+\setminus\{2\},
\end{aligned} 
\end{equation}
for all $\tau,\varepsilon\in(0,1]$ and $\varphi\in\sS(\bR^2)^\bFF\otimes_{\mathrm{alg}}\sG^-$. For $m\in\bN_+$ we define an antisymmetric distribution $T^m\in \sS'(\bR^{2m})^{\bFF^m}$ by the equalities  
\begin{equation}
\begin{aligned}
 \langle T^2 ,\varphi^{\otimes 2}\rangle&:=
 \langle \varphi,G\ast\varphi\rangle
 -
 2\,\langle V^2_{1},(\fJ(G\ast\varphi))^{\otimes 2}\rangle,
 \\
 \langle T^m,\varphi^{\otimes m}\rangle&:=(-1)^{m(m-1)/2}m!\, \langle V^m_{1},(\fJ (G\ast\varphi))^{\otimes m}\rangle,\qquad m\in\bN_+\setminus\{2\},
\end{aligned} 
\end{equation}
for all $\varphi\in\sS(\bR^2)^\bFF\otimes_{\mathrm{alg}}\sG^-$, where $V_1=(V_{1}^m)_{m\in\bN_+}\in\sN$ is defined by~\eqref{eq:schwinger_V_def} with $\tau=0$, $\varepsilon=0$ and $t=1$. 
By Lemma~\ref{lem:G_L1_infty} and dominated convergence,
\begin{equation}
 \lim_{\tau,\varepsilon\searrow0} \langle \varphi_1,G_{\tau,\varepsilon}\ast\varphi_2\rangle
 =
 \langle \varphi_1,G\ast\varphi_2\rangle
\end{equation}
for all $\varphi_1,\varphi_2\in\sS(\bR^2)^\bFF$. Thus, in order to prove Item~(A') it suffices to show that
\begin{multline}
 \lim_{\tau,\varepsilon\searrow0}
 \langle V^m_{\tau,\varepsilon;1},\fJ(G_{\varepsilon}\ast\varphi_1)\otimes \fJ(G_{\tau,\varepsilon}\ast\varphi_2)\otimes\ldots\otimes \fJ(G_{\tau,\varepsilon}\ast\varphi_m)\rangle
 \\
 =
 \langle V^m_{1},\fJ(G\ast\varphi_1)\otimes\ldots\otimes \fJ(G\ast\varphi_m)\rangle
\end{multline}
for all $m\in\bN_+$ and $\varphi_1,\ldots,\varphi_m\in\sS(\bR^2)^\bFF$. To this end, we observe that
\begin{multline}
 |\langle V^m_{1},\fJ(G\ast\varphi_1)\otimes\ldots\otimes \fJ(G\ast\varphi_m)\rangle
 -
 \langle V^m_{\tau,\varepsilon;1},\fJ(G_{\varepsilon}\ast\varphi_1)\otimes \fJ(G_{\tau,\varepsilon}\ast\varphi_2)\otimes\ldots\otimes \fJ(G_{\tau,\varepsilon}\ast\varphi_m)\rangle|
 \\
 \leq
 |\langle V^m_{\tau,\varepsilon;1},
 \fJ(G\ast\varphi_1)\otimes\ldots\otimes \fJ(G\ast\varphi_m)
 -
 \fJ(G_{\varepsilon}\ast\varphi_1)\otimes \fJ(G_{\tau,\varepsilon}\ast\varphi_2)\otimes\ldots\otimes \fJ(G_{\tau,\varepsilon}\ast\varphi_m)\rangle|
 \\
 +
 |\langle V^m_{1}-V^m_{\tau,\varepsilon;1},\fJ(G\ast\varphi_1)\otimes\ldots\otimes \fJ(G\ast\varphi_m)\rangle|
\end{multline}

Set
\begin{equation}
	F_j:=\fJ(G\ast\varphi_j),\qquad
	F_{\varepsilon,1}:=\fJ(G_\varepsilon\ast\varphi_1),\qquad
	F_{\tau,\varepsilon;j}:=\fJ(G_{\tau,\varepsilon}\ast\varphi_j),
\end{equation}
where $j\in\{1,\ldots,m\}$ in the first definition and $j\in\{2,\ldots,m\}$ in the last one. We telescope
\begin{multline}\label{eq:correlations_telescopic}
	F_1\otimes\ldots\otimes F_m
	-
	F_{\varepsilon,1}\otimes F_{\tau,\varepsilon;2}
	\otimes
	\ldots
	\otimes F_{\tau,\varepsilon;m}
	=
	(F_1-F_{\varepsilon,1})\otimes F_2\otimes\ldots\otimes F_m
	\\
	+
	\sum_{j=2}^m  F_{\varepsilon,1}\otimes
	F_{\tau,\varepsilon;2}\otimes\ldots\otimes
	F_{\tau,\varepsilon;j-1}\otimes (F_j-F_{\tau,\varepsilon;j})
	\otimes
	F_{j+1}\otimes\ldots\otimes F_m.
\end{multline}
By Lemma~\ref{lem:G_L1_infty}, applied componentwise,
\begin{equation}
\|(F_1-F_{\varepsilon,1})/\tilde w\|_{L^1}\longrightarrow0,
\end{equation}
the families $F_{\varepsilon,1}/\tilde w$ and $F_{\tau,\varepsilon;j}$ are uniformly bounded in $L^1$ and $L^\infty$, respectively, and
\begin{equation}
F_{\tau,\varepsilon;j}\longrightarrow F_j
\end{equation}
locally uniformly for $j\geq2$.

Note that by the bounds~\eqref{eq:cor_bound_2} established in Corollary~\ref{cor:contraction} and
Def.~\ref{dfn:sN} and~\ref{dfn:V_space} there exists $C\in(0,\infty)$ such that
\begin{equation}\label{eq:bounds_V_N}
	\|V^m_{\tau,\varepsilon;1}\|_{\sN^m} \leq C\,,
	\qquad
	\|V^m_{1}-V^m_{\tau,\varepsilon;1}\|_{\tilde\sN^m}
	\leq C\,\lambda_{\tau\vee\varepsilon}^\kappa
\end{equation}
for all $\tau,\varepsilon\in(0,1]$ and $m\in\bN_+$.

Therefore, the first term in the telescopic sum~\eqref{eq:correlations_telescopic} converges to zero after pairing with $V^m_{\tau,\varepsilon;1}$ by Remark~\ref{rem:sN_estimate}. Consider one of the remaining terms and set
\begin{equation}
	H_{\tau,\varepsilon;j}:=F_j-F_{\tau,\varepsilon;j}.
\end{equation}
For $R,M\in(0,\infty)$, Remark~\ref{rem:sN_estimate} and $\rD(x_1,\ldots,x_m)\geq|x_j-x_1|$ give
\begin{multline}
	|\langle V^m_{\tau,\varepsilon;1},
	F_{\varepsilon,1}\otimes\cdots\otimes
	H_{\tau,\varepsilon;j}\otimes\cdots\otimes F_m\rangle|
	\\
	\leq
	C\|V^m_{\tau,\varepsilon;1}\|_{\tilde\sN^m}
	\left(  \exp(-R^\zeta)
	+
	\|H_{\tau,\varepsilon;j}\|_{L^\infty([-M-R,M+R]^2)}
	+
	\int_{\{|x|>M\}}
	\frac{|F_{\varepsilon,1}(x)|}{\tilde w(x)}\,\rd x  \right).
\end{multline}
Here and below the norms and absolute values of jet-valued functions are understood componentwise as in Remark~\ref{rem:Lp_K}. For fixed $R,M$ the middle term converges to zero as $\tau,\varepsilon\searrow0$, while the last term can be made arbitrarily small by choosing $M$ large, uniformly for sufficiently small $\varepsilon$. Letting finally $R\to\infty$ shows that every term in the telescopic sum converges to zero. Hence the first term in the above decomposition converges to zero. The second one converges to zero by Remark~\ref{rem:sN_estimate} and~\eqref{eq:bounds_V_N}. This proves Item~(A'). As argued above, this proves Item~(B') as well as Items~(A) and~(B).

Let us turn to the proof of Item~(C). Note that $V_1\in\sN$ defined by~\eqref{eq:schwinger_V_def} is invariant under the symmetries of the plane by Def.~\ref{dfn:U_functional},~\ref{dfn:sX} and~\ref{dfn:sW}. If $\varphi'$ denotes the image of $\varphi$ under the dual action associated with $\fT(x,R)$, then
\begin{equation}
	\fJ(G\ast\varphi')
	=
	\fT(x,R)\fJ(G\ast\varphi).
\end{equation}
Consequently, the invariance of $V_1$ and the formulas defining $T^m$ imply that $T^m$ is invariant under the dual action. The free term in $T^2$ is invariant by the same intertwining property and the definition of the dual action. Finally, the relation between Schwinger and truncated Schwinger functions implies the corresponding invariance of $S^m$. This proves Item~(C).

Item~(D) follows from the definition of $T^m$ and the boundary
conditions encoded in Defs.~\ref{dfn:X_map} and~\ref{dfn:R_map},
which give
\begin{equation}
 \fL V_1^2=0,
 \qquad
 \fL V_1^4=\lambda.
\end{equation}
Moreover, by Def.~\ref{dfn:covariance},
\begin{equation}
 (\fF G^{+,\varsigma_1;-,\varsigma_2})(0)
 =\delta_{\varsigma_1,\varsigma_2},
 \qquad
 (\fF G^{-,\varsigma_1;+,\varsigma_2})(0)
 =-\delta_{\varsigma_1,\varsigma_2},
\end{equation}
while the other two blocks vanish. Since $\ubar\psi=\psi^+$ and
$\bar\psi=\psi^-$, Grassmann anti-commutation shows that the
zero-momentum action of $G$ yields $\ubar\psi\cdot\bar\psi$ in the
quadratic term and preserves $(\ubar\psi\cdot\bar\psi)^2$ in the
quartic term. The formulas defining $T^2$ and $T^4$ therefore give
the two equalities in Item~(D).

In order to prove Item~(E) it is enough to show that
\begin{equation}
 \lim_{|x|\to\infty}\exp(|x|^{1/2})\,|\langle \varphi_x,G\ast\psi\rangle|=0
\end{equation}
for all $\varphi\in C^\infty_\rc(\bR^{2})^{\bFF}$ and $\psi\in C^\infty_\rc(\bR^{2})^{\bFF}$ and
\begin{equation}
 \lim_{|x|\to\infty}\exp(|x|^{1/2})\,|\langle V^m_{1},\fJ^{\otimes m} (G^{\otimes m}\ast(\varphi_x\otimes \psi))\rangle|
 =0
\end{equation}
for all $m\in\bN_+$, $n\in\{1,\ldots,m-1\}$, $\varphi\in C^\infty_\rc(\bR^{2n})^{\bFF^n}$, $\psi\in C^\infty_\rc(\bR^{2(m-n)})^{\bFF^{m-n}}$. Both of the above equalities follow from Remarks~\ref{rem:sN_estimate} and~\ref{rem:sN_G} and the bounds~\eqref{eq:bounds_V_N} for $\|V^m_{1}\|_{\tilde\sN^m}\leq \|V^m_{1}\|_{\sN^m}$. This finishes the proof.
\end{proof}

\begin{rem}\label{rem:Lp_K}
Given $p\in[1,\infty]$ and a finite tuple $\varphi=(\varphi^k)_{k\in\bK}$ of measurable functions on $\bR^2$ we define $\|\varphi\|_{L^p(\bR^2)^\bK}:=\sum_{k\in\bK}\|\varphi^k\|_{L^p(\bR^2)}$. 
\end{rem}

\begin{rem}\label{rem:sN_estimate}
Let $m\in\bN_+$ and $V=(V^{a,\sigma})_{a\in\bA^m,\sigma\in\bFF^m}\in\sN^m$. Recall that 
\begin{equation}
 \|V\|_{\tilde\sN^m}=\sum_{a\in\bA^m}\sum_{\sigma\in\bFF^m}\|\tilde w_1^{m}V^{a,\sigma}\|_{\sM^m}
 \leq
 \|V\|_{\sN^m},
\end{equation}
where
\begin{equation}
 \tilde w^{m}_{1}(x_1,\ldots,x_m) = 
 \tilde w(x_1) \exp(\rD(x_1,\ldots,x_m)^\zeta),
 \qquad
 \tilde w(x)=(1+|x|)^{-1/2},
 \qquad
 \zeta=4/5,
\end{equation}
for all $x,x_1,\ldots,x_m\in\bR^2$ and $\rD(x_1,\ldots,x_m)\geq 0$ is the diameter of $\{x_1,\ldots,x_m\}\subset\bR^2$. The following bound
\begin{equation}
 |\langle V,\varphi\rangle|\leq \|V\|_{\tilde\sN^m}
 \sup_{a\in\bA^m}\sup_{\sigma\in\bFF^m}\int_{\bR^2}\sup_{x_2,\ldots,x_m\in\bR^2} \frac{|\varphi^{a,\sigma}(x_1,\ldots,x_m)|}{\tilde w_1^{m}(x_1,\ldots,x_m)}\,\rd x_1
\end{equation}
is true for all $\varphi\in\sS(\bR^{2m})^{\bA^{m}\times\bFF^{m}}$.
\end{rem}

\begin{rem}\label{rem:sN_G}
By Lemma~\ref{lem:dot_G_infty} and Remark~\ref{rem:G}~(E),~(F) there exists $C\in(0,\infty)$ such that
\begin{equation}
 \|x\mapsto(1+|x|)^{1/2}\,\exp(|x|^\zeta)\,G(x)\|_{L^1(\bR^2)^{\bFF^2}}\leq C,\qquad
 \zeta=4/5.
\end{equation}
Consequently, by Lemma~\ref{lem:weights}~(c) the following bound
\begin{equation}
 \|G^{\otimes m}\ast V\|_{\tilde\sN^m}\leq C^m\,\|V\|_{\tilde\sN^m}
\end{equation}
holds for all $m\in\bN_+$ and $V\in\sN^m$.
\end{rem}
\begin{lem}\label{lem:G_L1_infty}
For every $\varphi\in\sS(\bR^2)^\bFF$ it holds
\begin{equation}
\lim_{\varepsilon\searrow0}
\|((G-G_{\varepsilon})\ast\varphi)/\tilde w\|_{L^1(\bR^2)^\bFF}
= 0 .
\end{equation}
Moreover,
\begin{equation}
\sup_{\tau,\varepsilon\in(0,1]}
\|G_{\tau,\varepsilon}\ast\varphi\|_{L^\infty(\bR^2)^\bFF}
< \infty
\end{equation}
and for every $R\in(0,\infty)$
\begin{equation}
\lim_{\tau,\varepsilon\searrow0}
\|(G-G_{\tau,\varepsilon})\ast\varphi\|_{L^\infty([-R,R]^2)^\bFF}
=0 .
\end{equation}
\end{lem}
\begin{proof}
By Remark~\ref{rem:w_tilde_young},
\begin{equation}
	\|((G-G_{\varepsilon})\ast\varphi)/\tilde w\|_{L^1}
	\leq
	\|(G-G_{\varepsilon})/\tilde w\|_{L^1}
	\|\varphi/\tilde w\|_{L^1}.
\end{equation}
Lemma~\ref{lem:dot_G_infty} and Remark~\ref{rem:G}~(F) imply
\begin{equation}
	\|G_{\varepsilon}/\tilde w\|_{L^1}\leq C,
	\qquad
	\|(G - G_{\varepsilon})/\tilde w\|_{L^1}\leq C\varepsilon
\end{equation}
for all sufficiently small $\varepsilon>0$.
This proves the first assertion. By Remark~\ref{rem:convolution_tau},
\begin{equation}
	G_{\tau,\varepsilon}\ast\varphi
	=
	G_\varepsilon\ast\fP_\tau\varphi.
\end{equation}
Since $\varphi\in\sS(\bR^2)^\bFF$,
\begin{equation}
	\sup_{\tau\in(0,1]}
	\|\fP_\tau\varphi\|_{L^\infty}<\infty
\end{equation}
and
\begin{equation}
	\lim_{\tau\searrow0}
	\|\fP_\tau\varphi - \varphi\|_{L^\infty([-R,R]^2)}=0
\end{equation}
for every $R\in(0,\infty)$. Together with the uniform weighted $L^1$ bound on $G_\varepsilon$, this implies the remaining two assertions.
\end{proof}

\section{Interacting Gross--Neveu field}\label{sec:interacting_field}

In this section we construct the interacting Gross--Neveu field in the non-commutative probability space $(\sF,\fE)$ introduced in Sec.~\ref{sec:free_field_decomposition}. To this end, we study the following stochastic differential equation
\begin{equation}\label{eq:fbsde}
 \varPhi_{\tau,\varepsilon;t} = \int_t^1 \dot G_{\tau,\varepsilon;s}\ast_\tau \rD U_{\tau,\varepsilon;s}[\varPhi_{\tau,\varepsilon;s}]\,\rd s + \varPsi_{\tau,\varepsilon;t}
\end{equation}
for $(0,1]\ni t\mapsto \,\varPhi_{\tau,\varepsilon;t}\in \sC^\bFF$. The functional $U_{\tau,\varepsilon;s}[\varphi]$ is the effective potential, $\varPsi_{\tau,\varepsilon;t}$ is the scale decomposition of the free field, $\ast_\tau$ denotes the convolution on the torus $\bT_\tau$ and $G_{\tau,\varepsilon;s}$ is the periodization of the scale decomposition of the propagator $G_{\varepsilon;t}$. For the motivation behind the above equation, see Remark 2.2 in~\cite{BB21} and Sec.~4 of~\cite{BBD23}. At fixed cutoffs the stochastic identification corresponding to~\cite[Proposition~3.10]{DFG25} is proved below in Lemma~\ref{lem:finite_cutoff_identification}. In particular, every adapted solution of~\eqref{eq:fbsde} satisfies
\begin{equation}\label{eq:stochastic_identification}
	\fE F(\varPhi_{\tau,\varepsilon;t}) = \frac{\fE F(\varPsi_{\tau,\varepsilon;t})\exp(U_{\tau,\varepsilon;t}[\varPsi_{\tau,\varepsilon;t}])}{\fE \exp(U_{\tau,\varepsilon;t}[\varPsi_{\tau,\varepsilon;t}])}
\end{equation}
for all $F\in\sN(C^\infty(\bT_\tau^2)^\bFF)$ and $t\in[0,1]$. Consequently, after constructing the solution of~\eqref{eq:fbsde}, Lemma~\ref{lem:interacting_measure} identifies its value at $t\in(0,\varepsilon)$ with the interacting Gross--Neveu field. We study its properties and complete the proof of Theorem~\ref{thm:main2}.

\begin{dfn}\label{dfn:sZ_norm}
Fix $\lambda\in(0,\lambda_\star]$, where $\lambda_\star$ is as in Corollary~\ref{cor:contraction}. The vector space $\sZ$ consists of functions $\varPhi_\Cdot\in C((0,1],\sC^{\bFF})$ such that the following norm
\begin{equation}
 \|\varPhi_\Cdot\|_{\sZ}:=\sup_{a\in\bA}\sup_{\sigma\in\bFF}\sup_{t\in(0,1]}\lambda^{\kappa}\, t^{1/2+|a|}\,\|\partial^a \varPhi^\sigma_t\|_\sC
\end{equation}
is finite. For $\varPhi_\Cdot\in\sZ$ we define
\begin{equation}
  \|\varPhi_\Cdot\|_{\tilde\sZ}:=\sup_{a\in\bA}\sup_{\sigma\in\bFF}\sup_{t\in(0,1]}\lambda^{\kappa}\,\lambda_t^{\kappa}\, t^{1/2+|a|}\,\|\tilde w\,\partial^a \varPhi^\sigma_t\|_\sC,
\end{equation}
where $\tilde w\in C(\bR^2)$ is defined by the equality $\tilde w(x)=(1+|x|)^{-1/2}$ for all $x\in\bR^2$. We also define
\begin{equation}
 \sU :=\{\varPhi_\Cdot\in \sZ\,|\,\|\varPhi_\Cdot\|_{\sZ}\leq 1\}.
\end{equation}

\end{dfn}
\begin{rem}
Recall that the vector space $(\sC,\|\Cdot\|_\sC)$ was introduced in Def.~\ref{dfn:sC}. It is easy to see that $(\sC,\|\Cdot\|_\sC)$ and $(\sZ,\|\Cdot\|_\sZ)$ are Banach spaces.
\end{rem}

\begin{dfn}\label{dfn:DV}
Let $\lambda\in(0,\lambda_\star]$, where $\lambda_\star$ is as in Corollary~\ref{cor:contraction}. Let
\begin{equation}
 (0,1]\ni t\mapsto X_{\tau,\varepsilon;t}\equiv (g_{\tau,\varepsilon;t},r_{\tau,\varepsilon;t},z_{\tau,\varepsilon;t},W_{\tau,\varepsilon;t})\in \bR\times\bR\times\bR\times\sN
\end{equation}
be the fixed point of the map $\fX_{\tau,\varepsilon;\Cdot}\,:\,\sY_{\tau,\varepsilon}\to \sY_{\tau,\varepsilon}$, introduced in Def.~\ref{dfn:X_map}. For $\tau,\varepsilon\in[0,1]$ and $t\in(0,1]$ define
\begin{equation}
 V_{\tau,\varepsilon;t}\equiv (V^m_{\tau,\varepsilon;t})_{m\in\bN_+}:= U(\theta^2_{\varepsilon;t}/g_{\tau,\varepsilon;t},r_{\tau,\varepsilon;t},z_{\tau,\varepsilon;t})+W_{\tau,\varepsilon;t}\in \sN
\end{equation}
and set $V_t:=V_{\tau,\varepsilon;t}$ with $\tau=0$, $\varepsilon=0$. For $\varphi\in\sS(\bR^2)^{\bA\times\bFF}\otimes_{\mathrm{alg}}\sG$ we define
\begin{equation}
 \rD V_{\tau,\varepsilon;t}[\varphi]\equiv (\rD^{1,a,\sigma} V_{\tau,\varepsilon;t}[\varphi])_{a\in\bA,\sigma\in\bFF} \in L^\infty(\bR^2)^{\bA\times\bFF}\otimes_{\mathrm{alg}}\sG
\end{equation}
by the equality
\begin{equation}
 \langle\rD V_{\tau,\varepsilon;t}[\varphi],\psi\rangle
 \equiv
\sum_{a\in\bA,\sigma\in\bFF} \langle\rD^{1,a,\sigma} V_{\tau,\varepsilon;t}[\varphi],\psi^{a,\sigma}\rangle
:=
 \sum_{m\in\bN_+} m\,\langle V^m_{\tau,\varepsilon;t},\varphi^{\otimes (m-1)}\otimes \psi\rangle\in\sG
\end{equation}
for all $\psi\in\sS(\bR^2)^{\bA\times\bFF}$.
\end{dfn}

\begin{lem}\label{lem:DV_continuity}
Let $V=(V^m)_{m\in\bN_+}\in\sN$ be translationally invariant and let
$\varphi\in\sC^{\bA\times\bFF}$. If
\begin{equation}
 \sum_{m\in\bN_+}
 m\left(
 \sum_{a\in\bA^m}\sum_{\sigma\in\bFF^m}
 \|V^{m,a,\sigma}\|_{\sM^m}
 \right)
 \|\varphi\|_{\sC}^{m-1}<\infty,
\end{equation}
then $\rD V[\varphi]\in\sC^{\bA\times\bFF}$ and
\begin{equation}
 \sum_{a\in\bA}\sum_{\sigma\in\bFF}
 \|\rD^{1,a,\sigma}V[\varphi]\|_{\sC}
 \leq
 \sum_{m\in\bN_+}
 m\left(
 \sum_{a\in\bA^m}\sum_{\sigma\in\bFF^m}
 \|V^{m,a,\sigma}\|_{\sM^m}
 \right)
 \|\varphi\|_{\sC}^{m-1},
\end{equation}
where $\|\varphi\|_{\mathscr C}:=\sup_{a,\sigma}\|\varphi^{a,\sigma}\|_{\mathscr C}$. In particular, the same kernel formula defines
$\rD V[\varphi]\in\sC^{\bA\times\bFF}$ for every $\varphi$
satisfying the above condition.
\end{lem}
\begin{proof}
Using the canonical translation-covariant representatives from
Remark~\ref{rem:canonical_kernel}, each homogeneous contribution to
$\rD V[\varphi]$ is bounded by the corresponding $\sM^m$-norm.
It is continuous by dominated convergence. The displayed estimate follows
by summing over the components, and its RHS implies uniform
convergence of the series over $m$.
\end{proof}

\begin{rem}
Recall that $\theta_{\varepsilon;t}$ was introduced in Def.~\ref{dfn:theta_ep_t} and $\theta_{\varepsilon;t}=1$ unless $\dot G_{\varepsilon;t}=0$.
\end{rem}

\begin{lem}\label{lem:DU_DV}
Let $\tau,\varepsilon\in(0,1]$, $t\in(0,1]$ and let $U_{\tau,\varepsilon;t}$ be the effective potential constructed in Corollary~\ref{cor:polchinski}. Then, for every $\phi\in C^\infty(\bT_\tau^2)^\bFF\otimes_{\mathrm{alg}}\sG^-$ and
$\sigma\in\bFF$,
\begin{equation}\label{eq:DU_DV}
\big( \dot G_{\tau,\varepsilon;t}\ast_\tau
\rD_\phi U_{\tau,\varepsilon;t}(\phi) \big)^\sigma
=
\sum_{a\in\bA}\sum_{\tilde\sigma\in\bFF}
(-1)^{|a|}
\partial^a\dot G_{\varepsilon;t}^{\sigma,\tilde\sigma}
\ast
\rD^{1,a,\tilde\sigma} V_{\tau,\varepsilon;t}[\fJ\phi].
\end{equation}
\end{lem}
\begin{proof}
Let
\begin{equation}
	\hat V_{\tau,\varepsilon;t}
	:=
	U(1/g_{\tau,\varepsilon;t},r_{\tau,\varepsilon;t},
	z_{\tau,\varepsilon;t})+W_{\tau,\varepsilon;t},
\end{equation}
so that $\hat V_{\tau,\varepsilon;t}$ is the functional appearing in Theorem~\ref{thm:polchinski}. By polarization of \eqref{eq:relation_U_V} and Def.~\ref{dfn:DV},
\begin{equation}
	\langle
	\rD_\phi U_{\tau,\varepsilon;t}(\phi),\eta
	\rangle_\tau
	=
	\sum_{a\in\bA}\sum_{\tilde\sigma\in\bFF}
	\langle
	\rD^{1,a,\tilde\sigma}
	\hat V_{\tau,\varepsilon;t}[\fJ\phi],
	\partial^a\eta^{\tilde\sigma}
	\rangle_\tau .
\end{equation}
Integration by parts and Remark~\ref{rem:convolution_tau} therefore give~\eqref{eq:DU_DV} with $V_{\tau,\varepsilon;t}$ replaced by $\hat V_{\tau,\varepsilon;t}$. Finally, $V_{\tau,\varepsilon;t}=\hat V_{\tau,\varepsilon;t}$ whenever $\dot G_{\varepsilon;t}\neq0$, while both sides of~\eqref{eq:DU_DV} vanish if $\dot G_{\varepsilon;t}=0$. This proves the claim.
\end{proof}

\begin{dfn}\label{dfn:Z_map}
For $\lambda\in(0,\lambda_\star]$, where $\lambda_\star$ is as in Corollary~\ref{cor:contraction}, and $\tau,\varepsilon\in[0,1]$ we define
\begin{equation}
	\fZ_{\tau,\varepsilon;\Cdot}\,:\,\sU\to\sZ
\end{equation}
by the equality
\begin{equation}
 \fZ_{\tau,\varepsilon;t}[\varPhi_\Cdot]
 :=
 \int_t^1  \fJ\dot G_{\varepsilon;s}\ast
 \rD V_{\tau,\varepsilon;s}[\fJ\varPhi_s]\,\rd s
 +
 \varPsi_{\tau,\varepsilon;t},
\end{equation}
where
\begin{equation}
\big(
\fJ\dot G_{\varepsilon;t}\ast
\rD V_{\tau,\varepsilon;t}[\fJ\varPhi_t]
\big)^\sigma
:=
\sum_{a\in\bA}\sum_{\tilde\sigma\in\bFF}
(-1)^{|a|}
\partial^a\dot G_{\varepsilon;t}^{\sigma,\tilde\sigma}
\ast
\rD^{1,a,\tilde\sigma}
V_{\tau,\varepsilon;t}[\fJ\varPhi_t].
\end{equation}
We omit $\tau$ and $\varepsilon$ if $\tau=0$ and $\varepsilon=0$.
\end{dfn}

\begin{rem}
The fact that the map $\fZ_{\tau,\varepsilon;\Cdot}$ is well-defined is non-trivial and is a consequence of the estimates stated below.
\end{rem}

\begin{thm}\label{thm:contraction_fbsde}
There exists $\lambda_\star\in(0,1]$ and $C\in(0,\infty)$ such that for all $\lambda\in(0,\lambda_\star]$, all $\tau,\varepsilon\in[0,1]$ and all $\varPhi_\Cdot,\varPsi_\Cdot\in \sU$ it holds:
\begin{itemize}
 \item[(A)] $\|s\mapsto\fZ_{\tau,\varepsilon;s}(\varPhi_\Cdot)\|_{\sZ}
 \leq
 C\,\lambda^\kappa$,
 \item[(B)] $\|s\mapsto(\fZ_{\tau,\varepsilon;s}(\varPhi_\Cdot)-\fZ_{\tau,\varepsilon;s}(\varPsi_\Cdot))\|_{\sZ}
 \leq
 C\,\lambda^\kappa\, \|\varPhi_\Cdot-\varPsi_\Cdot\|_{\sZ}$,
 \item[(C)] $\|s\mapsto(\fZ_{\tau,\varepsilon;s}(\varPhi_\Cdot)
 -
 \fZ_{\tau,\varepsilon;s}(\varPsi_\Cdot))\|_{\tilde\sZ}
 \leq
 C\,\lambda^\kappa\,
 \|\varPhi_\Cdot-\varPsi_\Cdot\|_{\tilde\sZ}$,

 \item[(D)] $\|s\mapsto(\fZ_s(\varPhi_\Cdot)
 -
 \fZ_{\tau,\varepsilon;s}(\varPhi_\Cdot))\|_{\tilde\sZ}
 \leq
 C\,\lambda_{\tau\vee\varepsilon}^\kappa$,
 \item[(E)] $\|s\mapsto \lambda_s^{40\kappa-1}(\fZ_{\tau,\varepsilon;s}(\varPhi_\Cdot)-\varPsi_{\tau,\varepsilon;s})\|_{\sZ}
 \leq
 C\,\lambda^{3\kappa}$.
\end{itemize}
\end{thm}
\begin{proof}
We write $\fZ_{\tau,\varepsilon;t}[\varPhi_\Cdot]=\fZ^{(1)}_{\tau,\varepsilon;t}[\varPhi_\Cdot]+\fZ^{(2)}_{\tau,\varepsilon;t}[\varPhi_\Cdot]$, where
\begin{equation}
 \fZ^{(1)}_{\tau,\varepsilon;t}[\varPhi_\Cdot]:=\varPsi_{\tau,\varepsilon;t},
 \qquad
 \fZ^{(2)}_{\tau,\varepsilon;t}[\varPhi_\Cdot]:=\int_t^1 \fJ\dot G_{\varepsilon;s}\ast \rD V_{\tau,\varepsilon;s}[\fJ\varPhi_s]\,\rd s,
\end{equation}
and prove the bounds~(A),~(B),~(C),~(D) separately for $\fZ^{(1)}_{\tau,\varepsilon;t}$ and $\fZ^{(2)}_{\tau,\varepsilon;t}$. The bound~(E) involves only $\fZ^{(2)}_{\tau,\varepsilon;t}$. The bounds~(A),~(B),~(C),~(D) for $\fZ^{(1)}_{\tau,\varepsilon;t}$ follow immediately from Lemma~\ref{lem:Psi_bounds}. Thus, it remains to prove the bounds for $\fZ^{(2)}_{\tau,\varepsilon;t}$.

First note that the weights introduced in Def.~\ref{dfn:weights} satisfy the estimate
\begin{equation}\label{eq:Z_weights}
 \tilde w(x_1)\leq\tilde w^m_s(x_1,\ldots,x_m)\leq w^m_s(x_1,\ldots,x_m)\tilde w(x_m)
\end{equation}
for all $m\in\bN_+$, $x_1,\ldots,x_m\in\bR^2$ and $s\in[0,1]$. Using Lemma~\ref{lem:H_dot_estimates}, Def.~\ref{dfn:sC} of $\|\Cdot\|_{\sC}$ and $\|\Cdot\|_{\tilde\sC}$ and the above estimate with $m=2$ we obtain
\begin{equation}
\begin{gathered}
 \|\partial^a\dot G_{\varepsilon;s}\ast\varphi\|_{\sC} \leq C \,s^{-|a|}\,\|\varphi\|_\sC,
 \qquad
 \|\partial^a\dot G_{\varepsilon;s}\ast\varphi\|_{\tilde\sC} \leq C \,s^{-|a|}\,\|\varphi\|_{\tilde\sC},
 \\
 \|\partial^a(\dot G^{\sigma}_{s}-\dot G^{\sigma}_{\varepsilon;s})\ast \varphi\|_{\tilde\sC} \leq C\,\lambda_\varepsilon^\kappa\,\lambda_s^{-\kappa}\,s^{-|a|}\,\|\varphi\|_{\tilde\sC}.
\end{gathered}
\end{equation}

Set $\gamma:=1-40\kappa$. For
$F=(F^{a,\sigma})_{a\in\bA,\sigma\in\bFF}$ define
\begin{equation}
 \|F\|_{s;\sC}
 :=
 \sum_{a\in\bA}\sum_{\sigma\in\bFF}
 s^{-|a|}\|F^{a,\sigma}\|_{\sC},
 \qquad
 \|F\|_{s;\tilde\sC}
 :=
 \sum_{a\in\bA}\sum_{\sigma\in\bFF}
 s^{-|a|}\|F^{a,\sigma}\|_{\tilde\sC}.
\end{equation}
For $\varPhi_\Cdot\in\sZ$ let
\begin{equation}
 P_s(\varPhi_\Cdot)
 :=
 \sup_{a\in\bA}\sup_{\sigma\in\bFF}
 s^{1/2+|a|}
 \|\partial^a\varPhi_s^\sigma\|_{\sC},
 \qquad
 \widetilde P_s(\varPhi_\Cdot)
 :=
 \sup_{a\in\bA}\sup_{\sigma\in\bFF}
 s^{1/2+|a|}
 \|\tilde w\,\partial^a\varPhi_s^\sigma\|_{\sC}.
\end{equation}
By Def.~\ref{dfn:sZ_norm},
\begin{equation}\label{eq:FBSDE_P_bound}
 P_s(\varPhi_\Cdot)
 \leq
 \lambda^{-\kappa}\|\varPhi_\Cdot\|_{\sZ},
 \qquad
 \widetilde P_s(\varPhi_\Cdot)
 \leq
 \lambda^{-\kappa}\lambda_s^{-\kappa}
 \|\varPhi_\Cdot\|_{\tilde\sZ}.
\end{equation}
For $V_s\in\sN$ set
\begin{equation}
 K_m(V_s)
 :=
 \sum_{a\in\bA^m}\sum_{\sigma\in\bFF^m}
 s^{-|a|}\|w_s^mV_s^{m,a,\sigma}\|_{\sM^m},
 \qquad
 \widetilde K_m(V_s)
 :=
 \sum_{a\in\bA^m}\sum_{\sigma\in\bFF^m}
 s^{-|a|}\|\tilde w_s^mV_s^{m,a,\sigma}\|_{\sM^m}.
\end{equation}
Using Def.~\ref{dfn:DV} componentwise, and using
Lemma~\ref{lem:DV_continuity} to justify that the resulting functions
belong to $\sC$, we obtain
\begin{equation}\label{eq:FBSDE_DV}
 \|\rD V_s[\fJ\varPhi_s]\|_{s;\sC}
 \leq
 \sum_{m\in\bN_+}
 m\,s^{-(m-1)/2}\,
 K_m(V_s)\,
 P_s(\varPhi_\Cdot)^{m-1}.
\end{equation}
Moreover, expanding the difference of the products gives
\begin{multline}\label{eq:FBSDE_DV_diff}
 \|\rD V_s[\fJ\varPhi_s]
 -
 \rD V_s[\fJ\varPsi_s]\|_{s;\sC}
 \\
 \leq
 \sum_{m\geq2}
 m^2\,s^{-(m-1)/2}\,
 K_m(V_s)\,
 \big(
 P_s(\varPhi_\Cdot)+P_s(\varPsi_\Cdot)
 \big)^{m-2}
 P_s(\varPhi_\Cdot-\varPsi_\Cdot).
\end{multline}
Using~\eqref{eq:Z_weights} in the same argument yields
\begin{multline}\label{eq:FBSDE_DV_diff_weighted}
 \|\rD V_s[\fJ\varPhi_s]
 -
 \rD V_s[\fJ\varPsi_s]\|_{s;\tilde\sC}
 \\
 \leq
 \sum_{m\geq2}
 m^2\,s^{-(m-1)/2}\,
 K_m(V_s)\,
 \big(
 P_s(\varPhi_\Cdot)+P_s(\varPsi_\Cdot)
 \big)^{m-2}
 \widetilde P_s(\varPhi_\Cdot-\varPsi_\Cdot),
\end{multline}
as well as
\begin{equation}\label{eq:FBSDE_DV_weighted}
 \|\rD V_s[\fJ\varPhi_s]\|_{s;\tilde\sC}
 \leq
 \sum_{m\in\bN_+}
 m\,s^{-(m-1)/2}\,
 \widetilde K_m(V_s)\,
 P_s(\varPhi_\Cdot)^{m-1}.
\end{equation}
Finally,
\begin{equation}\label{eq:FBSDE_DV_cutoff}
 \|\rD V_s[\fJ\varPhi_s]
 -
 \rD V_{\tau,\varepsilon;s}[\fJ\varPhi_s]\|_{s;\tilde\sC}
 \leq
 \sum_{m\in\bN_+}
 m\,s^{-(m-1)/2}\,
 \widetilde K_m(V_s-V_{\tau,\varepsilon;s})\,
 P_s(\varPhi_\Cdot)^{m-1}.
\end{equation}

Next, recall that by Corollary~\ref{cor:contraction}
\begin{equation}
 \|V_{\tau,\varepsilon;\Cdot}\|_{\sV^{8,4;1-40\kappa}} \leq C\,,
 \qquad
 \|V_{\Cdot}-V_{\tau,\varepsilon;\Cdot}\|_{\tilde\sV^{2,3;1-40\kappa}} \leq C\,\lambda_{\tau\vee\varepsilon}^\kappa.\,
\end{equation}
Moreover, by Def.~\ref{dfn:V_space},
\begin{equation}\label{eq:FBSDE_K}
 K_m(V_{\tau,\varepsilon;s})
 \leq
 C\,8^{-m}m^{-4}\,
 \lambda_s^{\gamma+2\kappa m}\,
 s^{-2+m/2},
\end{equation}
and, since $\tilde w_s^m\leq w_s^m$,
\begin{equation}\label{eq:FBSDE_K_tilde}
 \widetilde K_m(V_{\tau,\varepsilon;s})
 \leq
 C\,8^{-m}m^{-4}\,
 \lambda_s^{\gamma+2\kappa m}\,
 s^{-2+m/2}.
\end{equation}
Similarly,
\begin{equation}\label{eq:FBSDE_K_cutoff}
 \widetilde K_m(V_s-V_{\tau,\varepsilon;s})
 \leq
 C\,2^{-m}m^{-3}\,
 \lambda_{\tau\vee\varepsilon}^\kappa\,
 \lambda_s^{\gamma+2\kappa m}\,
 s^{-2+m/2}.
\end{equation}
Since $\varPhi_\Cdot,\varPsi_\Cdot\in\sU$, by~\eqref{eq:FBSDE_P_bound},
\begin{equation}
 P_s(\varPhi_\Cdot)+P_s(\varPsi_\Cdot)
 \leq 2\lambda^{-\kappa}.
\end{equation}
Using $\lambda_s\leq\lambda$ in
\eqref{eq:FBSDE_DV}--\eqref{eq:FBSDE_K_cutoff}, and summing the
convergent series in $m$, we obtain
\begin{equation}
\begin{aligned}
 \|\rD V_{\tau,\varepsilon;s}[\fJ\varPhi_s]\|_{s;\sC}
 +
 \|\rD V_{\tau,\varepsilon;s}[\fJ\varPhi_s]\|_{s;\tilde\sC}
 &\leq
 C\,\lambda^{2\kappa}\,
 \lambda_s^\gamma s^{-3/2},
 \\
 \|\rD V_{\tau,\varepsilon;s}[\fJ\varPhi_s]
 -
 \rD V_{\tau,\varepsilon;s}[\fJ\varPsi_s]\|_{s;\sC}
 &\leq
 C\,\lambda^{2\kappa}\,
 \lambda_s^\gamma s^{-3/2}\,
 \|\varPhi_\Cdot-\varPsi_\Cdot\|_{\sZ},
 \\
 \|\rD V_{\tau,\varepsilon;s}[\fJ\varPhi_s]
 -
 \rD V_{\tau,\varepsilon;s}[\fJ\varPsi_s]\|_{s;\tilde\sC}
 &\leq
 C\,\lambda^{2\kappa}\,
 \lambda_s^{\gamma-\kappa}s^{-3/2}\,
 \|\varPhi_\Cdot-\varPsi_\Cdot\|_{\tilde\sZ},
 \\
 \|\rD V_s[\fJ\varPhi_s]
 -
 \rD V_{\tau,\varepsilon;s}[\fJ\varPhi_s]\|_{s;\tilde\sC}
 &\leq
 C\,\lambda_{\tau\vee\varepsilon}^\kappa\,
 \lambda^{2\kappa}\,
 \lambda_s^\gamma s^{-3/2}.
\end{aligned}
\end{equation}
Combining these bounds with the convolution estimates above, for every
$b\in\bA$ we obtain
\begin{align}
 \|\partial^b\fJ\dot G_{\varepsilon;s}\ast
 \rD V_{\tau,\varepsilon;s}[\fJ\varPhi_s]\|_{\sC}
 &\leq
 C\,\lambda^{2\kappa}\,
 \lambda_s^\gamma s^{-3/2-|b|},
 \label{eq:FBSDE_integrand_A}
 \\
 \|\partial^b\fJ\dot G_{\varepsilon;s}\ast
 \big(
 \rD V_{\tau,\varepsilon;s}[\fJ\varPhi_s]
 -
 \rD V_{\tau,\varepsilon;s}[\fJ\varPsi_s]
 \big)\|_{\sC}
 &\leq
 C\,\lambda^{2\kappa}\,
 \lambda_s^\gamma s^{-3/2-|b|}
 \|\varPhi_\Cdot-\varPsi_\Cdot\|_{\sZ},
 \label{eq:FBSDE_integrand_B}
 \\
 \|\partial^b\fJ\dot G_{\varepsilon;s}\ast
 \big(
 \rD V_{\tau,\varepsilon;s}[\fJ\varPhi_s]
 -
 \rD V_{\tau,\varepsilon;s}[\fJ\varPsi_s]
 \big)\|_{\tilde\sC}
 &\leq
 C\,\lambda^{2\kappa}\,
 \lambda_s^{\gamma-\kappa}s^{-3/2-|b|}
 \|\varPhi_\Cdot-\varPsi_\Cdot\|_{\tilde\sZ}.
 \label{eq:FBSDE_integrand_C}
\end{align}
For the cutoff difference we write
\begin{multline}
 \fJ\dot G_s\ast\rD V_s[\fJ\varPhi_s]
 -
 \fJ\dot G_{\varepsilon;s}\ast
 \rD V_{\tau,\varepsilon;s}[\fJ\varPhi_s]
 \\
 =
 \fJ\dot G_s\ast
 \big(
 \rD V_s[\fJ\varPhi_s]
 -
 \rD V_{\tau,\varepsilon;s}[\fJ\varPhi_s]
 \big)
 +
 \fJ(\dot G_s-\dot G_{\varepsilon;s})\ast
 \rD V_{\tau,\varepsilon;s}[\fJ\varPhi_s].
\end{multline}
Hence, since $\lambda_\varepsilon^\kappa\leq \lambda_{\tau\vee\varepsilon}^\kappa$,
\begin{equation}\label{eq:FBSDE_integrand_D}
 \left\|
 \partial^b\left(
 \fJ\dot G_s\ast\rD V_s[\fJ\varPhi_s]
 -
 \fJ\dot G_{\varepsilon;s}\ast
 \rD V_{\tau,\varepsilon;s}[\fJ\varPhi_s]
 \right)
 \right\|_{\tilde\sC}
 \\
 \leq
 C\,\lambda_{\tau\vee\varepsilon}^\kappa\,
 \lambda^{2\kappa}\,
 \lambda_s^{\gamma-\kappa}
 s^{-3/2-|b|}.
\end{equation}
Note also that $\lambda_t\leq \lambda$ and
\begin{equation}
 \int_t^1 \lambda_s^{1-40\kappa}\, s^{-3/2}\,\rd s \leq C\,\lambda_t^{1-40\kappa}\,t^{-1/2},
 \qquad
 \int_t^1 \lambda_s^{1-41\kappa}\, s^{-3/2}\,\rd s \leq C\,\lambda_t^{1-41\kappa}\,t^{-1/2},
\end{equation}
for $\lambda\in(0,1]$ small enough by Lemma~\ref{lem:bounds_relevant_irrelevant}~(C).

Since $s^{-|b|}\leq t^{-|b|}$ for $s\in[t,1]$,
\eqref{eq:FBSDE_integrand_A} and the first estimate above imply
\begin{equation}
 \lambda^\kappa t^{1/2+|b|}
 \|\partial^b\fZ^{(2)}_{\tau,\varepsilon;t}[\varPhi_\Cdot]\|_{\sC}
 \leq
 C\,\lambda^{3\kappa}\lambda_t^\gamma
 \leq
 C\,\lambda^\kappa.
\end{equation}
This proves~(A) for $\fZ^{(2)}_{\tau,\varepsilon;\Cdot}$.
The same argument using~\eqref{eq:FBSDE_integrand_B} proves~(B).

Using~\eqref{eq:FBSDE_integrand_C} and the second integral estimate gives
\begin{equation}
 \lambda^\kappa\lambda_t^\kappa t^{1/2+|b|}
 \|\partial^b(
 \fZ^{(2)}_{\tau,\varepsilon;t}[\varPhi_\Cdot]
 -
 \fZ^{(2)}_{\tau,\varepsilon;t}[\varPsi_\Cdot])
 \|_{\tilde\sC}
 \leq
 C\,\lambda^{3\kappa}\lambda_t^\gamma
 \|\varPhi_\Cdot-\varPsi_\Cdot\|_{\tilde\sZ}
 \leq
 C\,\lambda^\kappa
 \|\varPhi_\Cdot-\varPsi_\Cdot\|_{\tilde\sZ},
\end{equation}
which proves~(C).

Similarly,~\eqref{eq:FBSDE_integrand_D} gives
\begin{equation}
 \lambda^\kappa\lambda_t^\kappa t^{1/2+|b|}
 \|\partial^b(
 \fZ^{(2)}_{t}[\varPhi_\Cdot]
 -
 \fZ^{(2)}_{\tau,\varepsilon;t}[\varPhi_\Cdot])
 \|_{\tilde\sC}
 \leq
 C\,\lambda_{\tau\vee\varepsilon}^\kappa
 \lambda^{3\kappa}\lambda_t^\gamma
 \leq
 C\,\lambda_{\tau\vee\varepsilon}^\kappa,
\end{equation}
which proves~(D).

Finally,~\eqref{eq:FBSDE_integrand_A} yields
\begin{equation}
 \lambda^\kappa
 \lambda_t^{40\kappa-1}
 t^{1/2+|b|}
 \|\partial^b\fZ^{(2)}_{\tau,\varepsilon;t}[\varPhi_\Cdot]\|_{\sC}
 \leq
 C\,\lambda^{3\kappa},
\end{equation}
which proves~(E). This finishes the proof.
\end{proof}

\begin{lem}\label{lem:finite_cutoff_identification}
Let $\tau,\varepsilon\in(0,1]$ and let $U_{\tau,\varepsilon;t}$ be the effective potential constructed in Corollary~\ref{cor:polchinski}. Suppose that
\begin{equation}
 [0,1]\ni t\mapsto \varPhi_{\tau,\varepsilon;t}
 \in
  \mathcal C_{\tau,\varepsilon}\otimes\sF^-
\end{equation}
is a solution of~\eqref{eq:fbsde} such that $\varPhi_{\tau,\varepsilon;t}\in \mathcal C_{\tau,\varepsilon}\otimes\sF_t^-$ for every $t\in[0,1]$. Then~\eqref{eq:stochastic_identification} holds for every $F\in\sN(C^\infty(\bT_\tau^2)^\bFF)$ and every $t\in[0,1]$.
\end{lem}
\begin{proof}
The only point in the proof of \cite[Proposition~3.10]{DFG25} which is not directly covered by our setting is the use of norm compatibility in order to obtain the increment estimate \cite[Eq.~(26)]{DFG25}. We verify this estimate directly.

Fix $\tau,\varepsilon\in(0,1]$. By Def.~\ref{dfn:Psi}, Remark~\ref{rem:bound_noise}, and the fact that the regularized kernels have finite Fourier support and vanish below scale $\varepsilon/2$, we have
\begin{equation}
 \sup_{\sigma\in\bFF}
 \|\varPsi_{\tau,\varepsilon;t}^\sigma
 -\varPsi_{\tau,\varepsilon;s}^\sigma\|_{\sC}
 \leq C_{\tau,\varepsilon}|t-s|^{1/2},
 \qquad s,t\in[0,1].
\end{equation}

At fixed cutoffs, \eqref{eq:fbsde} is a finite-dimensional Grassmann stochastic equation on $\mathcal C_{\tau,\varepsilon}$. After changing the orientation of the time variable, it has the forward-time form appearing in~\cite{DFG25}. At fixed cutoffs its drift is a polynomial map, locally Lipschitz on bounded sets, with coefficients continuous in the scale variable. Hence the above estimate and~\eqref{eq:fbsde} imply
\begin{equation}
 \|\varPhi_{\tau,\varepsilon;t}
 -\varPhi_{\tau,\varepsilon;s}\|_{\sC}
 \leq C_{\tau,\varepsilon,\varPhi}|t-s|^{1/2},
\end{equation}
and the drift integral admits the first-order expansion with remainder $o(|t-s|)$.
Thus the estimate used in \cite[Eq.~(26)]{DFG25} and the corresponding bound on the terms of degree at least three in the Grassmann Taylor expansion hold in our setting.

The remaining steps of the proof of \cite[Proposition~3.10]{DFG25} apply without change. Indeed, increments of $\varPsi_{\tau,\varepsilon}$ on disjoint scale intervals are independent with respect to the conditional expectations of Def.~\ref{dfn:conditional_exp}, their linear contribution has zero conditional expectation, and Remark~\ref{rem:Wick} identifies their quadratic contribution with the increment of $G_{\tau,\varepsilon}$. Hence the Grassmann It\^o formula used in \cite{DFG25} holds for $\varPhi_{\tau,\varepsilon}$.

Finally, using the Polchinski equation~\eqref{eq:polchinski_corollary}, the moments of $\varPhi_{\tau,\varepsilon;t}$ and the normalized Gaussian moments on the RHS of~\eqref{eq:stochastic_identification} satisfy the same finite system of linear ODEs and have the same initial values, exactly as in the proof of \cite[Proposition~3.10]{DFG25}. The denominator is non-zero by Def.~\ref{dfn:effective_potential} and Remark~\ref{rem:Z_not_zero}. Uniqueness of this system proves~\eqref{eq:stochastic_identification}.
\end{proof}

\begin{cor}\label{cor:fbsde}
There exists $\lambda_\star\in(0,1]$ and $C\in(0,\infty)$ such that for all $\lambda\in(0,\lambda_\star]$ and all $\tau,\varepsilon\in[0,1]$ the map $\fZ_{\tau,\varepsilon;\Cdot}\,:\,\sU\to \sU$ is well-defined and has the unique fixed point denoted by $\varPhi_{\tau,\varepsilon;\Cdot}$ such that
\begin{equation}\label{eq:cor_bound_fbsde}
 \|\varPhi_{\Cdot}-\varPhi_{\tau,\varepsilon;\Cdot}\|_{\tilde\sZ}
 \leq
 C\,\lambda_{\tau\vee\varepsilon}^\kappa\,,
\end{equation}
where $\varPhi_{\Cdot}:=\varPhi_{\tau,\varepsilon;\Cdot}$ with $\tau=0$, $\varepsilon=0$. For every $\tau,\varepsilon\in[0,1]$ there exists $\varPhi_{\tau,\varepsilon}\in\sS'(\bR^2)^\bFF$ such that
\begin{equation}\label{eq:cor_limit_fbsde}
 \lim_{s\searrow0}\varPhi_{\tau,\varepsilon;s}
 =
 \varPhi_{\tau,\varepsilon}
\end{equation}
in $\sS'(\bR^2)^\bFF$. We write $\varPhi:=\varPhi_{0,0}$.
For all $\tau,\varepsilon\in(0,1]$ the field $\varPhi_{\tau,\varepsilon}\in C^\infty(\bT^2_\tau)^\bFF$ is the interacting Gross--Neveu field with the cutoffs $\tau,\varepsilon$, i.e. for all $F\in\sN(C^\infty(\bT_\tau^2)^\bFF)$ it holds
\begin{equation}\label{eq:cor_measure_fbsde}
 \mu_{\tau,\varepsilon}(F)
 \equiv
 \frac{\int F(\vartheta_\varepsilon\ast\psi_{\tau,\varepsilon})\,\exp\!\big(-A_{\tau}(\psi_{\tau,\varepsilon})+U_{\tau,\varepsilon}(\vartheta_\varepsilon\ast\psi_{\tau,\varepsilon})\big)\,\rd\psi_{\tau,\varepsilon}}{\int \exp\!\big(-A_{\tau}(\psi_{\tau,\varepsilon})+U_{\tau,\varepsilon}(\vartheta_\varepsilon\ast\psi_{\tau,\varepsilon})\big)\,\rd\psi_{\tau,\varepsilon}}
 =
 \fE F(\varPhi_{\tau,\varepsilon}).
\end{equation}
Moreover, $\varPhi^\sigma\in \sC^{-1/2}$, $\lim_{\tau,\varepsilon\searrow0}\|\varPhi^\sigma-\varPhi^\sigma_{\tau,\varepsilon}\|_{\tilde\sC^\alpha}=0$ and
 \begin{equation}\label{eq:cor_improved_fbsde}
  \sup_{i\in\{-1,0,1,\ldots\}}(i+2)^{1-40\kappa}\,2^{-i/2}
  \|\Delta_i(\varPhi^\sigma-\varPsi^\sigma)\|_{\sC}<\infty
 \end{equation}
for all $\sigma\in\bFF$ and $\alpha\in(-\infty,-1/2)$.
\end{cor}
\begin{rem}
Recall that $(\Delta_i)_{i\in\{-1,0,1,\ldots\}}$ are the Littlewood--Paley blocks and the Hölder--Besov spaces $\sC^\alpha,\tilde\sC^\alpha$ were introduced in Def.~\ref{dfn:Besov}.
\end{rem}
\begin{proof}
The existence and uniqueness of the fixed point $\varPhi_{\tau,\varepsilon;\Cdot}$ of $\fZ_{\tau,\varepsilon;\Cdot}\,:\,\sU\to \sU$ for $\lambda\in(0,1]$ small enough satisfying the bound~\eqref{eq:cor_bound_fbsde} follows from Theorem~\ref{thm:contraction_fbsde} and the Banach fixed-point theorem.

Let $\tau,\varepsilon\in(0,1]$. For $t\in[0,1]$ define the functional $U_{\tau,\varepsilon;t}\in\sN(C^\infty(\bT_\tau^2)^\bFF)$ as in Corollary~\ref{cor:polchinski}. Note that $U_{\tau,\varepsilon;t}$ satisfies the Polchinski equation~\eqref{eq:polchinski_corollary}. Moreover, by Lemma~\ref{lem:DU_DV}, the fixed point equation for $\fZ_{\tau,\varepsilon;\Cdot}$ is precisely~\eqref{eq:fbsde}.
We next verify the adaptedness required in Lemma~\ref{lem:finite_cutoff_identification}. Start the Picard iteration for $\fZ_{\tau,\varepsilon;\Cdot}$ from the zero trajectory. If $\varPhi^{(n)}_{\tau,\varepsilon;s}\in\mathcal C_{\tau,\varepsilon}\otimes\sF_s^-$ for every $s\in[0,1]$, then the evenness of $U_{\tau,\varepsilon;s}$ implies that $\rD U_{\tau,\varepsilon;s}[\varPhi^{(n)}_{\tau,\varepsilon;s}]$ is odd. Since $\dot G_{\tau,\varepsilon;s}$ has Fourier support in $\Lambda_{\tau,\varepsilon}$ and $\sF_s\subset\sF_t$ for $s\geq t$, \eqref{eq:fbsde} shows that
\begin{equation}
\varPhi^{(n+1)}_{\tau,\varepsilon;t} \in \mathcal C_{\tau,\varepsilon}\otimes\sF_t^-
\end{equation}
for every $t\in[0,1]$. Hence all Picard iterates are adapted and take values in the finite Fourier space.  Their limit in $\sZ$ has the same property, so Lemma~\ref{lem:finite_cutoff_identification} applies to the fixed point $\varPhi_{\tau,\varepsilon;\Cdot}$.
We have $\varPhi_{\tau,\varepsilon;t}=\varPhi_{\tau,\varepsilon;\varepsilon/2}$ for $t\in(0,\varepsilon/2]$ since $\varPsi_{\tau,\varepsilon;t} = \varPsi_{\tau,\varepsilon;\varepsilon/2}$ for $t\in(0,\varepsilon/2]$ by Def.~\ref{dfn:Psi} and Remark~\ref{rem:G}~(B) and $\dot G_{\varepsilon;t}=0$ for $t\in(0,\varepsilon]$ by Remark~\ref{rem:G}~(A). The same cutoff property implies $U_{\tau,\varepsilon;\varepsilon/2}=U_{\tau,\varepsilon}$ and $\varPsi_{\tau,\varepsilon;\varepsilon/2}=\varPsi_{\tau,\varepsilon}$.  Therefore~\eqref{eq:stochastic_identification} at $t=\varepsilon/2$, together with Lemma~\ref{lem:interacting_measure}, shows that $\varPhi_{\tau,\varepsilon} = \varPhi_{\tau,\varepsilon;\varepsilon/2}$ is distributed according to the Gross--Neveu measure $\mu_{\tau,\varepsilon}$ with cutoffs $\tau,\varepsilon$, which proves~\eqref{eq:cor_measure_fbsde}.

Set
\begin{equation}
 F_{\tau,\varepsilon;t}
 :=
 \varPhi_{\tau,\varepsilon;t}-\varPsi_{\tau,\varepsilon;t}.
\end{equation}
By the fixed point equation,
\begin{equation}
 F_{\tau,\varepsilon;t}
 =
 \int_t^1
 \fJ\dot G_{\varepsilon;s}\ast
 \rD V_{\tau,\varepsilon;s}
 [\fJ\varPhi_{\tau,\varepsilon;s}]
 \,\rd s.
\end{equation}
Since the Fourier transform of $\dot G_{\varepsilon;s}$ is supported at frequencies of order $s^{-1}$, there exists a universal constant $c\in(0,\infty)$ such that
$
 \Delta_i F_{\tau,\varepsilon;s}
 =
 \Delta_i F_{\tau,\varepsilon;t}
$
for all $\tau,\varepsilon\in[0,1]$, $i\in\{-1,0,1,\ldots\}$ and $s,t\in(0,c\,2^{-i})$. Consequently, we define the Littlewood--Paley blocks of
$F_{\tau,\varepsilon}$ by
\begin{equation}\label{eq:LP_block_stabilization}
 \Delta_i F_{\tau,\varepsilon}
 :=
 \Delta_i F_{\tau,\varepsilon;t},
 \qquad
 t\in(0,c\,2^{-i}).
\end{equation}
By Theorem~\ref{thm:contraction_fbsde}~(E) and Def.~\ref{dfn:sZ_norm},
\begin{equation}
 \|F_{\tau,\varepsilon;t}^\sigma\|_\sC
 \leq
 C\, \lambda_t^{1-40\kappa}\, t^{-1/2}.
\end{equation}
Choosing $t$ proportional to $2^{-i}$ in the definition of $\Delta_iF_{\tau,\varepsilon}$ yields
\begin{equation}
 \|\Delta_iF_{\tau,\varepsilon}^\sigma\|_\sC
 \leq
 C\, (i+2)^{-1+40\kappa}\, 2^{i/2}.
\end{equation}
The preceding estimate implies that
$F_{\tau,\varepsilon}\in\sC^{-1/2}$ and
\begin{equation}
 \lim_{t\searrow0}F_{\tau,\varepsilon;t}
 =
 F_{\tau,\varepsilon}
\end{equation}
in $\sS'(\bR^2)^\bFF$. We then set
\begin{equation}
 \varPhi_{\tau,\varepsilon}
 :=\varPsi_{\tau,\varepsilon}+F_{\tau,\varepsilon}.
\end{equation}
Moreover,
$\varPsi_{\tau,\varepsilon;t}\to\varPsi_{\tau,\varepsilon}$ in
$\sS'(\bR^2)^\bFF$ as $t\searrow0$. Hence
\begin{equation}
 \lim_{t\searrow0}\varPhi_{\tau,\varepsilon;t}
 =
 \varPsi_{\tau,\varepsilon}+F_{\tau,\varepsilon}
 =
 \varPhi_{\tau,\varepsilon}
\end{equation}
in $\sS'(\bR^2)^\bFF$. By Lemma~\ref{lem:Psi_bounds}~(C), $\varPsi^\sigma\in\sC^{-1/2}$. The preceding bound on the
Littlewood--Paley blocks of $F$ gives
$F^\sigma\in\sC^{-1/2}$. Therefore
$
 \varPhi^\sigma=\varPsi^\sigma+F^\sigma\in\sC^{-1/2}.
$
Moreover, since $\varPhi^\sigma-\varPsi^\sigma=F^\sigma$, the same bound proves
\eqref{eq:cor_improved_fbsde}. Finally, by~\eqref{eq:LP_block_stabilization},
\eqref{eq:cor_bound_fbsde}, Lemma~\ref{lem:Psi_bounds}~(B),~(C), and the uniform boundedness of the Littlewood--Paley blocks on the weighted $L^\infty$ space, for every $i\in\{-1,0,1,\ldots\}$ it holds
\begin{equation}
 \lim_{\tau,\varepsilon\searrow0}
 \|\tilde w\Delta_i
 (\varPhi^\sigma-\varPhi^\sigma_{\tau,\varepsilon})\|_\sC=0.
\end{equation}
Moreover, the bounds above and Lemma~\ref{lem:Psi_bounds} imply
\begin{equation}
 2^{\alpha i}
 \|\tilde w\Delta_i
 (\varPhi^\sigma-\varPhi^\sigma_{\tau,\varepsilon})\|_\sC
 \leq C\,2^{(\alpha+1/2)i}
\end{equation}
uniformly in $\tau,\varepsilon\in[0,1]$. This proves $\lim_{\tau,\varepsilon\searrow0}
 \|\varPhi^\sigma-\varPhi^\sigma_{\tau,\varepsilon}\|_{\tilde\sC^\alpha}
 =0$.
\end{proof}

\section*{Acknowledgments}

The financial support by the grant `Sonata Bis' 2019/34/E/ST1/00053 of the National Science Centre, Poland, is gratefully acknowledged. I would like to thank Roland Bauerschmidt, David Brydges, Ajay Chandra, Rafael Greenblatt, Antti Kupiainen, Manfred Salmhofer, Christian Webb and Oleg Zaboronski for helpful discussions and suggestions. I also thank the anonymous referees for their careful reading of the manuscript, for pointing out errors, and for suggesting improvements.

\end{document}